\documentclass[12pt]{amsart}
\usepackage{fullpage,graphicx,color,mathrsfs}
\newtheorem{lemma}{Lemma}
\newtheorem{theorem}{Theorem}
\newtheorem{conjecture}{Conjecture}
\newtheorem{corollary}{Corollary}
\newenvironment{myremark}{$\triangleleft$ Remark:}{$\triangleright$}
\theoremstyle{remark}

\newcommand{\dd}{\mathrm{d}}
\newcommand{\ii}{\mathrm{i}}
\newcommand{\ee}{\mathrm{e}}
\newcommand{\del}{\partial}
\newcommand{\delbar}{\overline{\partial}}
\newcommand{\re}{\mathrm{Re}}
\newcommand{\imag}{\mathrm{Im}}

\title{A Study of the Direct Spectral Transform for the Defocusing Davey-Stewartson II Equation in the Semiclassical Limit }
\author{O. Assainova}
\address{Institut de Math\'ematiques, Universit\'e de Bourgogne, Dijon, France}
\author{C. Klein} 
\address{Institut de Math\'ematiques, Universit\'e de Bourgogne, Dijon, France}
\author{K. D. T.-R. McLaughlin}
\address{Department of Mathematics, Colorado State University, Fort Collins, CO}
\author{P. D. Miller}
\address{Department of Mathematics, University of Michigan, Ann Arbor, MI}
\date{\today}
\begin{document}
\begin{abstract}  
The defocusing Davey-Stewartson II equation has been shown in numerical experiments to exhibit behavior in the semiclassical limit that qualitatively resembles that of its one-dimensional reduction, the defocusing nonlinear Schr\"odinger equation, namely the generation from smooth initial data of regular rapid oscillations occupying domains of space-time that become well-defined in the limit.  As a first step to study this problem analytically using the inverse-scattering transform, we consider the direct spectral transform for the defocusing Davey-Stewartson II equation for smooth initial data in the semiclassical limit.  The direct spectral transform involves a singularly-perturbed elliptic Dirac system in two dimensions.  We introduce a WKB-type method for this problem, prove that it makes sense formally for sufficiently large values of the spectral parameter $k$ by controlling the solution of an associated nonlinear eikonal problem, and we give numerical evidence that the method is accurate for such $k$ in the semiclassical limit.  Producing this evidence requires both the numerical solution of the singularly-perturbed  Dirac system and the numerical solution of the eikonal problem.  The former is carried out using a method previously developed by two of the authors and we give in this paper a new method for the numerical solution of the eikonal problem valid for sufficiently large $k$.  For a particular potential we are able to solve the eikonal problem in closed form for all $k$, a calculation that yields some insight into the failure of the WKB method for smaller values of $k$.  Informed by numerical calculations of the direct spectral transform we then begin a study of the singularly-perturbed  Dirac system for values of $k$ so small that there is no global solution of the eikonal problem.  We provide a rigorous semiclassical analysis of the solution for real radial potentials at $k=0$, which yields an asymptotic formula for the reflection coefficient at $k=0$ and suggests an annular structure for the solution that may be exploited when $k\neq 0$ is small.  The numerics also suggest that for some potentials the reflection coefficient converges pointwise as $\epsilon\downarrow 0$ to a limiting function that is supported in the domain of $k$-values on which the eikonal problem does not have a global solution.  It is expected that singularities of the eikonal function play a role similar to that of turning points in the one-dimensional theory.
\end{abstract}
\maketitle

\section{Introduction}
By the semiclassical limit for the defocusing Davey-Stewartson II (DS-II) equation\footnote{A more physically relevant way to write \eqref{eq:DSII} is to introduce the real-valued mean flow $\Phi:=\re(g)-|q|^2$.  Then applying $\del$ to the second equation in \eqref{eq:DSII} and taking the real part one arrives at the form
\[
\begin{split}
\ii\epsilon q_t + \epsilon^2(q_{xx}-q_{yy}) + 2(\Phi + |q|^2)q&=0\\
\Phi_{xx}+\Phi_{yy} +2(|q|^2)_{xx}&=0.
\end{split}
\]%
} we mean the following
Cauchy initial-value problem parametrized by $\epsilon>0$:
\begin{equation}
\begin{split}
\ii\epsilon q_t + 2\epsilon^2\left(\delbar^2+\del^2\right)q + \left(g+\overline{g}\right)q&=0\\
\delbar g +\del\left(|q|^2\right)&=0,
\end{split}
\label{eq:DSII}
\end{equation}
for a complex-valued field $q=q^\epsilon (x,y,t)$ where
\begin{equation}
\del:=\frac{1}{2}\left(\frac{\partial}{\partial x}-\ii\frac{\partial}{\partial y}\right)\quad\text{and}\quad
\delbar:=\frac{1}{2}\left(\frac{\partial}{\partial x}+\ii\frac{\partial}{\partial y}\right),
\label{eq:d-dbar}
\end{equation}
subject to an initial condition of  ``oscillatory wavepacket'' or ``WKB'' form:
\begin{equation}
q^\epsilon(x,y,0)=A(x,y)\ee^{\ii S(x,y)/\epsilon},\quad A>0,\quad S\in\mathbb{R}.
\label{eq:initial-data}
\end{equation}
Here $A$ and $S$ are functions independent of $\epsilon$, and the precise meaning of the equation $\delbar g + \del (|q|^2)=0$ in \eqref{eq:DSII} is that for each $t\ge 0$, $g$ is eliminated from the equation governing $q$ by the solid Cauchy transform:  $g=-\delbar^{-1}(\del (|q|^2))$, where for a suitable function $F:\mathbb{R}^2\to\mathbb{C}$,
\begin{equation}
\overline{\partial}^{-1}F(x,y):=-\frac{1}{\pi}\iint_{\mathbb{R}^2}\frac{F(x',y')\,\dd \mathbb{A}(x',y')}{(x'-x)+\ii (y'-y)},
\label{eq:solid-Cauchy}
\end{equation}
where $\dd\mathbb{A}(x,y)$ denotes the area differential in the plane.
The dependence on $\epsilon>0$ in this problem enters both through the phase factor in the initial data and the coefficients of the DS-II equation.  

The initial-value problem \eqref{eq:DSII}--\eqref{eq:initial-data} is globally well-posed in $H^{1,1}(\mathbb{R}^2)$ \cite{Perry12}, a result that has recently been extended to $L^2(\mathbb{R}^2)$ \cite{NachmanRT17}.  Significantly, these results are completely insensitive to the value of $\epsilon>0$.  Thus, for $A\in L^2(\mathbb{R}^2)$ there exists a unique global solution of the initial-value problem for every value of $\epsilon>0$, 
and the question we wish to address is \emph{how does this well-defined solution $q=q^\epsilon(x,y,t)$ behave asymptotically as $\epsilon\downarrow 0$?}  This problem is
interesting because it sets up a competition between two space-time scales:
\begin{itemize}
\item On the one hand, the system \eqref{eq:DSII} admits solutions with rapid space-time variations:  indeed setting $q^\epsilon(x,y,t)=Q(x/\epsilon,y/\epsilon,t/\epsilon)$, one sees that $Q$ satisfies
\eqref{eq:DSII} with $\epsilon=1$.  Thus one may expect that ``typical'' solutions of \eqref{eq:DSII} exhibit wavelengths and periods proportional to $\epsilon$.
\item On the other hand, the functions $A$ and $S$ in the initial data vary on spatial scales that are fixed as $\epsilon$ tends to zero.  
\end{itemize}
Therefore, the ``natural'' space-time scales of the system 
\eqref{eq:DSII} are $\Delta x\sim\Delta y\sim\Delta t=O(\epsilon)$ while those of the functions $A$ and $S$ in the initial data are $\Delta x\sim\Delta y=O(1)$.  Consequently, one expects that the solution $q^\epsilon(x,y,t)$ will exhibit a multiscale structure when $\epsilon$ is small and the solution is examined on $O(1)$ space-time scales.  In dispersive partial differential equations formulated in one space dimension, this competition of scales results in a back-and-forth process in which the solution first evolves for a time according to a simplified dispersionless model (e.g., Madelung quantum hydrodynamic system, see Section~\ref{sec:Madelung} below) until one or more singularities form in the approximating solution.   The singularities are resolved by dispersive terms in the full equation which produce small-scale oscillations.  Then the process repeats, as the oscillations develop smooth modulations that may be captured by a more complicated dispersionless approximating system (eg., Whitham's modulation equations \cite{DubrovinN89,Whitham65}), solutions of which may also become singular at a later time, and so on.  The wild oscillations are confined to space-time domains that are increasingly well-defined in the limit $\epsilon\downarrow 0$.  The modulated oscillatory structures that appear in such a way are called \emph{dispersive shock waves}.  Numerical simulations show that similar phemonena also occur in $2+1$ systems such as \eqref{eq:DSII}; see Figure~\ref{fig:KR-numerics}.
\begin{figure}[h]
\begin{center}
\includegraphics[width=0.3\linewidth]{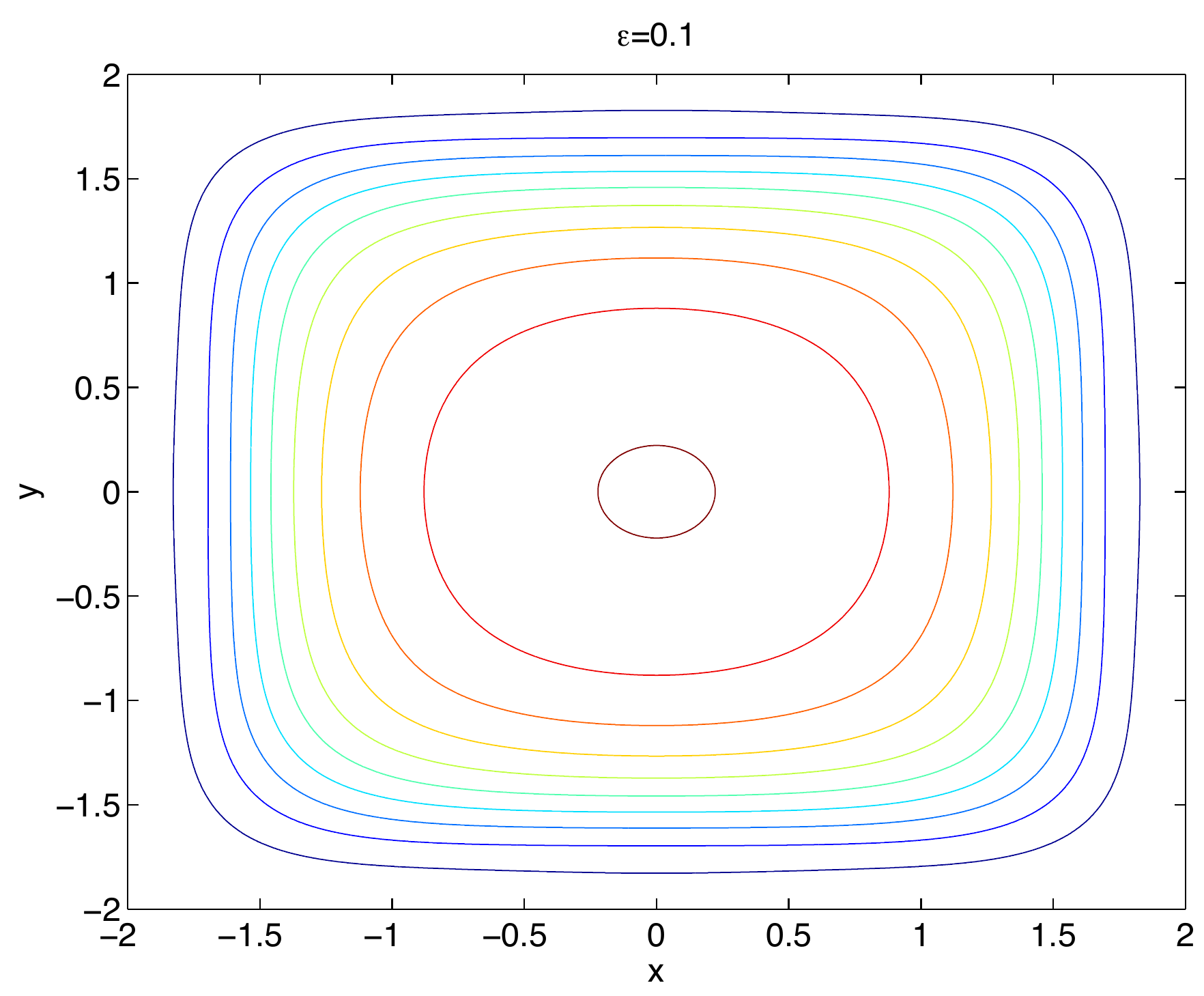}%
\hspace{0.01\linewidth}%
\includegraphics[width=0.3\linewidth]{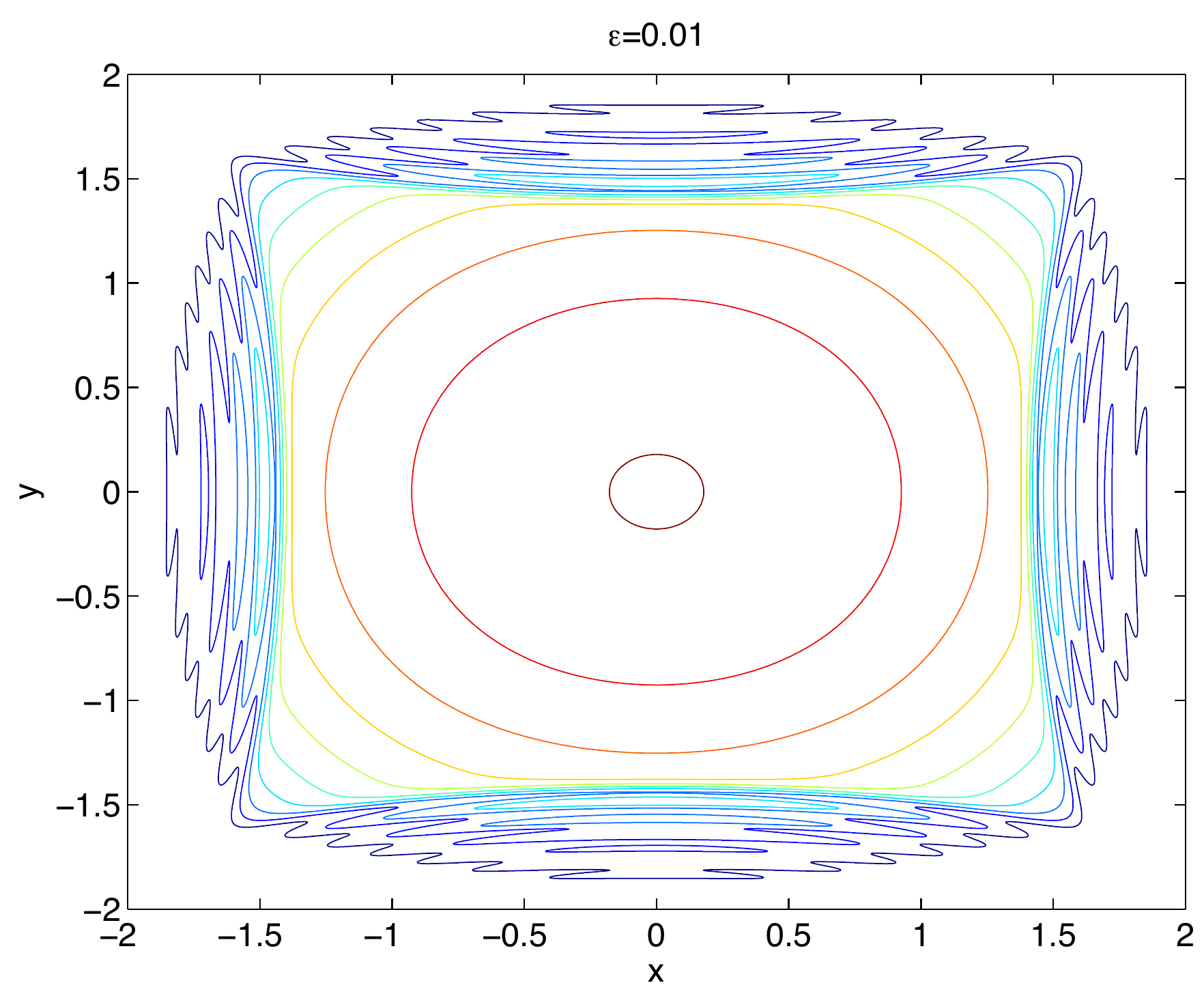}%
\hspace{0.01\linewidth}%
\includegraphics[width=0.3\linewidth]{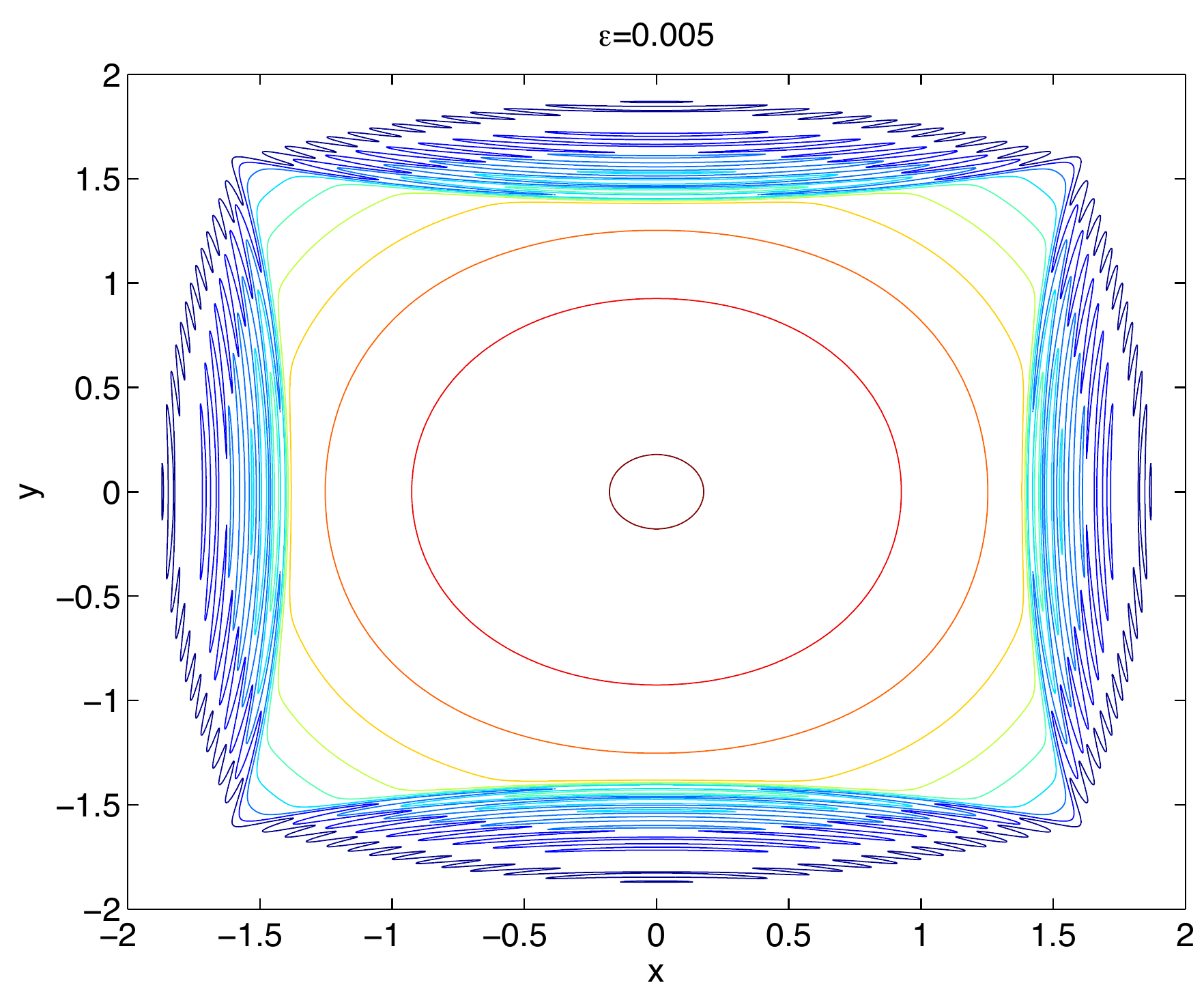}
\end{center}
\caption{Contour plots of $|q^\epsilon(x,y,1)|^2$ on the fixed region $(x,y)\in [-2,2]\times [-2,2]$ for Gaussian initial data of the form \eqref{eq:initial-data} with $S(x,y)\equiv 0$ and $A(x,y):=\ee^{-(x^2+y^2)}$.  Left:  $\epsilon=0.1$.  Center:  $\epsilon=0.01$.  Right:  $\epsilon=0.005$.  Figure taken from \cite{KleinR14}.  Observe that as $\epsilon$ decreases, the four lens-shaped regions that appear to confine the $O(\epsilon)$-wavelength oscillations become better defined in the $(x,y)$-plane for fixed $t=1$.}
\label{fig:KR-numerics}
\end{figure}

Another way to think about the semiclassical limit is to consider the system \eqref{eq:DSII} with $\epsilon=1$ but for a family of initial data parametrized by $\epsilon$ that becomes large in a certain sense as $\epsilon\to 0$.  For example, by the scaling of $q^\epsilon(x,y,t)=Q(x/\epsilon,y/\epsilon,t/\epsilon)$ we  
arrive at the system \eqref{eq:DSII} with $\epsilon=1$ for $Q$, but also with the side effect of stretching the initial data for $Q$ in the $(x,y)$-directions by a factor of $\epsilon^{-1}$; this makes the norm of the initial data for $Q$ very large as $\epsilon\downarrow 0$.  Alternatively, 
one may take the $\epsilon=1$ form of \eqref{eq:DSII} for $Q$ with pointwise large initial data $Q(x,y,0)$ proportional to $\epsilon^{-1}$.
In this interpretation, the semiclassical limit at first seems similar to the \emph{strong coupling limit} discussed by Ablowitz and Clarkson \cite[Section 5.5.4]{AblowitzC91}.  However, the simplified dynamics of the strong coupling limit essentially transpires on the time scale of length $t=O(\epsilon)$ in the variables of \eqref{eq:DSII}, rendering the interesting dynamics observed numerically on time scales $t=O(1)$  out of reach. 

\subsection{Quantum hydrodynamics}
\label{sec:Madelung}
Eliminating the imaginary part of $g$ from \eqref{eq:DSII} yields the form
\begin{equation}
\begin{split}
\ii\epsilon q_t +\epsilon^2\left(q_{xx}-q_{yy}\right) + 2Mq&=0\\
M_{xx}+M_{yy}&=\left(|q|^2\right)_{yy}-\left(|q|^2\right)_{xx},
\end{split}
\label{eq:DSII-again}
\end{equation}
where $M=M^\epsilon(x,y,t):=\re(g)$.  This form reveals a certain formal connection with dispersive nonlinear equations in $1+1$ dimensions.  Indeed, if one takes $q$ and $M$ independent of $y$ and accepts the solution $M=-|q|^2$ of the Poisson equation governing $M$ one finds
that $q$ solves the \emph{defocusing nonlinear Schr\"odinger equation}
\begin{equation}
\ii\epsilon q_t +\epsilon^2q_{xx} - 2|q|^2q=0.
\label{eq:defocNLS}
\end{equation}
Similarly, taking $q$ and $M$ independent of $x$ and related by $M=|q|^2$ one arrives instead at 
\begin{equation}
-\ii\epsilon q_t +\epsilon^2q_{yy}-2|q|^2q=0,
\end{equation}
which is also a (conjugated, or time-reversed) defocusing nonlinear Schr\"odinger equation.  Note however, that these reductions are not obviously consistent with the elimination of $g$ using $\delbar^{-1}$ given by \eqref{eq:solid-Cauchy}, which assumes some sort of decay of $\del(|q|^2)$ \emph{in all directions of the $(x,y)$-plane}.

The simplest interpretation of the semiclassical limit is that afforded by the quantum hydrodynamic system that one can derive from \eqref{eq:DSII} by following the ideas of Madelung \cite{Madelung26}.  
%
%
Let us assume only that $|q|>0$ for all $(x,y,t)$, and represent $q$ in the form (resembling the initial data):
\begin{equation}
q^\epsilon(x,y,t)=A^\epsilon(x,y,t)\ee^{\ii S^\epsilon(x,y,t)/\epsilon}.
\end{equation}
Inserting this form into \eqref{eq:DSII-again}, dividing out the common phase factor from the first equation and separating it into real and imaginary parts gives, without approximation, the following system governing the three real-valued fields $A=A^\epsilon(x,y,t)$, $S=S^\epsilon(x,y,t)$, and $M=M^\epsilon(x,y,t)$:
\begin{equation}
\begin{split}
S_t +S_x^2 -S_y^2 -2M &= \epsilon^2\frac{A_{xx}-A_{yy}}{A}\\
A_t + 2A_xS_x +AS_{xx} - 2A_yS_y-AS_{yy}&=0\\
M_{xx}+M_{yy}+\left(A^2\right)_{xx}-\left(A^2\right)_{yy}&=0.
\end{split}
\label{eq:exactsystem}
\end{equation}
This system is to be solved with the $\epsilon$-independent initial data $A^\epsilon(x,y,0)=A(x,y)$
and $S^\epsilon(x,y,0)=S(x,y)$.  This situation obviously invites the neglect of the formally small terms proportional to $\epsilon^2$ on the right-hand side of \eqref{eq:exactsystem}.  Setting $\epsilon=0$ in \eqref{eq:exactsystem} yields the
 \emph{dispersionless DS-II system}, which may be expected to govern the semiclassical evolution of $q^\epsilon(x,y,t)$ in the initial phase of the dynamics (until singularities form in its solution).  To write the dispersionless DS-II system in quantum hydrodynamic form, we introduce
Madelung's quantum fluid density $\rho$ and quantum fluid velocity $\mathbf{u}$ by 
\begin{equation}
\rho=\rho^\epsilon(x,y,t):=A^\epsilon(x,y,t)^2\quad\text{and}\quad
\mathbf{u}=\mathbf{u}^\epsilon(x,y,t):=\nabla S^\epsilon(x,y,t),
\end{equation}
where $\nabla$ is the gradient in the spatial variables $(x,y)$.
The dispersionless DS-II system then becomes the \emph{quantum hydrodynamic system} associated to \eqref{eq:DSII}:
\begin{equation}
\begin{split}
\mathbf{u}_t +\nabla \left(\mathbf{u}\cdot\sigma_3\mathbf{u}\right) -2\nabla M&=\mathbf{0}\\
\rho_t + 2\mathrm{div}\left(\rho\sigma_3\mathbf{u}\right)&=0\\
\Delta M+\mathrm{div}\left(\sigma_3\nabla\rho\right)&=0.
\end{split}
\label{eq:dispersionless-DSII}
\end{equation}
Here $\sigma_3=\mathrm{diag}(1,-1)$ is a Pauli matrix.  We see that $M$ has the interpretation of a kind of fluid pressure, and if we were to replace $\sigma_3$ by the identity matrix, these would essentially be the Euler equations of motion for a physical compressible fluid.  


If the dispersionless problem (i.e., \eqref{eq:exactsystem} with $\epsilon=0$, or equivalently \eqref{eq:dispersionless-DSII}, with $\epsilon$-independent initial data) is locally well-posed, then it is reasonable to expect that the solution of
\eqref{eq:DSII} can be approximated for small $\epsilon$ by the dispersionless limit over some finite time interval independent of $\epsilon$.  In the case of the defocusing nonlinear Schr\"odinger equation in $1+1$ dimensions, the corresponding $1+1$ reduction of \eqref{eq:dispersionless-DSII} is a hyperbolic quasilinear system, and local well-posedness is guaranteed; the accuracy of the dispersionless approximation as $\epsilon\downarrow 0$ has been proven in this case by several different methods including energy estimates applied to a Madelung-type ansatz \cite{Grenier98}, Lax-Levermore variational theory \cite{JinLM99}, and matrix steepest-descent type Riemann-Hilbert techniques  \cite{Miller16} following similar steps as were earlier developed for the small-dispersion limit of the Korteweg-de Vries equation \cite{DeiftVZ97}.  In all these cases, shock formation and the accompanying dispersive regularization precludes global well-posedness.

Integrable dispersionless systems in higher dimensions such as \eqref{eq:dispersionless-DSII} admit certain specialized techniques \cite{FerapontovK04,ManakovS06}, and some aspects of these techniques have been developed in the specific setting of the dispersionless Davey-Stewartson system \cite{Konopelchenko07,Yi14}.  However, even if one has local well-posedness for \eqref{eq:dispersionless-DSII}, one does not expect 
to have global well-posedness.  One expects instead that the solution of the dispersionless system develops singularities (shocks, gradient catastrophes, or caustics) in finite time.  As the singularity is approached, the terms
proportional to $\epsilon^2$ on the right-hand side of \eqref{eq:exactsystem} 
can no longer be discarded and must instead be included at the same order, resulting in the generation of short-wavelength oscillations near the shock point.  
Once these structures form near the shock point, a different kind of ansatz is required locally for $q$, and a more complicated system obtained by Whitham averaging would be expected to take the place of \eqref{eq:dispersionless-DSII}.  

\subsection{Inverse scattering transform}
In order to study the semiclassical limit for the Cauchy problem \eqref{eq:DSII}--\eqref{eq:initial-data}, we wish to exploit the complete integrability of this problem to express its solution via the corresponding \emph{inverse scattering transform}, which was introduced by Ablowitz and Fokas \cite{AblowitzF82, AblowitzF83, AblowitzF84, Fokas83} and refined by many others, including Beals and Coifman \cite{BealsC85,BealsC89}, Sung \cite{Sung94}, Perry \cite{Perry12}, and Nachman, Regev, and Tataru \cite{NachmanRT17}.

To introduce the necessary formulae, we follow the notation
of \cite{Perry12} and introduce the appropriate $\epsilon$ scalings.
\subsubsection{Direct transform}
Consider the Dirac system of linear equations
\begin{equation}
\begin{split}
\epsilon\delbar\psi_1&=\tfrac{1}{2}q\psi_2\\
\epsilon\del\psi_2 &=\tfrac{1}{2}\overline{q}\psi_1
\end{split}
\label{eq:1.10}
\end{equation}
to which we seek for each fixed time $t$ the unique (complex geometrical optics) solution $\psi_j=\psi_j^\epsilon(z;k,t)$ parametrized by the additional complex parameter $k\in\mathbb{C}$ that satisfies the asymptotic conditions:
\begin{equation}
\begin{split}
\lim_{|z|\to\infty}\psi_1^\epsilon(z;k,t)\ee^{-kz/\epsilon}&=1\\
\lim_{|z|\to\infty}\psi_2^\epsilon(z;k,t)\ee^{-\overline{k}\overline{z}/\epsilon}&=0,
\end{split}
\label{eq:psi-norm}
\end{equation}
where $z=x+\ii y$.  
The \emph{reflection coefficient} $R=R^\epsilon(k;t)$ is defined in terms of $\psi_2^\epsilon(z;k,t)$ as follows:
\begin{equation}
\ee^{-kz/\epsilon}\overline{\psi_2^\epsilon(z;k,t)}=\tfrac{1}{2}R^\epsilon(k;t) z^{-1}+ O(|z|^{-2}),\quad |z|\to\infty.
\label{eq:r-def}
\end{equation}
It can be easily shown that if $q$ has radial symmetry, i.e., depends only on $|z|$ and not $\arg(z)$, then also $R^\epsilon$ has radial symmetry, i.e., depends only on $|k|$ and not $\arg(k)$.

\begin{myremark}
The notation $\psi_j^\epsilon(z;k,t)$ is shorthand for $\psi_j^\epsilon((x,y);(\re(k),\imag(k)),t)$ and is not meant to suggest analytic dependence on either $z$ or $k$.  Similar notational conventions hold for other functions throughout this paper.
\end{myremark}

\subsubsection{Time dependence}
As $q$ evolves in time $t$ according to \eqref{eq:DSII}, the reflection coefficient evolves by a trivial phase factor:
\begin{equation}
R^\epsilon(k;t)=R^\epsilon_0(k)\ee^{4\ii t\re(k^2)/\epsilon},\quad R^\epsilon_0(k):=R^\epsilon(k;0).
\end{equation}
For convenience we define
\begin{equation}
R^\epsilon(k;z,t):=R^\epsilon(k;t)\ee^{2\ii \imag(kz)/\epsilon}.
\end{equation}
\subsubsection{Inverse transform}
%
The related quantities defined by 
\begin{equation}
\nu_1=\nu^\epsilon_1(k;z,t):=
\ee^{-kz/\epsilon}\psi_1\quad\text{and}\quad
\nu_2=\nu^\epsilon_2(k;z,t):=
\ee^{-kz/\epsilon}\psi_2
\label{eq:nu-psi}
\end{equation}
can then be shown to satisfy, for each fixed $z\in\mathbb{C}$,
the linear differential (with respect to $k$) equations
\begin{equation}
\begin{split}
\epsilon\delbar_k\nu_1&=\tfrac{1}{2}\overline{R^\epsilon(k;z,t)}\overline{\nu}_2\\
\epsilon\delbar_k\nu_2&=\tfrac{1}{2}\overline{R^\epsilon(k;z,t)}\overline{\nu}_1
\end{split}
\label{eq:dbar-k}
\end{equation}
where, writing $k=\kappa+\ii\sigma$ for $(\kappa,\sigma)\in\mathbb{R}^2$,
\begin{equation}
\delbar_k:=\frac{1}{2}\left(\frac{\partial}{\partial\kappa}+\ii\frac{\partial}{\partial\sigma}\right),
\end{equation}
and the asymptotic conditions
\begin{equation}
\lim_{|k|\to\infty}\nu^\epsilon_1(k;z,t)=1\quad\text{and}\quad
\lim_{|k|\to\infty}\nu^\epsilon_2(k;z,t)=0.
\label{eq:nu-asymp}
\end{equation}
The inverse scattering problem is then to recover $\nu_{1,2}^\epsilon(k;z,t)$ given $R^\epsilon_0(k)$ from \eqref{eq:dbar-k}--\eqref{eq:nu-asymp}, a problem that formally very closely resembles the direct scattering problem \eqref{eq:1.10}--\eqref{eq:psi-norm}.
Using the definitions \eqref{eq:nu-psi}
and the complex conjugate of the second equation of the system \eqref{eq:1.10} gives the reconstruction formula
\begin{equation}
q^\epsilon(x,y,t)=2\epsilon\overline{\left[\frac{\del\psi_2}{\psi_1}\right]}=2\epsilon\frac{\delbar\,\overline{\psi}_2}{\overline{\psi}_1}=2\frac{\overline{k}\,\overline{\nu}_2 +\epsilon\delbar\,\overline{\nu}_2}{\overline{\nu}_1}.
\label{eq:q-find-1}
\end{equation}
The right-hand side is in fact independent of $k$, so we can let $|k|\to\infty$
and use the asymptotics for $\nu_j$ with respect to $k$ to get the reconstruction formula for the solution of the Cauchy problem \eqref{eq:DSII}--\eqref{eq:initial-data}.
\begin{equation}
q^\epsilon(x,y,t)=2\lim_{|k|\to\infty}\overline{k}\,\overline{\nu^\epsilon_2(k;x+\ii y,t)}
\label{eq:q-find-2}
\end{equation}
\subsection{Results and outline of the paper}
Both the direct and inverse transforms involve singularly-perturbed linear elliptic problems in the plane.  This paper concerns the development of tools for the study of the direct scattering problem in the semiclassical limit $\epsilon\downarrow 0$.  The ultimate goal is to determine an asymptotic formula for the reflection coefficient $R_0^\epsilon(k)$ associated with suitably general real-valued amplitude and phase functions $A$ and $S$.  Such a formula should contain sufficient information about the latter functions to allow their reconstruction via the inverse problem\footnote{An important observation is that, even though the direct and inverse problems are formally similar, they may be quite different in character in the semiclassical limit, if it happens that the reflection coefficient depends on $\epsilon$ in a more subtle way than does the initial data $q^\epsilon(x,y,0) = A(x,y)\ee^{\ii S(x,y)/\epsilon}$.}, also considered in the semiclassical limit $\epsilon\downarrow 0$.  

\subsubsection{Aside:  an analogous, better-understood, problem} 
It is useful to have in mind the simpler example of the integrable defocusing nonlinear Schr\"odinger equation in the form \eqref{eq:defocNLS} with initial data $q^\epsilon(x,0)=A(x)\ee^{\ii S(x)/\epsilon}$, assuming, say, that $A$ and $S''$ are Schwartz functions. 
The direct transform in this case involves the calculation of the Jost solution $\boldsymbol{\psi}$ for $\lambda\in\mathbb{R}$ of the Zakharov-Shabat system \cite{ZakharovS73}
\begin{equation}
\epsilon\frac{\dd\boldsymbol{\psi}}{\dd x}=\begin{bmatrix}-\ii\lambda & A(x)\ee^{\ii S(x)/\epsilon}\\ A(x)\ee^{-\ii S(x)/\epsilon} & \ii\lambda\end{bmatrix}\boldsymbol{\psi}
\label{eq:ZS-system}
\end{equation}
that satisfies the boundary conditions
\begin{equation}
\boldsymbol{\psi}(x)=\begin{cases} \displaystyle\begin{bmatrix}\ee^{-\ii\lambda x/\epsilon}\\0 \end{bmatrix}+R_0^\epsilon(\lambda)\begin{bmatrix}0 \\ \ee^{\ii\lambda x/\epsilon}\end{bmatrix}+o(1),&\quad x\to +\infty\\
\displaystyle T_0^\epsilon(\lambda)\begin{bmatrix}\ee^{-\ii\lambda x/\epsilon}\\0\end{bmatrix}+o(1),&\quad x\to -\infty
\end{cases}
\label{eq:Jost-BC}
\end{equation}
defining the reflection coefficient $R_0^\epsilon(\lambda)$ and transmission coefficient $T_0^\epsilon(\lambda)$ for $\lambda\in\mathbb{R}$.  We note, in comparison with the direct spectral problem \eqref{eq:1.10}--\eqref{eq:psi-norm}, the association $k=-\ii\lambda$.
Asymptotic formulae for $R_0^\epsilon(\lambda)$ and $|T_0^\epsilon(\lambda)|^2$ valid in the semiclassical limit $\epsilon\downarrow 0$ may be obtained under some further conditions on $A$ and $S$ by using the WKB method to study the problem \eqref{eq:ZS-system}--\eqref{eq:Jost-BC} and construct an approximation to $\boldsymbol{\psi}$ uniformly valid for $x\in\mathbb{R}$ as $\epsilon\downarrow 0$.
The simplest ansatz for $\boldsymbol{\psi}$ is to assume that for some scalar exponent function $E=E(x,\lambda)$ independent of $\epsilon$, $\boldsymbol{\psi}=\ee^{E/\epsilon}\boldsymbol{\phi}$ where the vector function $\boldsymbol{\phi}$ has an asymptotic power series expansion in $\epsilon$.  This ansatz forces $E$ to satisfy the eikonal equation
\begin{equation}
\left(\frac{\dd E}{\dd x}\right)^2 +\left(\lambda+\frac{1}{2}S'(x)\right)^2=A(x)^2.
\label{eq:Jost-eikonal}
\end{equation}
It is convenient to assume now that the functions 
\begin{equation}
a(x):=-\frac{1}{2}S'(x)-A(x)\quad\text{and}\quad b(x):=-\frac{1}{2}S'(x)+A(x)
\end{equation}
each have only one critical point, a minimizer $x_a$ for $a$ with value $\lambda^-:=a(x_a)$ and a maximizer $x_b$ for $b$ with value $\lambda^+:=b(x_b)$.  Then, if either $\lambda<\lambda^-$ or $\lambda>\lambda^+$, the exponent $E$ is purely imaginary and  therefore the Jost solution $\boldsymbol{\psi}(x)$ is rapidly oscillatory and one can prove that its WKB approximation is indeed uniformly accurate for all $x\in\mathbb{R}$, leading to the conclusion that $R_0^\epsilon(\lambda)$ is small beyond all orders in $\epsilon$ (similar to ``above barrier'' reflection in quantum mechanics).  However, if $\lambda^-<\lambda<\lambda^+$, there exist exactly two \emph{turning points} $x_-(\lambda)<x_+(\lambda)$ such that $E$ is imaginary giving rapidly oscillatory WKB approximations for $x<x_-(\lambda)$ and $x>x_+(\lambda)$, but $E$ is real in the intermediate region $x_-(\lambda)<x<x_+(\lambda)$ and the solutions are instead rapidly exponentially growing or decaying.  The WKB method (in its traditional Liouville-Green form, that is) fails in the vicinity of the two turning points, but uniform accuracy may be recovered with the use of Langer transformations and approximations based on Airy functions \cite[Section 7.2]{Miller06}.  It is the connection, through the two turning points, of oscillatory solutions onto exponential solutions and back again that yields nontrivial reflection in the semiclassical limit.  One obtains from this procedure the asymptotic formulae (see \cite[Appendix B]{MillerQ15} for all details of virtually the same calculation)
\begin{equation}
R_0^\epsilon(\lambda)=\ee^{2\ii\theta_0(\lambda)/\epsilon}(1+O(\epsilon)) \quad\text{and}\quad
|T_0^\epsilon(\lambda)|^2=1-|R_0^\epsilon(\lambda)|^2=\ee^{-2m(\lambda)/\epsilon}(1+O(\epsilon)),\quad\epsilon\downarrow 0,
\label{eq:1-D-reflection}
\end{equation}
where
\begin{equation}
m(\lambda):=\int_{x_-(\lambda)}^{x_+(\lambda)}\sqrt{(\lambda-a(x))(b(x)-\lambda)}\,\dd y
\label{eq:m-lambda}
\end{equation}
and
\begin{equation}
\theta_0(\lambda):=-\frac{1}{2}S(x_+(\lambda))-\lambda x_+(\lambda) + \int_{x_+(\lambda)}^{+\infty}\left[\sigma\sqrt{(\lambda-a(x))(\lambda-b(x))}-(\lambda+\tfrac{1}{2}S'(y))\right]\,\dd y
\label{eq:theta0-lambda}
\end{equation}
where $\sigma:=\mathrm{sgn}(\lambda+\tfrac{1}{2}S'(+\infty))$.  The only points $\lambda\in\mathbb{R}$ not covered by these approximations are those very close to the values $\lambda^\pm$ at which the two turning points coalesce.  In any case, the point is that these formulae represent the leading term of the reflection coefficient and its exponentially small deviation from unit modulus ($|T_0^\epsilon(\lambda)|^2$ is exponentially small, an analogue of a quantum tunneling amplitude) in the interval $(\lambda^-,\lambda^+)$ as explicit integral transforms of the given functions $A$ and $S$, via also the related functions $a$ and $b$.  This is sufficient information to allow these functions to be recovered by using the leading term in place of the actual reflection coefficient in the inverse spectral transform, in this case a Riemann-Hilbert problem of quite a different character than the direct spectral problem.  See \cite{Miller16} for details of these calculations.  

\subsubsection{Generalization to DS-II}
It should be noted that the WKB approach to the calculation of $R_0^\epsilon(\lambda)$ for the problem \eqref{eq:ZS-system}--\eqref{eq:Jost-BC} can also be motivated by the existence of several potentials $A(x)\ee^{\ii S(x)/\epsilon}$ for which the solution $\boldsymbol{\psi}$ can be calculated explicitly for all $\lambda$ and $\epsilon>0$ in terms of special functions.  Examples include piecewise constant potentials for which $\boldsymbol{\psi}$ is constructed by solving constant-coefficient systems in different intervals joined by continuity at the junction points, and Schwartz potentials of the form $A(x)=A\,\mathrm{sech}(x)$ and $S(x)=\mu\tanh(x)$ for arbitrary constants $A$ and $\mu$, for which the equation \eqref{eq:ZS-system} can be reduced to the Gauss hypergeometric equation and hence $R^\epsilon_0(\lambda)$ expressed explicitly in terms of gamma functions (see \cite{SatsumaY74} and \cite{TovbisV00} where this is done for the focusing case; the calculations for the defocusing case are similar).  By contrast, for the two-dimensional problem \eqref{eq:1.10}--\eqref{eq:psi-norm}, \emph{there exist no known potentials other than $q\equiv 0$ for which the direct spectral problem can be solved and the reflection coefficient $R^\epsilon_0(k)$ recovered for all $k\in\mathbb{C}$ and a set of $\epsilon>0$ with accumulation point $\epsilon=0$.}

In absence of any exact solutions, a natural approach to the two-dimensional generalization \eqref{eq:1.10}--\eqref{eq:psi-norm} of this problem is to mimic the WKB ansatz that produces such useful and explicit formulae as \eqref{eq:1-D-reflection}--\eqref{eq:theta0-lambda}.  As will be shown in Section \ref{sec:WKB}, this leads one to consider a certain complex-valued \emph{eikonal function} $f(x,y;k)$ that is an analogue of the WKB exponent $E$ trivially computed for the Zakharov-Shabat system \eqref{eq:ZS-system} as the antiderivative of an eigenvalue of the coefficient matrix in the system obtained from \eqref{eq:ZS-system} by a simple gauge transformation to remove the oscillatory factors $\ee^{\pm\ii S(x)/\epsilon}$.  In the two-dimensional setting, this function satisfies the \emph{eikonal equation}:
\begin{equation}
\left[2\delbar f + \ii\delbar S\right]\left[\vphantom{\delbar}2\del f - \ii\del S\right]=A^2,
\label{eq:eikonal}
\end{equation}
a nonlinear partial differential equation in the $(x,y)$-plane.  If $A$, $S$, and $E:=f-\re(kz)$ are independent of $y$, with $k=-\ii\lambda$ and $\lambda\in\mathbb{R}$ this equation reduces to the far simpler (solvable by quadrature) one-dimensional eikonal equation \eqref{eq:Jost-eikonal}.  For formal validity of the WKB ansatz near $|z|=\infty$ we insist that $f$ satisfy the asymptotic condition
\begin{equation}
\lim_{|z|\to\infty} \left(f+\frac{\ii}{2}S-kz\right)=0,\quad z=x+\ii y.
\label{eq:f-norm}
\end{equation}
We refer to the problem of finding a function $f=f(x,y;k)$ that satisfies \eqref{eq:eikonal}--\eqref{eq:f-norm} as the \emph{eikonal problem}.  An important point is that unlike the direct scattering problem \eqref{eq:1.10}--\eqref{eq:psi-norm}, the eikonal problem is independent of the parameter $\epsilon>0$, so although it is nonlinear it is not a singularly perturbed problem at all.  This may be viewed as a distinctive advantage\footnote{Note, however, that unlike the one-dimensional reduction of the eikonal problem which is explicitly solvable by quadratures and square roots, to our knowledge there is no analogous elementary integration procedure for the two-dimensional eikonal problem \eqref{eq:eikonal}--\eqref{eq:f-norm}, which is studied in detail in Section~\ref{sec:eikonal-problem}.} of the WKB approach to the direct scattering problem.  In Section~\ref{sec:eikonal-problem} we prove the following result.  Here $\mathrm{W}(\mathbb{R}^2)$ is the Wiener algebra of functions with Lebesgue integrable Fourier transforms, equipped with the norm $\|\cdot\|_\mathrm{W}$ defined in \eqref{eq:W-norm}.
\begin{theorem}
Suppose that $u=A^2\in L^p(\mathbb{R}^2)\cap \mathrm{W}(\mathbb{R}^2)$ for some $p\in [1,2)$, that $S\in C^1(\mathbb{R}^2)$, and that $v=\partial S\in \mathrm{W}(\mathbb{R}^2)$.  Then for every $B>\|v\|_\mathrm{W}$, if $k$ satisfies the inequality 
\begin{equation}
|k|>B + \max\left\{\frac{1}{4}\frac{\|u\|_\mathrm{W}}{B-\|v\|_\mathrm{W}},\frac{1}{2}\sqrt{\|u\|_\mathrm{W}}\right\},
\label{eq:mod-k-lower-bound}
\end{equation}
there is a unique global classical ($C^1(\mathbb{R}^2)$) solution $f(x,y;k)$ of the eikonal problem \eqref{eq:eikonal}--\eqref{eq:f-norm} that satisfies the estimate
\begin{equation}
\left\|\del f-k-\frac{1}{2}\ii\del S\right\|_\mathrm{W}\le B.
\label{eq:f-estimate}
\end{equation}
In particular, \eqref{eq:mod-k-lower-bound} and \eqref{eq:f-estimate} imply that $\ii\partial S-2\partial f$ is bounded away from zero on $\mathbb{R}^2$.
\label{theorem:eikonal-solve}
\end{theorem}
One interpretation of Theorem~\ref{theorem:eikonal-solve} is that the eikonal problem \eqref{eq:eikonal}--\eqref{eq:f-norm} is of nonlinear elliptic type for sufficiently large $|k|$.
The leading term of the WKB approximation explained in Section~\ref{sec:WKB} is proportional to a complex-valued function $\alpha_0=\alpha_0(x,y;k)$ that is required to solve the following linear equation in which the eikonal function $f$ appears as a coefficient:
\begin{equation}
\mathcal{L}\alpha_0:=(2\delbar f+\ii\delbar S)\del (A\alpha_0 )+A\delbar ((2\del f-\ii\del S)\alpha_0)=0.
\label{eq:alpha-equation-general}
\end{equation}
\begin{theorem}
Under the same conditions on $u=A^2$, $S$, $v=\partial S$, and $k$ as in Theorem~\ref{theorem:eikonal-solve}, there is a unique solution $\alpha_0$ of \eqref{eq:alpha-equation-general} for which $\alpha_0^2-1=m\in \mathrm{W}(\mathbb{R}^2)$.
\label{theorem:alpha-solve}
\end{theorem}
These two results provide conditions on $k$ sufficient to guarantee the \emph{formal} validity of the WKB expansion in the whole $(x,y)$-plane.  As will be shown in Section~\ref{sec:WKB}, global validity of the WKB expansion for a given $k\in\mathbb{C}$ implies that the reflection coefficient $R^\epsilon_0(k)$ tends to zero with $\epsilon$.  This situation is therefore completely analogous to the fact that in the one-dimensional analogue of this problem, if $|\lambda|$ is sufficiently large there are no turning points and hence a globally defined (purely imaginary) WKB exponent function $E(x;\lambda)$ exists and leads to negligible reflection.  \emph{Theorems~\ref{theorem:eikonal-solve} and \ref{theorem:alpha-solve} would therefore provide a two-dimensional analogue of the fact that in the one-dimensional setting the reflection coefficient is asymptotically supported on the finite interval $[\lambda^-,\lambda^+]$.} 

In Section~\ref{sec:numerical-examples} we give convincing numerical evidence that in the situation covered by Theorems~\ref{theorem:eikonal-solve} and \ref{theorem:alpha-solve} (and more generally, that the eikonal problem \eqref{eq:eikonal}--\eqref{eq:f-norm} has a global classical solution) the leading term of the WKB expansion indeed gives the expected order of relative accuracy as $\epsilon\downarrow 0$.  
Unfortunately, a proof of accuracy of the method, even in the favorable situation of global existence of the eikonal function, eludes us. Nonetheless the numerical results suggest the following conjecture:
\begin{conjecture}
Suppose (for instance) that $A$ and $S-S_\infty$ are Schwartz-class functions for some constant $S_\infty\in\mathbb{R}$, and that $k\in\mathbb{C}\setminus\{0\}$ is such that there exists a global classical solution $f(x,y;k)$ of the eikonal problem \eqref{eq:eikonal}--\eqref{eq:f-norm}.  Then the solution $\boldsymbol{\psi}^\epsilon(z;k)$ of the direct scattering problem \eqref{eq:1.10}--\eqref{eq:psi-norm} at $t=0$, well defined for all $\epsilon>0$, satisfies
\begin{equation}
\ee^{-f(x,y;k)/\epsilon}\ee^{-\ii S(x,y)\sigma_3/(2\epsilon)}\boldsymbol{\psi}^\epsilon(x+\ii y;k)=\frac{\alpha_0(x,y;k)}{2k}\begin{bmatrix}2\del f(x,y;k)-\ii \del S(x,y)\\ A(x,y)\end{bmatrix}+o(1),\quad\epsilon\downarrow 0
\label{eq:conjecture-equation}
\end{equation}
with the convergence measured in a suitable norm and the $o(1)$ symbol on the right-hand side can be uniquely continued to a full asymptotic power series in positive integer powers of $\epsilon$.
\label{conjecture:WKB}
\end{conjecture}
Some of the issues that would need to be addressed to give a proper 
proof of Conjecture~\ref{conjecture:WKB} are mentioned in 
Section~\ref{sec:accuracy-proof}. The accuracy of the WKB approximation predicted by Conjecture~\ref{conjecture:WKB} is illustrated in 
Figure~\ref{fig:conjecture}, in which the solution of the Dirac 
problem \eqref{eq:1.10}--\eqref{eq:psi-norm} for a Gaussian potential $A(x,y)\ee^{\ii S(x,y)/\epsilon}=\ee^{-(x^2+y^2)}$ at $k=1$ and 
$\epsilon=1/16$ is plotted in the upper row (for numerical reasons we plot the 
modulus of the components multiplied by $\ee^{-kz/\epsilon}$ in order to have functions 
bounded at infinity), while plots for the corresponding WKB 
approximation indicated in the conjecture are shown in the lower row. The qualitative accuracy of the 
approximation is obvious from these plots, however a more systematic numerical study of 
these questions is presented in Section~\ref{sec:numerical-examples}.
\begin{figure}[htb!]
  \includegraphics[width=0.49\textwidth]{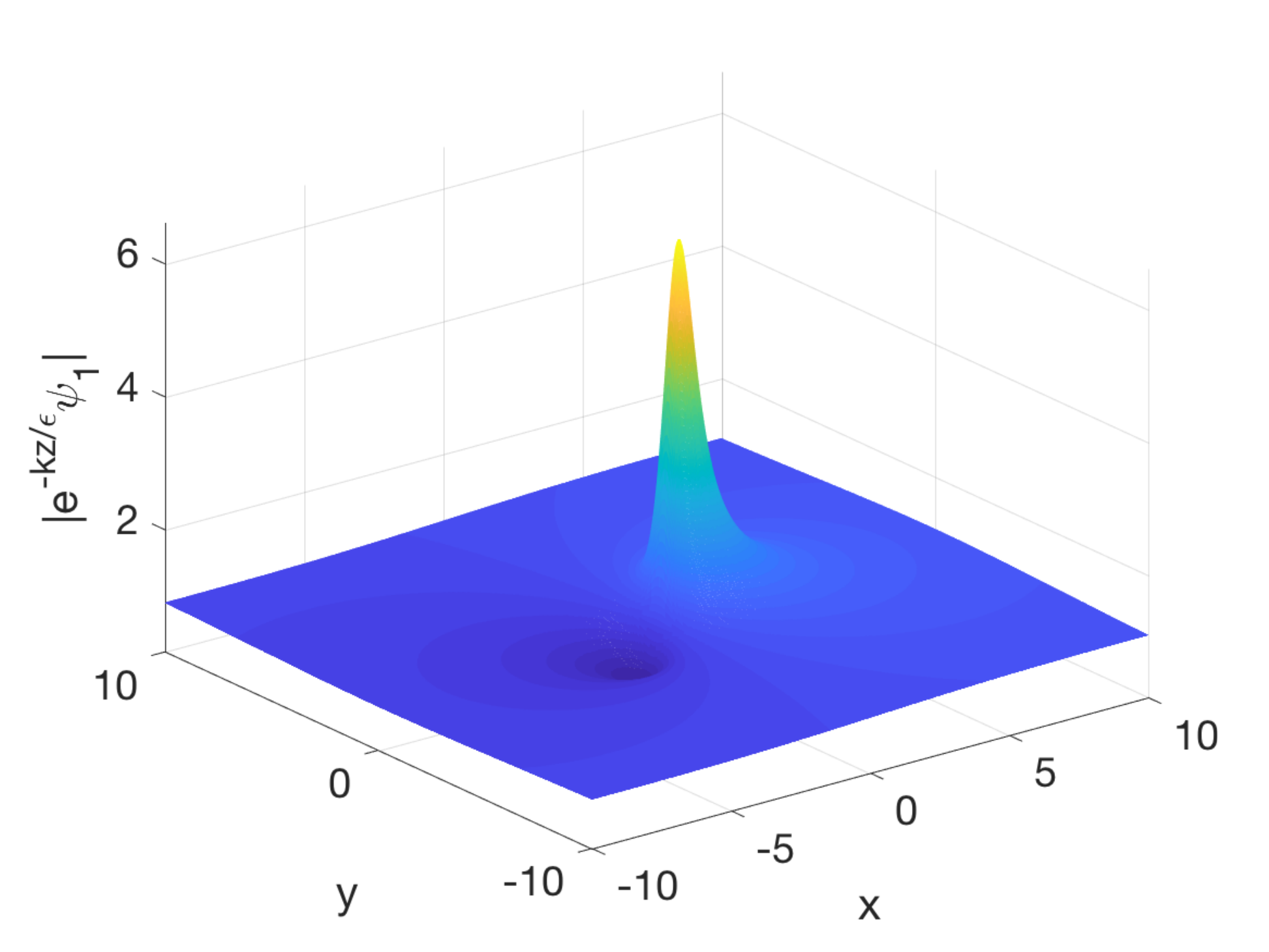}
  \includegraphics[width=0.49\textwidth]{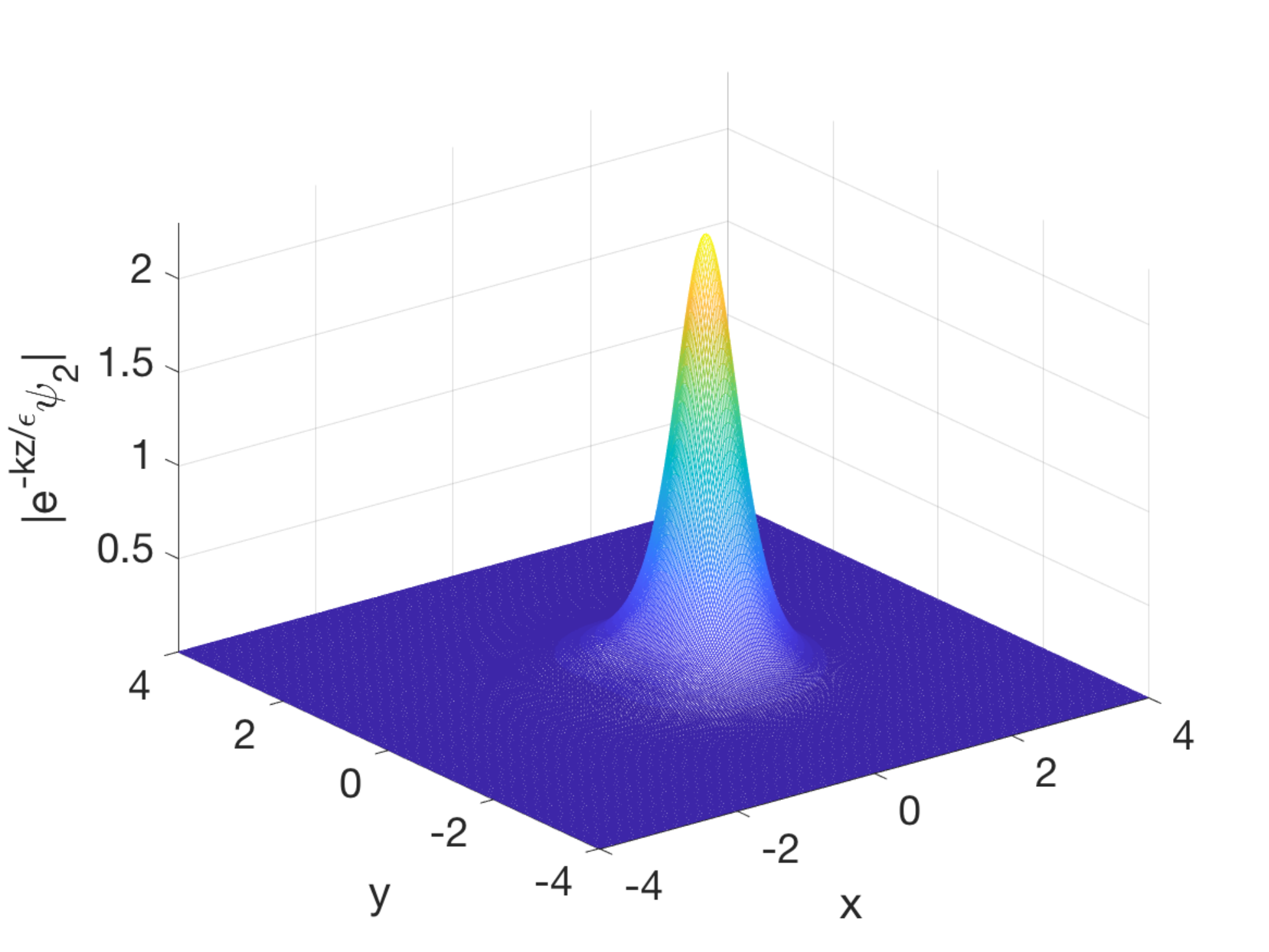}\\
  \includegraphics[width=0.49\textwidth]{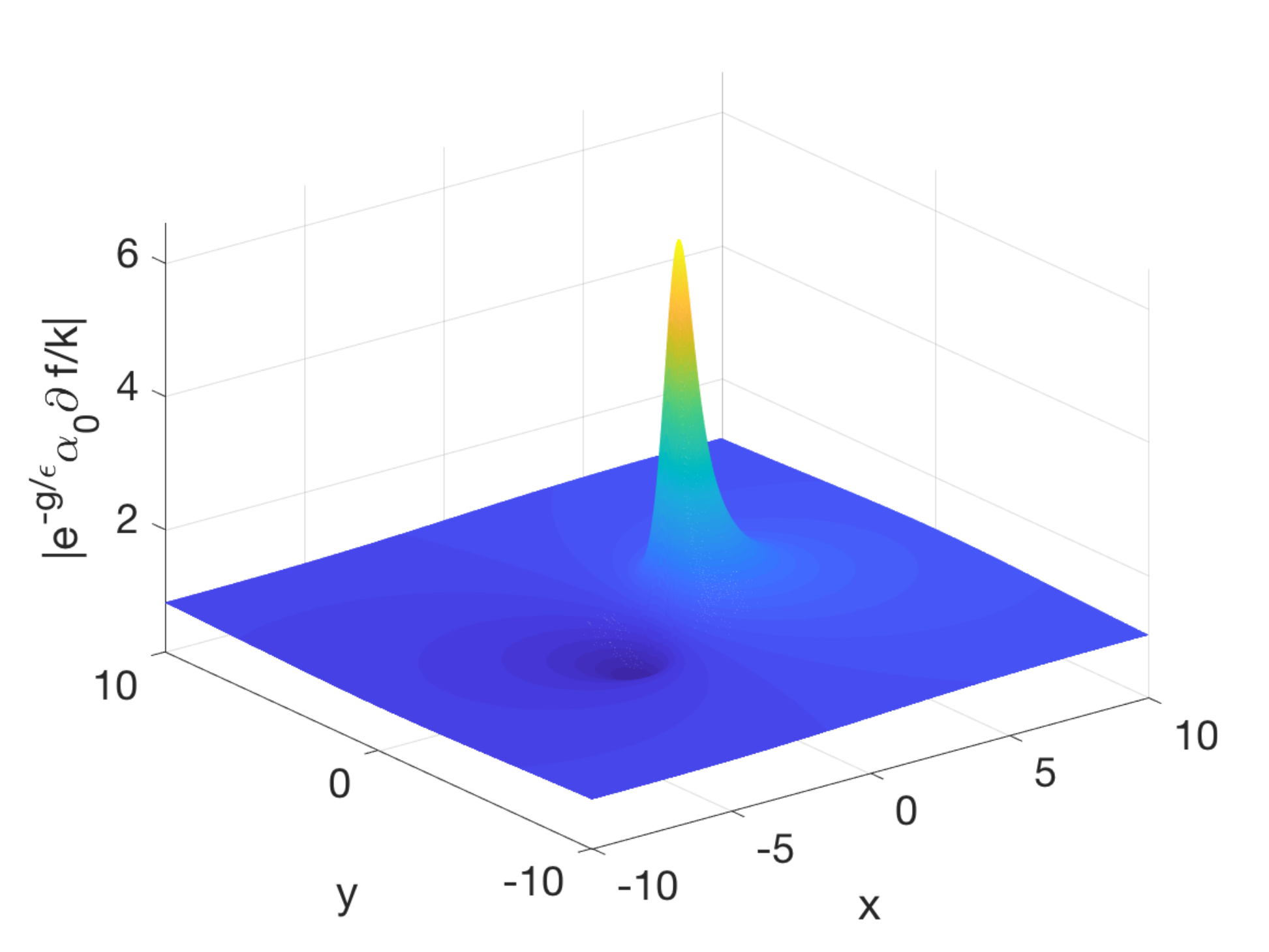}
  \includegraphics[width=0.49\textwidth]{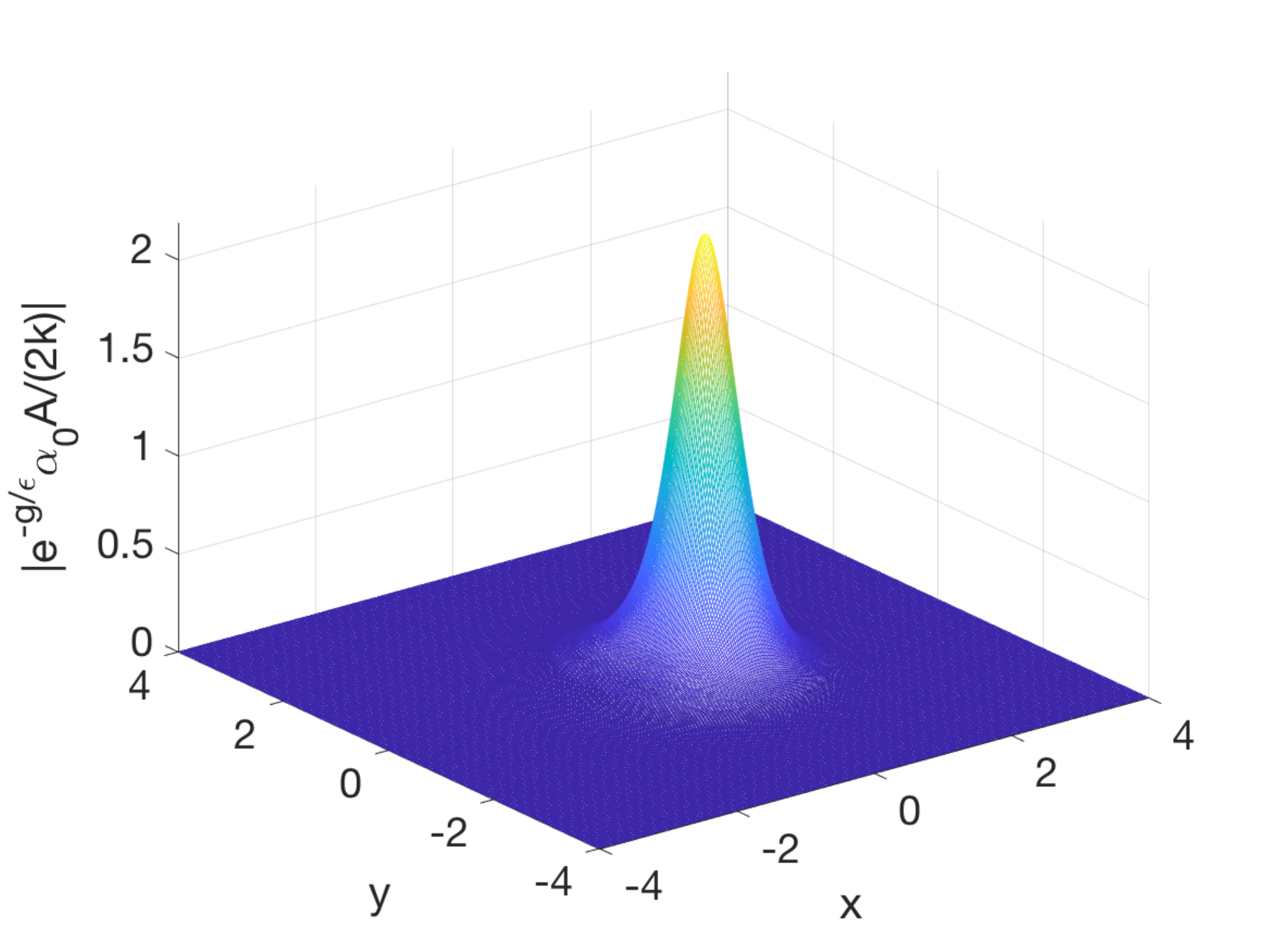}
 \caption{Comparison between the solution to the Dirac system \eqref{eq:1.10}--\eqref{eq:psi-norm} with Gaussian 
 potential $\ee^{-(x^2+y^2)}$ for $k=1$ and $\epsilon=1/16$ with the WKB approximation. First row:  the modulus of 
 $\ee^{-kz/\epsilon}\psi_{1}$ (left) and of  
 $\ee^{-kz/\epsilon}\psi_{2}$ (right). Second row: the corresponding WKB 
 approximations of Conjecture~\ref{conjecture:WKB}. }
 \label{fig:conjecture}
\end{figure}

A corollary of the existence of the full asymptotic expansion anticipated by Conjecture~\ref{conjecture:WKB} is the following stronger control on the reflection coefficient.
\begin{corollary}
Under the same conditions as in Conjecture~\ref{conjecture:WKB}, the reflection coefficient satisfies $R_0^\epsilon(k)=O(\epsilon^p)$ as $\epsilon\downarrow 0$ for all $p$.
\label{corollary:negligible-reflection}
\end{corollary}
The proof is simply based on the a priori existence of the reflection coefficient and is given in Section~\ref{sec:WKB-formalism}.

If $|k|$ becomes too small, we can no longer guarantee the existence of a global solution to the eikonal problem \eqref{eq:eikonal}--\eqref{eq:f-norm}.  As long as $k\neq 0$ it is, however, possible
to find a solution in a $k$-dependent neighborhood of $z=\infty$:
\begin{theorem}
Suppose that $u=A^2\in L^p(\mathbb{R}^2)\cap \mathrm{W}(\mathbb{R}^2)$ for some $p\in [1,2)$, that $S\in C^1(\mathbb{R}^2)$, $S-S_\infty\in \mathrm{W}(\mathbb{R}^2)$ for some constant $S_\infty\in\mathbb{R}$, and $v=\partial S\in \mathrm{W}(\mathbb{R}^2)$, and let $k\neq 0$ be a given complex value.  Then there exists $\rho>0$ such that there is a classical solution $f(x,y;k)$ of \eqref{eq:eikonal}--\eqref{eq:f-norm} defined for $|z|\ge \rho$.  There is also a corresponding classical solution $\alpha_0(x,y;k)$ of \eqref{eq:alpha-equation-general} in the same domain $|z|\ge \rho$ satisfying $\alpha_0\to 1$ as $|z|\to\infty$.
\label{theorem:squash}
\end{theorem}

The proof is given in Section~\ref{sec:squash}.  This result begs the question of what goes wrong with the eikonal problem if, given $k\neq 0$ with $|k|$ sufficiently small, one tries to continue the solution inwards from $z=\infty$.  Here we cannot say much yet; however we can present a potentially illustrative example.  Namely, if $A(x,y)=(1+x^2+y^2)^{-1}$ (a Lorentzian potential) and $S(x,y)\equiv 0$, we show in Section~\ref{sec:Lorentzian} that for $|k|>\tfrac{1}{2}$, 
the eikonal problem \eqref{eq:eikonal}--\eqref{eq:f-norm} has the explicit global solution
\begin{equation}
f(x,y;k)=kz +\frac{1}{2}\arcsin(W)+\frac{(1-W^2)^{1/2}-1}{2W},\quad W:=\frac{\overline{z}}{k(1+z\overline{z})},
\label{eq:fW-Lorentzian}
\end{equation}
and the equation \eqref{eq:alpha-equation-general} has the explicit solution
\begin{equation}
\alpha_0(x,y;k)=\sqrt{2}\left(\left(1-W^2\right)^{1/2}\left(1+\left(1-W^2\right)^{1/2}\right)\right)^{-1/2}
\label{eq:alphaW-Lorentzian}
\end{equation}
for which $\alpha_0(x,y;k)\to 1$ as $|z|\to\infty$.  Note that $\alpha_0$ is well-defined and smooth as long as $|k|>1/2$, i.e., exactly the same condition under which $f(x,y;k)$ is smooth.

With these explicit formulae in hand, we can begin to address the question of what happens to $f(x,y;k)$ when $|k|<1/2$, a necessary condition for the reflection coefficient to be non-negligible in the limit $\epsilon\to 0$.  It is easy to see that the whole complex $z$-plane is mapped onto the closed disk $D_{1/(2|k|)}$ in the $W$-plane of radius $(2|k|)^{-1}$ centered at $W=0$.  Each point in the interior of 
$D_{1/(2|k|)}$ has exactly two preimages in the $z$-plane along the ray satisfying $\arg(z)+\arg(k)+\arg(W)=0$, one with $|z|<1$ and one with $|z|>1$, while the map is one-to-one from the unit circle in the $z$-plane onto the boundary of $D_{1/(2|k|)}$.  When $|k|<1/2$, the disk $D_{1/(2|k|)}$ necessarily intersects both branch cuts emanating from $W=\pm 1$.  Pulling the parts of the branch cuts in $D_{1/(2|k|)}$ back to the $z$-plane, one sees that $f(x,y;k)$ is well-defined and smooth with the exception of two cuts, each of which connects the two preimages in the $z$-plane of the branch points $W=\pm 1$, joining them through the point on the unit circle in the $z$-plane corresponding to where the branch cut in the $W$-plane meets the boundary of $D_{1/(2|k|)}$.  Assuming $|k|<1/2$, the two preimages of $W=\pm 1$ are
\begin{equation}
z=\pm\frac{1}{2k}\left[1+\sigma\sqrt{1-4|k|^2}\right],\quad\sigma^2=1.
\end{equation}
This calculation is interesting because it shows that at branch points of $f$, which may be compared with turning points in the one-dimensional problem, the amplitude function $\alpha_0$ given by \eqref{eq:alphaW-Lorentzian} exhibits $-1/4$ power singularities, exactly as in the one-dimensional problem (see \cite[Section 7.2]{Miller06} and \cite[Appendix B.2]{MillerQ15}).  \emph{This suggests that the branch points might play the role in the two-dimensional problem that turning points play in the one-dimensional problem.}  
The branch points and cuts for $f(x,y;k)$ are shown in the $z$-plane for two values of $k$ in Figure~\ref{fig:LorentzianBranchCuts}.
\begin{figure}[h]
\begin{center}
\includegraphics{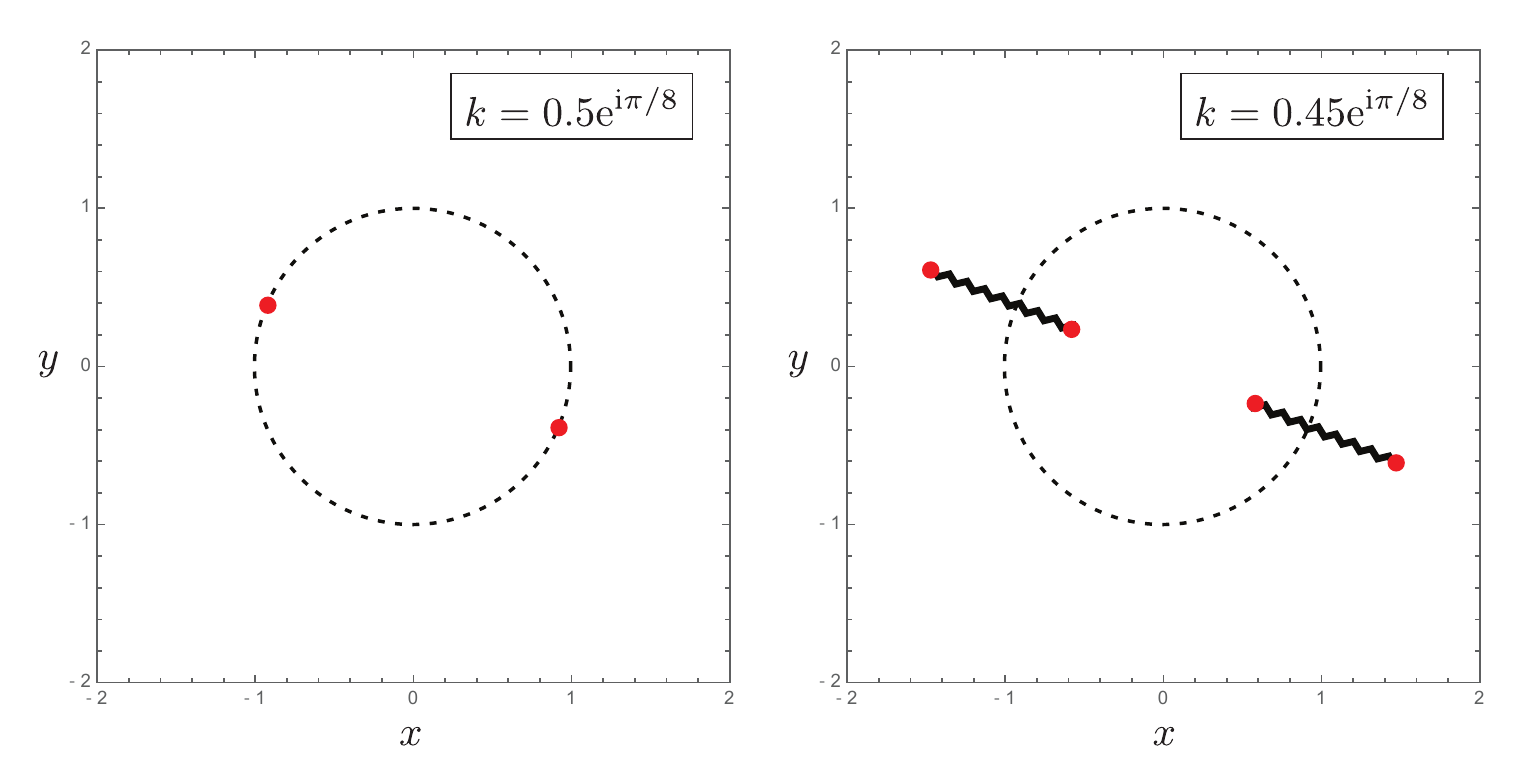}
\end{center}
\caption{The branch points (red) and cuts (wavy lines) in the $z$-plane for the continuation below $|k|=\tfrac{1}{2}$ of the solution \eqref{eq:fW-Lorentzian} of the eikonal problem \eqref{eq:eikonal}--\eqref{eq:f-norm} for the Lorentzian potential $A(x,y)=(1+x^2+y^2)^{-1}$ with $S(x,y)\equiv 0$.  Left:  $k=\tfrac{1}{2}\ee^{\ii\pi/8}$.  Right:  $k=0.45\ee^{\ii\pi/8}$.  For reference, the unit circle is shown with a dashed line.}
\label{fig:LorentzianBranchCuts}
\end{figure}

In the one-dimensional problem, the reflection coefficient fails to converge to zero with $\epsilon$ as soon as turning points appear in the problem, and one might therefore be led to believe that in the two-dimensional problem something similar occurs when $k$ decreases within a finite radius (e.g., $|k|=1/2$ for $A(x,y)=(1+x^2+y^2)^{-1}$ and $S(x,y)\equiv 0$) at which point singularities first appear in the solution of the eikonal problem.  \emph{Numerical reconstructions of the reflection coefficient for small $\epsilon$ suggest that this is indeed the case.}  In Figure~\ref{fig:reflection-Gaussian} we plot the reflection coefficient as a function of $|k|$ for the Gaussian potential $A(x,y)=\ee^{-(x^2+y^2)}$ with $S(x,y)\equiv 0$.  
\begin{figure}[h]
\begin{center}
\includegraphics{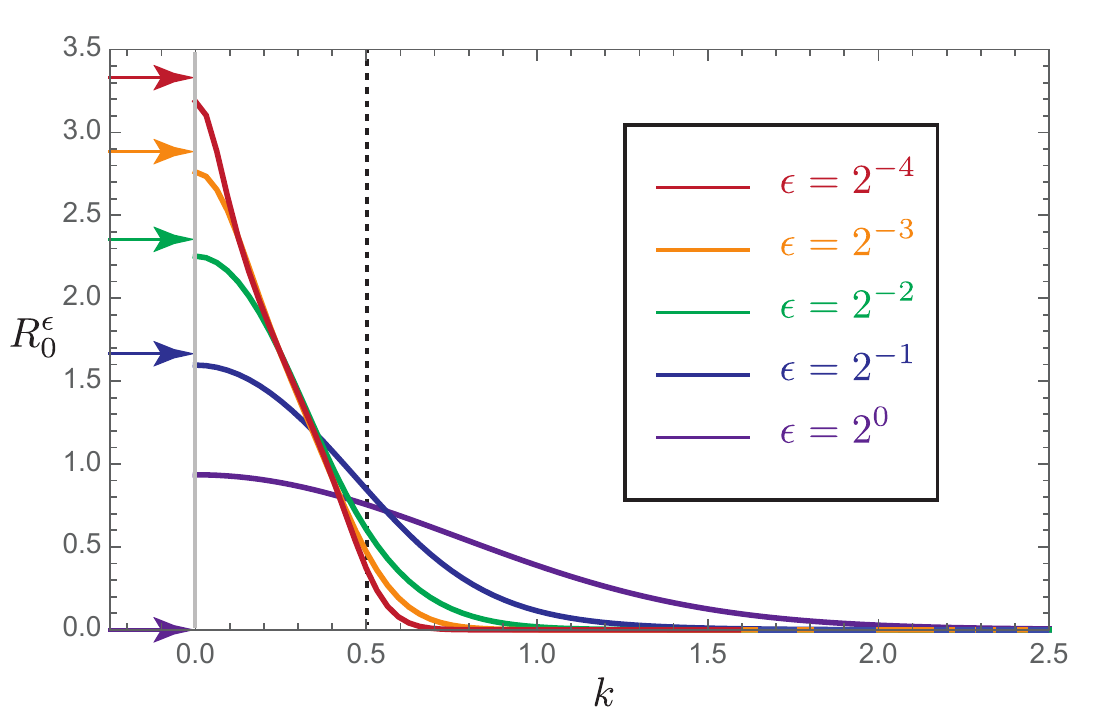}
\end{center}
\caption{Numerical calculations of the (radially symmetric and real-valued) reflection coefficient $R_0^\epsilon(k)$ for the Gaussian potential $q=A(x,y)\ee^{\ii S(x,y)/\epsilon}$ with $S(x,y)\equiv 0$ and $A(x,y)=\ee^{-(x^2+y^2)}$, plotted as a function of $k>0$ for several values of $\epsilon>0$.  The dotted vertical line at $k=1/2$ is the numerically-predicted threshold below which the eikonal problem does not have a smooth global solution and the support of the reflection coefficient appears to concentrate with decreasing $\epsilon$.  Also indicated with arrows on the vertical axis are the corresponding values of the relatively accurate approximate formula $2\sqrt{\ln(\epsilon^{-1})}$ for $R_0^\epsilon(0)$ as predicted by Theorem~\ref{theorem:Riccati-rigorous}.}
\label{fig:reflection-Gaussian}
\end{figure}
The reflection coefficient was calculated by solving the direct scattering problem \eqref{eq:1.10}--\eqref{eq:psi-norm} numerically using the scheme of \cite{KleinMcL} summarized in Section~\ref{sec:numerical-dbar}.  
These plots show that as $\epsilon\downarrow 0$ the support of the reflection coefficient $R^\epsilon_0(k)$ appears to reduce to a bounded region as $\epsilon\downarrow 0$, perhaps the domain $|k|\le\tfrac{1}{2}$.  
Now, as will be explained in Section~\ref{sec:Gaussian}, Theorem~\ref{theorem:eikonal-solve} predicts the existence of a global solution of the eikonal problem \eqref{eq:eikonal}--\eqref{eq:f-norm} if $|k|>1$ for the potential $A(x,y)=\ee^{-(x^2+y^2)}$ with $S(x,y)\equiv 0$, but our numerical calculations described in Section~\ref{sec:numerical-Gaussian} suggest that this is not a sharp bound, and moreover they suggest that the correct value at which singularities first form in $f(x,y;k)$ is again $|k|=\tfrac{1}{2}$, exactly as is known to be true for the Lorentzian potential.
\emph{Therefore,
like in the one-dimensional problem, we expect that the existence of singularities in the solution of the eikonal problem \eqref{eq:eikonal}--\eqref{eq:f-norm} leads to nontrivial reflection in the semiclassical limit}.

Despite this connection with the one-dimensional problem, it is worth dwelling on the stark qualitative differences between the asymptotic behavior of $R_0^\epsilon(k)$ for the two-dimensional problem as illustrated in Figure~\ref{fig:reflection-Gaussian}
and that of $R_0^\epsilon(\lambda)$ for the one-dimensional problem as given by \eqref{eq:1-D-reflection}
for $\lambda_-<\lambda<\lambda_+$ (and $R_0^\epsilon(\lambda)=o(1)$ as $\epsilon\downarrow 0$ for $\lambda\in\mathbb{R}\setminus [\lambda_-,\lambda_+]$).  Apparently, $R_0^\epsilon(k)$ is real-valued, non-oscillatory, and develops a growing peak near $k=0$ as $\epsilon\downarrow 0$, while $R_0^\epsilon(\lambda)$ is complex, rapidly oscillatory, and essentially of unit modulus within its asymptotic support.  \emph{Moreover, it seems obvious to the eye that as $\epsilon\downarrow 0$, $R_0^\epsilon(k)$ is converging pointwise to a real, radially-symmetric function with compact support on the disk $|k|\le\tfrac{1}{2}$ and that blows up as $|k|\downarrow 0$.}

We do not yet have a good explanation for most of these interesting features of $R^\epsilon_0(k)$.  However, motivated by the numerical observation of the growth of $R_0^\epsilon(k)$ near $k=0$, in Section \ref{sec:Riccati} we show how the solution of the direct spectral problem can be calculated for small $\epsilon$ at $k=0$ for radially symmetric potentials $A$ (and $S\equiv 0$).  This analysis is based on a radial ordinary differential equation, and it results in an asymptotic formula for $R_0^\epsilon(0)$ that we prove is accurate as $\epsilon\downarrow 0$.  
\begin{theorem}
Suppose that $S(x,y)\equiv 0$ and that $A=A(r=\sqrt{x^2+y^2})$, where $A(r)$ is a continuous nonincreasing function with $A(r)>0$ for all $r>0$ such that the function $rA(r)$ has a single maximum. Assume further that for some positive constants $L\le U$, $b$ and $p$, the inequalities $L\ee^{-br^p}\le A(r)\le U\ee^{-br^p}$ hold for $r$ sufficiently large.  Then $R^\epsilon_0(0)=2(b^{-1}\ln(\epsilon^{-1}))^{1/p}(1+o(1))$ as $\epsilon\downarrow 0$.
\label{theorem:Riccati-rigorous}
\end{theorem}
The Gaussian $A(r)=\ee^{-r^2}$ satisfies the hypotheses of Theorem~\ref{theorem:Riccati-rigorous} with $L=U=b=1$ and $p=2$, and we conclude that $R^\epsilon_0(0)=2\sqrt{\ln(\epsilon^{-1})}(1+o(1))$
as $\epsilon\downarrow 0$.  The divergence of this approximation as $\epsilon\downarrow 0$ explains the rising peak at $k=0$ seen in Figure~\ref{fig:reflection-Gaussian}; the exact values of the approximate formula for $R_0^\epsilon(0)$ are indicated with arrows for comparison. 
The heuristic analysis in Section~\ref{sec:Riccati-formal} leading up to the proof of Theorem~\ref{theorem:Riccati-rigorous} indicates that a similar approximation of $R_0^\epsilon(0)$ holds true for compactly supported amplitude functions $A=A(r)$, in which $(b^{-1}\ln(\epsilon^{-1}))^{1/p}$ is replaced with the largest value of $r>0$ in the support of $A$, which in this case is independent of $\epsilon$.  The proof of Theorem~\ref{theorem:Riccati-rigorous} is given in Section~\ref{sec:Riccati-rigorous}, and in Section~\ref{sec:characteristic-function-exact} we show how the direct spectral problem \eqref{eq:1.10}--\eqref{eq:psi-norm} can be solved explicitly in terms of special functions when $S(x,y)\equiv 0$ and $A$ is a positive multiple of the characteristic function of the disk of radius $\rho$ centered at the origin yielding the rigorous (but specialized to this particular example) result that $R^\epsilon_0(0)=2\rho+o(1)$ as $\epsilon\downarrow 0$, consistent with the general principle for compactly supported radial amplitudes indicated above.   

The analysis in Section~\ref{sec:Riccati} shows that the solution of the direct scattering problem \eqref{eq:1.10}--\eqref{eq:psi-norm} at $k=0$ for radial potentials $A=A(r)$ with $S\equiv 0$ is only consistent with the WKB expansion method in an annulus in the $(x,y)$-plane centered at the origin with an inner radius proportional to $\epsilon$ and an outer radius proportional to $r_\mathrm{Match}(\epsilon)$, defined as the largest solution $r$ of the equation $rA(r)=\epsilon$.  In this case, the eikonal problem \eqref{eq:eikonal}--\eqref{eq:f-norm} has an exact radial solution that is smooth except for a conical singularity at the origin (this solution is described in Section~\ref{sec:radial}). Our analysis shows that the $\epsilon$-dependent problem \eqref{eq:1.10}--\eqref{eq:psi-norm} regularizes the effect of this singularity within a small neighborhood of the origin, and behaves as if $A\equiv 0$ for $r>r_\mathrm{Match}(\epsilon)$.  This observation suggests that if one wants to capture the behavior of the reflection coefficient for values of $k$ of modulus sufficiently small that the eikonal problem does not have a global smooth solution, it may be necessary to construct the solution in nested approximately annular domains as is known to yield accurate approximations for $k=0$.  This is a subject for future investigation.

In Section~\ref{sec:numerics} we provide new numerical algorithms for computing the eikonal function $f(x,y;k)$ and WKB amplitude $\alpha_0(x,y;k)$, assuming that $|k|$ is sufficiently large.  These algorithms are tested on the known exact solutions \eqref{eq:fW-Lorentzian} and \eqref{eq:alphaW-Lorentzian} respectively.  One of the algorithms for computing the eikonal function $f(x,y;k)$ (a series-based method applicable to radial potentials with $S(x,y)\equiv 0$ that is described in Section~\ref{sec:numerical-Fourier-series}) also gives a method of estimating the critical radius for $|k|$ below which singularities of some sort certainly appear in the eikonal function.  This method predicts the threshold value of $|k|=\tfrac{1}{2}$ for the Gaussian $A(x,y)=\ee^{-(x^2+y^2)}$ that matches with the numerical computations of $R_0^\epsilon(k)$ shown in Figure~\ref{fig:reflection-Gaussian}. We also briefly review the method advanced in an earlier work \cite{KleinMcL} of two of the authors for solving the $\epsilon$-dependent direct scattering problem \eqref{eq:1.10}--\eqref{eq:psi-norm}.  In Section~\ref{sec:numerical-examples} we use the developed numerical methods to make quantitative comparisons with the WKB method and provide quantitative justification of Conjecture~\ref{conjecture:WKB}.

\subsection{Acknowledgements}
OA and CK acknowledge support by the program PARI and the FEDER 
2016 and 2017 as well as  the  I-QUINS project. The research of KDTRM was supported by the 
National Science Foundation under grants DMS-1401268 and DMS-1733967. 
The research of PDM was supported by the National Science Foundation 
under grants DMS-1206131 and DMS-1513054 and by the Simons Foundation 
under grant 267106.  The authors are grateful to Kari Astala, Sarah 
Hamilton, Michael Music, Peter Perry, Samuli Siltanen, Johannes 
Sj\"ostrand and Nikola Stoilov for useful discussions.  

The authors benefited from participation in a Focused Research Group on ``Inverse Problems, Nonlinear Waves, and Random Matrices'' at the Banff International Research Station in 2012, the ``Exceptional Circle'' workshop at the University of Helsinki in 2013, a conference on ``Scattering and Inverse Scattering in Multi-Dimensions'' at the University of Kentucky in 2014 (funded by the National Science Foundation under grant DMS-1408891), and a Research in Pairs meeting entitled ``Semiclassical Limit of the Davey-Stewartson Equations'' at the Mathematisches Forschungsinstitut Oberwolfach in 2016.

\section{WKB Method for Calculating the Reflection Coefficient}
\label{sec:WKB}
\subsection{WKB formalism}
\label{sec:WKB-formalism}
If the initial data is given in the form \eqref{eq:initial-data}, then \eqref{eq:1.10} takes the form of a
linear system of partial differential equations with highly oscillatory coefficients:
\begin{equation}
\epsilon\mathcal{D}\boldsymbol{\psi}=\frac{1}{2}\begin{bmatrix}0 & A(x,y)\ee^{\ii S(x,y)/\epsilon}\\
A(x,y)\ee^{-\ii S(x,y)/\epsilon} & 0\end{bmatrix}\boldsymbol{\psi},\quad\boldsymbol{\psi}=\begin{bmatrix}\psi_1\\\psi_2\end{bmatrix},\quad\mathcal{D}:=\begin{bmatrix}\delbar&0\\0 & \del\end{bmatrix}.
\label{eq:semi-system}
\end{equation}
Let us assume for simplicity that $A(x,y)$ is a strictly positive Schwartz-class function, and that
the real-valued phase is asymptotically linear: $S(x,y)=wz+\overline{w}\overline{z}+O(1)$ as $z\to\infty$ for some $w\in\mathbb{C}$, in the sense that
\begin{equation}
\partial S(x,y)\to w\quad\text{and}\quad\overline{\partial}S(x,y)\to \overline{w},\quad z\to\infty.
\label{eq:S-asymp}
\end{equation}
The parameter $w\in\mathbb{C}$ has the effect of introducing a shift of the value of the spectral parameter $k\in\mathbb{C}$.    Indeed if $S=wz+\overline{w}\overline{z}+\tilde{S}$ and $\tilde{\boldsymbol{\psi}}(z;k)$ corresponds to $(A,\tilde{S})$ while $\boldsymbol{\psi}(z;k)$ corresponds to $(A,S)$, then $\tilde{\psi}_1(z;k-\ii w)=\psi_1(z;k)\ee^{-\ii wz/\epsilon}$ and $\tilde{\psi}_2(z;k-\ii w)=\psi_2(z;k)\ee^{\ii\overline{w}\overline{z}/\epsilon}$.  Without loss of generality, we will therefore assume throughout this paper that $w=0$.  For classical solutions of \eqref{eq:semi-system}
we require $\boldsymbol{\psi}\in C^1(\mathbb{R}^2)$, and similarly for $\boldsymbol{\chi}$ and $\boldsymbol{\phi}$ to be defined shortly.

The oscillatory factors $\ee^{\pm \ii S(x,y)/\epsilon}$ can be removed from the coefficients in \eqref{eq:semi-system} by the substitution
\begin{equation}
\boldsymbol{\psi}=\ee^{\ii S(x,y)\sigma_3/(2\epsilon)}\boldsymbol{\chi}
\label{eq:psi-chi}
\end{equation}
leading to the equivalent system
\begin{equation}
\epsilon \mathcal{D}\boldsymbol{\chi}=\frac{1}{2}
\begin{bmatrix}-\ii\delbar S & A\\A & \ii\del S\end{bmatrix}
\boldsymbol{\chi}.
\label{eq:chi-system}
\end{equation}
This problem is not directly amenable to a perturbation approach, because if $\epsilon=0$ there can only exist nonzero solutions $\boldsymbol{\chi}$ 
if the coefficient matrix on the right-hand side is singular, which can be assumed to be a non-generic (with respect to $(x,y)\in\mathbb{R}^2$) phenomenon.

One way around this difficulty is to introduce a complex scalar field $f:\mathbb{R}^2\to\mathbb{C}$ and make an exponential gauge transformation of the form
\begin{equation}
\boldsymbol{\chi} = \ee^{f(x,y)/\epsilon}\boldsymbol{\phi}.
\label{eq:chi-phi}
\end{equation}
This transforms \eqref{eq:chi-system} into the form
\begin{equation}
\epsilon\mathcal{D}\boldsymbol{\phi}
=\mathbf{M}(x,y)\boldsymbol{\phi},
\label{eq:phi-system}
\end{equation}
where $\mathbf{M}(x,y)$ is the $\epsilon$-independent matrix
\begin{equation}
\mathbf{M}(x,y):=\frac{1}{2}\begin{bmatrix}
-\ii\delbar S-2\delbar f & A\\A & \ii\del S-2\del f
\end{bmatrix}.
\label{eq:M-def}
\end{equation}
Now we have both the vector unknown $\boldsymbol{\phi}$ and the scalar unknown $f$, but we may now take advantage of the extra degree of freedom by choosing $f$ in such a way that the augmented coefficient matrix $\mathbf{M}(x,y)$ on the right-hand side of \eqref{eq:phi-system}
is singular for all $(x,y)\in\mathbb{R}^2$.  A direct calculation shows that the condition $\det(\mathbf{M}(x,y))=0$ is precisely the eikonal equation \eqref{eq:eikonal} for $f$.  If $f$ is any solution of this nonlinear partial differential equation, it follows that there exist nonzero solutions
of \eqref{eq:phi-system} when $\epsilon=0$, and such a solution can be used as the leading term in a formal asymptotic power series expansion in powers of $\epsilon$.

Next recall the asymptotic normalization conditions \eqref{eq:psi-norm} on the functions $\psi_j$ as $z\to\infty$, which in terms of $\boldsymbol{\phi}$ imply
\begin{equation}
\begin{split}
\lim_{|z|\to\infty}\phi_1\ee^{f/\epsilon}\ee^{\ii S/(2\epsilon)}\ee^{-kz/\epsilon}&=1\\
\lim_{|z|\to\infty}\phi_2\ee^{f/\epsilon}\ee^{-\ii S/(2\epsilon)}\ee^{-\overline{k}\overline{z}/\epsilon}&=0
\end{split}
\end{equation}
Since $S$ is real, and since the second limit is zero, these two conditions can be combined to read
\begin{equation}
\lim_{|z|\to\infty}\boldsymbol{\phi}\exp\left(\frac{1}{\epsilon}\left[f+\frac{\ii}{2}S-kz\right]\right)=\begin{bmatrix}1\\0\end{bmatrix}.
\label{eq:phi-norm}
\end{equation}
Since we want to be able to accurately represent $\boldsymbol{\phi}$ using asymptotic power series in $\epsilon$, in particular we want $\boldsymbol{\phi}$ to have simple asymptotics as $z\to\infty$, so we now impose on the eikonal function $f$ the normalization condition \eqref{eq:f-norm}.  
Under this condition, \eqref{eq:phi-norm} becomes simply
\begin{equation}
\lim_{|z|\to\infty}\boldsymbol{\phi}=\begin{bmatrix}1\\0\end{bmatrix}.
\label{eq:phi-norm-2}
\end{equation}
Since the conditions \eqref{eq:eikonal}--\eqref{eq:f-norm} on the eikonal function $f$ explicitly involve the spectral parameter $k\in\mathbb{C}$, we denote any solution of the eikonal problem by $f=f(x,y;k)$.  Similarly, the matrix $\mathbf{M}$ defined in \eqref{eq:M-def} now depends on $k$ via $f$ and will be denoted $\mathbf{M}(x,y;k)$, a singular matrix for all $(x,y)\in\mathbb{R}^2$.

Given a suitable value of $k\in\mathbb{C}$ and a corresponding solution $f(x,y;k)$ of the eikonal problem \eqref{eq:eikonal}--\eqref{eq:f-norm}, we may now try to determine the terms in an asymptotic power series expansion of $\boldsymbol{\phi}=\boldsymbol{\phi}^\epsilon(x,y;k)$:
\begin{equation}
\boldsymbol{\phi}^\epsilon(x,y;k)\sim\sum_{n=0}^\infty\boldsymbol{\phi}^{(n)}(x,y;k)\epsilon^n,\quad\epsilon\to 0.
\label{eq:phi-series}
\end{equation}
Substituting into \eqref{eq:phi-system} and matching the terms with the same powers of $\epsilon$
one finds firstly that 
\begin{equation}
\boldsymbol{\phi}^{(0)}(x,y;k)\in\mathrm{ker}(\mathbf{M}(x,y;k))
=\mathop{\mathrm{span}}_{\mathbb{C}(x,y)}\begin{bmatrix}2\del f(x,y;k)-\ii\del S(x,y)\\
A(x,y)\end{bmatrix}.
\end{equation}
This determines $\boldsymbol{\phi}^{(0)}(x,y;k)$ up to a scalar multiple, which in general can depend on $(x,y)\in\mathbb{R}^2$ and $k\in\mathbb{C}$.  We may therefore write
$\boldsymbol{\phi}^{(0)}(x,y;k)$ in the form
\begin{equation}
\boldsymbol{\phi}^{(0)}(x,y;k)=\frac{\alpha_0(x,y;k)}{2k}\begin{bmatrix}2\del f(x,y;k)-\ii\del S(x,y)\\ A(x,y)
\end{bmatrix}
\label{eq:Phi-zero}
\end{equation}
for a scalar field $\alpha_0(x,y;k)$ to be determined.
Then from the higher-order terms one obtains the recurrence relations:
\begin{equation}
\mathbf{M}(x,y;k)\boldsymbol{\phi}^{(n+1)}(x,y;k)=\mathcal{D}
\boldsymbol{\phi}^{(n)}(x,y;k), \quad n=0,1,2,3,\dots.
\label{eq:WKB-recurrence}
\end{equation}
As $\mathbf{M}(x,y;k)$ is singular, at each order there is a solvability condition to be enforced, namely that
\begin{equation}
\mathcal{D}\boldsymbol{\phi}^{(n)}(x,y;k)\in 
\mathrm{ran}(\mathbf{M}(x,y))=\mathop{\mathrm{span}}_{\mathbb{C}(x,y)}
\begin{bmatrix}-\ii\delbar S(x,y)-2\delbar f(x,y;k)\\A(x,y)\end{bmatrix},\quad n=0,1,2,3,\dots,
\end{equation}
which we write in Wronskian form as
\begin{equation}
\det\left(\begin{bmatrix}-\ii\delbar S(x,y)-2\delbar f(x,y;k)\\A(x,y)\end{bmatrix},\mathcal{D}\boldsymbol{\phi}^{(n)}(x,y;k)\right)=0,\quad n=0,1,2,3,\dots.
\label{eq:alpha-equation}
\end{equation}
Assuming that \eqref{eq:alpha-equation} holds for a given $n$, the general solution of \eqref{eq:WKB-recurrence} is
\begin{equation}
\boldsymbol{\phi}^{(n)}(x,y;k)=\boldsymbol{\phi}^{(n)}_\mathrm{p}(x,y;k)  + \frac{\alpha_n(x,y;k)}{2k}\begin{bmatrix}2\del f(x,y;k)-\ii\del S(x,y)\\A(x,y)\end{bmatrix},\quad n=1,2,3,\dots,
\label{eq:Phi-n}
\end{equation}
where 
\begin{equation}
\boldsymbol{\phi}^{(n)}_\mathrm{p}(x,y;k):=\frac{2\del\phi_2^{(n-1)}(x,y;k)}{\ii\del S(x,y)-2\del f(x,y;k)}\begin{bmatrix}0\\1\end{bmatrix}
\label{eq:particular}
\end{equation}
is a particular solution and
$\alpha_n(x,y;k)$ is a scalar field to be determined that parametrizes the homogeneous solution.

The calculation of the terms in the formal series \eqref{eq:phi-series} therefore has been reduced to the sequential solution of the scalar equation \eqref{eq:alpha-equation} for $\alpha_{n-1}(x,y;k)$, for $n=1,2,3,\dots$.  We interpret the boundary condition \eqref{eq:phi-norm-2} in light of the formal series \eqref{eq:phi-series} as:
\begin{equation}
\lim_{|z|\to\infty}\boldsymbol{\phi}^{(0)}(x,y;k)=\begin{bmatrix}1\\0\end{bmatrix},\quad
\lim_{|z|\to\infty}\boldsymbol{\phi}^{(n)}(x,y;k)=\mathbf{0},\quad n=1,2,3,\dots.
\label{eq:phi-norm-terms}
\end{equation}
Taking into account \eqref{eq:f-norm}, \eqref{eq:S-asymp} for $w=0$, and \eqref{eq:phi-norm-terms} we require the solution of \eqref{eq:alpha-equation} subject to the boundary condition: 
\begin{equation}
\lim_{|z|\to\infty}\alpha_n(x,y;k)=\begin{cases}1,&\quad n=0\\
0,&\quad n=1,2,3,\dots.
\end{cases}
\label{eq:alpha-n-BC}
\end{equation}


A direct calculation shows that (suppressing the arguments)
\begin{equation}
\det\left(\begin{bmatrix}-\ii\delbar S - 2\delbar f\\ A\end{bmatrix},
\mathcal{D}\frac{\alpha_{n}}{2k}\begin{bmatrix}2\del f-\ii\del S\\A\end{bmatrix}\right)=-\frac{1}{2k}\mathcal{L}\alpha_{n}
\end{equation}
where the differential operator $\mathcal{L}$ is defined in \eqref{eq:alpha-equation-general}.  Therefore, assuming $k\neq 0$, taking $n=0$ in \eqref{eq:alpha-equation} and using \eqref{eq:Phi-zero} immediately yields \eqref{eq:alpha-equation-general} for $\alpha_0$, which by \eqref{eq:alpha-n-BC} is to be solved subject to the boundary condition $\alpha_0\to 1$ as $|z|\to\infty$.  Similarly, taking $n>0$ in \eqref{eq:alpha-equation} and using \eqref{eq:Phi-n} gives a related non-homogeneous equation
\begin{equation}
\mathcal{L}\alpha_n = 2k\det\left(\begin{bmatrix}-\ii\delbar S-2\delbar f\\A\end{bmatrix},\mathcal{D}\boldsymbol{\phi}^{(n)}_\mathrm{p}\right)=2k(-\ii\delbar S-2\delbar f)\del\left[\frac{2\del\phi^{(n-1)}_2}{\ii\del S-2\del f}\right],\quad n=1,2,3,\dots,
\end{equation}
which by \eqref{eq:alpha-n-BC} is to be solved subject to the boundary condition $\alpha_n\to 0$ as $|z|\to\infty$.  We remark that under the conditions of Theorem~\ref{theorem:eikonal-solve} the denominator $\ii\del S-2\del f$ is bounded away from zero, so we may expect that the forcing term on the right-hand side is a smooth function of $(x,y)\in\mathbb{R}^2$ that decays as $|z|\to\infty$ due in part to the fact that $\ii\delbar S+2\delbar f\to 0$ as $|z|\to\infty$.  Therefore, invertibility of $\mathcal{L}$ on a suitable space of decaying functions is sufficient to guarantee the existence of all terms of the WKB expansion.  

Now we give the proof (conditioned on Conjecture~\ref{conjecture:WKB}) of Corollary~\ref{corollary:negligible-reflection}.
Observe that using \eqref{eq:r-def}, \eqref{eq:f-norm}, \eqref{eq:psi-chi}, and \eqref{eq:chi-phi}, the reflection coefficient $R^\epsilon_0(k)$ can be written in terms of the (well-defined for all $(x,y)\in\mathbb{R}^2$, $k\in\mathbb{C}$, and $\epsilon>0$) solution $\boldsymbol{\phi}^\epsilon(x,y;k)$ of \eqref{eq:phi-system} and \eqref{eq:phi-norm-2} as 
\begin{equation}
R^\epsilon_0(k)=2\lim_{z\to\infty}z\ee^{-2\ii\imag(kz)/\epsilon}\overline{\phi^\epsilon_2(x,y;k)}.
\label{eq:reflection-phi}
\end{equation}
Suppose that the WKB expansion can be successfully and uniquely constructed through terms of order $\epsilon^N$, in which case we may write $\boldsymbol{\phi}^\epsilon(x,y;k)$ unambiguously in the form
\begin{equation}
\boldsymbol{\phi}^\epsilon(x,y;k)=\sum_{n=0}^N\boldsymbol{\phi}^{(n)}(x,y;k)\epsilon^n + \tilde{\boldsymbol{\phi}}^{(N),\epsilon}(x,y;k).
\label{eq:WKB-with-error}
\end{equation}
Suppose also that the remainder term $\tilde{\boldsymbol{\phi}}^{(N),\epsilon}(x,y;k) = o(\epsilon^N)$ uniformly in $(x,y)\in\mathbb{R}^2$.  Then, since
the rapidly oscillatory factor $\ee^{-2\ii\imag(kz)/\epsilon}$ is bounded despite having no limit as $|z|\to\infty$ unless $k=0$, the (known) existence of $R^\epsilon_0(k)$ for all $k\in\mathbb{C}$ and $\epsilon>0$ implies that
$\phi^{(n)}_2(x,y;k)=o(z^{-1})$ as $|z|\to\infty$ for all $n=0,1,2,\dots,N$, and we conclude that
$R_0^\epsilon(k)=o(\epsilon^N)$ as $\epsilon\downarrow 0$.  This is rather obvious for the case of $N=0$;
indeed, replacing $\phi^\epsilon_2$ with its leading-order approximation $\phi^{(0)}_2(x,y;k)=\alpha_0(x,y;k)A(x,y)/(2k)$ yields under the assumption that $\tilde{\boldsymbol{\phi}}^{0,\epsilon}(x,y;k)$ is uniformly $o(1)$ the approximate formula
\begin{equation}
R^\epsilon_0(k)=\frac{1}{\overline{k}}\lim_{z\to\infty}z\ee^{-2\ii\imag(kz)/\epsilon}\overline{\alpha_0(x,y;k)}A(x,y)+o(1) = o(1)
\end{equation}
(the explicit limit is zero because $\alpha_0\to 1$ and $A$ is Schwartz-class).

\subsection{Some notes on rigorous analysis}
\label{sec:accuracy-proof}
Assuming for a given $k\in\mathbb{C}\setminus\{0\}$ that the terms $\boldsymbol{\phi}^{(0)},\dots,\boldsymbol{\phi}^{(N)}$ have been determined, the error term $\tilde{\boldsymbol{\phi}}^{(N),\epsilon}(x,y;k)$ in \eqref{eq:WKB-with-error} satisfies the equation
\begin{equation}
\left[\epsilon\mathcal{D}-\mathbf{M}\right]\tilde{\boldsymbol{\phi}}^{(N),\epsilon} =\epsilon^{N+1}\boldsymbol{\gamma}^{(N)}(x,y;k),\quad \boldsymbol{\gamma}^{(N)}(x,y;k):=-\mathcal{D}\boldsymbol{\phi}^{(N)}.
\end{equation}
Note that $\boldsymbol{\gamma}^{(N)}(x,y;k)$ is independent of $\epsilon>0$ and is, for each $(x,y)\in\mathbb{R}^2$, a vector in $\mathrm{ran}(\mathbf{M}(x,y;k))$ as a consequence of the equation (cf., \eqref{eq:alpha-equation}) satisfied by $\alpha_N(x,y;k)$.  

In general, the singularly-perturbed differential operator $\epsilon\mathcal{D}-\mathbf{M}$, although certainly invertible on suitable spaces ultimately as a consequence of the Fredholm theory and vanishing lemma described in \cite[Lemma 2.3]{Perry12}, will have a very large inverse when $\epsilon$ is small.  
Controlling this inverse is obviously the fundamental analytical challenge in establishing the validity of the WKB expansion.  

Here we offer only the following advice to assist in the necessary estimation:  the inverse $(\epsilon\mathcal{D}-\mathbf{M})^{-1}$ need only be controlled on the subspace of vector-valued functions $\boldsymbol{\gamma}^{(N)}$ that lie pointwise in the $\mathbb{C}^2(x,y)$ subspace $\mathrm{ran}(\mathbf{M}(x,y;k))$.  Such uniform control would automatically imply that the norm of $\tilde{\boldsymbol{\phi}}^{(N),\epsilon}$ is $O(\epsilon^{N+1})$, as continuing the WKB expansion to higher order would suggest.  

\section{The Eikonal Problem}
\label{sec:eikonal-problem}
In this section, we consider the problem of how to construct solutions of the eikonal problem consisting of the nonlinear equation \eqref{eq:eikonal} and the boundary condition \eqref{eq:f-norm}.  We also 
consider the related problem of finding the leading-order WKB amplitude function $\alpha_0(x,y;k)$.

\subsection{Global existence of $f(x,y;k)$ and $\alpha_0(x,y;k)$ for $|k|$ sufficiently large}
\label{sec:k-large}
We first consider solving the eikonal problem \eqref{eq:eikonal}--\eqref{eq:f-norm} for $f(x,y;k)$.
To study a function that tends to zero at infinity, we define 
\begin{equation}
g(x,y;k):=f(x,y;k)-kz+\ii\frac{1}{2}S(x,y),
\label{eq:f-to-g}
\end{equation}
upon which \eqref{eq:eikonal} can be rearranged to read
\begin{equation}
\delbar g=\frac{u}{4(k-\ii v+\del g)},\quad\text{where $u:=A^2$ and $v:=\del S$.}
\end{equation}
Differentiation via the operator $\del$ and assuming that $g$ is twice continuously differentiable gives an equation for $b=\del g-\ii v$:
\begin{equation}
\delbar b = -\ii\delbar v+\del\left[\frac{u}{4(k+b)}\right].
\end{equation}
Since we expect $\del g\to 0$, and we may assume $v=\del S\to 0$ as $|z|\to\infty$, we may invert $\delbar$ with the solid Cauchy transform $\delbar^{-1}$ defined by \eqref{eq:solid-Cauchy},
and hence obtain the fixed-point equation
\begin{equation}
b = F(b),
\label{eq:fixed-point-equation}
\end{equation}
where $F$ is the nonlinear mapping
\begin{equation}
F(b):= -\ii v + \mathcal{B}\left[\frac{u}{4(k+b)}\right],
\label{eq:b-fixed-point}
\end{equation}
in which $\mathcal{B}$ denotes the \emph{Beurling transform} defined by $\mathcal{B}:=\delbar^{-1}\del=\del\delbar^{-1}$.

We will seek a solution $b\in \mathrm{W}(\mathbb{R}^2)$, where $\mathrm{W}(\mathbb{R}^2)$ denotes the \emph{Wiener space} \cite{DymMcK72} defined as the completion of the Schwartz space $\mathscr{S}(\mathbb{R}^2)$ under the norm
\begin{equation}
\|b\|_\mathrm{W}:=\iint_{\mathbb{R}^2}|\hat{b}(\xi_x,\xi_y)|\,\dd\xi_x\,\dd\xi_y,\quad\hat{b}(\xi_x,\xi_y):=\frac{1}{4\pi^2}\iint_{\mathbb{R}^2}b(x,y)\ee^{-\ii(\xi_xx+\xi_yy)}\,\dd x\,\dd y,
\label{eq:W-norm}
\end{equation}
i.e., the Wiener norm is just the $L^1$ norm in the Fourier transform domain.  Observe that since the inverse Fourier transform is given by
\begin{equation}
b(x,y)=\iint_{\mathbb{R}^2}\hat{b}(\xi_x,\xi_y)\ee^{\ii(\xi_xx+\xi_yy)}\,\dd\xi_x\,\dd\xi_y,
\end{equation}
it follows that whenever $b$ is a function with a nonnegative Fourier transform $\hat{b}(\xi_x,\xi_y)\ge 0$, the Wiener norm is given simply by the value of $b$ at the origin:  $\|b\|_\mathrm{W}=b(0,0)$.  By the Riemann-Lebesgue lemma, functions in $\mathrm{W}(\mathbb{R}^2)$ are continuous and decay to zero as $|z|\to\infty$, and $\|b\|_\infty\le \|b\|_\mathrm{W}$.  A key property of the Wiener space is that it is a Banach algebra as a consequence of the convolution theorem:
\begin{equation}
\begin{split}
\|b_1b_2\|_\mathrm{W}&=\iint_{\mathbb{R}^2}|\widehat{b_1b_2}(\xi_x,\xi_y)|\,\dd\xi_x\,\dd\xi_y\\
&=\iint_{\mathbb{R}^2}|\hat{b}_1\ast\hat{b}_2(\xi_x,\xi_y)|\,\dd\xi_x\,\dd\xi_y\\
&=\iint_{\mathbb{R}^2}\left|\iint_{\mathbb{R}^2}\hat{b}_1(\xi_x',\xi_y')
\hat{b}_2(\xi_x-\xi_x',\xi_y-\xi_y')\,\dd\xi_x'\,\dd\xi_y'\right|\,\dd\xi_x\,\dd\xi_y\\
&\le\iint_{\mathbb{R}^2}\iint_{\mathbb{R}^2}|\hat{b}_1(\xi_x',\xi_y')|
|\hat{b}_2(\xi_x-\xi_x',\xi_y-\xi_y')|\,\dd\xi_x\,\dd\xi_y\,\dd\xi_x'\,\dd\xi_y'\\
&=\|b_1\|_\mathrm{W}\|b_2\|_\mathrm{W}.
\end{split}
\label{eq:Banach-algebra}
\end{equation}
Another important property obvious from the definition \eqref{eq:W-norm} is scale invariance: if $b\in \mathrm{W}(\mathbb{R}^2)$ and for $\rho>0$, $b_\rho(x,y):=b(x/\rho,y/\rho)$, then $\|b_\rho\|_\mathrm{W}=\|b\|_\mathrm{W}$ for all $\rho>0$.
The Wiener space is also well-behaved with respect to the Beurling transform, whose action in the Fourier domain is given by 
\begin{equation}
\widehat{\mathcal{B}b}(\xi_x,\xi_y)=-\frac{\xi_x+\ii\xi_y}{\xi_x-\ii\xi_y}\hat{b}(\xi_x,\xi_y),
\end{equation}
so as the Fourier multiplier has unit modulus for all $(\xi_x,\xi_y)\in\mathbb{R}^2$, $|\widehat{\mathcal{B}b}(\xi_x,\xi_y)|=|\hat{b}(\xi_x,\xi_y)|$, and therefore
\begin{equation}
\|\mathcal{B}b\|_\mathrm{W}=\|b\|_\mathrm{W},\quad\forall b\in \mathrm{W}(\mathbb{R}^2).
\label{eq:Beurling-norm}
\end{equation}
While all of these properties are useful to us, it is really the combination of the Banach algebra property \eqref{eq:Banach-algebra} with the unitarity of the Beurling transform expressed in \eqref{eq:Beurling-norm} that makes $\mathrm{W}(\mathbb{R}^2)$ a useful space for us to work with when  dealing with nonlinear problems involving the operator $\mathcal{B}$ such as \eqref{eq:fixed-point-equation}--\eqref{eq:b-fixed-point}.

To view \eqref{eq:fixed-point-equation}--\eqref{eq:b-fixed-point} as a fixed-point equation on $\mathrm{W}(\mathbb{R}^2)$, we first assume that $u\in \mathrm{W}(\mathbb{R}^2)$ and $v\in \mathrm{W}(\mathbb{R}^2)$.  We then need to guarantee that $F(b)\in \mathrm{W}(\mathbb{R}^2)$ provided that $b\in \mathrm{W}(\mathbb{R}^2)$.  We write $F(b)$ in the slightly-modified form
\begin{equation}
F(b)=-\ii v +\frac{1}{4k}\mathcal{B}u +\mathcal{B}\left[\frac{1}{4}u\left(\frac{1}{k+b}-\frac{1}{k}\right)\right].
\end{equation}
Due to \eqref{eq:Banach-algebra} and \eqref{eq:Beurling-norm}, it is sufficient that $b\mapsto (k+b)^{-1}-k^{-1}$ takes $\mathrm{W}(\mathbb{R}^2)$ into itself.  This will be the case provided that $|k|$ is sufficiently large.  Indeed,
consider the geometric series
\begin{equation}
\frac{1}{k+b}-\frac{1}{k}=\frac{1}{k}\cdot \frac{1}{1-(-b/k)}-\frac{1}{k} = -\sum_{n=1}^{\infty}\left(-\frac{1}{k}\right)^{n+1}b^n.
\label{eq:geometric-series}
\end{equation}
Since due to the homogeneity property of the norm and the Banach algebra property \eqref{eq:Banach-algebra},
\begin{equation}
\left\|\left(-\frac{1}{k}\right)^{n+1}b^n\right\|_\mathrm{W} = \frac{1}{|k|^{n+1}}\|b^n\|_\mathrm{W}
\le \frac{\|b\|_\mathrm{W}^n}{|k|^{n+1}},\quad n=1,2,3,\dots,
\end{equation}
the geometric series on the right-hand side of \eqref{eq:geometric-series} converges in the Wiener space $\mathrm{W}(\mathbb{R}^2)$ provided that $|k|>\|b\|_\mathrm{W}$.  Moreover, given any $B>0$, under the condition $\|b\|_\mathrm{W}\le B$ and $|k|>B$
we have $(k+b)^{-1}-k^{-1}\in \mathrm{W}(\mathbb{R}^2)$ with 
\begin{equation}
N_k[b]:=\left\|\frac{1}{k+b}-\frac{1}{k}\right\|_\mathrm{W}\le \sum_{n=1}^\infty \frac{\|b\|_\mathrm{W}^n}{|k|^{n+1}}
\le\sum_{n=1}^\infty\frac{B^n}{|k|^{n+1}}
= \frac{1}{|k|-B}-\frac{1}{|k|}=\frac{B}{|k|(|k|-B)}.
\label{eq:fraction-estimate}
\end{equation}
Under the same conditions, an estimate for the action of the nonlinear operator $F$ given by \eqref{eq:b-fixed-point} is
\begin{equation}
\|F(b)\|_\mathrm{W}\le
\|v\|_\mathrm{W} + \frac{\|u\|_\mathrm{W}}{4(|k|-B)},\quad \|b\|_\mathrm{W}\le B,\quad |k|>B.
\end{equation}
It follows that $F$ is a mapping from the closed $B$-ball in $\mathrm{W}(\mathbb{R}^2)$ into itself provided that $k$ and $B$ are chosen so that 
\begin{equation}
\|v\|_\mathrm{W} + \frac{\|u\|_\mathrm{W}}{4(|k|-B)}\le B\quad\text{and}\quad |k|>B.
\end{equation}
This proves the following result.
\begin{lemma}
Suppose that $u$ and $v$ are functions in the Wiener space $\mathrm{W}(\mathbb{R}^2)$ with $\|u\|_\mathrm{W}>0$.  Then, for every $B>\|v\|_\mathrm{W}$, the mapping $b\mapsto F(b)$ defined by \eqref{eq:b-fixed-point} takes the closed $B$-ball in $\mathrm{W}(\mathbb{R}^2)$ into itself if
\begin{equation}
|k|\ge B+\frac{\|u\|_\mathrm{W}}{4}\cdot\frac{1}{B-\|v\|_\mathrm{W}}>B.
\end{equation}
\label{lemma:ball-to-ball}
\end{lemma}
We next consider under what additional conditions the mapping $F$ defines a contraction on the $B$-ball in $\mathrm{W}(\mathbb{R}^2)$.  Suppose that $b,b'\in \mathrm{W}(\mathbb{R}^2)$ with $\|b\|_\mathrm{W}\le B$ and $\|b'\|_\mathrm{W}\le B$.  Then,
\begin{equation}
\begin{split}
\|F(b')-F(b)\|_\mathrm{W} &= \left\|\mathcal{B}\left[\frac{u}{4(k+b')}-\frac{u}{4(k+b)}\right]\right\|_\mathrm{W} \\
&= \left\|\frac{u}{4(k+b')}-\frac{u}{4(k+b)}\right\|_\mathrm{W}
\\ &=\frac{1}{4}\left\|
\frac{(b'-b)u}{(k+b')(k+b)}\right\|_\mathrm{W}.
\end{split}
\end{equation}
Adding and subtracting $k^{-1}$ from $(k+b')^{-1}$ and $(k+b)^{-1}$, the triangle inequality and
the Banach algebra property \eqref{eq:Banach-algebra} give
\begin{equation}
\|F(b')-F(b)\|_\mathrm{W}\le\frac{1}{4}\|u\|_\mathrm{W}\left(\frac{1}{|k|^2}+\frac{N_k[b']}{|k|} +
\frac{N_k[b]}{|k|}+N_k[b']N_k[b]\right)
\|b'-b\|_\mathrm{W},
\end{equation}
where the notation in the parentheses is defined in \eqref{eq:fraction-estimate}.
Using the inequality \eqref{eq:fraction-estimate} and the given bounds on $b$ and $b'$, we therefore get
\begin{equation}
\|F(b')-F(b)\|_\mathrm{W}\le \frac{\|u\|_\mathrm{W}}{4(|k|-B)^2}\|b'-b\|_\mathrm{W}.
\end{equation}
Combining this estimate with Lemma~\ref{lemma:ball-to-ball}, we have proved the following.
\begin{lemma}
Suppose that $u$ and $v$ are functions in the Wiener space $\mathrm{W}(\mathbb{R}^2)$.  Then, for every $B>\|v\|_\mathrm{W}$, the mapping $b\mapsto F(b)$ defined by \eqref{eq:b-fixed-point} is a contraction mapping on the closed $B$-ball in $\mathrm{W}(\mathbb{R}^2)$ if $k$ satisfies the inequality \eqref{eq:mod-k-lower-bound}.
\label{lemma:contraction}
\end{lemma}
Now we may give the proof of Theorem~\ref{theorem:eikonal-solve}.
\begin{proof}[Proof of Theorem~\ref{theorem:eikonal-solve}]
Because $u\in \mathrm{W}(\mathbb{R}^2)$ and $v\in\mathbb{R}^2$, the given condition on $k$ implies, via the contraction mapping theorem and Lemma~\ref{lemma:contraction}, the existence of a unique solution $b$ of the fixed-point equation $b=F(b)$ with $\|b\|_\mathrm{W}\le B$.  To obtain $f$ from $b$, recall that $f=g+kz+\tfrac{1}{2}\ii S$ where $\del g=b+\ii v=F(b)+\ii v$.  Applying $\del^{-1}$ as defined by the conjugate solid Cauchy transform and using $\del^{-1}\mathcal{B}=\delbar^{-1}$,  we obtain
\begin{equation}
g=\overline{\partial}^{-1}\left[\frac{u}{4(k+b)}\right].
\end{equation}
Because $\|b\|_\infty\le \|b\|_\mathrm{W}\le B$, the condition \eqref{eq:mod-k-lower-bound} on $k$ implies that $(k+b)^{-1}\in L^\infty(\mathbb{R}^2)$, so since $u\in L^p(\mathbb{R}^2)$ and $u\in \mathrm{W}(\mathbb{R}^2)\subset L^\infty(\mathbb{R}^2)$, $g$ is $\delbar^{-1}$ applied to a function that is in $L^{p'}(\mathbb{R}^2)$ for every $p'\ge p$.  It follows from \cite[Theorem 4.3.11]{AstalaIM09} that $g$ is continuous and tends to zero as $|z|\to\infty$, proving the asymptotic boundary condition \eqref{eq:f-norm}.  
Now, as $\del g = b+\ii v\in \mathrm{W}(\mathbb{R}^2)$, in particular $\del g$ is continuous.  Furthermore,
$\delbar g=\mathcal{B}^{-1}\del g$ so since $\mathcal{B}^{-1}$ maps $\mathrm{W}(\mathbb{R}^2)$ onto itself, $\delbar g$ is also in $\mathrm{W}(\mathbb{R}^2)$ and hence continuous.  It follows that $g$ is actually of class $C^1(\mathbb{R}^2)$, and so is $f=g+kz+\tfrac{1}{2}\ii S$.  Therefore $f$ is a classical solution of \eqref{eq:eikonal}.  Finally, the estimate \eqref{eq:f-estimate} follows from $\|b\|_\mathrm{W}\le B$ because $b=\partial f-k-\tfrac{1}{2}\ii v=\partial f-k-\tfrac{1}{2}\ii \partial S$.
As $b$ is the unique Wiener space solution of the fixed-point equation $b=F(b)$ with $\|b\|_\mathrm{W}\le B$, $f$ is the only classical solution of \eqref{eq:eikonal} satisfying the condition \eqref{eq:f-estimate}.
\end{proof}
Some comments:
\begin{itemize}
\item The lower bound \eqref{eq:mod-k-lower-bound} on $|k|$ that implies existence of a global solution depends on $B$, and it is attractive to try to choose $B$ in order to guarantee a solution for $|k|$ as small as possible.  The lower bound on $|k|$ is continuous with respect to $B$ and grows both as $B\downarrow \|v\|_\mathrm{W}$ and as $B\uparrow\infty$, guaranteeing a strictly positive minimum value depending only on $\|u\|_\mathrm{W}$ and $\|v\|_\mathrm{W}$. There exists a solution of the eikonal problem \eqref{eq:eikonal}--\eqref{eq:f-norm} with the desired asymptotics whenever $|k|$ exceeds this minimum value. 
When $v=0$, the lower bound for $|k|$ can be made as small as $\sqrt{\|u\|_\mathrm{W}}$ by taking the optimal value of $B=\tfrac{1}{2}\sqrt{\|u\|_\mathrm{W}}$.
\item The contraction mapping theorem guarantees that there is exactly one solution within the $B$-ball in $\mathrm{W}(\mathbb{R}^2)$.  There could in principle be other solutions as well, with larger Wiener norms.
\end{itemize}
Next, we consider the existence of the leading-order WKB amplitude $\alpha_0(x,y;k)$.  We will show that under the same conditions that a unique $f$ is determined for sufficiently large $|k|$, we also obtain a suitable function $\alpha_0$ solving \eqref{eq:alpha-equation-general}  under the boundary condition $\alpha_0\to 1$ as $|z|\to\infty$.  That this problem has a solution when $|k|$ is sufficiently large is the content of Theorem~\ref{theorem:alpha-solve} which we now prove.
\begin{proof}[Proof of Theorem~\ref{theorem:alpha-solve}]
Multiplying \eqref{eq:alpha-equation-general} by $A$ and using the eikonal equation \eqref{eq:eikonal} gives
\begin{equation}
(2\delbar f+\ii\delbar S)\left[A\partial(A\alpha_0) + (2\del f-\ii\del S)\delbar((2\del f-\ii\del S)\alpha_0)\right]=0.
\end{equation}
We can choose to satisfy this equation by equating the second factor to zero; multiplying through by $2\alpha_0$ (assuming $\alpha_0\neq 0$) we obtain the equation
\begin{equation}
\delbar((2\del f-\ii v)^2\alpha_0^2)+\del(u\alpha_0^2)=0,\quad u:=A^2\;\text{and}\;v:=\partial S.
\end{equation}
Now in terms of the quantity $b$ satisfying the fixed point equation $b=F(b)$ \eqref{eq:fixed-point-equation}--\eqref{eq:b-fixed-point} equivalent to the eikonal problem \eqref{eq:eikonal}--\eqref{eq:f-norm}, we have $2\del f-\ii v =2(k+b)$, so the equation for $\alpha$ can be written as
\begin{equation}
4\delbar((k+b)^2\alpha_0^2)+\partial(u\alpha_0^2)=0.
\end{equation}
Noting that since $b\in \mathrm{W}(\mathbb{R}^2)$ decays to zero as $|z|\to\infty$, we have $(k+b)^2\alpha^2\to k^2$ as $|z|\to\infty$,  and taking this into account we invert the operator $\delbar$ and obtain
\begin{equation}
4(k+b)^2\alpha_0^2 = 4k^2-\mathcal{B}(u\alpha^2).
\end{equation}
Now, to get into the Wiener space, we seek $\alpha_0^2$ in the form $\alpha_0^2=1+m$ with $m\in \mathrm{W}(\mathbb{R}^2)$.  Thus the problem becomes
\begin{equation}
m-\mathcal{K}m=h,\quad \mathcal{K}m:=-\frac{\mathcal{B}(um)}{4(k+b)^2},\quad h:=-\frac{8kb+4b^2+\mathcal{B}(u)}{4(k+b)^2}.
\end{equation}
Now observe that under the inequality \eqref{eq:mod-k-lower-bound}, we have $h\in \mathrm{W}(\mathbb{R}^2)$ with
\begin{equation}
\|h\|_\mathrm{W}\le \frac{8|k|B+4B^2+\|u\|_\mathrm{W}}{4(|k|-B)^2}.
\end{equation}
Also, since
\begin{equation}
\|\mathcal{K}m\|_\mathrm{W}\le \frac{\|u\|_\mathrm{W}}{4(|k|-B)^2}\|m\|_\mathrm{W},
\end{equation}
the inequality \eqref{eq:mod-k-lower-bound} implies that the operator norm of $\mathcal{K}$ on $\mathrm{W}(\mathbb{R}^2)$ satisfies $\|\mathcal{K}\|_\mathrm{W}<1$.  Hence $1-\mathcal{K}$ has a bounded inverse on $\mathrm{W}(\mathbb{R}^2)$ given by the Neumann series $1+\mathcal{K}+\mathcal{K}^2+\cdots$.
\end{proof}
We remark that this proof shows the bounded invertibility of the linear differential operator $\mathcal{L}$ defined in \eqref{eq:alpha-equation-general} on a space of functions whose squares differ from unity by a function in $\mathrm{W}(\mathbb{R}^2)$.

\subsection{Existence of $f(x,y;k)$ for $|z|$ sufficiently large given arbitrary $k\neq 0$.}
\label{sec:squash}
%
\begin{proof}[Proof of Theorem~\ref{theorem:squash}]  
%
Let $n\in \mathrm{W}(\mathbb{R}^2)$ be a function with compact support in the unit disk satisfying $n(0,0)=1$, and suppose that $\del n\in \mathrm{W}(\mathbb{R}^2)$ as well.  For $\rho>0$, denote by $n_\rho\in \mathrm{W}(\mathbb{R}^2)$ the function defined by
$n_\rho(x,y):=n(x/\rho,y/\rho)$.  Then for each $h\in \mathrm{W}(\mathbb{R}^2)$, $n_\rho h\to h$ in $\mathrm{W}(\mathbb{R}^2)$ as $\rho\to\infty$.
Indeed, we have 
\begin{equation}
\|n_\rho h-h\|_\mathrm{W}:=\iint_{\mathbb{R}^2}|\widehat{n_\rho h}(\xi_x,\xi_y)-\hat{h}(\xi_x,\xi_y)|\,\dd\xi_x\,\dd\xi_y
=\iint_{\mathbb{R}^2}|\widehat{n_\rho}\ast\hat{h}(\xi_x,\xi_y)-\hat{h}(\xi_x,\xi_y)|\,\dd\xi_x\,\dd\xi_y.
\label{eq:cutoff-estimate}
\end{equation}
Also, note that $\widehat{n_\rho}(\xi_x,\xi_y)=\rho^2\hat{n}(\rho\xi_x,\rho\xi_y)$ behaves as an approximate delta function when $\rho$ is large, having unit integral on $\mathbb{R}^2$ independently of $\rho$.  Since 
$\hat{n},\hat{h}\in L^1(\mathbb{R}^2)$, it follows from \eqref{eq:cutoff-estimate} that $\|n_\rho h-h\|_\mathrm{W}\to 0$ as $\rho\to\infty$; see \cite[Theorem 2.16]{LiebL01}.  
Therefore $(1-n_\rho)h\to 0$ in $\mathrm{W}(\mathbb{R}^2)$ as $\rho\to\infty$, and $(1-n_\rho)h$ agrees exactly with $h$ for $|z|>\rho$.

Given $k\neq 0$, we use the function $1-n_\rho$ for $\rho$ sufficiently large (given $k$) to modify the functions $u$ and $v$ appearing in the fixed-point iteration for \eqref{eq:eikonal}--\eqref{eq:f-norm} in such a way that the inequality \eqref{eq:mod-k-lower-bound} holds and therefore Theorem~\ref{theorem:eikonal-solve} applies to the modified $u$ and $v$.  Concretely, given $\rho$ we set
\begin{equation}
\tilde{u}:=(1-n_\rho)u\quad\text{and}\quad\tilde{S}:=(1-n_\rho)S.
\end{equation}
Recalling $v=\partial S$, the latter definition implies that 
\begin{equation}
\tilde{v}:=\del\tilde{S}=(1-n_\rho)v-\frac{1}{\rho}S\del n(x/\rho,y/\rho).
\end{equation}
Note that the second term above has a Wiener norm of order $O(\rho^{-1})$ because $\del n\in \mathrm{W}(\mathbb{R}^2)$ and $S$ differs from a function in $\mathrm{W}(\mathbb{R}^2)$ by a constant, so the claim follows from the scale invariance and Banach algebra properties of the Wiener space.
The value of $\rho$ will be chosen as follows.  Choose $B\in (0,\tfrac{1}{2}|k|)$.  Then take $\rho>0$ so large that 
$\|\tilde{u}\|_\mathrm{W}\le |k|B$ and $\|\tilde{v}\|_\mathrm{W}\le \tfrac{1}{2}B$.  
It then follows that $B>\|\tilde{v}\|_\mathrm{W}$, and that
\begin{equation}
B+\frac{1}{4}\cdot\frac{\|\tilde{u}\|_\mathrm{W}}{B-\|\tilde{v}\|_\mathrm{W}}<\frac{1}{2}|k|+\frac{1}{2}|k|= |k|
\end{equation}
and
\begin{equation}
B+\frac{1}{2}\sqrt{\|\tilde{u}\|_\mathrm{W}}<\frac{1}{2}|k| +\frac{1}{2\sqrt{2}}|k|<|k|
\end{equation}
so the inequality \eqref{eq:mod-k-lower-bound} holds true.  Therefore, by Theorem~\ref{theorem:eikonal-solve}, there is a unique global classical solution $\tilde{f}(x,y;k)$ of \eqref{eq:eikonal}--\eqref{eq:f-norm}, in which $u=A^2$ is replaced with $\tilde{u}$ and $S$ is replaced with $\tilde{S}$,
that satisfies $\|\partial \tilde{f}-k-\tfrac{1}{2}\ii\tilde{v}\|_\mathrm{W}\le B$.  Since $\tilde{u}(x,y)=u(x,y)=A(x,y)^2$ and $\tilde{S}(x,y)=S(x,y)$
both hold for $|z|>\rho$ due to the compact support in the unit disk of $n$, the construction of $f(x,y;k)$ given $k\neq 0$ is finished upon defining $f:=\tilde{f}$ for $|z|>\rho$.  According to Theorem~\ref{theorem:alpha-solve}, corresponding to $\tilde{f}$ defined on $\mathbb{R}^2$ there is a unique classical solution $\tilde{\alpha}_0$ of \eqref{eq:alpha-equation-general} with the appropriate substitutions for which $\tilde{\alpha}_0\to 1$ as $|z|\to\infty$, and defining $\alpha_0:=\tilde{\alpha}_0$ for $|z|>\rho$ finishes the proof.
\end{proof}

\subsection{Series solutions of the eikonal problem}\label{series}
Here, we develop a method based on infinite series that reproduces some of the above results by different means, and that can lead to an effective, sometimes explicit, solution of the eikonal problem.

\subsubsection{Series expansions of $f(x,y;k)$ for $S(x,y)=0$}
\label{sec:series-for-f}
Suppose that $S(x,y)\equiv 0$.  If also $A(x,y)\equiv 0$, then the exact solution of the eikonal problem \eqref{eq:eikonal}--\eqref{eq:f-norm} is  $f(x,y;k)=kz$ regardless of the value of $k\in\mathbb{C}$.  This fact suggests a perturbative approach to the latter problem in which, for fixed $k$, a measure of the amplitude $A(x,y)$ is taken to be the small parameter.  Such an approach is to be contrasted with that of Section~\ref{sec:k-large} in which for fixed $A$ and $S$, $k$ was taken to be a large parameter. 

Let $\delta>0$ be a parameter, and consider the $S\equiv 0$ form of the eikonal equation \eqref{eq:eikonal} in which $A^2/4$ is replaced with $\delta A^2$:
\begin{equation}
\overline{\partial}f(x,y;k)\partial f(x,y;k)=\delta A(x,y)^2,
\label{eq:fdeltaeqn}
\end{equation}
We try to solve \eqref{eq:fdeltaeqn} by a formal series
\begin{equation}
f(x,y;k)\sim kz+\sum_{n=1}^\infty\delta^nf_n(x,y;k),\quad \delta\to 0,
\label{eq:fseries}
\end{equation}
where the coefficient functions $f_n(x,y;k)$ are to be determined.  Since the leading term builds in the leading asymptotics of $f(x,y;k)$ for large $|z|$, we insist that $f_n(x,y;k)\to 0$ as $|z|\to\infty$ for all $n$ for consistency with \eqref{eq:f-norm}.  We intend to set $\delta=\tfrac{1}{4}$ once these have been determined and then assess the possible convergence of the series.

Substituting the series \eqref{eq:fseries} into \eqref{eq:fdeltaeqn} and collecting together the terms with the same powers of $\delta$ yields the following hierarchy of equations:
\begin{equation}
\delbar f_1(x,y;k)=\frac{1}{k}A(x,y)^2,
\label{eq:f1eqn}
\end{equation}
and
\begin{equation}
\delbar f_n(x,y;k)=-\frac{1}{k}\sum_{\ell=1}^{n-1}\delbar f_\ell(x,y;k)\del f_{n-\ell}(x,y;k),\quad n=2,3,4,\dots,
\label{eq:fneqn}
\end{equation}
The boundary condition $f_n(x,y;k)\to 0$ as $|z|\to\infty$ requires that we 
invert $\delbar$ on the right-hand side by the solid Cauchy transform \eqref{eq:solid-Cauchy}, however in certain situations the inversion can be carried out explicitly.
We will make this procedure effective in the special case that $A$ is a function with radial symmetry below in Section~\ref{sec:radial}.

Setting $u_n:=k^{-1}\delbar f_n$, the hierarchy \eqref{eq:f1eqn}--\eqref{eq:fneqn} becomes
\begin{equation}
u_1=k^{-2}A(x,y)^2,\quad u_n=-\sum_{\ell=1}^{n-1}u_\ell\mathcal{B}u_{n-\ell},\quad n=2,3,4,\dots.
\label{eq:un-recurrence}
\end{equation}
The space $\mathrm{W}(\mathbb{R}^2)$ is a convenient choice to analyze the terms $u_n$ for the same reasons as in the preceding study of the fixed point problem \eqref{eq:fixed-point-equation}--\eqref{eq:b-fixed-point}, namely the combination of nonlinearity with the presence of the Beurling transform $\mathcal{B}$ in the recurrence relation \eqref{eq:un-recurrence}.
Using the triangle inequality in the space $\mathrm{W}(\mathbb{R}^2)$ along with the Banach algebra property \eqref{eq:Banach-algebra}
and the identity \eqref{eq:Beurling-norm}, we then get
\begin{equation}
\|u_n\|_\mathrm{W}\le\sum_{\ell=1}^{n-1}\|u_\ell\|_\mathrm{W}\|u_{n-\ell}\|_\mathrm{W},\quad n=2,3,4,\dots.
\label{eq:pre-Catalan}
\end{equation} 
Now we renormalize $u_n$ as follows:  $u_n=\|u_1\|_\mathrm{W}^n v_n$, such that \eqref{eq:pre-Catalan} becomes
\begin{equation}
\|v_1\|_\mathrm{W}=1,\quad \|v_n\|_\mathrm{W}\le\sum_{\ell=1}^{n-1}\|v_\ell\|_\mathrm{W}\|v_{n-\ell}\|_\mathrm{W},\quad n=2,3,4,\dots.
\end{equation}
Recall the \emph{Catalan numbers} that satisfy the recurrence relation
\begin{equation}
C_{n}=\sum_{\ell=0}^{n-1}C_{\ell} C_{n-1-\ell},\quad n=1,2,3,4,\dots
\label{eq:CatalanRecurrence}
\end{equation}
subject to the initial condition $C_0=1$.  Explicitly, the Catalan numbers are given by the formula
\begin{equation}
C_n=\frac{(2n)!}{(n+1)!n!},\quad n\ge 0.
\label{eq:Catalan-numbers}
\end{equation}
From these definitions, we see that
\begin{equation}
\|v_n\|_\mathrm{W}\le C_{n-1}=\frac{(2n-2)!}{n!(n-1)!},\quad n\ge 1.
\label{eq:vn-estimate}
\end{equation}
Now we consider the convergence of the series (using $\delta=1/4$)
\begin{equation}
\sum_{n=1}^\infty\delta^n\delbar f_n = k\sum_{n=1}^\infty 4^{-n}u_n = k\sum_{n=1}^\infty (\tfrac{1}{4}\|u_1\|_\mathrm{W})^nv_n.
\label{eq:series-I}
\end{equation}
Since, by Stirling's formula, 
\begin{equation}
C_{n-1}=\frac{4^{n}}{4\sqrt{\pi}n^{3/2}}(1+O(n^{-1})),\quad n\to\infty,
\label{eq:Catalan-Stirling}
\end{equation}
the series \eqref{eq:series-I} is convergent in the space $\mathrm{W}(\mathbb{R}^2)$ provided that $\|u_1\|_\mathrm{W}\le 1$, i.e., that
\begin{equation}
|k|^2\ge \|A^2\|_\mathrm{W}.
\label{eq:k-assumption-series}
\end{equation}
Under this assumption on $|k|$, we then set
\begin{equation}
f(x,y;k)=kz + \delbar^{-1}\sum_{n=1}^\infty\delta^n\delbar f_n(x,y)=kz+k\delbar^{-1}\sum_{n=1}^\infty \frac{1}{4^n}u_n(x,y),
\label{eq:f-dbar-series}
\end{equation}
under the additional assumption that 
$\delbar^{-1}$ makes sense when applied to the particular element of $\mathrm{W}(\mathbb{R}^2)$ given by the convergent series.  
Note that the assumption \eqref{eq:k-assumption-series} on $k$ coincides with \eqref{eq:mod-k-lower-bound} in the case that $v=0$ and $B=\tfrac{1}{2}\sqrt{\|u\|_\mathrm{W}}$.  As has been pointed out, the latter is the optimal choice of $B$ given $v=0$ in \eqref{eq:mod-k-lower-bound}.

\subsubsection{Explicit inversion of $\delbar$ for $A(x,y)$ radially symmetric}
\label{sec:radial}
Specializing further, let us now suppose that $A(x,y)$ is a radially-symmetric function, that is,
\begin{equation}
A(x,y)=a(m),\quad m:=x^2+y^2=z\overline{z}
\label{eq:Aradial}
\end{equation}
for a suitable function $a:\mathbb{R}_+\to\mathbb{R}_+$.  We will show how in this case the iterative construction of series terms $f_n$ can be made explicit, avoiding the solution of partial differential equations or convolution with the Cauchy kernel (cf., \eqref{eq:solid-Cauchy}) at each order.

In Section~\ref{sec:Riccati} we will be interested in the solution of the eikonal problem \eqref{eq:eikonal}--\eqref{eq:f-norm} for radial phase-free potentials at $k=0$, so before implementing the series procedure described in Section~\ref{sec:series-for-f} we briefly discuss this special case.  With $S\equiv 0$ and $A$ given in the form \eqref{eq:Aradial}, observe that for $k=0$ one may
seek $f$ as a function of $m=x^2+y^2=z\overline{z}$ alone by writing $f(x,y;0)=F(m)$ by analogy with \eqref{eq:Aradial}.  The eikonal equation \eqref{eq:eikonal} for $S\equiv 0$ and $A$ of the form \eqref{eq:Aradial} then becomes simply
\begin{equation}
4mF'(m)^2=a(m)^2.
\label{eq:eikonal-S-zero-radial}
\end{equation}
This equation has two solutions that are smooth for all $m>0$ and that decay to zero as $m\to\infty$:
\begin{equation}
F(m)=\pm\frac{1}{2}\int_m^\infty\frac{a(\mu)}{\mu^{1/2}}\,\dd\mu=\pm\int_{m^{1/2}}^\infty a(s^2)\,\dd s.
\label{eq:eikonal-S-zero-radial-kzero-two-solutions}
\end{equation}
On the other hand, both of these solutions $f(x,y;0)=F(x^2+y^2)$ exhibit conical singularities at the origin $r=0$ unless $a(0)=0$.

Now we return to the series approach described in Section~\ref{sec:series-for-f}.
The equation \eqref{eq:f1eqn} for $f_1$ in the current setting reads
\begin{equation}
\delbar f_1=\frac{1}{k}a(z\overline{z})^2.
\end{equation}
This equation is easily integrated under the condition that $f_1$ should be smooth at the origin:
\begin{equation}
f_1=\frac{1}{kz}\int_0^{z\overline{z}}a(m)^2\,\dd m.
\end{equation}
Assuming that $a\in L^2(\mathbb{R}_+)$, we see easily that 
\begin{equation}
|f_1|\le\frac{\|a\|_2^2}{|kz|},
\end{equation}
an estimate that provides decay as $z\to\infty$.  Assuming also that $a$ is continuous down to $m=0$ shows that
\begin{equation}
f_1=\frac{a(0)^2}{k}\overline{z} +o(|z|),\quad z\to 0,
\end{equation}
indicating that $f_1$ is smooth near $z=0$ as well.  We next claim that for all $n=1,2,3,\dots$ it is consistent with \eqref{eq:f1eqn} and \eqref{eq:fneqn} to write $f_n$ in the form
\begin{equation}
f_n=\frac{G_n(m)}{(2n-1)(kz)^{2n-1}},\quad m=z\overline{z},
\label{eq:fnradialform}
\end{equation}
where $G_n$ is a smooth function.  (Precisely, the assertion is that $(2n-1)(kz)^{2n-1}f_n$ is a radial function of $(x,y)$, i.e., depending only on the product $m=z\overline{z}$.)  Indeed, this holds for $n=1$ with 
\begin{equation}
G_1(m):=\int_0^m a(\mu)^2\,\dd\mu.
\label{eq:G1general}
\end{equation}
Furthermore, substituting \eqref{eq:fnradialform} into \eqref{eq:fneqn} gives a recurrence relation on the functions $G_n$:
\begin{equation}
G_n'(m)=\sum_{\ell=1}^{n-1}K_{n\ell}\left[(2(n-\ell)-1)G_\ell'(m)G_{n-\ell}(m)-mG_\ell'(m)G_{n-\ell}'(m)\right],\quad n\ge 2,
\label{eq:Gnprimerecurrence}
\end{equation}
where
\begin{equation}
K_{n\ell}:=\frac{2n-1}{(2\ell-1)(2(n-\ell)-1)}.
\end{equation}
In order to ensure that $f_n$ is smooth at the origin, we need to insist that $G_n$ vanish at $m=0$, and so once $G_n'(m)$ is known from \eqref{eq:Gnprimerecurrence}, we obtain $G_n$ itself by
\begin{equation}
G_n(m)=\int_0^mG_n'(\mu)\,\dd\mu.
\label{eq:Gnconstant}
\end{equation}
This guarantees only that $G_n(0)=0$ but sufficiently high-order vanishing at $m=0$ will be required to cancel the factor of $z^{2n-1}$ in the denominator of $f_n$ as given by \eqref{eq:fnradialform}.  We will need $G_n(m)=O(m^{2n-1})$ as $m\to 0$ to have the necessary smoothness.  We will also need to avoid rapid growth in $G_n(m)$ as $m\to\infty$ in order that $f_n$ decay as $z\to\infty$.  Although there is no additional freedom available once the recurrence \eqref{eq:Gnprimerecurrence} is solved and the integration constant is determined by \eqref{eq:Gnconstant}, these additional properties of $G_n$ are indeed present as can be confirmed in examples, to which we now proceed.

\begin{myremark}
The form \eqref{eq:fnradialform} shows that, in polar coordinates $z=r\ee^{\ii\phi}$, $f_n=\tilde{f}_n(r)\ee^{-\ii(2n-1)\phi}$, and thus the infinite series $f(x,y;k)-kz=\sum_{n=1}^\infty \delta^n f_n(x,y;k)$
is nothing but a Fourier series consisting of only negative odd harmonics $\ee^{-\ii\phi}$, $\ee^{-3\ii\phi}$, $\ee^{-5\ii\phi}$, etc.  Another important observation clear from \eqref{eq:fnradialform} and the fact that $G_n$ is independent of $k$ is that $f(x,y;k)-kz$ is a power series in negative odd powers of $k$ with coefficients depending on $(x,y)\in\mathbb{R}^2$.  These observations lead to a numerical approach to the eikonal problem for radial potentials with $S(x,y)\equiv 0$ that will be explained in Section~\ref{sec:numerical-Fourier-series}.  It is also clear that it is the asymptotic behavior of $G_n(m)$ as $n\to\infty$ that determines for a given $|z|$ the minimum value of $|k|$ for which the series \eqref{eq:fseries} converges.  
\end{myremark}

\subsubsection{Example:  Gaussian amplitude}
\label{sec:Gaussian}
Suppose that $A(x,y)=\ee^{-(x^2+y^2)}$, which we can write in the form \eqref{eq:Aradial}
with $a(m)=\ee^{-m}$.  
Since the Fourier transform of $A(x,y)^2=\ee^{-2(x^2+y^2)}$ by the definition \eqref{eq:W-norm} is 
$\ee^{-|\xi|^2/8}/(8\pi)>0$ where $|\xi|^2:=\xi_x^2+\xi_y^2$, it is easy to compute the Wiener norm of $A^2$ and we hence conclude that the series \eqref{eq:series-I} is convergent in $\mathrm{W}(\mathbb{R}^2)$ provided $|k|\ge \sqrt{\|A^2\|_\mathrm{W}}=\sqrt{A(0,0)^2}=1$.  Later in Section~\ref{sec:numerical-Gaussian} we will see convincing numerical evidence that this condition on $k$ is not sharp, and that the related series \eqref{eq:f-dbar-series} is convergent in $L^\infty(\mathbb{R}^2)$ for $|k|\ge\tfrac{1}{2}$.

Let us illustrate the analytical calculation of the terms in the series for this case.  From \eqref{eq:fnradialform}--\eqref{eq:G1general} we have
\begin{equation}
G_1(m)=\int_0^m\ee^{-2\mu}\,\dd\mu=\frac{1}{2}\left[1-\ee^{-2m}\right]\quad\implies\quad
f_1=\frac{1-\ee^{-2z\overline{z}}}{2kz}.
\end{equation}
With $G_1$ determined, \eqref{eq:Gnprimerecurrence} for $n=2$ reads
\begin{equation}
\begin{split}
G_2'(m)&=3G_1'(m)G_1(m)-3mG_1'(m)^2\\
&=3\ee^{-2m}\frac{1}{2}\left[1-\ee^{-2m}\right]-3m\ee^{-4m}\\
&=
\frac{3}{2}\ee^{-2m}-\left[\frac{3}{2}+3m\right]\ee^{-4m},
\end{split}
\end{equation}
and hence using \eqref{eq:Gnconstant} we get
\begin{equation}
G_2(m)=\frac{3}{16}\left[1-4\ee^{-2m}+(3+4m)\ee^{-4m}\right]\quad\implies\quad f_2=\frac{1-4\ee^{-2z\overline{z}}+(3+4z\overline{z})\ee^{-4z\overline{z}}}{16(kz)^3}.
\end{equation}
It can be checked by Taylor expansion that $f_2$ is smooth at the origin and it decays as $z\to\infty$.

This procedure can be continued explicitly to arbitrary order because one needs only to be able to integrate in closed form expressions of the form $m^p\ee^{-2qm}$ for non-negative integers $p$ and $q$:
\begin{equation}
\int_0^m \mu^p\ee^{-2q\mu}\,\dd\mu=\frac{p!}{(2q)^{p+1}}\left(1-\ee^{-2qm}\sum_{\ell=0}^p\frac{(2qm)^\ell}{\ell !}\right).
\end{equation}
Unfortunately, it seems difficult to deduce a closed form expression for $G_n(m)$ for general $n\ge 2$ (and prove its correctness by an induction argument). Rather than proceed in this direction, we turn to another example of a radial amplitude function $A(x,y)$ for which this procedure yields dramatic results.

\subsubsection{Example:  Lorentzian amplitude}  
\label{sec:Lorentzian}
Suppose now that $A(x,y)=(1+x^2+y^2)^{-1}$, which can be written in the form \eqref{eq:Aradial} with
$a(m)=(1+m)^{-1}$.  Using the definition \eqref{eq:W-norm}, the Fourier transform of $A(x,y)^2$ in this case turns out to be $|\xi|K_1(|\xi|)/(4\pi)$ where $|\xi|:=\sqrt{\xi_x^2+\xi_y^2}$ and $K_1$ is a modified Bessel function of order $1$ \cite[\S10.25]{dlmf}.  By an integral representation formula \cite[Eqn.\@ 10.32.9]{dlmf} it is obvious that $K_1(|\xi|)>0$, so again it is easy to calculate the Wiener norm of $A^2$ and hence observe that the series \eqref{eq:series-I} converges in $\mathrm{W}(\mathbb{R}^2)$ whenever $|k|\ge\sqrt{\|A^2\|_\mathrm{W}}=\sqrt{A(0,0)^2}=1$.  Again, this condition is not sharp, and we will see so below, without the need to resort to numerics, by explicit calculation of the terms $f_n$ given by \eqref{eq:fnradialform}.  

Indeed, from \eqref{eq:fnradialform}--\eqref{eq:G1general} we have
\begin{equation}
G_1(m)=\int_0^m\frac{\dd\mu}{(1+\mu)^2}=1-\frac{1}{1+m}=\frac{m}{1+m}\quad\implies\quad
f_1=\frac{\overline{z}}{k(1+z\overline{z})},
\end{equation}
and we note that $f_1$ is smooth at the origin and decays as $z\to\infty$.  We next claim that for general $n\ge 2$, the recurrence \eqref{eq:Gnprimerecurrence} and the normalization condition \eqref{eq:Gnconstant} are satisfied by taking $G_n$ in the form
\begin{equation}
G_n(m)=C_{n-1}\left(\frac{m}{1+m}\right)^{2n-1},\quad n\ge 1,
\label{eq:GnLorentzianGeneralForm}
\end{equation}
where $C_0,C_1,C_2,\dots$ are suitably chosen constants.  Indeed $G_n(0)=0$ for all $n\ge 1$, 
so \eqref{eq:Gnconstant} is obviously satisfied regardless of the choice of the constants $\{C_k\}_{k=0}^\infty$.  Also, the form \eqref{eq:GnLorentzianGeneralForm} is clearly correct for $n=1$ with the choice $C_0=1$.  Moreover, substituting \eqref{eq:GnLorentzianGeneralForm} into \eqref{eq:Gnprimerecurrence} shows that \eqref{eq:GnLorentzianGeneralForm} is correct for general $n$, provided that the constants $\{C_k\}_{k=0}^\infty$ satisfy the recurrence \eqref{eq:CatalanRecurrence} together with 
the initial condition $C_0=1$; i.e., the constant $C_n$ is the $n^\text{th}$ Catalan number, which is explicitly given by \eqref{eq:Catalan-numbers}.
Therefore, $G_n(m)$ has been determined in closed form for all $n$, and it follows that 
\begin{equation}
f_n=\frac{C_{n-1}}{2n-1}\left(\frac{1}{k}\cdot\frac{\overline{z}}{1+z\overline{z}}\right)^{2n-1},
\quad 
n\ge 1.
\label{eq:fn-Lorentzian}
\end{equation}
Note that $f_n$ is smooth at the origin and decays as $z\to\infty$ for every $n\ge 1$.

With the terms $f_n$ all explicitly determined, we directly analyze the convergence of the formal series \eqref{eq:fseries} for $f$ with $\delta=\tfrac{1}{4}$. 
Noting that by \eqref{eq:Catalan-Stirling} we have 
\begin{equation}
\frac{C_{n-1}}{2n-1}=\frac{4^n}{8\sqrt{\pi}n^{5/2}}(1+O(n^{-1})),\quad n\to\infty, 
\label{stirling}
\end{equation}
we see that the series \eqref{eq:fseries} with $\delta=\tfrac{1}{4}$ converges exactly when
\begin{equation}
|k|\ge \frac{|z|}{1+|z|^2}
\label{eq:convergencecondition}
\end{equation}
and diverges otherwise.  Moreover, given any $\sigma>1$, the convergence is absolute and uniform for $k$ and $z$ satisfying the condition
\begin{equation}
|k|\ge \sigma\frac{|z|}{1+|z|^2}.
\end{equation}
Since the function on the right-hand side of \eqref{eq:convergencecondition} achieves its maximum value of $1/2$ at $|z|=1$, we learn that if $|k|\ge 1/2$ the series on the right-hand side of \eqref{eq:fseries} converges uniformly on $\mathbb{R}^2$ to a continuous function vanishing at infinity.  


Proceeding further, the infinite series on the right-hand side of \eqref{eq:fseries} can be summed in closed form \cite{Mathematica} for those $k$ and $z$ for which it converges, yielding the explicit formula \eqref{eq:fW-Lorentzian}, in which the square root and the arcsin are both given by  principal branches.  This explicit expression for $g=f-kz$, which was originally defined as a function of $W$ by a power series convergent for $|W|<1$, defines an analytic continuation from the unit disk in the $W$-plane   to the whole complex $W$-plane with the exception of two slits joining the points $W=\pm 1$ to infinity (we may choose the branch cuts to be the real intervals $-\infty<W\le -1$ and $1\le W<+\infty$).  Moreover, one can directly check that regardless of whether $W$ is inside or outside of the unit disk, the explicit expression for $f(x,y;k)$ is an exact solution of the equation \eqref{eq:eikonal} in the case $S\equiv 0$ when $A(x,y)=(1+x^2+y^2)^{-1}$.

In this case, we can also solve explicitly for the scalar coefficient $\alpha_0(x,y;k)$, which completes the construction of the leading term $\boldsymbol{\phi}^{(0)}(x,y;k)$ in the WKB expansion.
By direct calculation using the definition of $W$ given in \eqref{eq:fW-Lorentzian},
\begin{equation}
\del W=-kW^2\quad\text{and}\quad \overline{z}^2\delbar W=kW^2.
\label{eq:W-derivs}
\end{equation}
Hence
\begin{equation}
\del f = \frac{k}{2}\left(1+(1-W^2)^{1/2}\right),
\label{eq:df-Lorentzian}
\end{equation}
and from the relevant eikonal equation $\del f\cdot\delbar f=\tfrac{1}{4}(1+z\overline{z})^{-2}$ we get
\begin{equation}
\overline{z}^2\delbar f=\frac{k}{2}\left(1-(1-W^2)^{1/2}\right).
\label{eq:dbarf-Lorentzian}
\end{equation}
Writing \eqref{eq:alpha-equation-general} in the special case of $S\equiv 0$ (and hence $w=0$) gives
\begin{equation}
A\delbar (\del f\cdot\alpha_0) + \delbar f\cdot\del (A\alpha_0)=0
\label{eq:alpha-equation-S-zero}
\end{equation}
as the equation to be solved by $\alpha_0(x,y;k)$ under the condition $\alpha_0(x,y;k)\to 1$ as $|z|\to\infty$.  Since $A=kW/\overline{z}$ according to \eqref{eq:fW-Lorentzian}, and $\del f$ and $\delbar f$ are given by \eqref{eq:df-Lorentzian}--\eqref{eq:dbarf-Lorentzian}, \eqref{eq:alpha-equation-S-zero} can be written as
\begin{equation}
W\overline{z}^2\delbar\left(\left(1+(1-W^2)^{1/2}\right)\alpha_0\right)+\left(1-(1-W^2)^{1/2}\right)\del (W\alpha_0)=0
\label{eq:alpha-PDE-Lorentzian}
\end{equation}
where we have used $k\neq 0$.  Now using \eqref{eq:W-derivs} it is clear that there is a solution of the form $\alpha_0=\alpha_0(W)$, i.e., that $\alpha_0$ depends on $(x,y)$ only via $W$.  Indeed, by the chain rule, the ansatz $\alpha_0=\alpha_0(W)$ in \eqref{eq:alpha-PDE-Lorentzian} leads to the ordinary differential equation
\begin{equation}
W\frac{d}{dW}\left(\left(1+(1-W^2)^{1/2}\right)\alpha_0(W)\right)-\left(1-(1-W^2)^{1/2}\right)\frac{\dd}{\dd W}(W\alpha_0(W))=0
\end{equation}
after canceling $kW^2$.  This can be rewritten in the equivalent form
\begin{equation}
\frac{\dd}{\dd W}\log\left(W(1-W^2)^{1/2}\alpha_0(W)^2\right)=W^{-1}(1-W^2)^{-1/2}=\frac{\dd}{\dd W}\log\left(\frac{W}{1+(1-W^2)^{1/2}}\right).
\end{equation}
Integrating, exponentiating, and solving for $\alpha_0$ gives
\begin{equation}\label{alphaexact}
\alpha_0(W)=C\left((1-W^2)^{1/2}(1+(1-W^2)^{1/2})\right)^{-1/2}
\end{equation}
where $C$ is an integration constant.  Since $|z|\to\infty$ means $W\to 0$, we need $C=\sqrt{2}$ to have $\alpha_0\to 1$ as $|z|\to\infty$, which gives \eqref{eq:alphaW-Lorentzian}.  

\begin{myremark}
For this example, we can explain the gap between the general sufficient condition for convergence, namely $|k|\ge\sqrt{\|A^2\|_\mathrm{W}}$ which works out to $|k|\ge 1$ in this case, and the actual condition $|k|\ge 1/2$ obtained by direct analysis of the explicit terms in the series.  Indeed, it is easy to check that when $f_n$ is given by \eqref{eq:fn-Lorentzian}, the corresponding functions $u_n:=k^{-1}\delbar f_n$ satisfy the identity $\mathcal{B}u_n=-\overline{z}^2 u_n$ for $n=1,2,3,\dots$.  If this specialized information is used in \eqref{eq:un-recurrence}, the recurrence becomes
\begin{equation}
u_1=k^{-2}A(x,y)^2,\quad u_n=\overline{z}^2\sum_{\ell=1}^{n-1}u_\ell u_{n-\ell}.
\end{equation}
Therefore, introducing $w_n:=\overline{z}^2 u_n$, we get a corresponding recurrence for $\{w_n\}_{n=1}^\infty$:
\begin{equation}
w_1=k^{-2}\overline{z}^2A(x,y)^2,\quad w_n=\sum_{\ell=1}^{n-1}w_\ell w_{n-\ell}.
\end{equation}
This recurrence relation can be studied in exactly the same way as \eqref{eq:un-recurrence}; one introduces $v_n$ by the rescaling $w_n=\|w_1\|^nv_n$ (here we can use the $L^\infty(\mathbb{R}^2)$ norm in place of the Wiener norm if desired because we need only the Banach algebra property having dispensed with the Beurling transform), and obtains the estimate \eqref{eq:vn-estimate}.  Hence the condition for convergence of the series \eqref{eq:series-I} now takes the form $\|w_1\|\le 1$.   Since by comparison with $u_1$, $w_1$ contains the additional factor of $\overline{z}^2$, it is easy to check that whereas $\|u_1\|_\mathrm{W}=\|u_1\|_\infty\le 1$ reads $|k|\ge 1$, the condition $\|w_1\|_\infty\le 1$ reads $|k|\ge\tfrac{1}{2}$.
\end{myremark}

\section{A Specialized Method for Radial Potentials with $S\equiv 0$ and $k=0$}
\label{sec:Riccati}
Suppose $S\equiv 0$, and fix the spectral parameter to be $k=0$.  It is well known that in this case the corresponding Zakharov-Shabat scattering problem \eqref{eq:ZS-system} that arises in the one-dimensional setting with $\lambda=0$ corresponding to $k=0$, namely $\epsilon\boldsymbol{\psi}'(x)=A(x)\sigma_1\boldsymbol{\psi}(x)$,  
can be solved explicitly by introducing a new coordinate $m$ satisfying $m'(x)=A(x)>0$.  Indeed, this monotone change of independent variable reduces the problem to the constant-coefficient system $\epsilon\boldsymbol{\psi}'(m)=\sigma_1\boldsymbol{\psi}(m)$.
Unfortunately, similar reasoning fails in the setting of the two-dimensional Davey-Stewartson scattering problem \eqref{eq:1.10}.  

In this section, we further assume that $A(x,y)$ is a function with radial symmetry, i.e., depending only on $|z|$, and show how the use of polar coordinates can be used to reduce the scattering problem to the study of a suitable ordinary differential equation.  We then study this equation in the semiclassical limit and obtain a formula for the reflection coefficient in this special case.  

We begin by writing the scattering problem \eqref{eq:1.10} in polar coordinates $(r,\phi)$, where $z=x+\ii y=r\ee^{\ii\phi}$ and $\overline{z}=r\ee^{-\ii\phi}$.  In polar coordinates, the operators defined by \eqref{eq:d-dbar} take the
form
\begin{equation}
\partial = \frac{\ee^{-\ii\phi}}{2r}\left(r\frac{\partial}{\partial r}-\ii\frac{\partial}{\partial\phi}\right)\quad\text{and}\quad
\overline{\partial}=\frac{\ee^{\ii\phi}}{2r}\left(r\frac{\partial}{\partial r} + \ii\frac{\partial}{\partial\phi}\right).
\label{eq:d-dbar-polar}
\end{equation}
Therefore, with $S\equiv 0$ and $A=A(r)$ being a smooth function with $A'(0)=0$, \eqref{eq:1.10} becomes
\begin{equation}
\begin{split}
\epsilon\frac{\ee^{\ii\phi}}{r}\left(r\psi_{1r} + \ii\psi_{1\phi}\right)&=A(r)\psi_2\\
\epsilon\frac{\ee^{-\ii\phi}}{r}\left(r\psi_{2r}-\ii\psi_{2\phi}\right)&=A(r)\psi_1.
\end{split}
\label{eq:dbar-system-real-radial}
\end{equation}
It is then convenient to introduce new dependent variables by $w_1:=\psi_1$ and $w_2=\overline{z}\psi_2$, so that the system takes the form
\begin{equation}
\begin{split}
\epsilon rw_{1r}+\ii\epsilon w_{1\phi}&=A(r)w_2\\
\epsilon rw_{2r}-\ii\epsilon w_{2\phi}&=r^2A(r)w_1.
\end{split}
\end{equation}
If $k=0$, then from \eqref{eq:psi-norm}--\eqref{eq:r-def}, we see that the solution we seek has the property that $w_1\to 1$ and $w_2\to\tfrac{1}{2}\overline{R_0^\epsilon(0)}$ as $r\to\infty$, thereby recovering the reflection coefficient evaluated at the origin.  Implicit is the assumption that $w_j$ are smooth functions on the plane.  We claim that in this situation, $w_j=w_j(r)$ are purely radial functions, reducing the problem to the study of the linear ordinary differential equations
\begin{equation}
\begin{split}
\epsilon r\frac{\dd w_1}{\dd r}&=A(r)w_2\\
\epsilon r\frac{\dd w_2}{\dd r}&=r^2A(r)w_1.
\end{split}
\label{eq:radial-ODEs-kzero}
\end{equation}
By the method of Frobenius, one can see that this system
has a one-dimensional space of solutions that are bounded with zero derivative at $r=0$, which is a regular singular point.  Indeed, assuming that $A(r)=A(0)+O(r^2)$ as $r\downarrow 0$, the system \eqref{eq:radial-ODEs-kzero} can be written in the form
\begin{equation}
\frac{\dd}{\dd r}\begin{bmatrix}w_1\\w_2\end{bmatrix}=\left(\frac{1}{r}  \begin{bmatrix}
0 & \epsilon^{-1}A(0)\\0&0\end{bmatrix}+O(1)\right)\begin{bmatrix}w_1\\w_2\end{bmatrix},\quad r\downarrow 0
\end{equation}
and hence the only indicial exponent for the origin is zero with non-diagonalizable coefficient matrix; therefore every solution is a linear combination of a solution analytic at $r=0$ proportional there to the nullvector $[1,0]^\top$ and a second independent solution that diverges logarithmically at the origin.  We can attempt to find $R_0^\epsilon(0)$ by normalizing an element of this subspace of solutions regular at the origin so that $w_1\to 1$ as $r\to\infty$.  Alternatively, we can take any nonzero element of this subspace and obtain $R_0^\epsilon(0)$ by the formula
\begin{equation}
R_0^\epsilon(0)=2\lim_{r\to\infty}\overline{w_2(r)}/\overline{w_1(r)}.
\label{eq:r-kzero-w1w2}
\end{equation}

\subsection{Riccati equation.  Formal asymptotic analysis}
\label{sec:Riccati-formal}
The formula \eqref{eq:r-kzero-w1w2} in turn motivates us to study the Riccati equation for $Q:=w_2/w_1$ implied by the coupled linear system \eqref{eq:radial-ODEs-kzero} for $w_j(r)$:
\begin{equation}
\epsilon \frac{\dd Q}{\dd r}=\frac{A(r)}{r}\left(r^2-Q^2\right).
\label{eq:Riccati}
\end{equation}
If \eqref{eq:Riccati} is solved subject to the initial condition $Q(r)=O(r^2)$ as $r\downarrow 0$ (corresponding to the regular subspace at the origin for \eqref{eq:radial-ODEs-kzero}), then the reflection coefficient $R_0^\epsilon(0)$ may be found as $R_0^\epsilon(0)=2\lim_{r\to\infty}Q(r)$ (using the fact that $Q(r)$ is real-valued).  Equivalently, we may introduce $X(r):=Q(r)/r$, which satisfies
\begin{equation}
\begin{split}
\epsilon \frac{\dd X}{\dd r}&= -A(r)X^2-\frac{\epsilon}{r}X + A(r)\\
&=-A(r)[X-X_+(r;\epsilon)][X-X_-(r;\epsilon)],\quad X(r)=O(r),\quad r\downarrow 0,
\end{split}
\label{eq:Riccati-2}
\end{equation}
for 
\begin{equation}
X_\pm(r;\epsilon):=\frac{1}{2A(r)}\left[-\frac{\epsilon}{r}\pm\sqrt{\frac{\epsilon^2}{r^2}+4A(r)^2}\right],
\label{eq:nullclines}
\end{equation}
and from which one obtains $R_0^\epsilon(0)$ by
\begin{equation}
R_0^\epsilon(0)=2\lim_{r\to\infty}rX(r).
\label{eq:r-zero-from-X}
\end{equation}
Note that, given the solution $X(r;\epsilon)$ of \eqref{eq:Riccati-2}, the solution of the original system \eqref{eq:dbar-system-real-radial} with the boundary conditions $\psi_1\to 1$ and $\psi_2=O(1/r)$ as $r\to\infty$ is given explicitly by
\begin{equation}
\begin{bmatrix}
\psi_1\\\psi_2\end{bmatrix}
=\tilde{\alpha}_0\begin{bmatrix}1\\ \ee^{\ii\phi}X(r;\epsilon)\end{bmatrix}\ee^{f/\epsilon},\quad\text{where}\quad
f=-\int_r^{+\infty}A(r')\,\dd r'
\label{eq:radial-kzero-psi-from-Riccati}
\end{equation}
and
\begin{equation}
\tilde{\alpha}_0=\exp\left(\frac{1}{\epsilon}\int_r^{+\infty}(1-X(r';\epsilon))A(r')\,\dd r'\right).
\label{eq:radial-kzero-psi-from-Riccati-2}
\end{equation}

Suppose that $A(r)$ is nonincreasing.  The nullclines for \eqref{eq:Riccati-2} are given by $X=X_\pm(r;\epsilon)$ (cf., \eqref{eq:nullclines}).  We have $X_+(r;\epsilon)>0>X_-(r;\epsilon)$, and $\dd X/\dd r>0$ for $X_-(r;\epsilon)<X<X_+(r;\epsilon)$ while $\dd X/\dd r<0$ if either $X>X_+(r;\epsilon)$ or $X<X_-(r;\epsilon)$.  The nullclines have the following asymptotic behavior for small $\epsilon$:
\begin{itemize}
\item If $r\ll\epsilon$, then $X_+(r;\epsilon)=A(0)r/\epsilon + O((r/\epsilon)^3)$ while $X_-(r;\epsilon)=-\epsilon/(A(0)r)+O(r/\epsilon)$.
\item If $\epsilon\ll r$ and $rA(r)\gg\epsilon$, then $X_\pm(r;\epsilon)=\pm 1 + o(1)$.
\item If $\epsilon\ll r$ and $rA(r)\ll\epsilon$, then $X_+(r;\epsilon)=[rA(r)/\epsilon](1+o(1))$ while $X_-(r;\epsilon)=-[\epsilon/(rA(r))](1+o(1))$.
\end{itemize}
Since $X=X(r;\epsilon)$ tends to zero as $r\downarrow 0$ for fixed $\epsilon$, only the nullcline $X_+(r,\epsilon)$ plays any role for small $r$ (given $\epsilon>0$ small).  Moreover, since $\dd X/\dd r$ is explicitly proportional to $\epsilon^{-1}$, $X(r;\epsilon)$ will very rapidly approach a small neighborhood of the nullcline $X=X_+(r;\epsilon)$ as $r$ increases; therefore for moderate values of $r$ in the regime where $r\gg\epsilon$ but $rA(r)\gg\epsilon$ (the latter condition avoiding the ``tail'' of the amplitude function $A(r)$) we will have $X(r;\epsilon)\approx X_+(r;\epsilon)\approx 1$ for small $\epsilon$.  
On the other hand, when $A(r)$ becomes small as $r$ increases, then \eqref{eq:Riccati-2} can be approximated by the linear equation
\begin{equation}
\epsilon\frac{\dd X}{\dd r}= -\frac{\epsilon}{r}X\quad\text{with general solution}\quad X(r;\epsilon)=\frac{C(\epsilon)}{r}.
\end{equation}
This approximation is exact wherever $A(r)\equiv 0$.
The constant $C(\epsilon)$ can be determined by matching the approximate solution $X(r;\epsilon)\approx C(\epsilon)/r$ onto the approximation $X(r;\epsilon)\approx 1$ at an appropriate value of $r$, say $r=r_\mathrm{Match}$.  If $A(r)$ has compact support, then we take the breakpoint $r_\mathrm{Match}$ to be the positive support endpoint; otherwise we take the breakpoint $r=r_\mathrm{Match}$ to be the root of the equation $rA(r)=\epsilon$ that is not small as $\epsilon\downarrow 0$.  In the latter case, $r_\mathrm{Match}\to\infty$ as $\epsilon\downarrow 0$ because $A$ is nonincreasing and $A(r)\to 0$ as $r\to\infty$.
Given $\epsilon\ll 1$ and the corresponding value of $r_\mathrm{Match}(\epsilon)>0$, we then determine $C=C(\epsilon)$ by setting $C/r_\mathrm{Match}=1$.  See Figures~\ref{fig:GaussianRiccati}--\ref{fig:TopHatRiccati} for further understanding of the solutions of the Riccati equation \eqref{eq:Riccati-2} and their relation to the nullcline $X=X_+(r;\epsilon)$ as $\epsilon$ decreases toward zero.%
\begin{figure}[h]
\begin{center}
\includegraphics{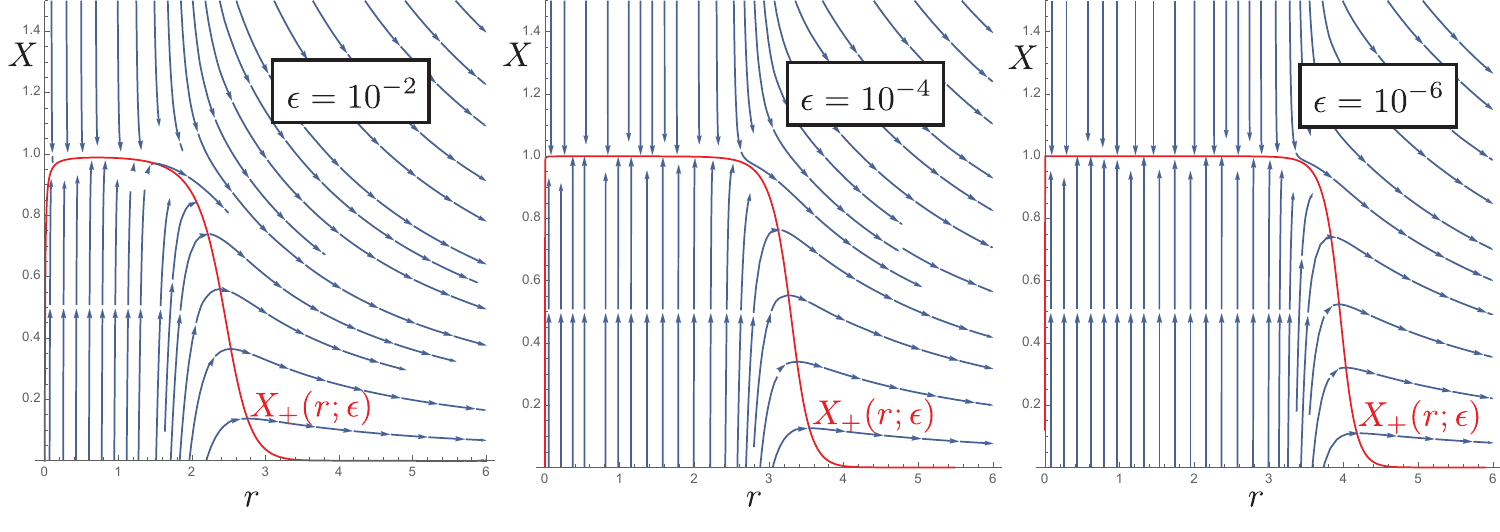}
\end{center}
\caption{The direction field of the Riccati equation \eqref{eq:Riccati-2} and its relation to the nullcline $X_+(r;\epsilon)$ in the case of a Gaussian amplitude $A(r)=\ee^{-r^2}$.  For small $\epsilon$, the solution $X(r;\epsilon)$ departs from the nullcline $X=X_+(r;\epsilon)$ near its ``shoulder,'' a feature that is increasingly well-defined as $\epsilon\to 0$ and is asymptotically located at $r=r_\mathrm{Match}(\epsilon)$.}
\label{fig:GaussianRiccati}
\end{figure}
\begin{figure}[h]
\begin{center}
\includegraphics{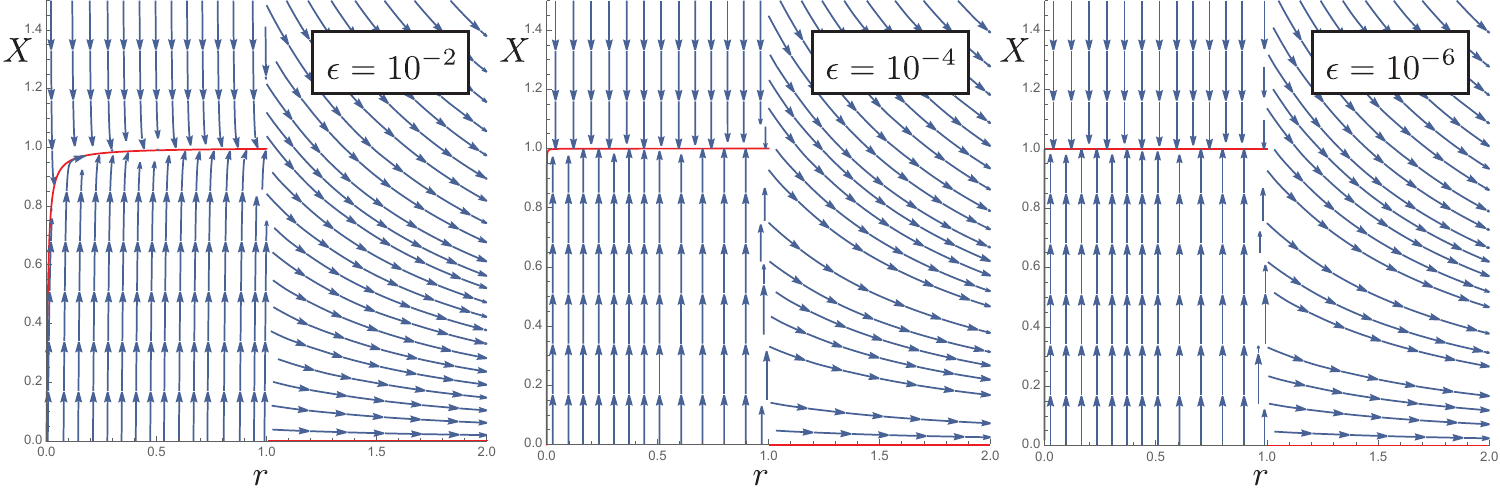}
\end{center}
\caption{The same as Figure~\ref{fig:GaussianRiccati} except for the potential $A(r)=\chi_{r\le 1}(r)$. As in Figure~\ref{fig:GaussianRiccati}, the (here, discontinuous) red curve is the nullcline $X=X_+(r;\epsilon)$. In this case for $r>1$ we have $X_+(r;\epsilon)\equiv 0$ and $X(r;\epsilon)=C/r$ exactly.}
\label{fig:TopHatRiccati}
\end{figure}
For the Gaussian example $A(r)=\ee^{-r^2}$, the implications of the behavior of $X(r;\epsilon)$ can be seen also in numerical solutions at $k=0$ of the direct spectral problem \eqref{eq:1.10}--\eqref{eq:psi-norm} carried out using the method described below in Section~\ref{sec:numerical-dbar}.  See Figure
~\ref{fig:Gaussianpsi}.
\begin{figure}[htb!]
\begin{center}
\includegraphics[width=0.32\textwidth]{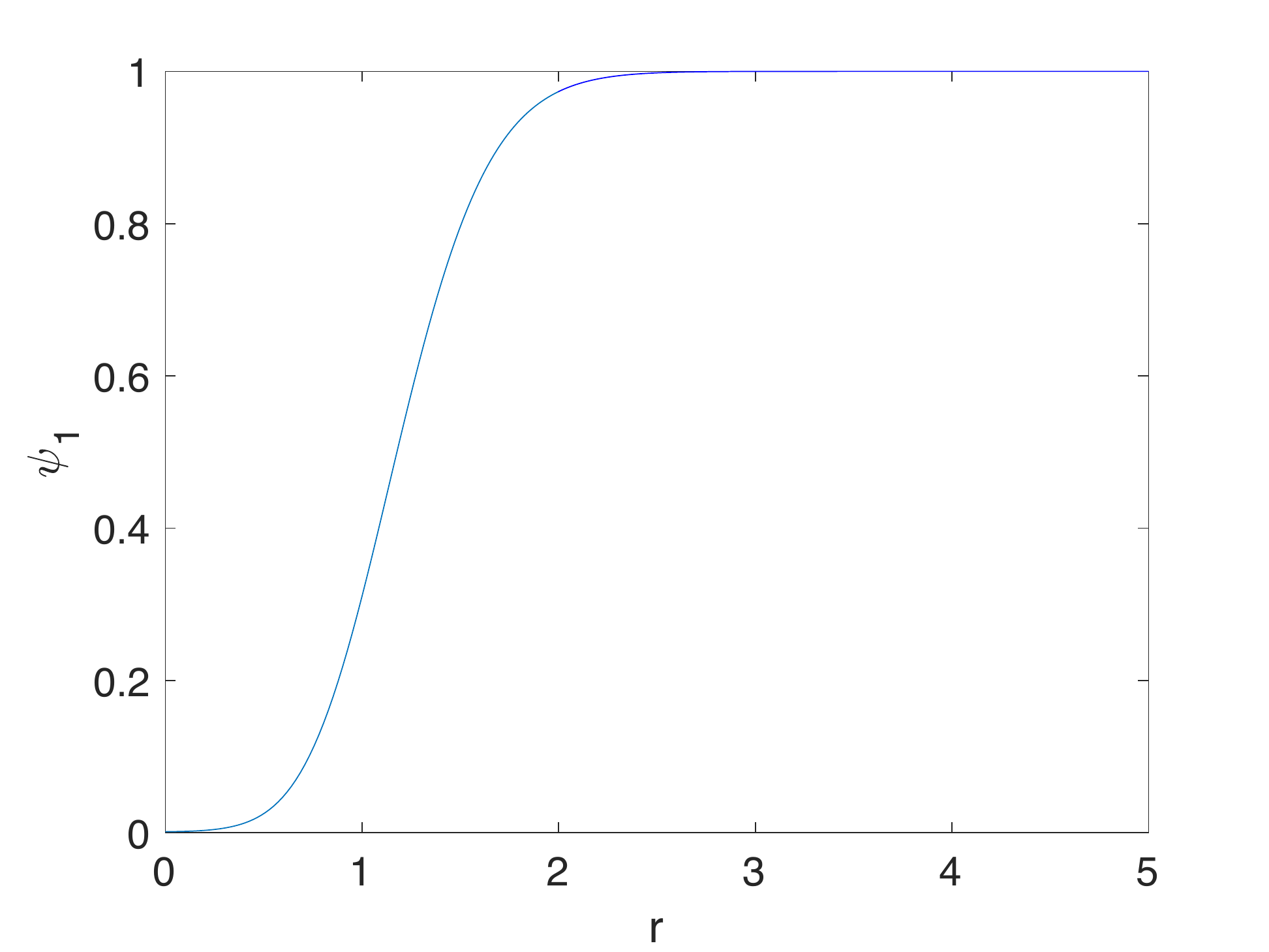}
\includegraphics[width=0.32\textwidth]{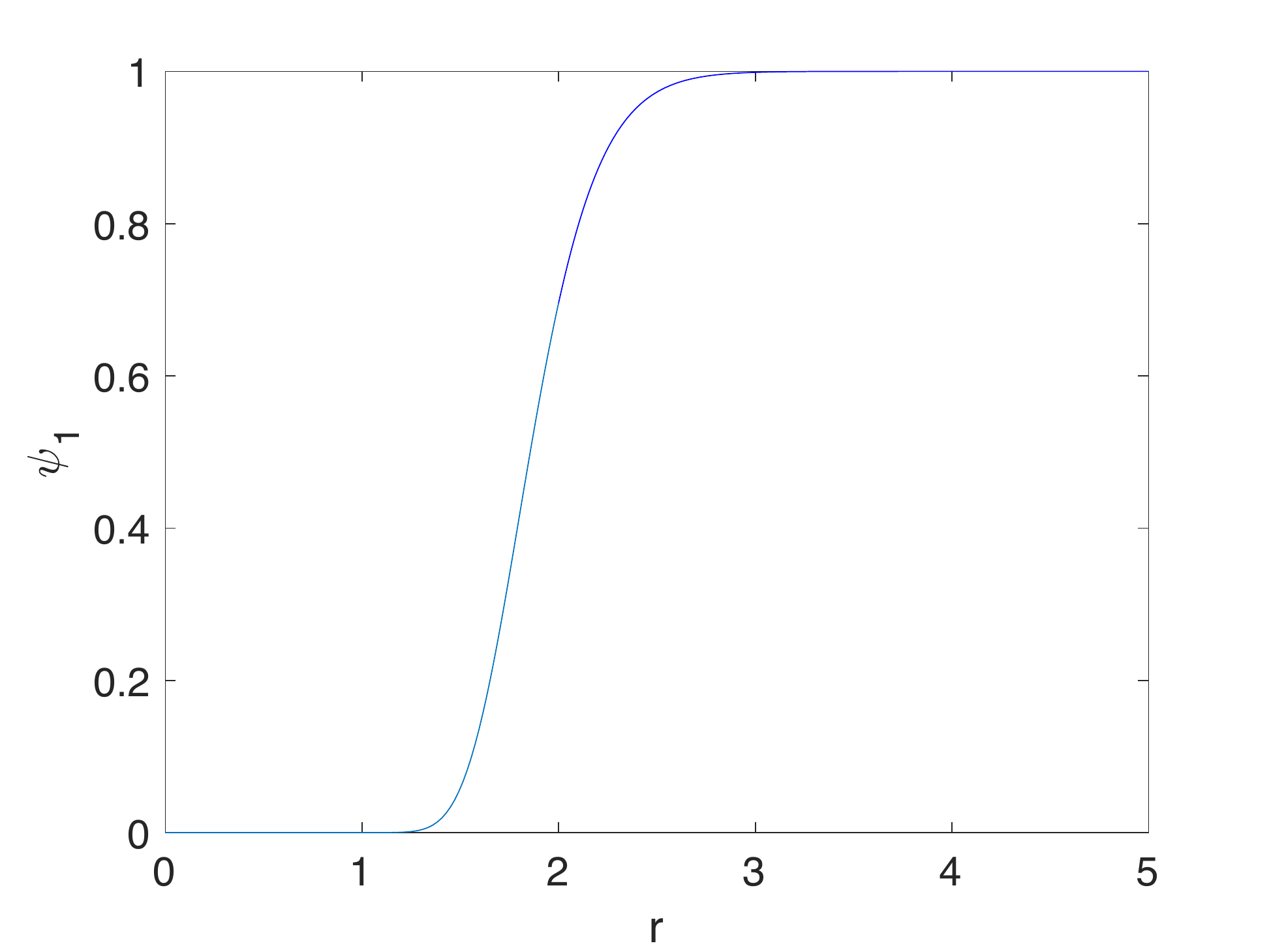}
\includegraphics[width=0.32\textwidth]{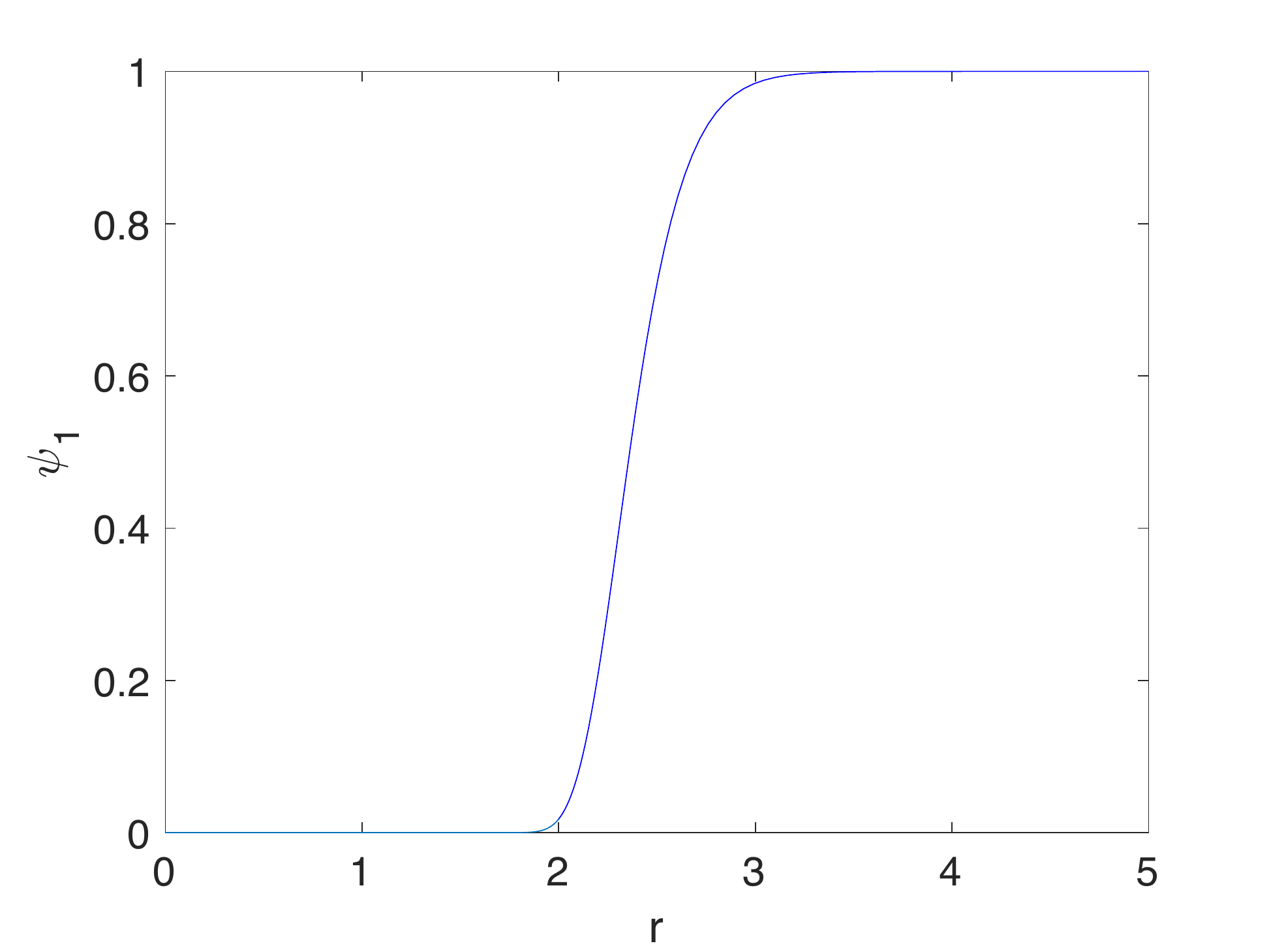}\\
\includegraphics[width=0.32\textwidth]{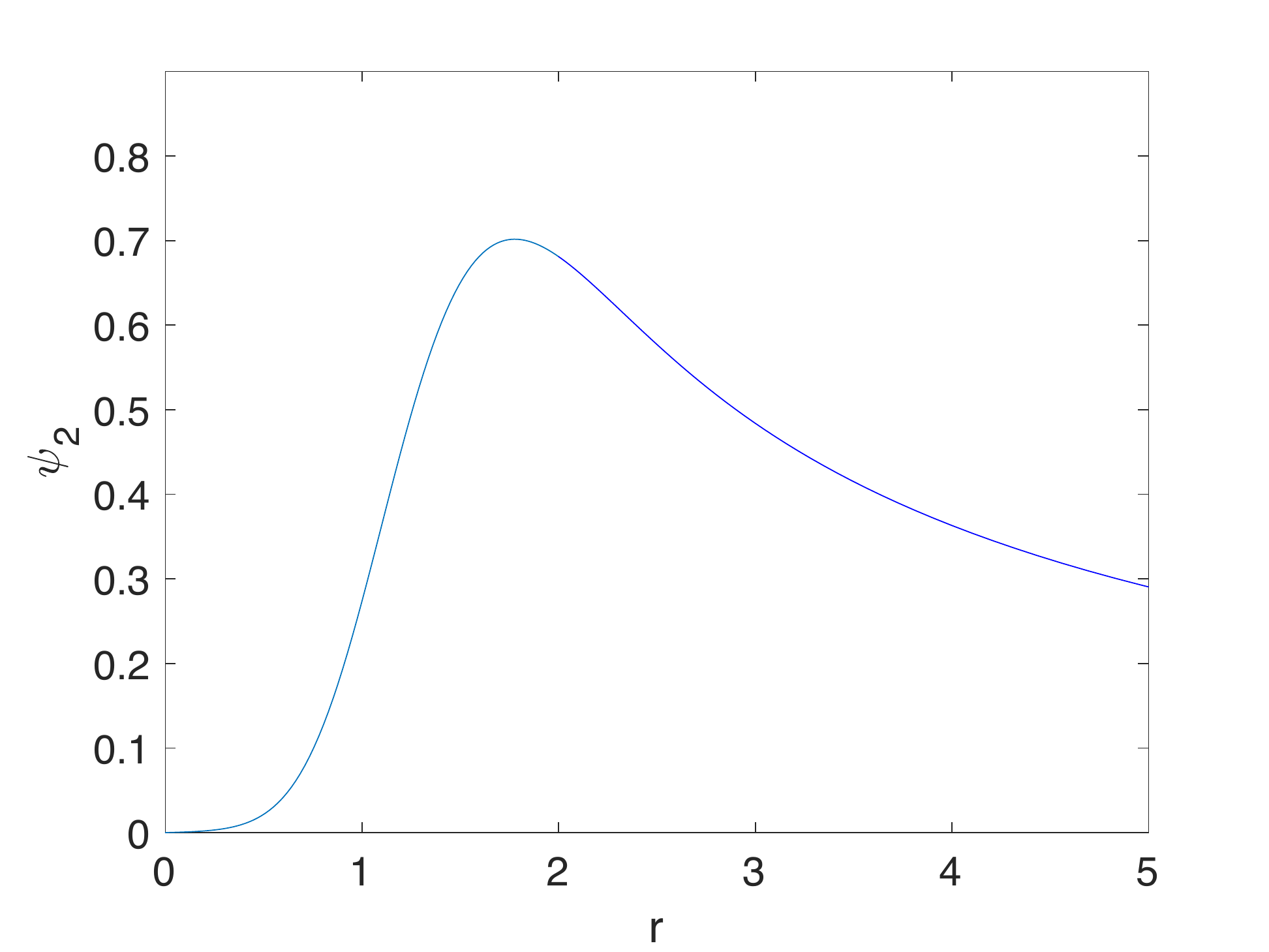}
\includegraphics[width=0.32\textwidth]{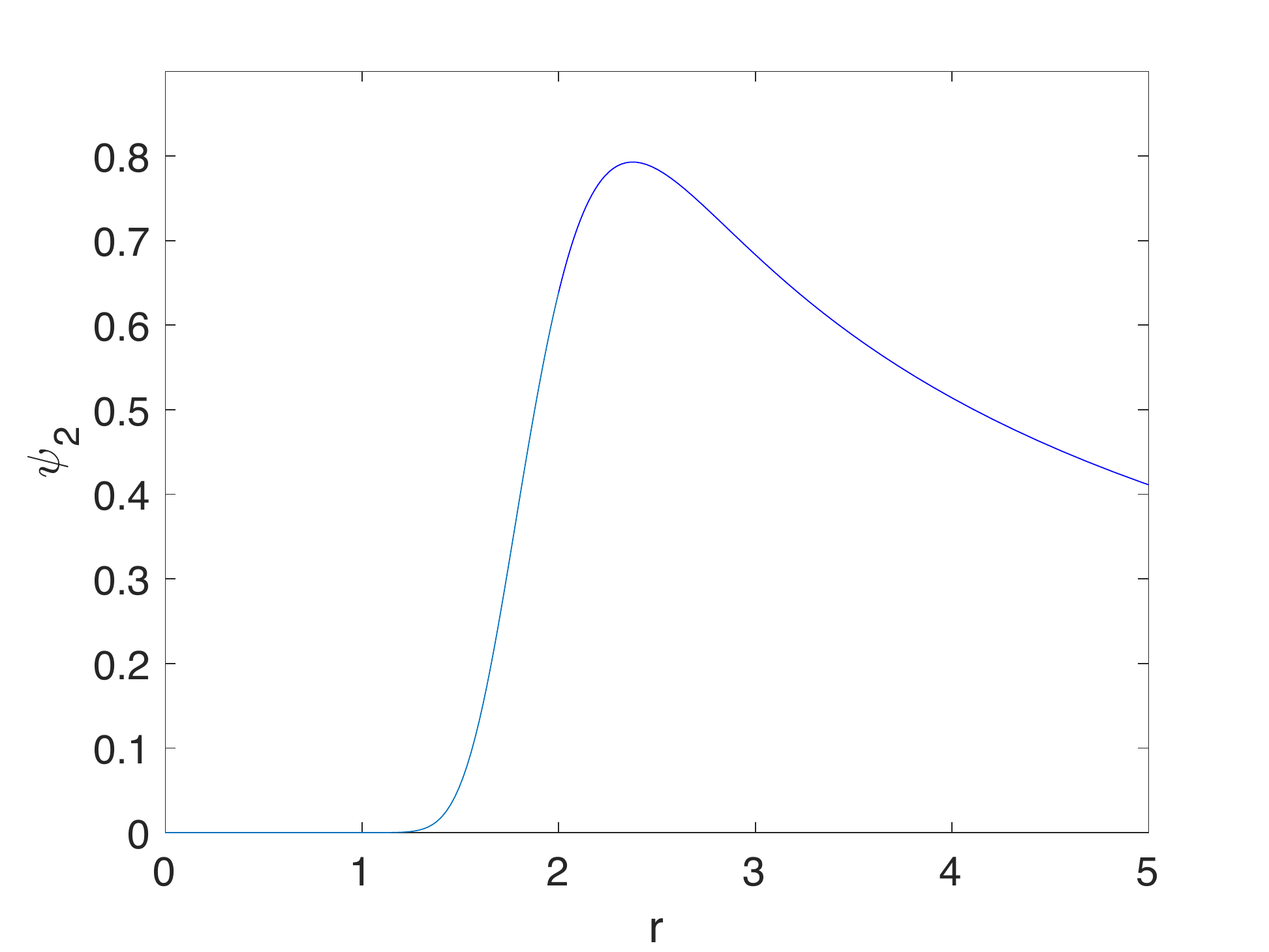}
\includegraphics[width=0.32\textwidth]{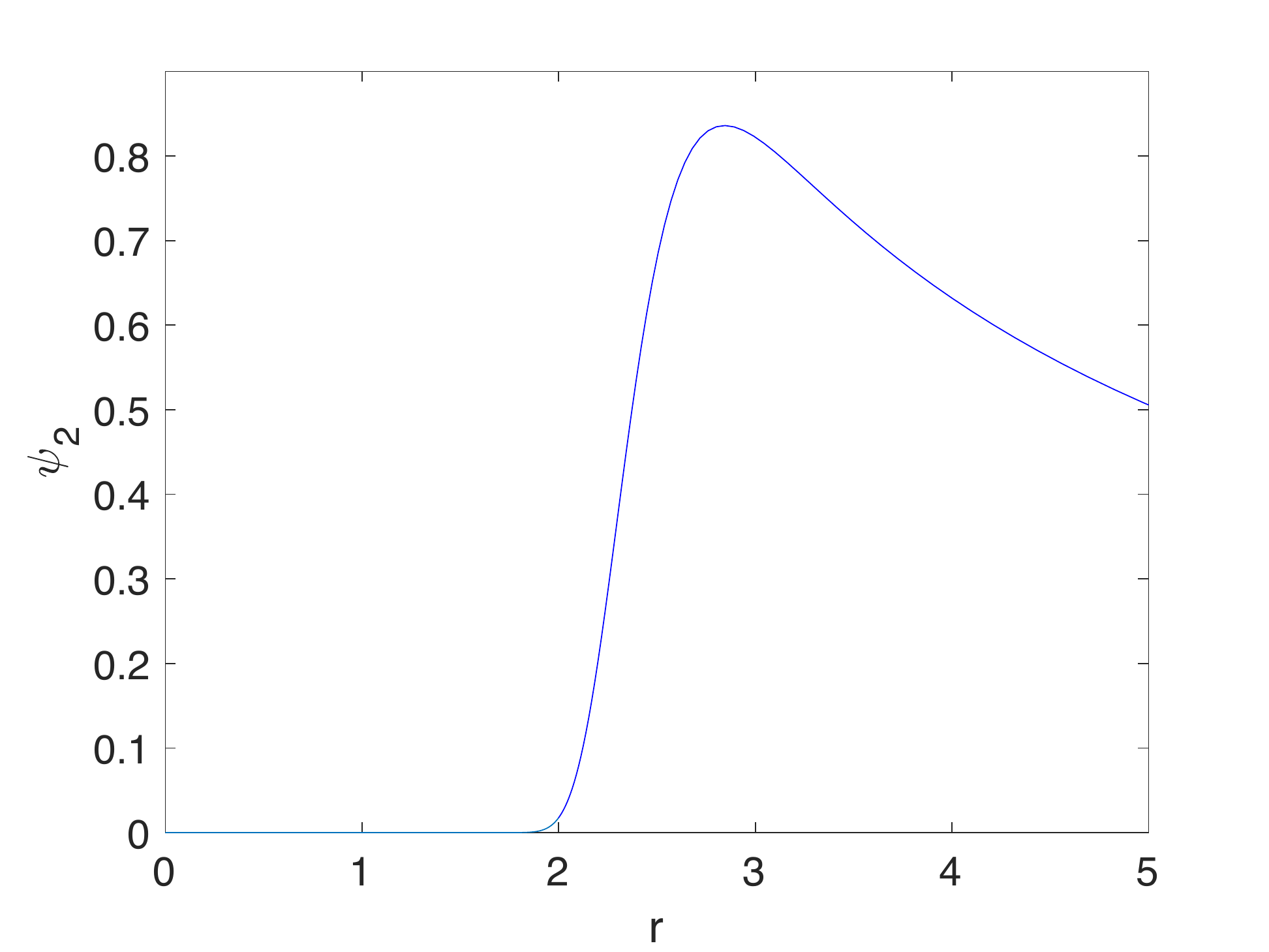}
\end{center}
\caption{Numerical solutions to the Dirac system (\ref{eq:1.10}) with normalization condition \eqref{eq:psi-norm} for the 
Gaussian potential $A(r)=\ee^{-r^2}$ at $k=0$ for the values $\epsilon=10^{-1}$, $10^{-2}$, 
$10^{-3}$  from left to right.  Upper row $\psi_{1}$, lower row 
$|\psi_{2}|$. }
\label{fig:Gaussianpsi}
\end{figure}

Our formal approximation of $X(r;\epsilon)$ in the limit $\epsilon\downarrow 0$ is then as follows:
\begin{itemize}
\item For $r=O(\epsilon)$, $X(r;\epsilon)$ makes a rapid transition from the initial value $X(0;\epsilon)=0$ to $X(r;\epsilon)\approx 1$.
\item $X(r;\epsilon)\approx 1$ for $r\gg\epsilon$ but $r\le r_\mathrm{Match}(\epsilon)$.
\item $X(r;\epsilon)\approx r_\mathrm{Match}(\epsilon)/r$ for $r>r_\mathrm{Match}(\epsilon)$.
\end{itemize}
Recalling \eqref{eq:r-zero-from-X} to calculate the reflection coefficient at $k=0$ gives
\begin{equation}
R_0^\epsilon(0)\approx 2r_\mathrm{Match}(\epsilon),\quad\epsilon\downarrow 0.
\label{eq:r-zero-approx}
\end{equation}
\begin{myremark}
Given this asymptotic description of $X(r;\epsilon)$, from the formula \eqref{eq:radial-kzero-psi-from-Riccati}--\eqref{eq:radial-kzero-psi-from-Riccati-2} we can see that the solution of \eqref{eq:dbar-system-real-radial} for $k=0$ is consistent with the approach based on the WKB method, \emph{but only in the intermediate regime $\epsilon\ll r\ll r_\mathrm{Match}(\epsilon)$} where $X(r;\epsilon)\approx 1$.  Note that the exponent $f$ satisfies the eikonal equation \eqref{eq:eikonal} in the form  \eqref{eq:eikonal-S-zero-radial} appropriate for radial potentials with $S\equiv 0$, and in particular the solution \eqref{eq:eikonal-S-zero-radial-kzero-two-solutions} with the lower sign is the one selected.  Indeed, it is easily checked that the vector $[1,\ee^{\ii\phi}X(r;\epsilon)]^\top$ lies nearly in $\mathrm{ker}(\mathbf{M})$ wherever $X(r;\epsilon)\approx 1$.  It should also be possible to prove that $\tilde{\alpha}_0$ given by \eqref{eq:radial-kzero-psi-from-Riccati-2} is $O(1)$ despite the explicit appearance of $\epsilon$ in the denominator of the exponent.  Indeed, except perhaps in small intervals near $r=0$ or near the ``shoulder'' or the nullcline $X_+(r;\epsilon)$, we will have $(1-X(r;\epsilon))A(r)=O(\epsilon)$ for $r\gg\epsilon$ away from the ``shoulder'' because either $X(r;\epsilon)=1+O(\epsilon)$ (for $\epsilon\ll r\le r_\mathrm{Match}(\epsilon)$) or $0<X(r;\epsilon)<1$ and $A(r)<\epsilon/r$ (for $r\ge r_\mathrm{Match}(\epsilon)$).  Finally, $(1-X(r;\epsilon))A(r)=O(1)$ near $r=0$, so one expects that with a bit more work the integral in the exponent in \eqref{eq:radial-kzero-psi-from-Riccati-2} can be shown to be uniformly $O(\epsilon)$ for all $r>0$.  This observation may help motivate the correct way to generalize the WKB formalism so that it applies for $|k|$ below the threshold where the eikonal function develops singularities.  
\end{myremark}

\subsubsection{Examples}
\label{sec:examples}
Before turning to a rigorous proof, let us apply \eqref{eq:r-zero-approx} in some examples.  
\medskip

\paragraph{\textit{Example 1:  characteristic function of a disk.}}  Suppose that $A$ is an arbitrary positive multiple of the characteristic function of the disk of radius $\rho$.  In this case $r_\mathrm{Match}=\rho$, and therefore $R_0^\epsilon(0)\approx 2\rho$ in the limit $\epsilon\downarrow 0$.  Observe that this result is independent of the amplitude of $A(r)$.  We prove that this result is accurate by an explicit calculation involving modified Bessel functions in Section~\ref{sec:characteristic-function-exact}.

\medskip

\paragraph{\textit{Example 2:  Gaussian amplitude.}}  Suppose that $A(r)=A_0\ee^{-r^2}$.  Then $r_\mathrm{Match}(\epsilon)$ satisfies the equation $\ln(r_\mathrm{Match}) +\ln(A_0) - r_\mathrm{Match}^2 = -\ln(\epsilon^{-1})$, and so $r_\mathrm{Match}\sim\sqrt{\ln(\epsilon^{-1})}$ as $\epsilon\downarrow 0$, and therefore also $R_0^\epsilon(0)\approx 2\sqrt{\ln(\epsilon^{-1})}$ in this limit.  Again, the leading order asymptotic is independent of the amplitude $A_0$.  We prove that this formula is accurate in the relative sense in Section~\ref{sec:Riccati-rigorous} below.

\subsection{Riccati equation.  Rigorous analysis}
\label{sec:Riccati-rigorous}
Theorem~\ref{theorem:Riccati-rigorous} amounts to a more careful formulation of \eqref{eq:r-zero-approx} under suitable conditions on the amplitude function $A(r)$.  
\begin{proof}[Proof of Theorem~\ref{theorem:Riccati-rigorous}]
Given the graph $X=\varphi(r)$ in the $(r,X)$-plane of an arbitrary function $\varphi(\cdot)$, we may compare the slope of the vector field of the Riccati equation \eqref{eq:Riccati-2} evaluated at a point on the graph with the slope of the graph itself.  If 
\begin{equation}
\left.\Delta X'\right|_{X=\varphi(r)}:=\left[\frac{A(r)}{\epsilon}(1-\varphi(r)^2)-\frac{1}{r}\varphi(r)\right]-\varphi'(r)
\label{eq:DeltaX-prime}
\end{equation}
is positive (negative) at a point $P=(r,\varphi(r))$, then the solution of \eqref{eq:Riccati-2} passing through $P$ enters the region above (below) the graph $X=\varphi(r)$ as $r$ increases.  By choosing appropriate functions $\varphi(\cdot)$ and calculating the sign of $\Delta X'$ we will be able to obtain upper and lower bounds on the unique solution $X(r;\epsilon)$ of \eqref{eq:Riccati-2} satisfying $X(r;\epsilon)\to 0$ as $r\downarrow 0$ that are sufficiently strong to establish the asymptotic behavior of the reflection coefficient $R_0^\epsilon(0)$ given by \eqref{eq:r-zero-from-X} up to a relative error term that vanishes with $\epsilon$.  

To get started, we need to first locate the desired solution $X(r;\epsilon)$ for small $r>0$.  Using $A(r)=A(0)+o(r)$ and $X(r;\epsilon)\to 0$ as $r\downarrow 0$ we see that $X(r;\epsilon)$ actually satisfies the stronger condition $X(r;\epsilon)=A(0)r/(2\epsilon)+o(r)$ as $r\downarrow 0$ (the $o(r)$ error term depends on $\epsilon$).

Now we look for simple bounds on the solution $X(r;\epsilon)$.  Consider firstly the quantity $\Delta X'$ defined by \eqref{eq:DeltaX-prime} for the graph of the constant function $X=\varphi_1(r):=1$.  Obviously, \begin{equation}
\left.\Delta X'\right|_{X=\varphi_1(r)}=-1/r<0,\quad \forall r>0, 
\end{equation}
so all solutions of \eqref{eq:Riccati-2} cross the horizontal line $X=1$ in the downward direction as $r$ increases.  (Equivalently, this horizontal line lies above the nullcline $X=X_+(r;\epsilon)$ for all $r>0$.)  Since for small $r$, the desired solution $X(r;\epsilon)$ certainly lies below this line, we obtain the inequality $X(r;\epsilon)<1$ for all $r>0$.

Next, observe that if $\epsilon<\tfrac{1}{2}A(0)$ we have the inequality $X(r;\epsilon)> r$ for sufficiently small $r>0$.  Computing the quantity $\Delta X'$ from \eqref{eq:DeltaX-prime} for the graph $X=\varphi_2(r):=r$ gives
\begin{equation}
\left.\Delta X'\right|_{X=\varphi_2(r)}=\frac{A(r)}{\epsilon}(1-r^2)-2.
\end{equation}
Clearly, $\left.\Delta X'\right|_{X=\varphi_2(r)}>0$ holds for small $r>0$ as a consequence of the inequality $\epsilon<\tfrac{1}{2}A(0)$, however it is equally clear that for $A(r)$ with exponential decay, $\left.\Delta X'\right|_{X=\varphi_2(r)}<0$ if $r$ is sufficiently large given $\epsilon>0$.  Let $r_0(\epsilon)$ denote the smallest positive value of $r$ for which $\left.\Delta X'\right|_{X=\varphi_2(r)}=0$.  It is easy to see that $r_0(\epsilon)=1-\epsilon A(1)^{-1} + o(\epsilon)$ as $\epsilon\to 0$.
Therefore, since $X(r;\epsilon)>r$ for small $r>0$ and since $\Delta X'$ for $X=\varphi_2(r):=r$ is positive for $0<r<r_0(\epsilon)$, the lower bound $X(r;\epsilon)>r$ persists for all $r\in (0,r_0(\epsilon))$.  In particular at $r=r_0(\epsilon)$ we learn that $X(r_0(\epsilon);\epsilon)\ge r_0(\epsilon) = 1-\epsilon A(1)^{-1}+o(\epsilon)$.  Combining this with the uniform upper bound $X(r;\epsilon)<1$ puts the solution $X(r;\epsilon)$ in an $O(\epsilon)$ neighborhood of the nullcline $X=X_+(r;\epsilon)$ for $r=r_0(\epsilon)\approx 1$.  

Now we try to get a lower bound on a larger interval, the length of which grows as $\epsilon\downarrow 0$.  For any constant $\delta\in (0,1)$, we consider the horizontal line $X=\varphi_3(r):=1-\delta$ and compute $\Delta X'$ from \eqref{eq:DeltaX-prime} for this graph:
\begin{equation}
\left.\Delta X'\right|_{X=\varphi_3(r)}=\frac{A(r)}{\epsilon}(2\delta-\delta^2)-\frac{1-\delta}{r}.
\end{equation}
Since $2\delta-\delta^2=\delta(1+(1-\delta))>0$ and $A(1)>0$ we have $\left.\Delta X'\right|_{X=\varphi_3(r)}>0$ for $r=r_0(\epsilon)$ and $\epsilon/\delta$ sufficiently small.  Because $rA(r)$ has a single maximum, the equation $\left.\Delta X'\right|_{X=\varphi_3(r)}=0$ has two roots when both $\delta$ and $\epsilon/\delta$ are small, obtained from 
\begin{equation}
rA(r)=\frac{\epsilon}{\delta}\cdot\frac{1-\delta}{2-\delta}.
\label{eq:r1-equation}
\end{equation}
(It is easy to see that these two roots coincide with the intersection points between the horizontal line $X=\varphi_3(r):=1-\delta$ and the graph of the nullcline $X=X_+(r;\epsilon)$.)
One of the roots obviously satisfies $r=O(\epsilon/\delta)$ and hence is less than $r_0(\epsilon)\approx 1$.  The other is large compared to $r_0(\epsilon)$ when $\epsilon/\delta$ is small.  Let us denote it by $r_1(\epsilon,\delta)$.  Now, given the bounds on the solution $X(r;\epsilon)$ established so far for $r=r_0(\epsilon)$, the assumption that $\epsilon/\delta$ is small implies in particular that $X(r_0(\epsilon);\epsilon)>1-\delta$, so since graphs of solutions of \eqref{eq:Riccati-2} cross the horizontal line $X=\varphi_3(r):=1-\delta$ in the upward direction for $r_0(\epsilon)\le r<r_1(\epsilon,\delta)$, it follows that the lower bound $X(r;\epsilon)\ge 1-\delta$ holds on the same interval.

To continue the lower bound for $r>r_1(\epsilon,\delta)$, we consider the graph $X=\varphi_4(r):=(1-\delta)r_1(\epsilon,\delta)/r$ and compute $\Delta X'$ for this graph from \eqref{eq:DeltaX-prime}:
\begin{equation}
\left.\Delta X'\right|_{X=\varphi_4(r)}=\frac{A(r)}{\epsilon}\left(1-\frac{(1-\delta)^2r_1(\epsilon,\delta)^2}{r^2}\right).
\end{equation}
Obviously we have $\left.\Delta X'\right|_{X=\varphi_4(r)}\ge 0$ for $r\ge r_1(\epsilon,\delta)>(1-\delta)r_1(\epsilon,\delta)$, so solutions of \eqref{eq:Riccati-2} cross the graph in the upwards direction provided $r\ge r_1(\epsilon,\delta)$.  Moreover, since $X(r;\epsilon)\ge 1-\delta$ holds at $r=r_1(\epsilon,\delta)$ the graph of the solution $X(r;\epsilon)$ lies above the graph of $X=\varphi_4(r):=(1-\delta)r_1(\epsilon,\delta)/r$ at $r=r_1(\epsilon,\delta)$, and therefore the lower bound $X(r;\epsilon)\ge (1-\delta)r_1(\epsilon,\delta)/r$ holds for all $r\ge r_1(\epsilon,\delta)$.

So far, the only upper bound we have is $X(r;\epsilon)<1$; however we can obtain an upper bound proportional to $r^{-1}$ for large $r$ by considering the graph of the function
\begin{equation}
X=\varphi_5(r):=\left(r_\mathrm{Match}(\epsilon)+\int_{r_\mathrm{Match}(\epsilon)}^r\frac{sA(s)}{\epsilon}\,\dd s\right)\frac{1}{r}.
\end{equation}
Note that $\varphi_5(r_\mathrm{match}(\epsilon))=1$ and that 
\begin{equation}
\varphi_5(r)=\frac{C}{r}(1+o(1)),\quad r\to\infty,\quad C:=r_\mathrm{Match}(\epsilon)+\int_{r_\mathrm{Match}(\epsilon)}^\infty\frac{sA(s)}{\epsilon}\,\dd s.
\end{equation}
The $o(1)$ error term depends on $\epsilon$ but this dependence is irrelevant for the calculation of the reflection coefficient.
Now, we calculate $\Delta X'$ from \eqref{eq:DeltaX-prime} for this graph:
\begin{equation}
\left.\Delta X'\right|_{X=\varphi_5(r)}=-\frac{A(r)}{\epsilon}X(r)^2<0,
\end{equation}
so trajectories of the Riccati equation \eqref{eq:Riccati-2} cross the graph of $X=\varphi_5(r)$ downwards.  Since $X(r_\mathrm{Match}(\epsilon);\epsilon)<1$ and since $\varphi_5(r_\mathrm{Match}(\epsilon))=1$, it follows that $X(r;\epsilon)<\varphi_5(r)$ for all $r\ge r_\mathrm{Match}(\epsilon)$.

To sum up, we have shown that the unique solution $X(r;\epsilon)$ of the Riccati equation \eqref{eq:Riccati-2} for which $X(r;\epsilon)\to 0$ as $r\downarrow 0$ satisfies, if $\epsilon>0$, $\delta>0$, and $\epsilon/\delta$ are all sufficiently small, the inequalities:
\begin{equation}
r<X(r;\epsilon)<1,\quad 0<r\le r_0(\epsilon),
\end{equation}
\begin{equation}
1-\delta<X(r;\epsilon)<1,\quad r_0(\epsilon)\le r\le r_1(\epsilon,\delta),
\end{equation}
\begin{equation}
\frac{(1-\delta)r_1(\epsilon,\delta)}{r}<X(r;\epsilon)<1,\quad r_1(\epsilon,\delta)\le r\le r_\mathrm{Match}(\epsilon),
\end{equation}
and finally,
\begin{equation}
\frac{(1-\delta)r_1(\epsilon,\delta)}{r}<X(r;\epsilon)<\left(r_\mathrm{Match}(\epsilon)+\int_{r_\mathrm{Match}(\epsilon)}^r\frac{sA(s)}{\epsilon}\,\dd s\right)\frac{1}{r},\quad r\ge r_\mathrm{Match}(\epsilon).
\end{equation}
See Figure~\ref{fig:ULBounds}. 
\begin{figure}[h]
\begin{center}
\includegraphics{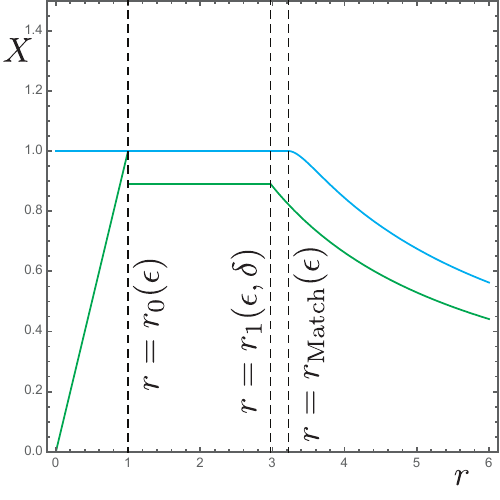}
\end{center}
\caption{The upper bounds (blue) and lower bounds (green) on the solution $X(r;\epsilon)$ in the case $A(r)=\ee^{-r^2}$ for $\epsilon=10^{-4}$ and $\delta=1/\ln(\epsilon^{-1})$.}
\label{fig:ULBounds}
\end{figure}
Setting $\delta=\delta(\epsilon):=1/\ln(\epsilon^{-1})$,
from \eqref{eq:r-zero-from-X} we then obtain the inequalities
\begin{equation}
\underline{R_0^\epsilon}:=2(1-\delta(\epsilon))r_1(\epsilon,\delta(\epsilon))<R_0^\epsilon(0) < 2r_\mathrm{Match}(\epsilon) + 2\int_{r_\mathrm{Match}(\epsilon)}^\infty
\frac{sA(s)}{\epsilon}\,\dd s =:\overline{R_0^\epsilon}.
\label{eq:r-zero-bounds}
\end{equation}
It remains to prove that the upper and lower bounds $\overline{R_0^\epsilon}$ and $\underline{R_0^\epsilon}$ may both be written in the form $2(b^{-1}\ln(\epsilon^{-1}))^{1/p}(1+o(1))$ in the limit $\epsilon\downarrow 0$.

First consider the upper bound $\overline{R_0^\epsilon}$.  The first term $2r_\mathrm{Match}(\epsilon)$ can be found from the logarithm of the defining relation for $r_\mathrm{Match}(\epsilon)$:
\begin{equation}
\ln(r_\mathrm{Match}(\epsilon))+\ln(A(r_\mathrm{Match}(\epsilon)))=-\ln(\epsilon^{-1}).
\end{equation}
Since $L\ee^{-br^p}\le A(r)\le U\ee^{-br^p}$ implies that $\ln(L)-br^p\le \ln(A(r))\le \ln (U)-br^p$, 
it follows that for large $r$, $\ln(A(r))=-br^p+O(1)$. Therefore, 
\begin{equation}
\ln(r_\mathrm{Match}(\epsilon))-br_\mathrm{Match}(\epsilon)^p+O(1)=-\ln(\epsilon^{-1}),
\end{equation}
and it is clear that the dominant balance occurs between the terms $-br_\mathrm{Match}(\epsilon)^p$ and $-\ln(\epsilon)^{-1}$, showing that $r_\mathrm{Match}(\epsilon)=(b^{-1}\ln(\epsilon^{-1}))^{1/p}(1+o(1))$ as $\epsilon\downarrow 0$.  
We estimate the (positive) second term in $\overline{R_0^\epsilon}$ as follows:
\begin{equation}
\begin{split}
\int_{r_\mathrm{Match}(\epsilon)}^\infty\frac{sA(s)}{\epsilon}\,\dd s&=\int_{r_\mathrm{Match}(\epsilon)}^\infty\frac{sA(s)}{r_\mathrm{Match}(\epsilon)A(r_\mathrm{Match}(\epsilon))}\,\dd s\\
&\le
\frac{U}{L}\int_{r_\mathrm{Match}(\epsilon)}^\infty\frac{s\ee^{-bs^p}}{r_\mathrm{Match}(\epsilon)\ee^{-br_\mathrm{Match}(\epsilon)^p}}\,\dd s\\
&=\frac{U}{L}r_{\mathrm{Match}}(\epsilon)\int_1^\infty t\ee^{-br_\mathrm{Match}(\epsilon)^p(t^p-1)}\,\dd t.
\end{split}
\end{equation}
It follows by dominated convergence that this upper bound is $o(r_\mathrm{Match}(\epsilon))$ in the limit $r_\mathrm{Match}(\epsilon)\uparrow\infty$, or equivalently, as $\epsilon\downarrow 0$.  This proves that $\overline{R_0^\epsilon}=2(b^{-1}\ln(\epsilon^{-1}))^{1/p}(1+o(1))$ as $\epsilon\downarrow 0$.

For the lower bound $\underline{R_0^\epsilon}$, since $\delta(\epsilon)=(\ln(\epsilon^{-1}))^{-1}\to 0$ as $\epsilon\downarrow 0$, it suffices to prove that $r_1(\epsilon,\delta(\epsilon))=(b^{-1}\ln(\epsilon^{-1}))^{1/p}(1+o(1))$ as $\epsilon\downarrow 0$.  For this we return to the defining relation \eqref{eq:r1-equation} for $r_1(\epsilon,\delta)$ and take a logarithm:
\begin{equation}
\ln(r_1(\epsilon,\delta(\epsilon))) +\ln(A(r_1(\epsilon,\delta(\epsilon))))=\ln(\ln(\epsilon^{-1}))-\ln(\epsilon^{-1})-\ln(2)+O((\ln(\epsilon^{-1}))^{-1}).
\end{equation}
Again using $\ln(A(r))=-br^p+O(1)$ as $r\uparrow\infty$, this becomes
\begin{equation}
\ln(r_1(\epsilon,\delta(\epsilon)))-br_1(\epsilon,\delta(\epsilon))^p=\ln(\ln(\epsilon^{-1}))-\ln(\epsilon^{-1})+O(1).
\end{equation}
As in the asymptotic calculation of $r_\mathrm{Match}(\epsilon)$, the dominant balance is between $-br_1(\epsilon,\delta(\epsilon))^p$ and $-\ln(\epsilon^{-1})$ and we indeed conclude that $r_1(\epsilon,\delta(\epsilon))=(b^{-1}\ln(\epsilon^{-1}))^{1/p}(1+o(1))$ as $\epsilon\downarrow 0$ as desired.
\end{proof}

The Gaussian $A(r)=\ee^{-r^2}$ satisfies the hypotheses of Theorem~\ref{theorem:Riccati-rigorous} with $L=U=1$, $b=1$, and $p=2$, and we are therefore guaranteed the corresponding relatively accurate approximation $R_0^\epsilon(0)=2\sqrt{\ln(\epsilon^{-1})}(1+o(1))$ as $\epsilon\downarrow 0$.  The upper and lower bounds $\overline{R^\epsilon_0}$ and $\underline{R^\epsilon_0}$ are compared with $2\sqrt{\ln(\epsilon^{-1})}$ and the numerical data for $R_0^\epsilon(0)$ from Figure~\ref{fig:reflection-Gaussian} in Figure~\ref{fig:UpperLowerReflectionBounds}.
\begin{figure}[h]
\begin{center}
\includegraphics{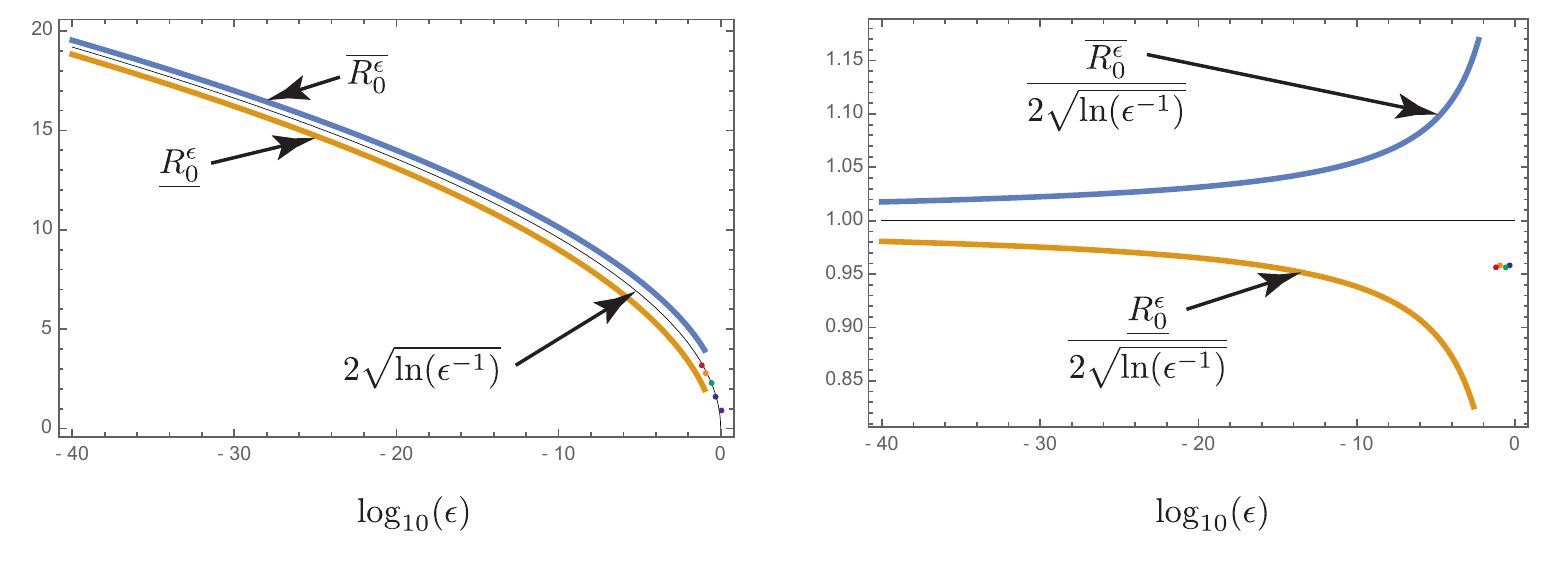}
\end{center}
\caption{The upper bound $\overline{R_0^\epsilon}$ and lower bound $\underline{R_0^\epsilon}$ for the Gaussian potential $A(r)=\ee^{-r^2}$, compared with the approximation $2\sqrt{\ln(\epsilon^{-1})}$ and the numerical data for $R_0^\epsilon(0)$ shown in Figure~\ref{fig:reflection-Gaussian} (the points are colored to correspond with the curves in that figure).  The left-hand panel illustrates absolute accuracy, while the right-hand panel illustrates relative accuracy.}
\label{fig:UpperLowerReflectionBounds}
\end{figure}
It is worth noting that the decay of the relative error as $\epsilon\downarrow 0$ is extremely slow.  Indeed, all of the numerical data that we have been able to reliably compute corresponds only to the colored points in the lower right-hand corner of the plot in the left-hand panel of Figure~\ref{fig:UpperLowerReflectionBounds}; although these points are apparently far from the asymptotic regime of convergence as $\epsilon\downarrow 0$, it is also clear that to the eye \emph{they lie nearly on top of the theoretically-predicted curve}.

\subsection{Exact direct scattering for $k=0$ with $S\equiv 0$ and $A$ being the characteristic function of a disk.}
\label{sec:characteristic-function-exact}
As it is formulated, Theorem~\ref{theorem:Riccati-rigorous} does not apply to compactly-supported potentials.  However,
the approximate formula \eqref{eq:r-zero-approx} for $R_0^\epsilon(0)$ can be confirmed by an exact calculation in the case that $A(x,y)$ is proportional to the characteristic function of the disk of radius $\rho$:  $A(r)=A_0\chi_{r<\rho}(r)$ .  Referring to \eqref{eq:radial-ODEs-kzero}, we have
\begin{equation}
\epsilon r\frac{\dd w_1}{\dd r}=A_0 w_2\quad\text{and}\quad \epsilon r\frac{\dd w_2}{\dd r}=A_0r^2w_1,\quad 0<r<\rho
\label{eq:characteristic-exact-kzero-interior}
\end{equation}
while $w_j(r)=w_j(\rho)$ for $r\ge \rho$ and $j=1,2$.  Eliminating $w_2$ from \eqref{eq:characteristic-exact-kzero-interior} gives
\begin{equation}
\left( r\frac{\dd}{\dd r}\right)^2w_1=\left(\frac{A_0r}{\epsilon}\right)^2w_1,\quad 0<r<\rho.
\end{equation}
With $A_0r/\epsilon = Z$, this equation becomes
\begin{equation}
\left(Z\frac{\dd}{\dd Z}\right)^2w_1=Z^2w_1,\quad 0<Z<\frac{A_0\rho}{\epsilon}.
\end{equation}
Thus $w_1$ is a solution of the modified Bessel equation of order $\nu=0$ \cite[Chapter 10]{dlmf}.
The general solution therefore is $w_1=c_1I_0(Z)+c_2K_0(Z)$.  In order that $w_1$ be bounded at the origin $r=0$ it is necessary to choose $c_2=0$ and then we may (without loss of generality, since only the ratio $w_2/w_1$ is important for the calculation of $R_0^\epsilon(0)$) take $c_1=1$.  Thus we have $w_1(r)=I_0(A_0r/\epsilon)$, and then from the first equation in \eqref{eq:characteristic-exact-kzero-interior} we get
\begin{equation}
w_1(r)=I_0(A_0r/\epsilon)\quad\text{and}\quad 
w_2(r)=\frac{\epsilon r}{A_0}\frac{\dd w_1}{\dd r} = rI_0'(A_0r/\epsilon),\quad 0\le r\le \rho.
\end{equation}
Then since $w_j(r)$ is independent of $r$ for $r>\rho$, we obtain from \eqref{eq:r-kzero-w1w2} the exact formula for the reflection coefficient at $k=0$:
\begin{equation}
R_0^\epsilon(0)=2\frac{\overline{w_2(\rho)}}{\overline{w_1(\rho)}} = 2\rho\frac{I_0'(A_0\rho/\epsilon)}{I_0(A_0\rho/\epsilon)}.
\end{equation}
According to \cite[eqns.\@ 10.40.1 and 10.40.3]{dlmf} (noting that in the notation of that reference $a_0(0)=b_0(0)=1$), we have $I_0'(Z)/I_0(Z)\to 1$ as $Z\to +\infty$, so it follows that 
\begin{equation}
R_0^\epsilon(0) = 2\rho + o(1),\quad \epsilon\downarrow 0,
\end{equation}
which agrees with the formal asymptotic result \eqref{eq:r-zero-approx} being as $r_\mathrm{Match}(\epsilon)=\rho$ by definition in the compact support case.

\section{Numerical approaches}
\label{sec:numerics}
In this section we discuss various numerical approaches to the 
problems appearing in the semiclassical limit of the defocusing DS-II 
equation: the solution of the eikonal problem \eqref{eq:eikonal}--\eqref{eq:f-norm}, the 
computation of the leading-order normalization function $\alpha_0$ appearing in \eqref{eq:Phi-zero}, and the 
solution of the full $\epsilon$-dependent direct scattering problem \eqref{eq:1.10}--\eqref{eq:psi-norm}. For the latter we just 
give a brief review of the approach for Schwartz class potentials in 
\cite{KleinMcL}. 

\begin{myremark}
In this section and the next the notation for Fourier transforms differs slightly from that defined in \eqref{eq:W-norm}.  Namely, here the Fourier and inverse Fourier transform operators denoted below as $\mathcal{F}$ and $\mathcal{F}^{-1}$ respectively are scaled by positive constants to be unitary on $L^2(\mathbb{R}^2)$.
\end{myremark}

For the ease of representation we 
concentrate on the case  $S\equiv 0$. 
Note, however, that it is straightforward to include a phase function $S$ bounded at infinity in the 
approaches discussed below.  With $S\equiv 0$, the relation $g=f-kz$ (cf., \eqref{eq:f-to-g}) defines a function vanishing at $|z|=\infty$, which is numerically convenient.
Using polar coordinates, we thus obtain from \eqref{eq:eikonal} with $S\equiv 0$ the following partial differential equation for $g$:
\begin{equation}
    g_{r}^{2}+\frac{1}{r^{2}}g_{\phi}^{2}+2k \ee^{\ii\phi}\left( 
    g_{r}+\frac{\ii}{r}g_{\phi}\right)=A^{2}
    \label{eikonalg}.
\end{equation}
This equation will be solved in the whole complex plane with a Fourier 
spectral method in $\phi$ and a multidomain spectral method in $r$. 
The ensuing system of nonlinear equations will be solved iteratively 
both with a fixed point method and a Newton iteration. The case of a radially 
symmetric potential $A$ is solved in addition with a series approach 
similar to Section~\ref{series}. 

This section is organized as follows: in Section~\ref{sec:spectral-approach} we collect some facts about the spectral methods to be 
used in the following. In Section~\ref{sec:numerical-eikonal} we present two iterative numerical approaches 
for the eikonal equation and an additional numerical approach based on Fourier series and adapted to radial potentials $A=A(r)$, and test them against the exact solution obtained in Section~\ref{sec:Lorentzian} for the case of the Lorentzian profile $A(x,y)=(1+x^2+y^2)^{-1}$. In Section~\ref{sec:numerical-alpha},
a numerical approach for computing the leading-order normalization function $\alpha_0$ for a given $f$
is presented and again checked against the corresponding exact solution for the Lorentzian profile. In 
Section~\ref{sec:numerical-dbar} we briefly summarize the approach of \cite{KleinMcL} for 
the problem \eqref{eq:1.10}--\eqref{eq:psi-norm} with a Schwartz class potential.

\subsection{Spectral methods}
\label{sec:spectral-approach}
To compute the derivatives in \eqref{eikonalg}, we use two different 
spectral techniques since spectral methods are known for their 
excellent approximation properties for smooth functions. In  
the situation that the eikonal equation is uniformly globally elliptic and the solution 
is regular, this should lead to a very efficient approach.

Since $g(r,\phi)$ is periodic in $\phi$, a Fourier spectral method is 
natural in this context. We write 
$g(r,\phi)=\sum_{n\in\mathbb{Z}}^{}a_{n}(r)\ee^{\ii n\phi}$ and 
approximate the Fourier series via a discrete Fourier transform, see 
for instance \cite{trefethen} and references therein, i.e., for even 
$N$
\begin{equation}
    g(r,\phi)\approx\sum_{n=-N/2+1}^{N/2}a_{n}(r)\ee^{\ii n\phi},\quad 
    g_{\phi}(r,\phi)\approx\sum_{n=-N/2+1}^{N/2}\ii na_{n}(r)\ee^{\ii n\phi}
    \label{gphi};
\end{equation}
and hence the derivative of $g(r,\phi)$ with respect to $\phi$ is approximated via the 
derivative of the sum approximating $g(r,\phi)$. Note that the \emph{Nyquist 
mode} $a_{N/2}(r)$ has to be put equal to zero in the approximation 
of $g_{\phi}$, see \cite{trefethen}. The discrete Fourier transform 
is computed efficiently via a Fast Fourier Transform (FFT). The numerical error in 
approximating the Fourier series with a truncated sum is of the order 
of the first neglected Fourier coefficient. Thus it decreases 
exponentially with $N$ for analytic functions, indicating the  
\emph{spectral convergence} of the method.

In order to obtain a spectral approach also in $r$, we consider two 
domains, I: $r\in[0,1]$ and II: $s=1/r\in[0,1]$ similar to 
\cite{birem} and references therein. In the coordinate 
$s$ equation \eqref{eikonalg} reads
\begin{equation}
    s^{4}g_{s}^{2}+s^{2}g_{\phi}^{2}+2k \ee^{\ii\phi}\left( 
    -s^{2}g_{s}+\ii sg_{\phi}\right)=A^{2}
    \label{eikonalgs}.
\end{equation}
It is assumed that $A$ vanishes as $s\to0$ at least as fast as $s$. Thus we 
can solve \eqref{eikonalgs} after division by $s^{2}$. Note 
that equation \eqref{eikonalgs} is singular for $s=0$ whereas 
equation \eqref{eikonalg} is singular for $r=0$. 

In both domains I and II we approximate the functions $a_{n}(r)$ (respectively 
$a_{n}(s)$; in an abuse of notation, we use the same symbol in both 
cases), $n=-N/2+1,\ldots,N/2$ via a sum of Chebychev polynomials. We only outline the 
approach for domain I, it is completely analogous for domain II. The 
idea of a \emph{Chebychev collocation method} is to introduce the 
collocation points $l_{j}=\cos(\pi j/N_{c})$, $j=0,1,\ldots,N_{c}$ 
and to approximate a function $F(l)$, $l\in[-1,1]$ via the sum
\begin{equation}
    F(l)\approx \sum_{m=0}^{N_{c}}b_{m}T_{m}(l),
    \label{cheb}
\end{equation}
where $T_{m}(l)=\cos(m\arccos(l))$ are the Chebychev polynomials \cite[\S18.3]{dlmf}. The 
\emph{spectral coefficients} $b_{m}$, $m=0,\ldots,N_{c}$ are 
determined by the relations following from imposing \eqref{cheb} as 
an equality at the collocation points, 
\begin{equation}
    F(l_{j}) = \sum_{m=0}^{N_{c}}b_{m}T_{m}(l_{j})\quad 
    j=0,\ldots,N_{c}
    \label{coll}.
\end{equation}
They can be determined conveniently via a \emph{Fast Cosine 
Transform} (FCT) which can be computed via the FFT, see \cite{trefethen}.  The 
numerical error in approximating a function via a truncated Chebychev 
series is  as in the case of discrete Fourier series: it 
decreases exponentially with $N_{c}$ for analytic functions making 
this again a spectral method. 

It is well known that the derivative of a Chebychev polynomial can be 
expressed itself in terms of Chebychev polynomials. The basis for 
this is the identity
\begin{equation}
    \frac{T_{m+1}'(l)}{m+1}-\frac{T_{m-1}'(l)}{m-1}=2T_{m}(l),\quad 
    m = 2,3,\ldots
    \label{chebder}
\end{equation}
and $T_{1}'(l)=T_{0}(l)$, $T_{0}'(l)=0$.
The action of a derivative on a Chebychev sum \eqref{cheb} can thus 
be expressed in terms of the action of a differentiation matrix $\mathbf{D}$ on 
the vector of spectral coefficients $b_{m}$, $m=0,\ldots,N_{c}$. 

In a similar way the multiplication of a function with $l$ can be 
expressed in terms of the action of a matrix on the vector of 
spectral coefficients. The 
approach, see for instance \cite{birem,joerg},  
is based on the well known recurrence formula for Chebyshev 
polynomials,
\begin{equation}
    T_{m+1}(l)+T_{m-1}(l)=2lT_{m}(l), \quad n=1,2,\ldots
    \label{chebrec}
\end{equation}
This identity allows multiplication and division in coefficient space by $l\pm1$. We 
define for given Chebyshev coefficients $b_{m}$ coefficients $\tilde{b}_{m}$ 
via 
$\sum_{m=0}^{\infty}\tilde{b}_{m}T_{m}(l):=\sum_{m=0}^{\infty}(l\pm1)b_{m}T_{m}(l)$. 

 We 
put $r=(1+l)/2$ in domain I. The coefficients $a_{n}(r)$, 
$n=-N/2+1,\ldots,N/2$ are thus 
approximated via the sum $a_{n}\approx 
\sum_{m=0}^{N_{c}}a_{nm}T_{m}(l)$. The action of the derivative with 
respect to $r$ is therefore approximated by the action of a matrix $\mathbf{D}$ 
following from \eqref{chebder} on 
the spectral coefficients, and similarly the action of division by $r$ 
becomes the action of a matrix $\mathbf{R}$ following from \eqref{chebrec} 
on the coefficients. Thus we 
approximate the derivatives via
\begin{equation}
    g_{r}\pm 
    \frac{\ii}{r}g_{\phi}\approx\sum_{n=-N/2+1}^{N/2}\sum_{m=0}^{N_{c}}\left(\sum_{j=0}^{N_{c}}(D_{mj}\mp n R_{mj})a_{nj}\right)
    \label{dbarapprox}.
\end{equation}

The same technique is used in domain II with $s=(1+l)/2$. The 
solutions obtained in domain I and II have to be matched for $r=s=1$ 
to be continuous. As in \cite{birem}, this is done with Lanczos' tau 
method \cite{tau}: one of the equations for each $n$ following from 
using the discretization \eqref{dbarapprox} in \eqref{eikonalg} is 
replaced by the condition that $a_{n}(r=1)=a_{n}(s=1)$, 
$n=-N/2+1,\ldots,N/2$. More concretely we replace for $n<0$ the 
equations corresponding to $m=N_{c}$ in domain I, and for $n>0$ the 
equations corresponding to $m=N_{c}$ in domain II. In addition the 
Nyquist mode is put equal to zero.

\subsection{Numerical approaches for the eikonal problem}
\label{sec:numerical-eikonal}
We now discuss two different numerical approaches for the eikonal 
problem \eqref{eq:eikonal}--\eqref{eq:f-norm}, each of which produces an approximation to the function $g=f-kz$ that solves \eqref{eikonalg} and satisfies $g\to 0$ as $|z|\to\infty$.

\subsubsection{Iterative methods for the discretized eikonal equation}
\label{sec:eikonal-iterative}
The spectral discretization described in Section~\ref{sec:spectral-approach} leads to an 
approximation of \eqref{eikonalg} in terms of a $(2N_{c}+2)N$-dimensional nonlinear system of equations. This system will be solved 
iteratively. 

A first approach is based on a fixed-point iteration. We write for 
$|k|>1/2$ the 
system corresponding to \eqref{eikonalg} in the form 
\begin{equation}
    \sum_{j=0}^{N_{c}}(D_{mj}\mp n R_{mj})a_{nj} = G(\{a_{nm}\})
        \label{fix},
\end{equation}
where 
\begin{equation}
    \begin{split}
   G(\{a_{nm}\})&:=
    \frac{1}{2k}\mathbb{F}\left(A^{2}-\sum_{n=-N/2+1}^{N/2}\sum_{m=0}^{N_{c}}
    \left(\sum_{j=0}^{N_{c}}(D_{mj}+ n R_{mj})a_{nj}\right)\right.\times\\
    &\left.\sum_{n=-N/2+1}^{N/2}
    \sum_{m=0}^{N_{c}}\left(\sum_{j=0}^{N_{c}}(D_{mj}- n 
    R_{mj})a_{nj}\right)\right),  
    \end{split}
        \label{Gdef}
\end{equation}
where $\mathbb{F}$ denotes the combined action of the FFT and the FCT on the angular and radial variables respectively. 
Since both FFT and FCT are fast, it is convenient when possible to switch between physical 
space and the space of spectral coefficients in order to compute 
products instead of convolutions in coefficient space.

We first solve \eqref{fix} by casting it in the form of a fixed-point iteration: 
$\sum_{j=0}^{N_{c}}(D_{mj}- n R_{mj})a^{K+1}_{nj} = 
G(\{a^{K}_{nm}\})$. Here we choose as the initial iterate the 
solution of the $\delbar$-problem $\sum_{j=0}^{N_{c}}(D_{mj}\mp n 
R_{mj})a^{0}_{nj} = \mathbb{F}[\ee^{-\ii\phi}A^{2}/(2k)]$. Numerical resolution in 
each domain is controlled via the decrease of the spectral 
coefficients with $N$ and $N_{c}$. As discussed for the examples below, 
numerical resolution is ideal if the coefficients decrease to the 
order of machine precision both in the Fourier and Chebyshev 
dependence. If $|k|$ is large enough, the fixed-point iteration 
converges linearly, i.e., $\|a_{nm}^{K+1}-a_{nm}^K\|_\infty=O(K^{-1})$, as might be expected.

Alternatively we can use a Newton iteration. To this end 
we write the equation following from \eqref{eikonalg} 
after the spectral discretization described in Section~\ref{sec:spectral-approach} in the form 
$F(\{a_{nm}\})=0$ and solve it with a standard Newton iteration:
\begin{equation}
    a_{nm}^{K+1}= a_{nm}^{K}-\mathrm{Jac}(F(\{a^{K}_{nm}\}))^{-1}F(\{a^{K}_{nm}\})
    \label{newton};
\end{equation}
here the tau method is applied in the inversion of the Jacobian, the 
action of which is computed as a convolution in the space of 
coefficients. Using again the 
solution of the $\delbar$-problem $\sum_{j=0}^{N_{c}}(D_{mj}- n 
R_{mj})a^{0}_{nj} = \mathbb{F}[\ee^{-\ii\phi}A^{2}/(2k)]$ for $|k|>1/2$ as the initial iterate, we 
observe the expected quadratic convergence typical for Newton's method, i.e., $\|a_{nm}^{K+1}-a_{nm}^K\|_\infty=O(K^{-2})$. The disadvantage of the 
approach is that the Jacobian is a $(2N_{c}+2)N\times(2N_{c}+2)N$ 
matrix, but the quadratic convergence implies that the iteration 
takes roughly the same amount of time as the fixed-point iteration to reach a 
residual of $10^{-10}$ which is generally where iterations are stopped.

\subsubsection{A Fourier series method for the radially symmetric case}
\label{sec:numerical-Fourier-series}
In the radially symmetric case $A=A(r)$, we can proceed as in
Section~\ref{series} and solve for a series in $\ee^{\ii\phi}$: 
writing\footnote{The fact that only negative odd harmonics appear in \eqref{gseries} is a consequence of the form of the coefficients $f_n$ in the series approach for radial potentials with $S(x,y)\equiv 0$ described in Section~\ref{series}.  Comparing with \eqref{eq:fseries} for $\delta=\tfrac{1}{4}$ and \eqref{eq:fnradialform} we see that in the notation of Section~\ref{series}, $c_n(r)=G_{n+1}(r^2)/[2(2n+1)r^{2n+1}]$.} 
\begin{equation}
    g(r,\phi)=\sum_{n=0}^{\infty}\frac{c_{n}(r)}{(2k)^{2n+1}}\ee^{-\ii(2n+1)\phi}
    \label{gseries},
\end{equation}
where $c_n=c_n(r)$, we find from \eqref{eikonalg} that
\begin{equation}
    c_{0}'+\frac{1}{r}c_{0}=A(r)^{2},
    \label{c0}
\end{equation}
and for $n>0$,
\begin{equation}
    c_{n}'+\frac{2n+1}{r}c_{n}=\sum_{j=0}^{n-1}\left(1-\frac{2j+1}{r}\right)
    \left(1+\frac{2(n-j-1)+1}{r}\right)c_jc_{n-j-1}
    \label{cn},
\end{equation}
where the prime denotes differentiation with respect to $r$, and it is required that $c_n(r)\to 0$ as $r\to\infty$ for all $n\ge 0$.
Equations \eqref{c0} and \eqref{cn} are solved again with the 
Chebychev collocation method described above. 
As noted in Section~\ref{sec:radial}, the series \eqref{gseries} is a power series in odd negative powers of $k$ and hence for given $(x,y)\in\mathbb{R}^2$ will converge if $|k|$ is 
large enough. For smaller $|k|$, it can only converge if the $c_{n}$ 
decrease rapidly enough as $n\to\infty$. In applications only 
convergent cases are interesting where the series can be effectively truncated 
for some $n=N_{\phi}$. In this case a coupled system of $N_{\phi}$ ordinary differential equations of the form 
\eqref{c0} and \eqref{cn} has to be solved with zero initial conditions at $r=\infty$.  \emph{On the other hand,
the numerical computation of the $L^\infty(\mathbb{R}_+)$-norms of $c_n(\cdot)$ and their analysis for increasing $n$ allows one to make a good prediction of the critical radius $|k|$ above which one has a global smooth solution of the eikonal problem and below which the latter solution necessarily develops singularities analogous to turning points in the one-dimensional problem.}

\subsubsection{Comparison with the exact solution for the Lorentzian profile}
Note that the numerical approaches for the eikonal equation presented 
above are essentially independent and can be thus used as 
mutual tests. To illustrate how the codes work in practice and to 
establish which accuracies can be expected, we test them for the 
example of the exact solution \eqref{eq:fW-Lorentzian} for a 
Lorentzian amplitude $A(x,y)=(1+x^2+y^2)^{-1}$.  To compare with the numerics, we get $g(x,y;k)$ from \eqref{eq:fW-Lorentzian} simply by omitting the term $kz$ on the right-hand side.  Since $W$ is invariant under $z\mapsto 1/\overline{z}$, the exact solution $g$ is symmetric with respect to reflection through the unit circle in the $z$-plane.  As described in Section~\ref{sec:Lorentzian}, $g$ is smooth provided $|k|>1/2$.  
We first plot the exact solution for $k=1$ in Figure~\ref{lorentzk1}.  
\begin{figure}[htb!]
  \includegraphics[width=0.49\textwidth]{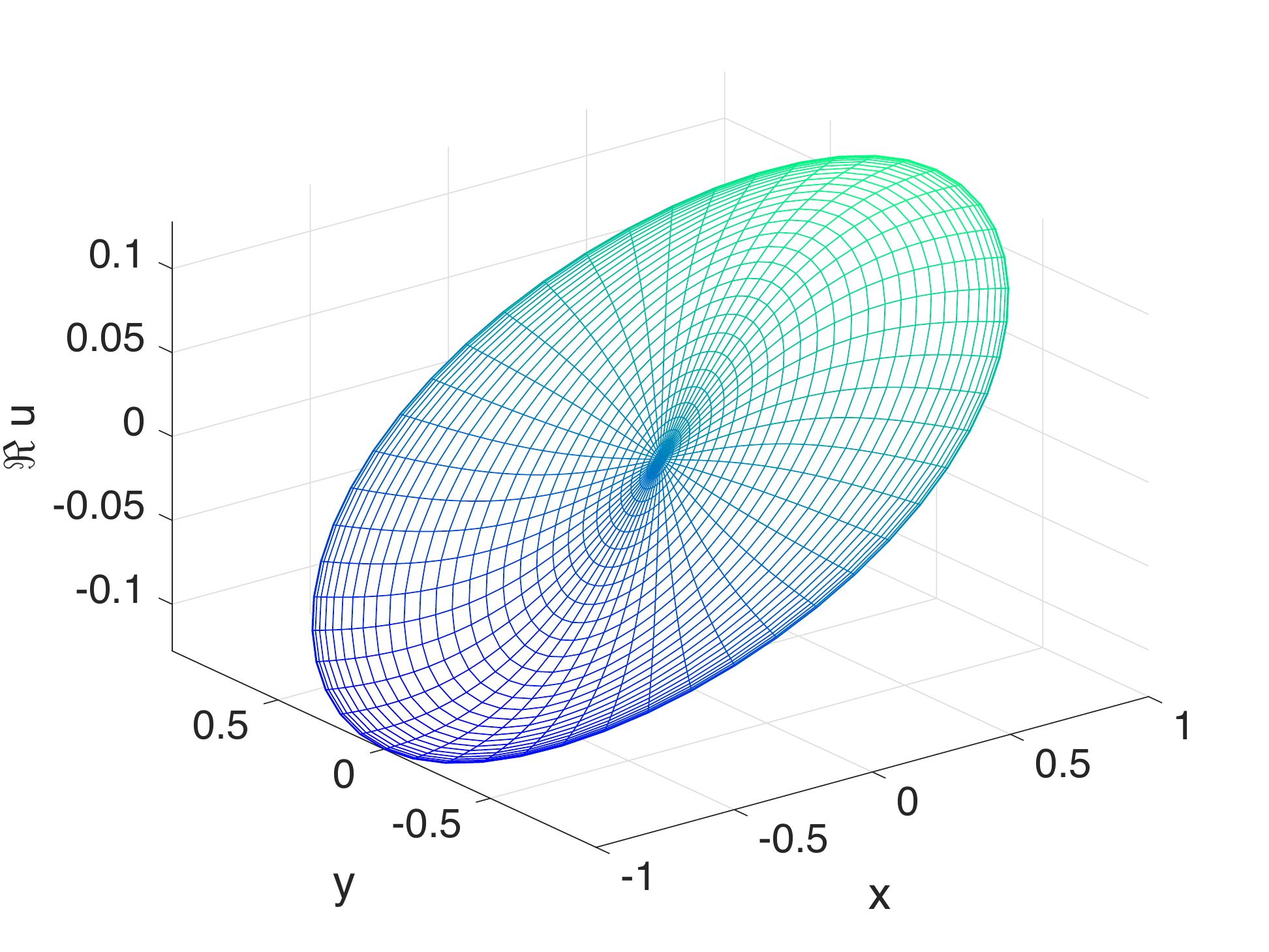}
  \includegraphics[width=0.49\textwidth]{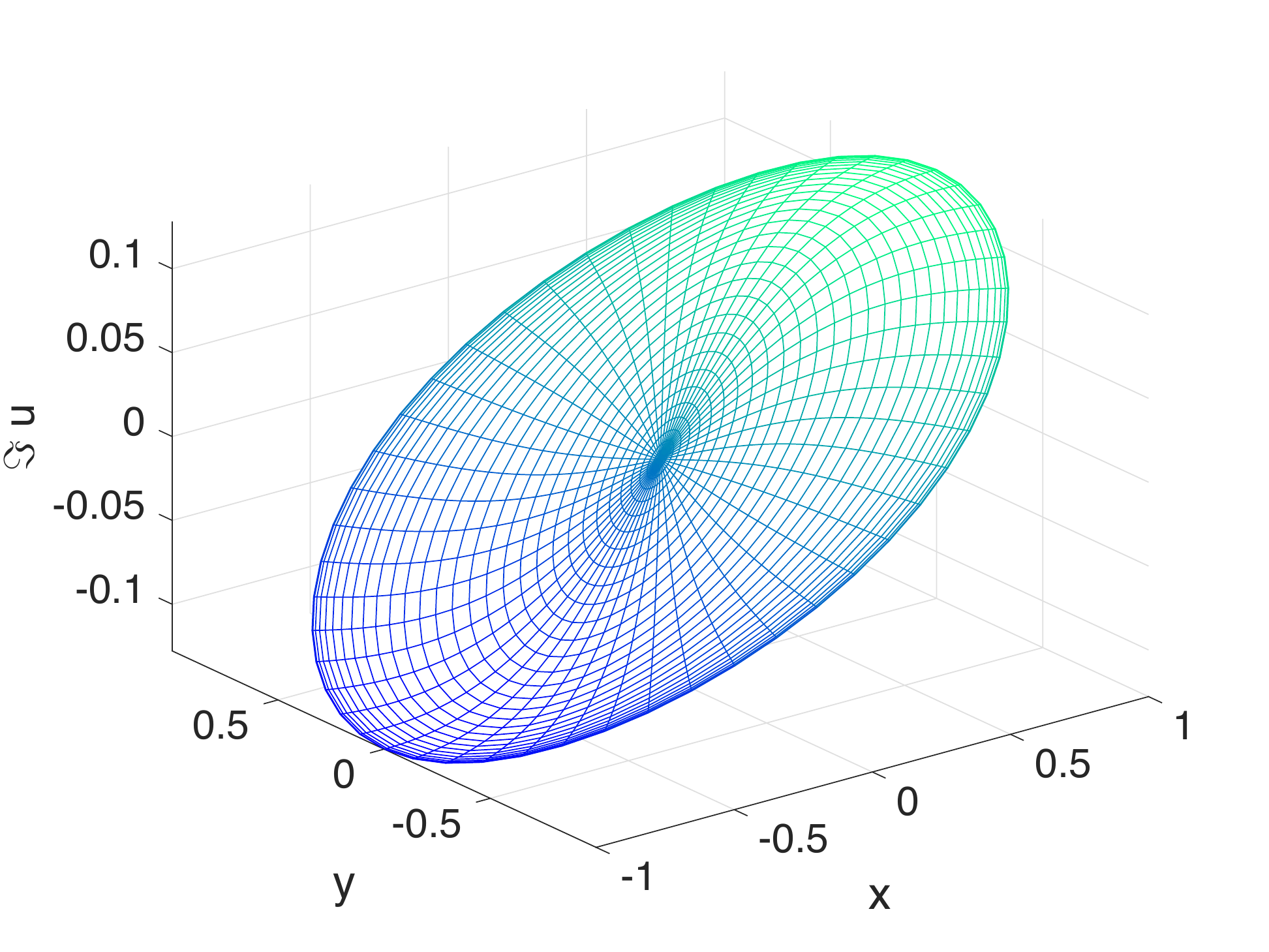}
 \caption{The exact solution $g(x,y;k)=f(x,y;k)-kz$ for the Lorentzian potential $A(x,y)=(1+x^2+y^2)^{-1}$ and $S(x,y)\equiv 0$ for $k=1$. Left: $\re(g(x,y;1))$.  Right:  $\imag(g(x,y;1))$.  By exact reflection symmetry through the unit circle, we only show it for $r\leq1$.}
 \label{lorentzk1}
\end{figure}

For the iterative solution of this problem, we use $N_{c}=32$ 
Chebychev polynomials and $N=50$ Fourier modes. It can be seen in the left-hand panel of 
Figure~\ref{lorentzk1coeff} that the coefficients $a_{nm}$ decrease 
exponentially in $(n,m)$ and that they reach machine precision well before the boundary of the spectral domain. This indicates that the solution is numerically resolved. The 
fixed-point iteration is stopped for this case when the difference 
between consecutive iterates is less than a given threshold, 
$\|a^{K+1}_{nm}-a^{K}_{nm}\|_{\infty}<10^{-10}$. This is achieved in 
this example in 10 iterations of the fixed-point method. The difference between the numerical result and 
the exact solution is shown in the right-hand panel of Figure~\ref{lorentzk1coeff}. It 
can be seen that the numerical error is largest near the origin and that 
it is of the order of $10^{-11}$. Note that this error is not 
affected if the iteration is stopped at a smaller threshold; it is 
due to the large condition numbers of the differentiation matrices which are for 
Chebychev differentiation of the order $N_{c}^{2}$, see e.g., the discussion in 
\cite{trefethen}. Thus the maximally achievable accuracy is of the 
order $10^{-11}$ with this approach for this example. If the problem 
required a higher numerical resolution (a larger $N_{c}$), the maximally achievable 
accuracy would be slightly lower. This problem can be addressed by 
introducing more than two domains in the $r$ variable, but this will not be needed for the  
examples studied here.

Note that the fixed-point code finds the 
symmetry of the solution with respect to $r\to1/r$ with the same 
accuracy, i.e., the same difference between numerical and exact 
solution will be found for $r>1$. Therefore we do not show the 
solution for $r>1$ in this case though it is obtained with a 
precision of the order of $10^{-11}$ in the whole complex plane. 
\begin{figure}[htb!]
  \includegraphics[width=0.49\textwidth]{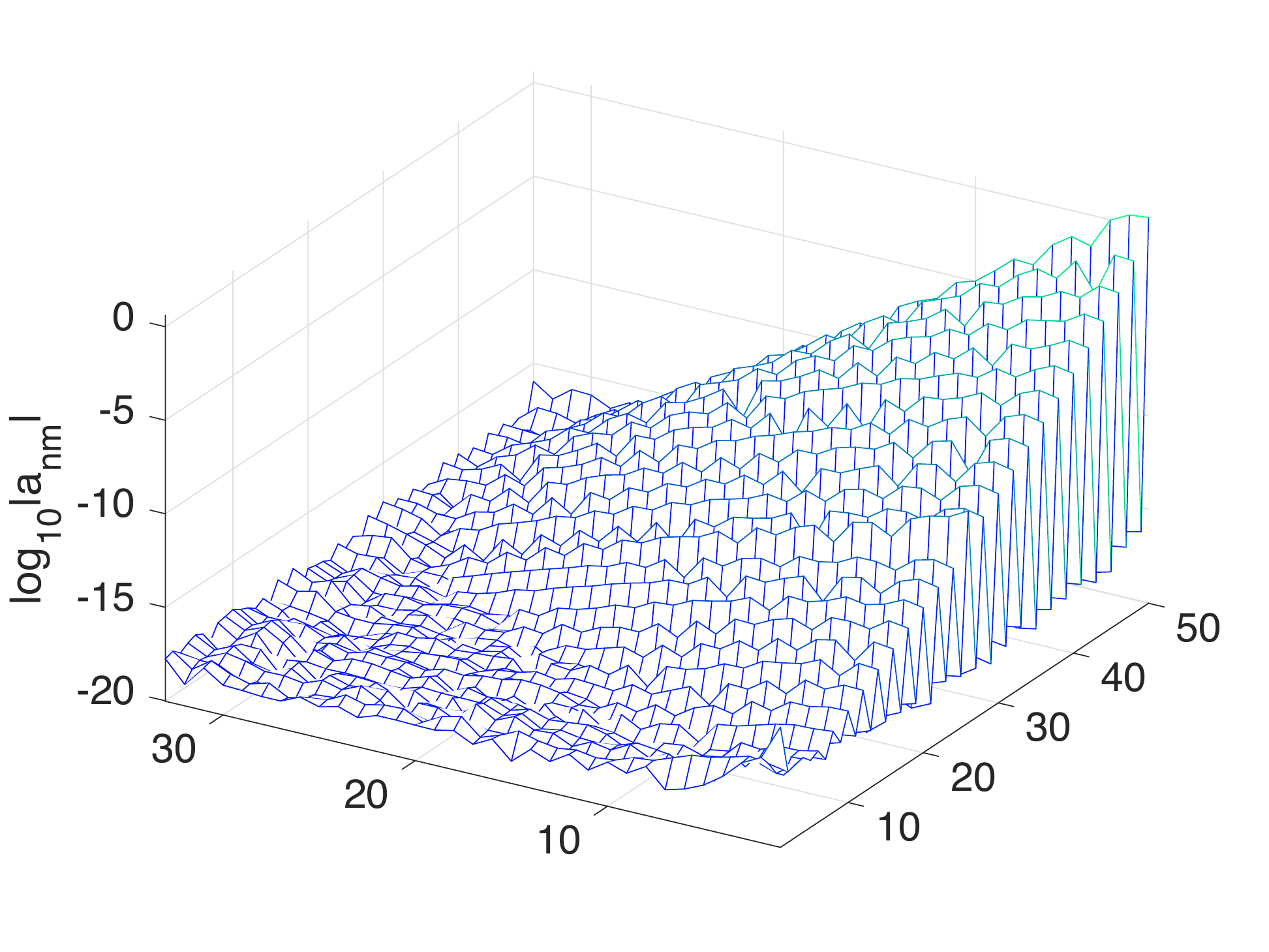}
  \includegraphics[width=0.49\textwidth]{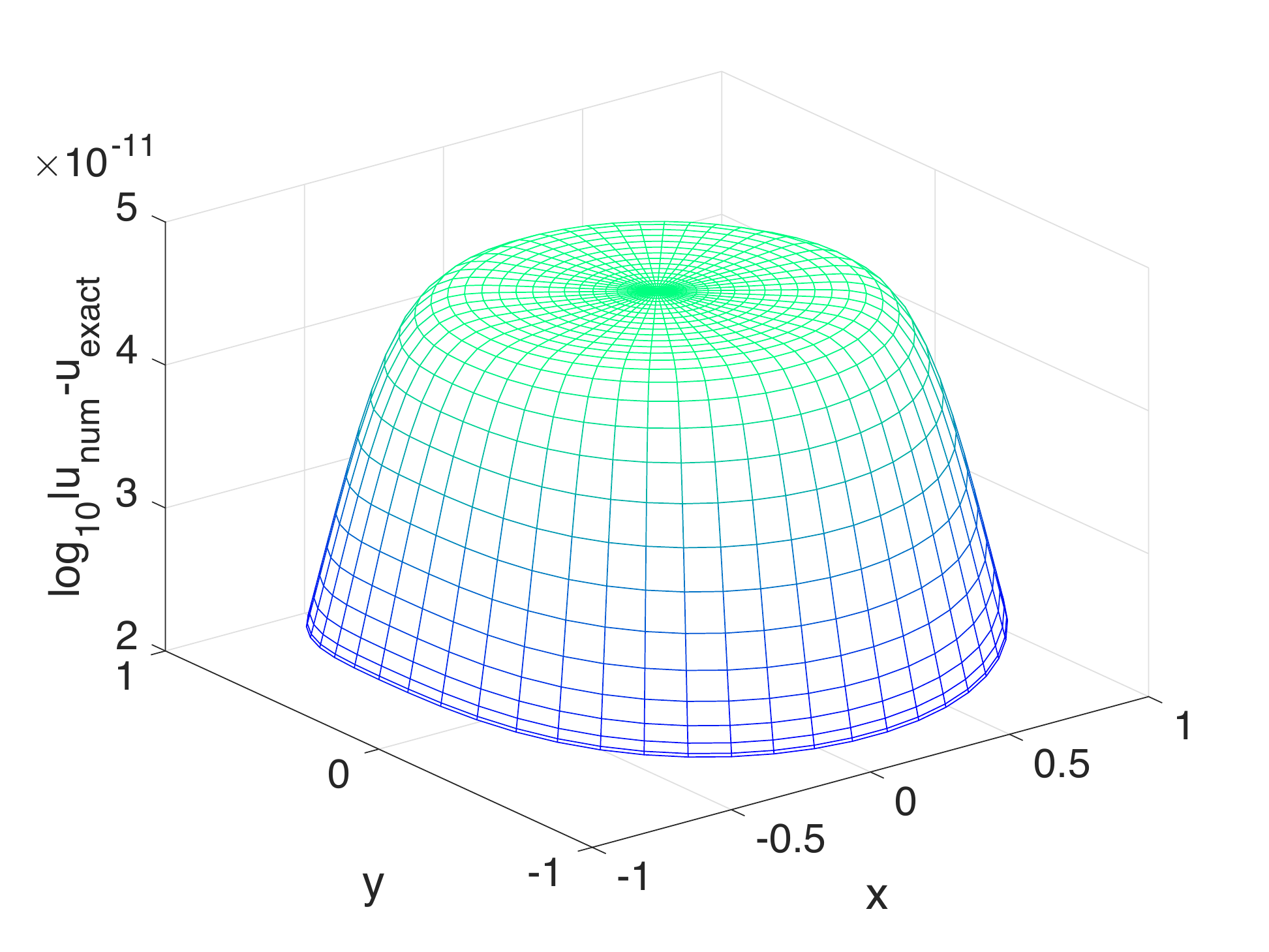}
 \caption{Chebychev and Fourier coefficients of the fixed-point approximation of the solution in 
 Figure~\ref{lorentzk1} plotted over the $(n,m)$-grid on the left, and the difference of the 
 numerical solution and the 
 exact solution on the right.}
 \label{lorentzk1coeff}
\end{figure}

The Newton iteration converges in this case after 3 iterations to the 
same precision. In practice it takes longer than the fixed-point 
iteration since the computation of the convolutions and the inversion 
of the Jacobian are computationally expensive. Krylov subspace techniques might be 
helpful in this context, but have not been explored so far. 

As the Lorentzian $A(x,y)=(1+x^2+y^2)^{-1}$ is radially symmetric, the numerical series approach described in Section~\ref{sec:numerical-Fourier-series} also applies, and we use the same discretization resolution parameters as in the iterative approaches, namely
$N_{c}=32$ and $N=50$ (recall that only odd powers of $\ee^{-\ii\phi}$ 
appear in this approach). The difference between numerical and exact 
solution in this case is of the order of $10^{-14}$, mainly near the 
origin as can be seen in the left-hand panel of Figure~\ref{lorentzk1rad}. The reason for a 
smaller error in this case is that only ordinary differential equations have to be solved as
the effect of the Fourier discretization of $\phi$ is essentially 
decoupled. Since there is no iteration, this method is also the 
fastest of the three discussed here.  The $L^{\infty}$ norm of the 
functions $c_{n}$ decreases as $n^{-5/2}$ as expected by the formulae \eqref{eq:fn-Lorentzian} 
and \eqref{stirling} (also noting that $|z|/(1+|z|^2)\le\tfrac{1}{2}$). Note that the finite precision employed delimits the 
number of coefficients $c_{n}$ which can be used in practice. 
\begin{figure}[htb!]
  \includegraphics[width=0.49\textwidth]{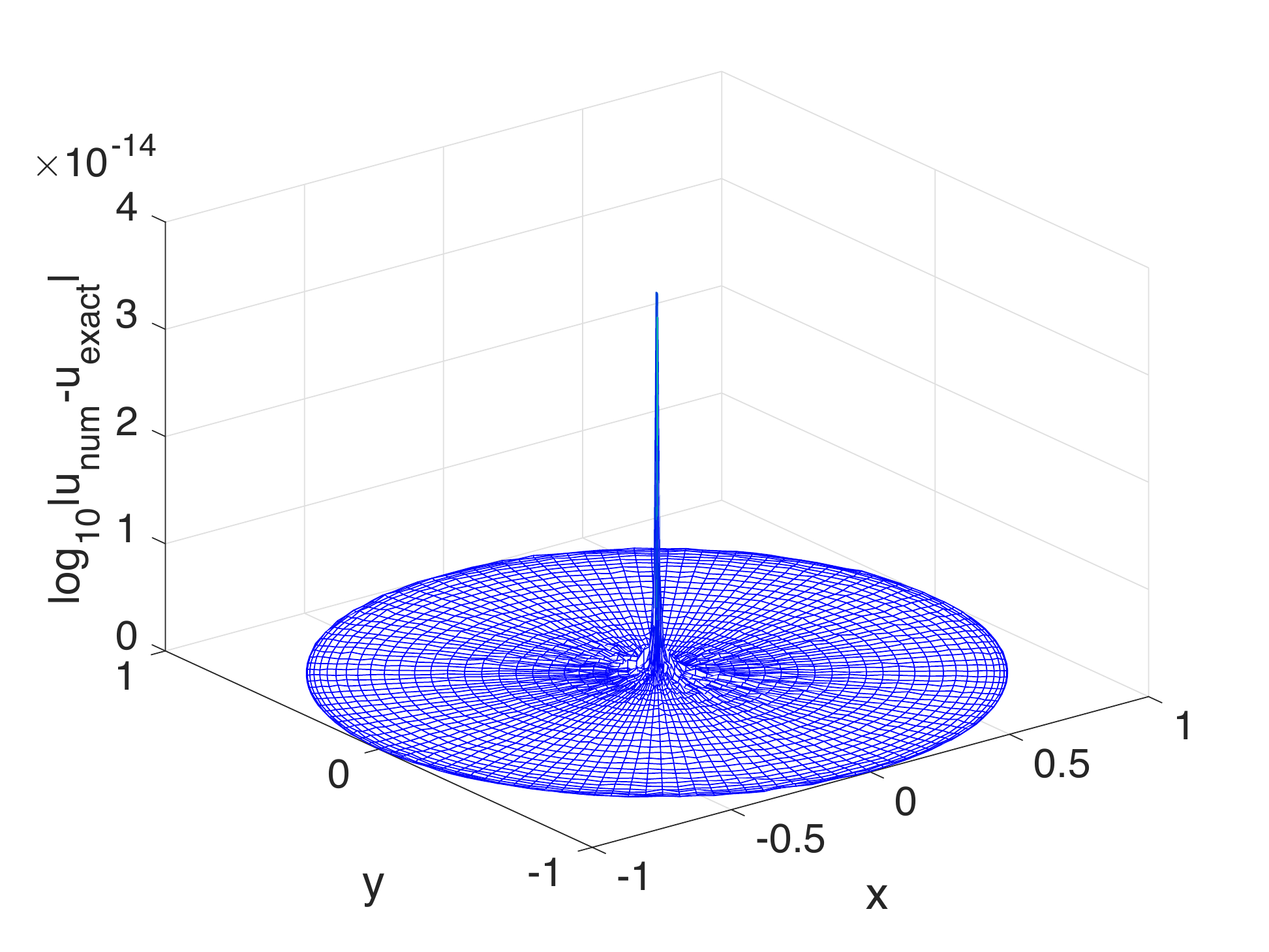}
  \includegraphics[width=0.49\textwidth]{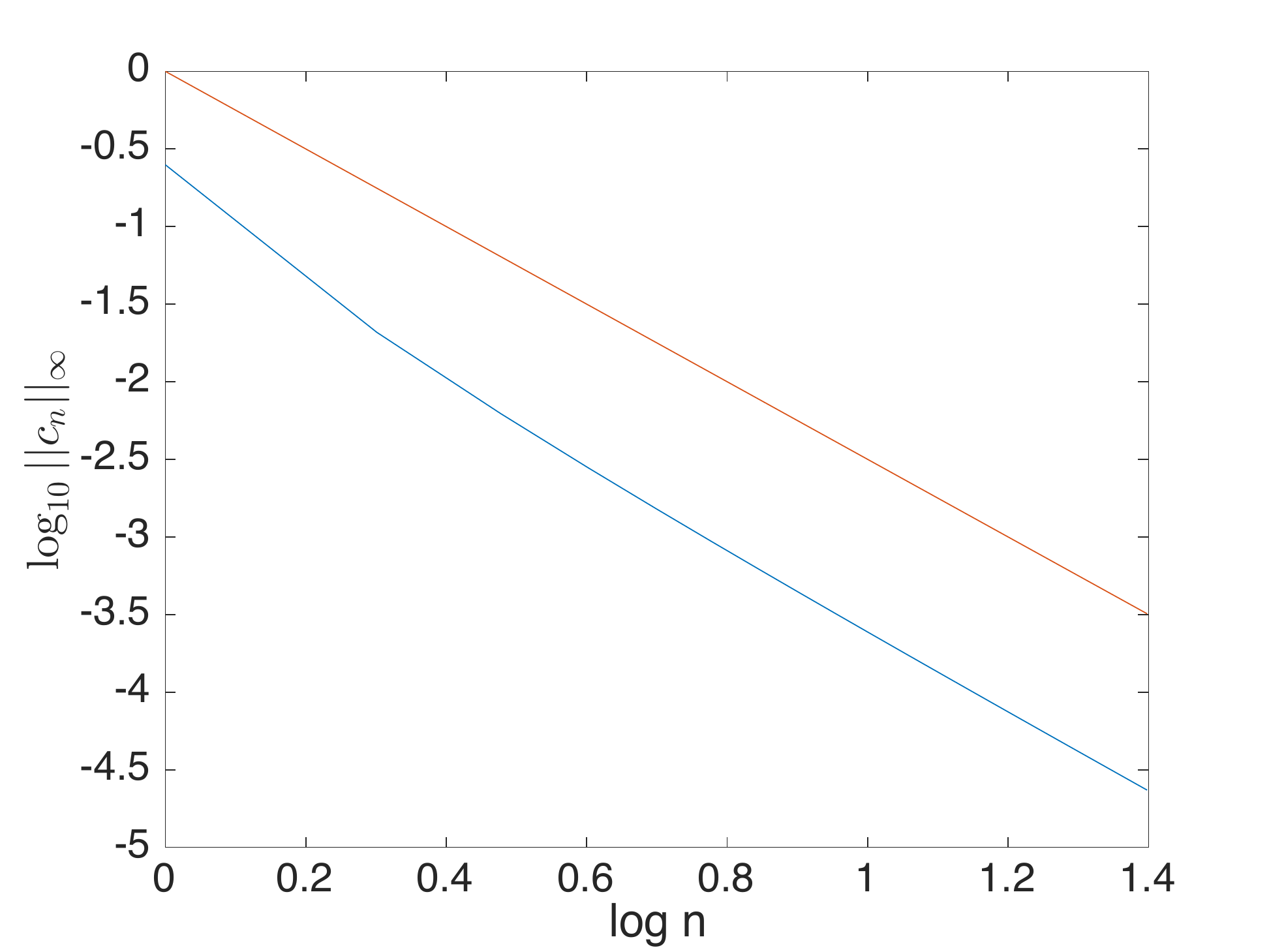}
 \caption{Left:  the difference between the exact solution $g(x,y;k)$ for the case of 
 the Lorentzian amplitude $A(x,y)=(1+x^2+y^2)^{-1}$ with $S(x,y)\equiv 0$ and the numerical solution 
 constructed with the series approach described in Section~\ref{sec:numerical-Fourier-series} for $k=1$. 
 Right:  the $L^{\infty}$ 
 norm of the coefficients $c_n(\cdot)$ in a log-log plot in blue (for reference the red line has slope $-5/2$).}
 \label{lorentzk1rad}
\end{figure}

The series approach also makes clear which problems are to be 
expected for smaller $k$. As $k$ decreases toward the critical value of $1/2$, 
more and more terms in the series \eqref{gseries} will be needed to obtain the same accuracy, and
for even smaller values of $k$ the series fails to converge (the continuation of the 
exact solution, however, will be bounded with jump 
discontinuities along some branch cuts as shown in Figure~\ref{fig:LorentzianBranchCuts}). For $k=0.6$, we thus need a 
considerably higher resolution in $\phi$ as can be seen in 
Figure~\ref{lorentzk06} where we use $N_{c}=40$ and $N=140$. The 
decrease of the Fourier modes is visibly much slower than before. The 
fixed-point iteration converges after 22 iterations, and the 
difference to the exact solution is of the order of $10^{-10}$ (this 
time it is largest near the rim of the disk). The Newton iteration 
converges in just 4 iterations, but becomes too slow at these 
parameters in comparison to the fixed-point iteration.  Using the same 
parameters for the series approach based on \eqref{gseries}, we get as before 
a difference of the order of $10^{-14}$ between the numerical and exact 
solutions. 
\begin{figure}[htb!]
  \includegraphics[width=0.49\textwidth]{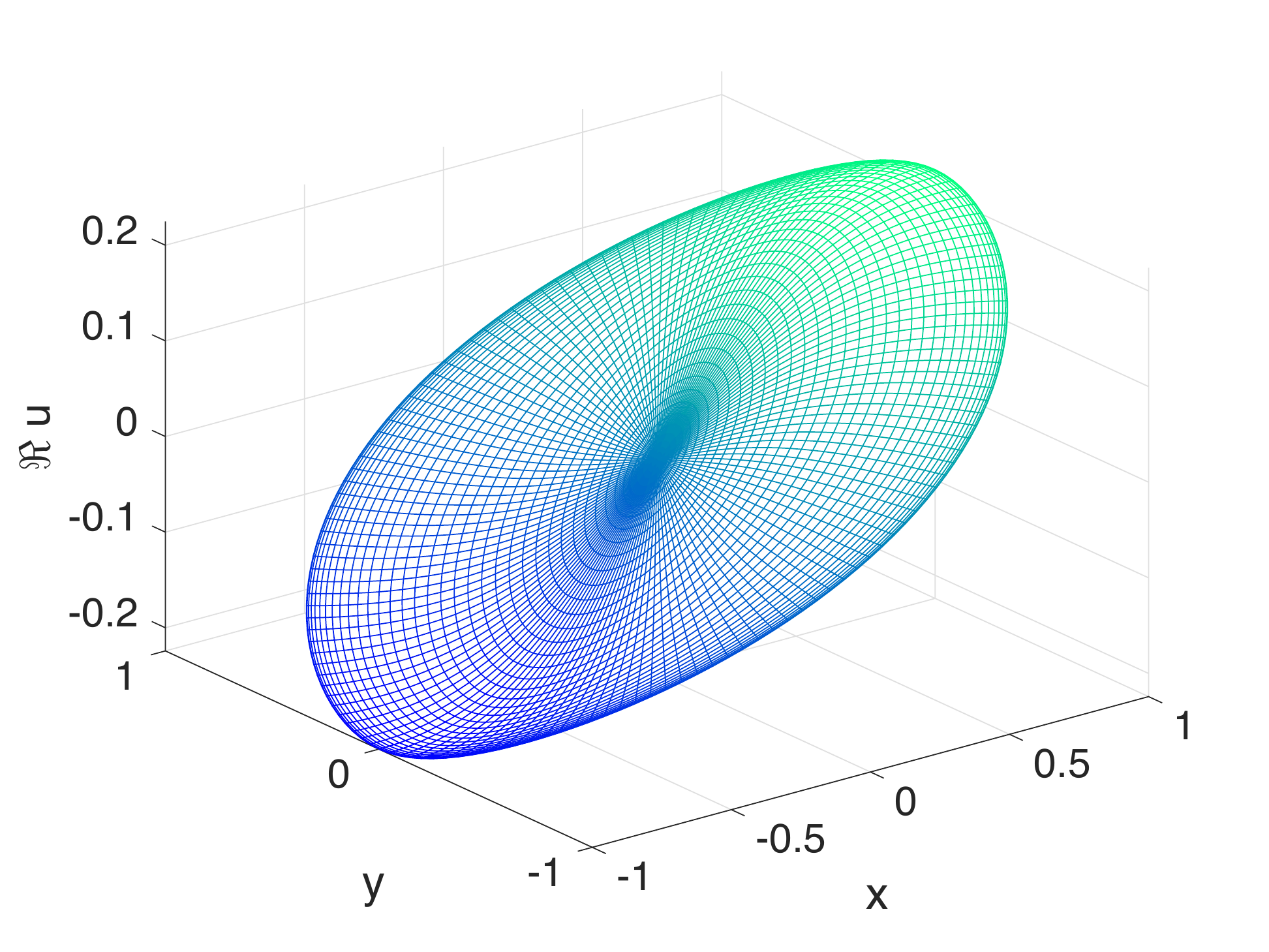}
  \includegraphics[width=0.49\textwidth]{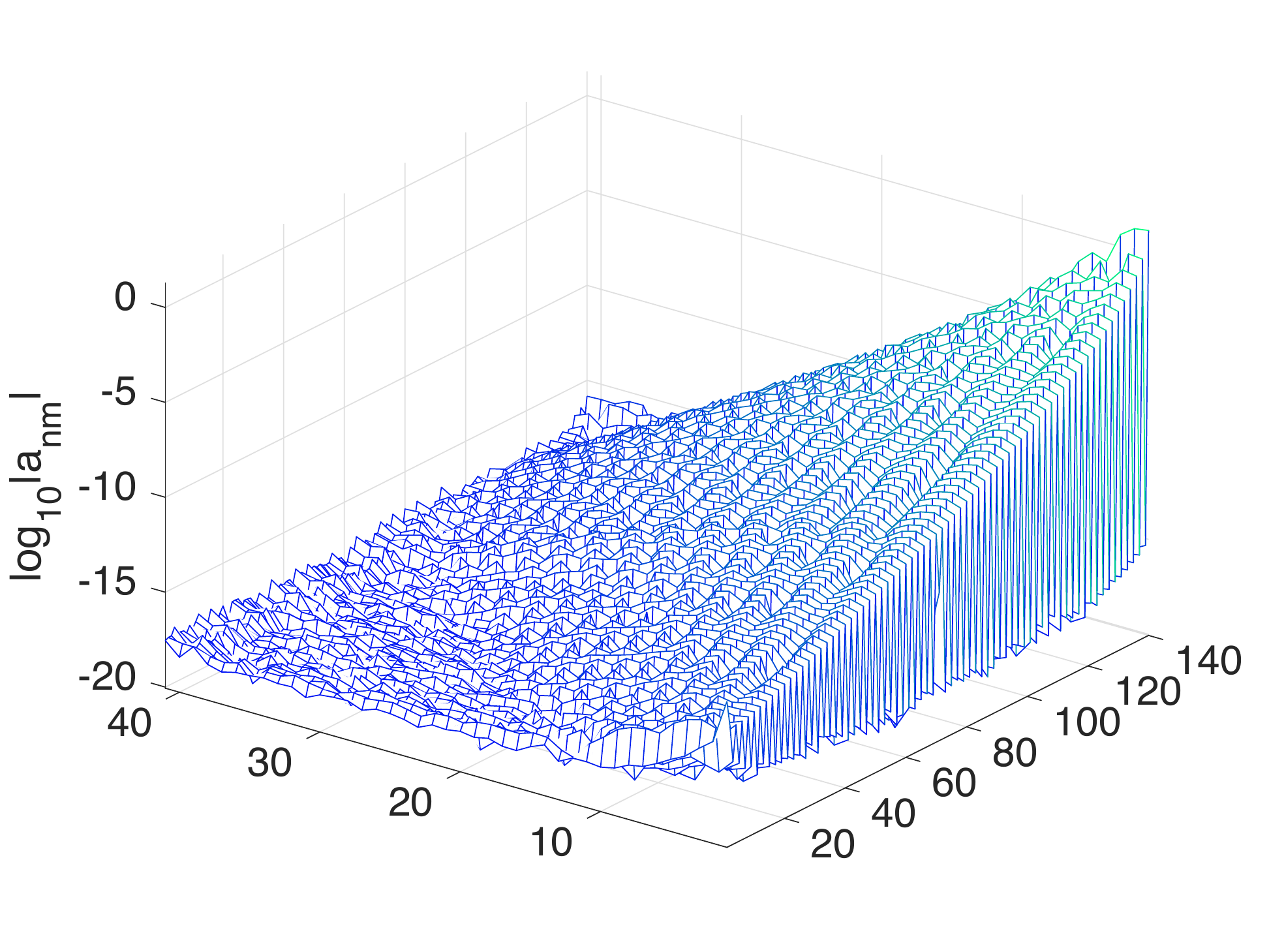}
 \caption{Left:  $\re(g(x,y;k))$ at $k=0.6$ for the case of the exact solution corresponding to the Lorentzian amplitude $A(x,y)=(1+x^2+y^2)^{-1}$ and $S(x,y)\equiv 0$.  Right: the 
 corresponding numerically-computed Fourier and Chebychev coefficients plotted over the $(n,m)$-grid.}
 \label{lorentzk06}
\end{figure}

\begin{myremark}
    Since the equations \eqref{Gdef} and \eqref{newton} are 
    nonlinear, \emph{aliasing errors}, see e.g.,~\cite{trefethen}, can 
    play a role in this context due to the use of truncated series in 
    the computation of products. To address this we use a filtering 
    in coefficient space: if the iteration is stopped at a given 
    level \texttt{tol} (typically $10^{-10}$), all spectral coefficients with 
    $|a_{nm}<\mathtt{tol}|$ are put equal to zero. 
\end{myremark}

Summing up, for radially symmetric potentials $A$, the series 
approach described in Section~\ref{sec:numerical-Fourier-series} is the most efficient of the ones presented 
here. If sufficient numerical resolution is provided, which can be 
controlled via the decrease of the $c_{n}(r)$ in $n$ and of their 
Chebychev coefficients, an accuracy of the order of $10^{-14}$ can be 
reached. We note that for potentials without radial symmetry, the series method is inapplicable but the fixed-point iteration method remains as an efficient option.

\subsection{Numerical computation of the leading-order normalization function $\alpha_0$}
\label{sec:numerical-alpha}
Once the eikonal problem is solved for $f(x,y;k)$, corresponding function $\alpha_0$ 
has to be determined in order to complete the construction of the leading term in the WKB expansion of the solution of the direct scattering problem.  In the case $S\equiv 0$, $\alpha_0$ is the solution of 
the linear equation
\eqref{eq:alpha-equation-S-zero} that satisfies $\alpha_0\to 1$ as $|z|\to\infty$. In polar coordinates, this equation 
reads 
\begin{equation}
    \begin{split}
    &2\left((\ee^{-\ii\phi}g_{r}+k)\alpha_{0r}+\frac{1}{r}\left(\ee^{-\ii\phi}\frac{g_{\phi}}{r}+\ii k\right)\alpha_{0\phi}\right)\\
    &
    +     \left[\left(g_{rr}+\frac{1}{r}g_{r}+\frac{1}{r^{2}}g_{\phi\phi}\right)
    +\left(g_{r}+\frac{\ii}{r}g_{\phi}\right)\left((\ln 
    A)_{r}-\frac{\ii}{r}(\ln A)_{\phi}\right) 
    \right]\ee^{-\ii\phi}\alpha_0=0.
    \label{alpharadial}
        \end{split}
\end{equation}
This equation is treated numerically in a similar way as was the eikonal problem in Section~\ref{sec:numerical-eikonal} which allows the use of the same numerical grid. The 
derivatives of both the potential $A$ and the function $g$ are computed as described in Section~\ref{sec:spectral-approach}
with spectral methods in the two radial domains that meet at the unit circle. 

\begin{myremark}
    The derivatives of the solution $f$ in \eqref{alpharadial} 
    contain divisions by $r$. As can be checked for the exact solution 
   \eqref{eq:fW-Lorentzian} in the case of the Lorentzian potential 
   $A$ the terms divided by $r$ appearing in the action of the Laplacian on $f$ 
   do not all vanish for $r=0$. This implies that analytically 
   unbounded terms cancel which is numerically challenging even for a
   spectral method to resolve. Thus a loss in accuracy near $r=0$ is to be 
   expected in the computation of $\alpha_0$ via \eqref{alpharadial}.  
\end{myremark}

The numerical solution of \eqref{alpharadial} with $\alpha_0\to 1$ as $|z|\to\infty$ in the case $A(x,y)=1/(1+x^{2}+y^2)$ and $S\equiv 0$ is 
shown for $r<1$ in the left-hand panel of Figure~\ref{alphalorentz} (note that the 
corresponding exact solution \eqref{eq:alphaW-Lorentzian} is symmetric under the reflection mapping 
$r\mapsto1/r$). The right-hand panel of the same figure shows a plot of the corresponding spectral coefficients. It can be seen that the solution is well resolved for 
$N_{r}=N_{\phi}=64$. The divisions by $r$ in the 
expressions in \eqref{alpharadial} lead, however, to a saturation 
level for the coefficients of the order of $10^{-10}$. Loosely 
speaking this is the level of the numerical error. 
\begin{figure}[htb!]
  \includegraphics[width=0.49\textwidth]{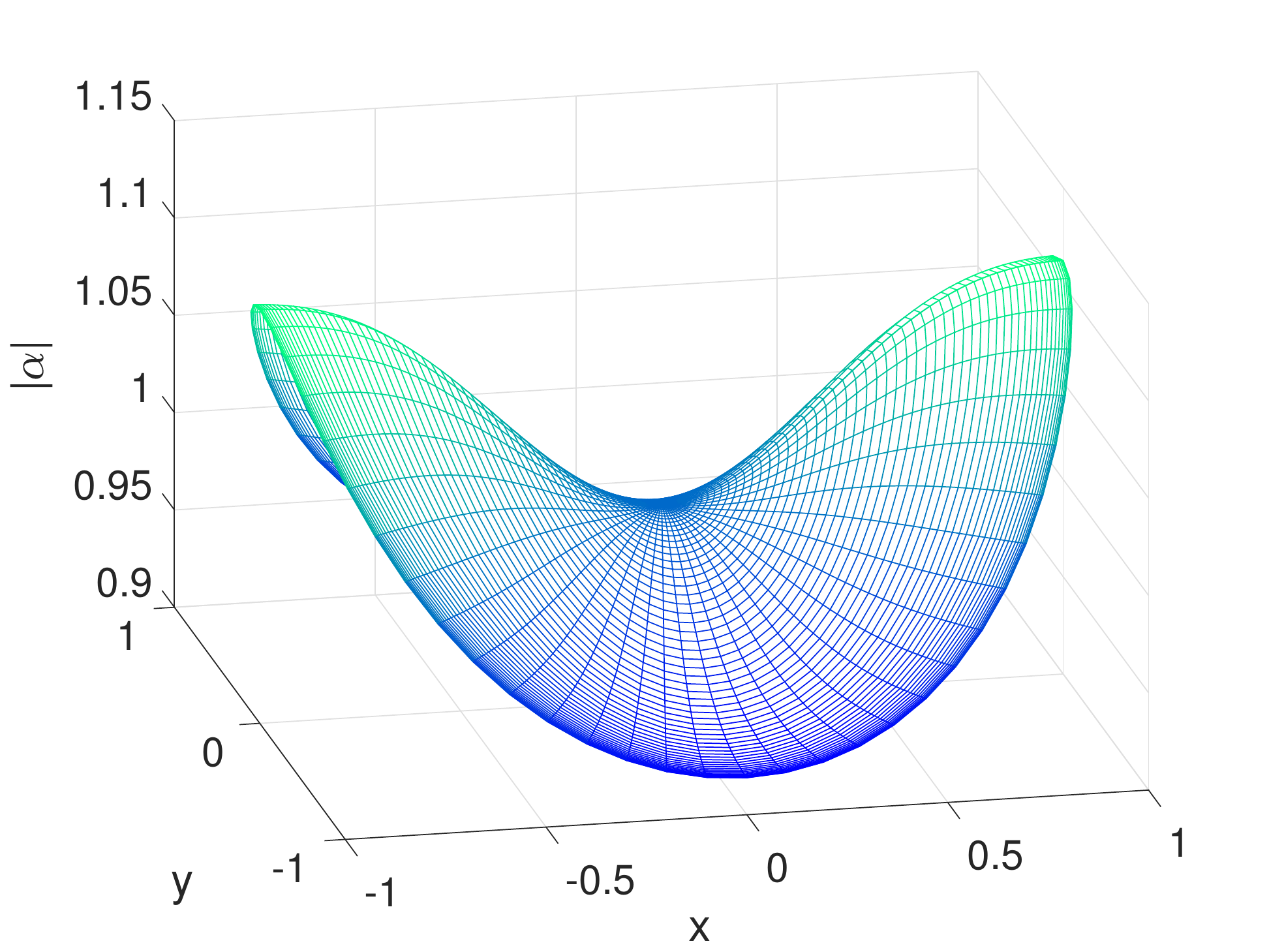}
  \includegraphics[width=0.49\textwidth]{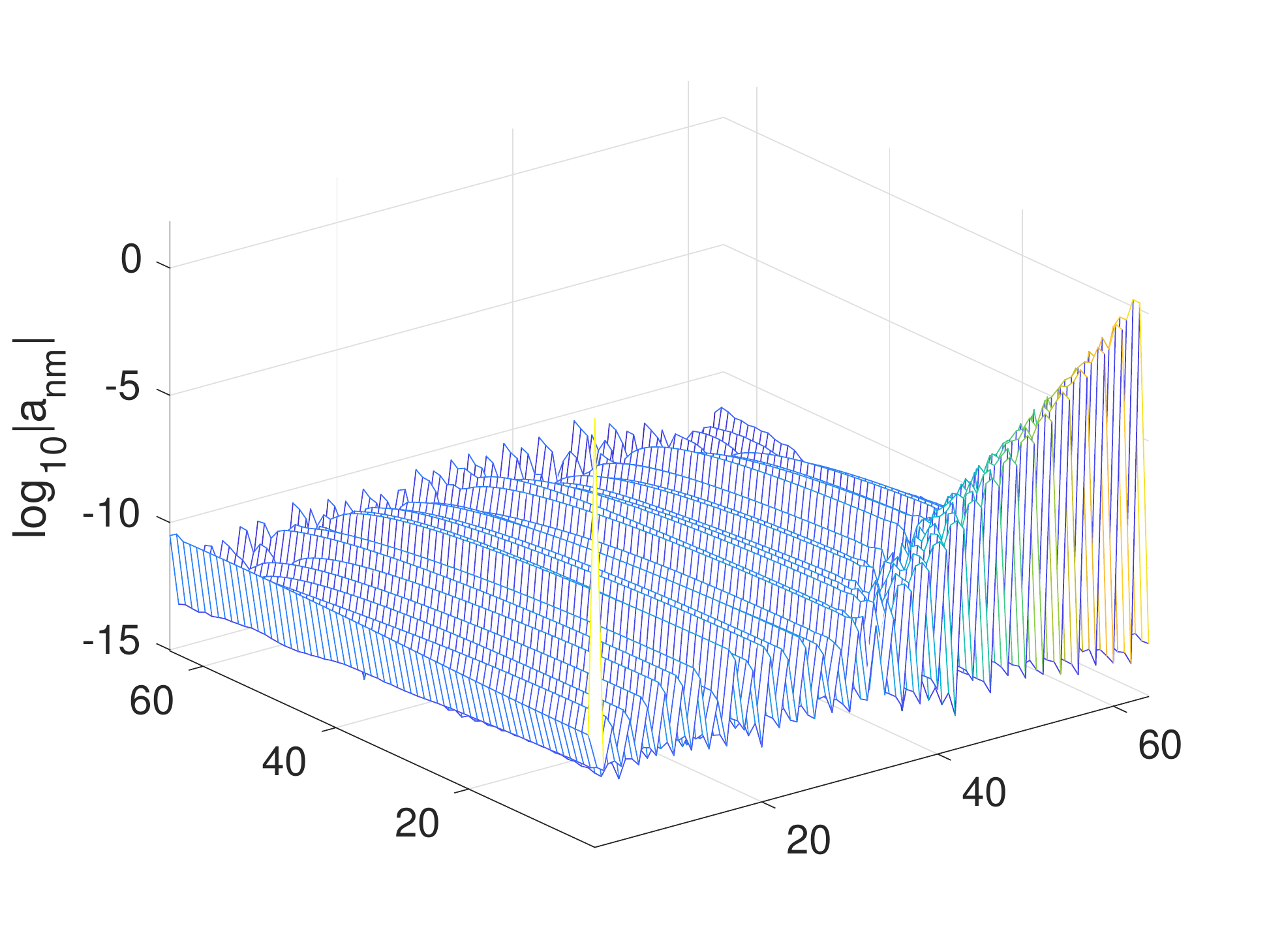}
 \caption{Left:  the numerical solution to the equation 
 \eqref{alpharadial} with $\alpha_0\to 1$ as $|z|\to\infty$ for the Lorentzian potential $A(r)=1/(1+r^{2})$ at $k=1$. Right: the modulus of the spectral 
 coefficients  in a logarithmic plot over the $(n,m)$-grid.}
 \label{alphalorentz}
\end{figure}

The expected accuracy is thus of the order of $10^{-10}$, and 
that this is indeed the case can be seen in Figure~\ref{alphatest} 
where the difference between the numerical and the exact solution \eqref{eq:alphaW-Lorentzian} is shown 
for $r<1$ on the left and for $r>1$ on the right. It appears that the largest 
errors occur for $r=0$. For other values of the radius, the 
numerical error is of the order $10^{-11}$. This shows that the 
solution of \eqref{alpharadial} can be obtained with high accuracy on 
the whole complex plane. 
\begin{figure}[htb!]
  \includegraphics[width=0.49\textwidth]{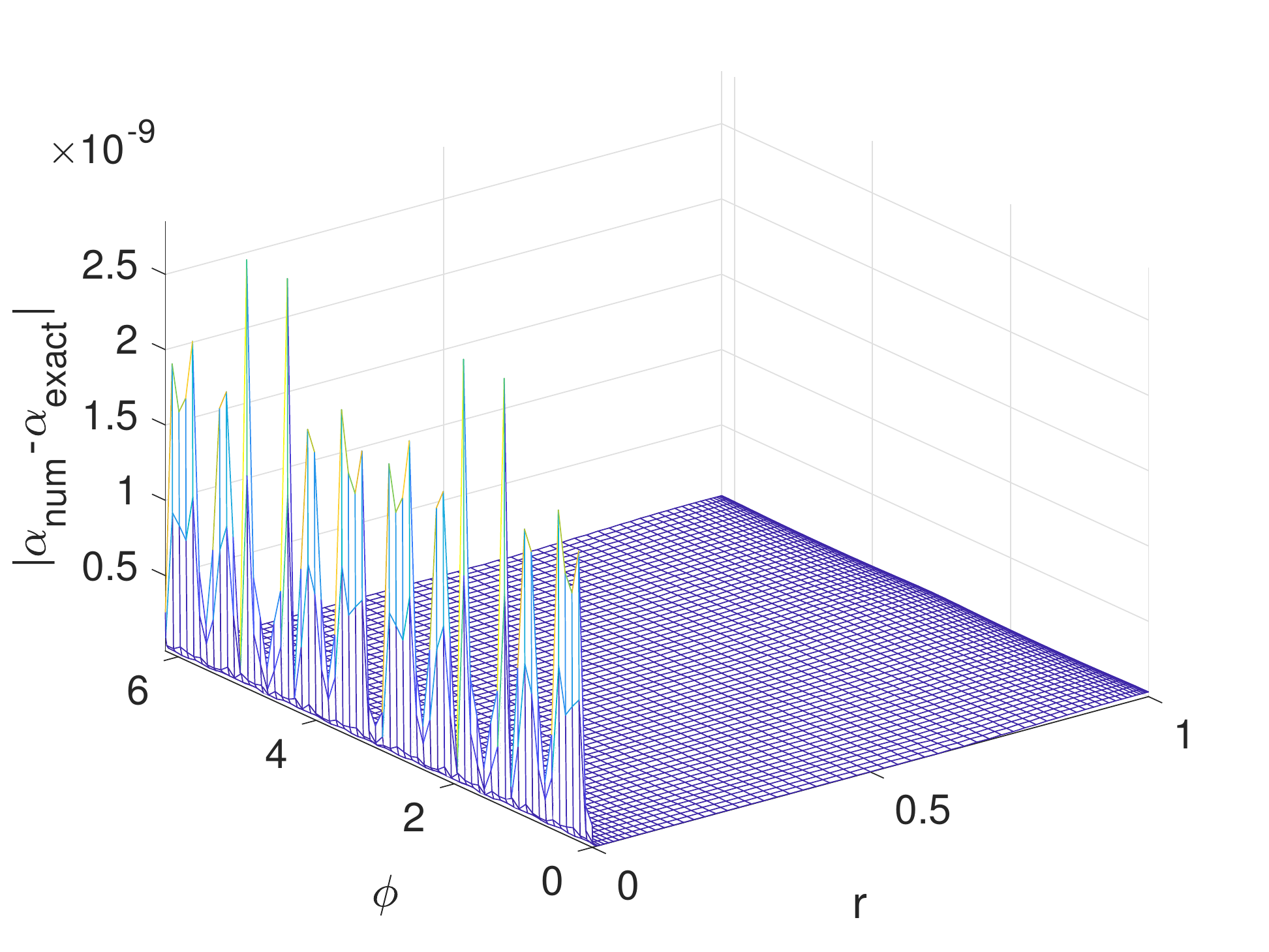}
  \includegraphics[width=0.49\textwidth]{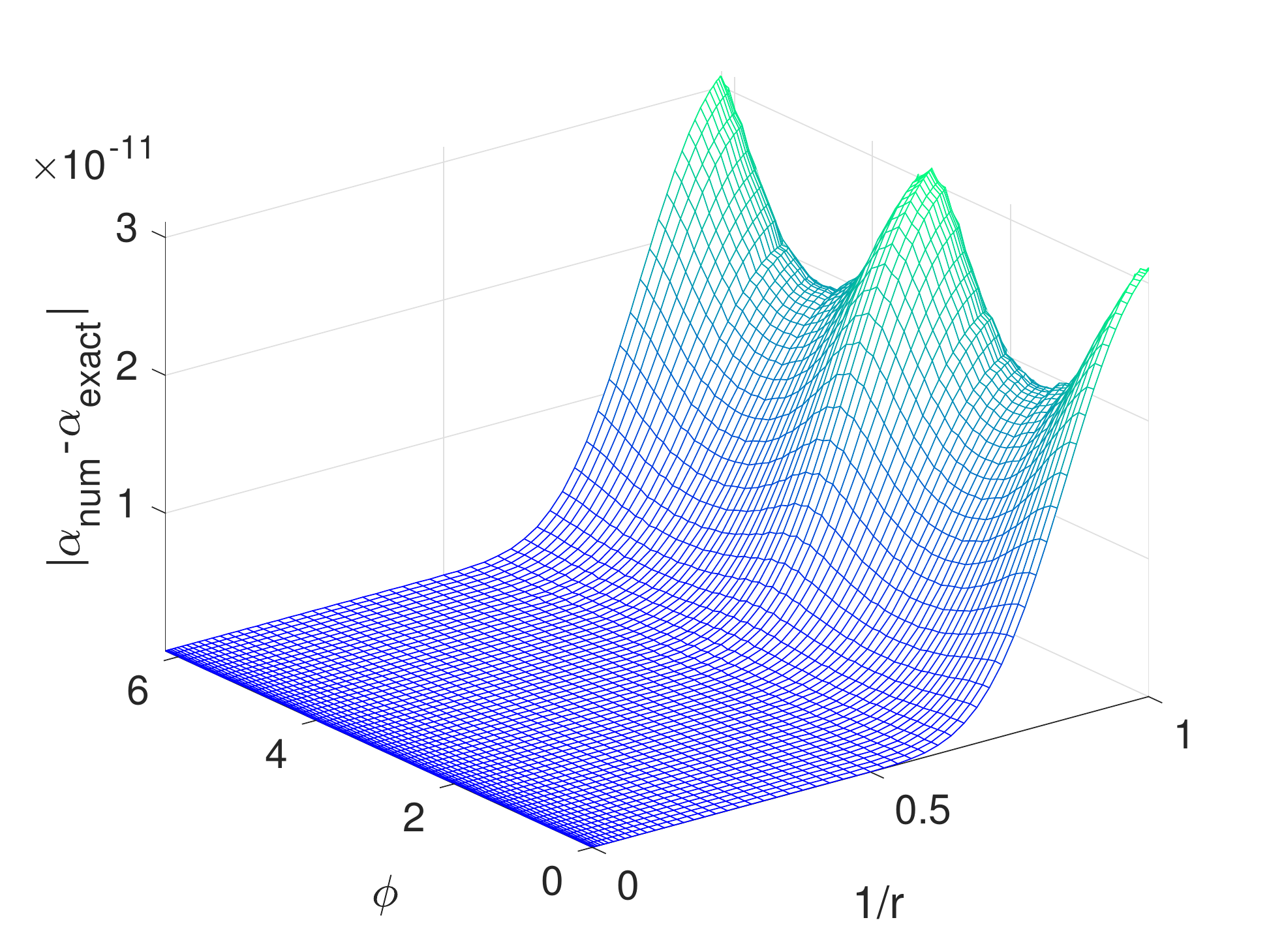}
 \caption{Difference between the numerical solution to the equation 
 \eqref{alpharadial} for the Lorentzian potential $A=1/(1+r^{2})$ and the exact 
 solution \eqref{eq:alphaW-Lorentzian} on the numerical grid. Left: $r<1$. Right: $r>1$.}
 \label{alphatest}
\end{figure}

\subsection{A spectral method for the $\epsilon$-dependent direct 
scattering problem for Schwartz class potentials}
\label{sec:numerical-dbar}
A spectral approach to solve the $\bar{\partial}$-problem for potentials in the Schwartz
class was developed by Klein and McLaughlin in \cite{KleinMcL}. We briefly summarize the approach here;
the reader is referred to \cite{KleinMcL} for details. 

The idea is to introduce the functions
\begin{equation}
    m^{\pm}(z;k) = \ee^{-kz/\epsilon}\left(\psi_{1}(z;k)\pm \overline{\psi_{2}(z;k})\right)-1,
    \label{mpm}
\end{equation}
satisfying the boundary conditions $\lim_{z\to\infty}m^{\pm}(z;k)=0$.   In these variables, the system 
\eqref{eq:1.10} becomes diagonal and takes the form
\begin{equation}
    \delbar m^{\pm}=\pm\frac{q}{2\epsilon}\ee^{(\overline{kz}-kz)/\epsilon}\cdot(\overline{m^{\pm}}+1).    \end{equation}
We write both of these equations in the common form
\begin{equation}
\delbar m = \frac{Q}{2\epsilon}\ee^{(\overline{kz}-kz)/\epsilon}\cdot(\overline{m}+1)
\label{mpmsys}
\end{equation}
where $m=m^\pm$ and $Q=\pm q$.  
Since under the Fourier transform $\mathcal{F}$ (cf., \eqref{eq:W-norm}, here scaled to be $L^2(\mathbb{R}^2)$-unitary) we have $\mathcal{F}\{\delbar m\}=\tfrac{1}{2}\ii (\xi_x+\ii\xi_y)\mathcal{F}\{m\}=\tfrac{1}{2}\ii\xi\mathcal{F}\{m\}$ for $\xi=\xi_x+\ii\xi_y$ the dual Fourier variable to $z=x+\ii y$, 
in the Fourier domain the system \eqref{mpmsys} becomes 
\begin{equation}
    S(\xi) =  - \ii \mathcal{F}\left\{ \frac{Q}{\epsilon}\ee^{(\overline{kz}-kz)/\epsilon}\cdot  \overline{ \mathcal{F}^{-1} \left\{ \frac{1}{\xi} S(\xi) \right\} }   \right\} 
 -\ii  \mathcal{F}\left\{  \frac{Q}{\epsilon}\ee^{(\overline{kz}-kz)/\epsilon}   \right\}
    \label{iefourier}
\end{equation}
where 
$S := \xi \mathcal{F}\{m\}=-2\ii\mathcal{F}\{\delbar m\}$. 

In the numerical approach \cite{KleinMcL} the Fourier transforms in 
\eqref{iefourier} are approximated by discrete Fourier transforms 
computed by a two-dimensional FFT. The integrand in \eqref{iefourier} 
is regularized in the form  
\begin{equation}
    \mathcal{F}^{-1} \left\{ \frac{1}{\xi} S(\xi) \right\} = \mathcal{F}^{-1} \left\{ \frac{1}{\xi} \left( S(\xi)  - G(\xi) \right) \right\}  + \mathcal{F}^{-1} \left\{ \frac{1}{\xi} G(\xi) \right\}, 
    \label{reg}
\end{equation}
where $G(\xi)$ is chosen such that 
$(S(\xi)-G(\xi))/\xi$ is regular to machine precision (as indicated by the fact that the Fourier coefficients decrease exponentially to the order of the roundoff
error), and 
also such that $\mathcal{F}^{-1}\{G(\xi)/\xi\}$  can be computed explicitly. 
A useful choice for $G$ is
\begin{equation}
    G(\xi)=\ee^{-|\xi|^{2}}\sum_{n=0}^{M}\frac{\delbar_\xi^{\,n}S(0)}{n!}\overline{\xi}^{\,n},\quad\delbar_\xi:=\frac{1}{2}\left(\frac{\partial}{\partial\xi_x}+\ii\frac{\partial}{\partial\xi_y}\right),
    \label{Gchoice}
\end{equation}
since it cancels the most offending terms in $S$ near the origin, while $\mathcal{F}^{-1}\{G(\xi)/\xi\}$ can be calculated with the help of the identity
\begin{equation}
     \mathcal{F}^{-1} \left\{\frac{ \overline{\xi}^{\,n} }{\xi}\ee^{ - |\xi 
    |^{2} } \right\}= 
    \ii(2\ii)^{n}\frac{n!}{z^{n+1}}\left[1-\ee^{-|z|^{2}/4}\sum_{k=0}^{n}
    \frac{1}{k!}\left(\frac{|z|^2}{4}\right)^{k}
    \right].
    \label{Gexplicit}
\end{equation}

The factor $e^{(\overline{kz}-kz)/\epsilon}$ appearing in \eqref{iefourier} leads to a shift in 
Fourier space of the Fourier transform of a function multiplied by 
it.  Indeed, if we introduce the shift operator $\mathcal{S}_{k/\epsilon}$ whose action on a function $f$ of $\xi$ is given by $\mathcal{S}_{k/\epsilon}f(\xi):=f(\xi+2\ii\overline{k}/\epsilon)$, then \eqref{iefourier} can be recast in the form
\begin{equation}
S(\xi)=\mathcal{S}_{k/\epsilon}\circ\mathcal{K}_0 S(\xi)+\mathcal{S}_{k/\epsilon}F(\xi),
\label{eq:shifted-equation}
\end{equation}
where the operator $\mathcal{K}_0$ and forcing function $F$ are independent of $k$:
\begin{equation}
\mathcal{K}_0S(\xi):=-\ii\mathcal{F}\left\{\frac{Q}{\epsilon}\cdot\overline{\mathcal{F}^{-1}\left\{\frac{1}{\xi}S(\xi)\right\}}\right\}\quad\text{and}\quad
F(\xi):=-\ii\mathcal{F}\left\{\frac{Q}{\epsilon}\right\}.
\end{equation}
As discussed in detail in \cite{KleinMcL}, for larger values of $|k|/\epsilon$ the effect of the shift is that the benefit of the regularization procedure \eqref{reg} is diminished because it is effectively removing a singularity that is not present at all since the shifted transform is large near the boundary of the (spectral) computational domain but vanishes to machine precision near the origin $\xi=0$.
%

To address this problem, the equation \eqref{eq:shifted-equation} may be replaced by a system of equations for two functions, $f$ and $h$:
\begin{equation}
\begin{split}
h&=\mathcal{K}_0f + F\\
f&=\mathcal{S}_{k/\epsilon}\circ\mathcal{K}_0\circ\mathcal{S}_{k/\epsilon}h.
\end{split}
\label{hf}
\end{equation}
It is a direct matter to check that if $(f,h)$ solves \eqref{hf}, then $S=f+\mathcal{S}_{k/\epsilon}h$ solves \eqref{eq:shifted-equation}.  However, since it turns out that for large $|k|/\epsilon$ both functions $f$ and $h$ are small near the boundary of the spectral computational domain, the system \eqref{hf} is better suited to regularization via \eqref{reg} than is \eqref{eq:shifted-equation} itself.
Moreover, to recover the reflection coefficient, the function $f$ is 
not needed, and it can therefore be explicitly eliminated from the first equation of \eqref{hf} using the second equation, leading to a closed equation for a single function $h$.  See \cite[Section 5.2]{KleinMcL} for details.

Numerically this integral equation is solved by standard discretization amenable to the two-dimensional FFT.
The resulting system of algebraic equations is not complex linear in 
 $h$ due to the complex conjugation present in the operator $\mathcal{K}_0$, but rather real linear in its real and
imaginary parts.  This linear system is solved with GMRES 
\cite{GMRES}, a Krylov subspace
approach that is especially useful in our setting because it avoids the necessity of storage of the coefficient matrix.  As discussed in \cite{KleinMcL}, the numerical error in the solution is of the order of the Fourier 
coefficients of the largest values of $\xi$ carried in the computation.    

Recall that $S=-2\ii\mathcal{F}\{\delbar m\}$, so that once $S$ is found, $\delbar m$ is available via the (spectrally-accurate) FFT. 
As discussed in \cite{KleinMcL},  
knowledge of $\delbar m^\pm$ is sufficient to compute the reflection coefficient. In 
order to obtain $m^{\pm}$, as will be needed to compare numerical solutions with the WKB approximations introduced in Section~\ref{sec:WKB}, we invert $\delbar$ 
in the Fourier domain for $S=f+\mathcal{S}_{k/\epsilon}h$ via division by $\xi$ and using again the 
regularization procedure \eqref{reg} with a shift in the 
Fourier domain for the second term in the expression for $S$ for $k\neq 0$. It is important to realize that 
the quantities $m^{\pm}$ decrease only as $1/|z|$ for $z\to\infty$ 
and are thus not themselves suitable for a Fourier spectral approach (the 
periodically continued functions would not be differentiable at the 
computational boundary), but that the function $(S(\xi)-G(\xi))/\xi$ is 
in the Schwartz class. For the latter term, Fourier spectral methods on a 
sufficiently large computational domain are very efficient and show 
spectral convergence, which is controlled as always by the decay of 
the modulus of the Fourier coefficients at the boundaries of the 
computational domain in Fourier space. 

As an example of the result of a computation using
this numerical approach, we show the 
solutions to the Dirac system \eqref{eq:1.10} with normalization \eqref{eq:psi-norm} obtained with a Gaussian 
potential $A(x,y)=\ee^{-(x^2+y^2)}$ and $S(x,y)\equiv 0$ for $k=0$ and $\epsilon=1$ in Figure~\ref{gausspsi12k0e1}. 
The function $\psi_{1}$ has minimal modulus at the origin and tends 
to $1$ at infinity, whereas the function $\psi_{2}$ vanishes at the 
origin and decreases slowly to $0$ as $z\to\infty$. 
\begin{figure}[htb!]
  \includegraphics[width=0.49\textwidth]{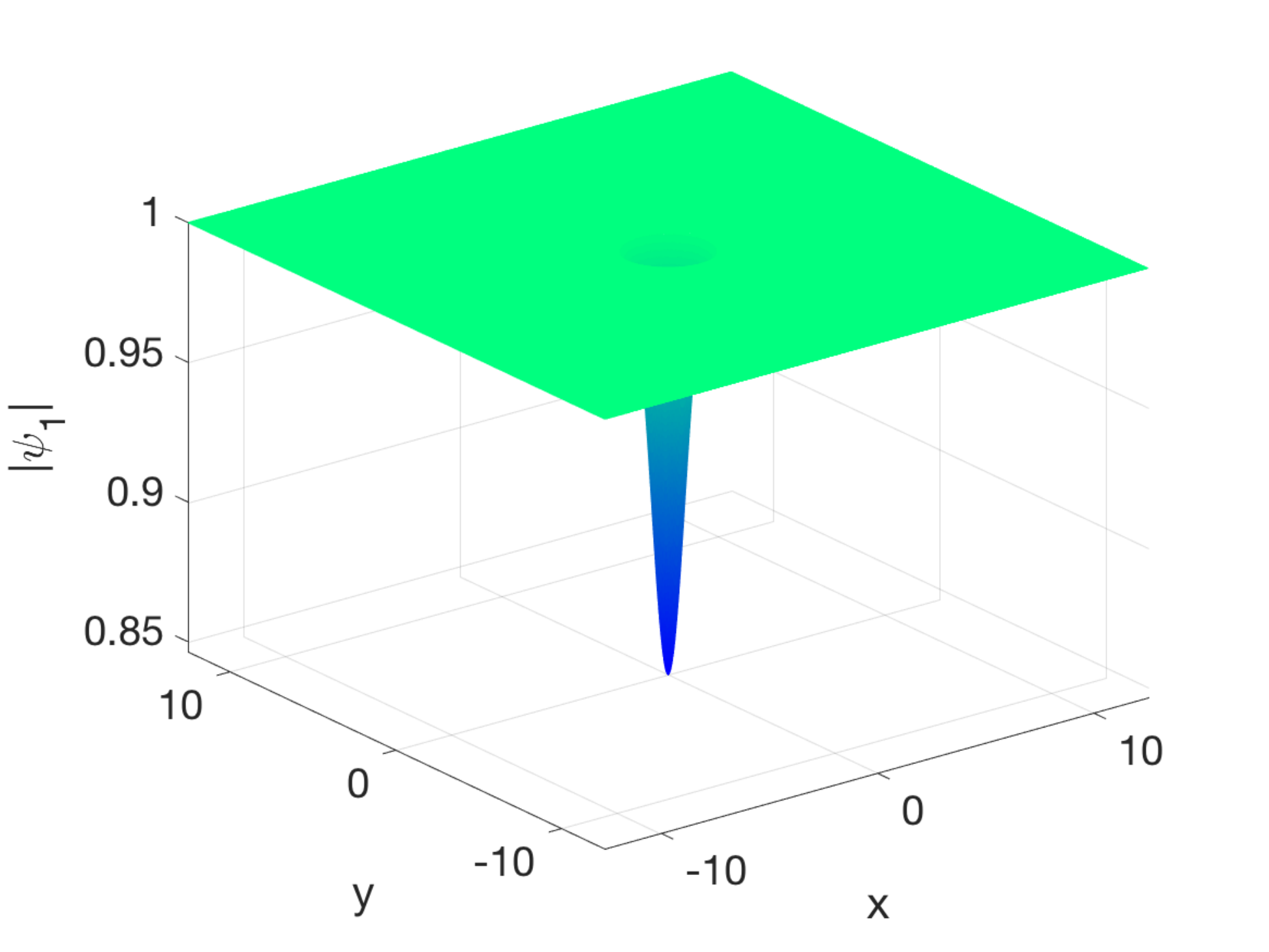}
  \includegraphics[width=0.49\textwidth]{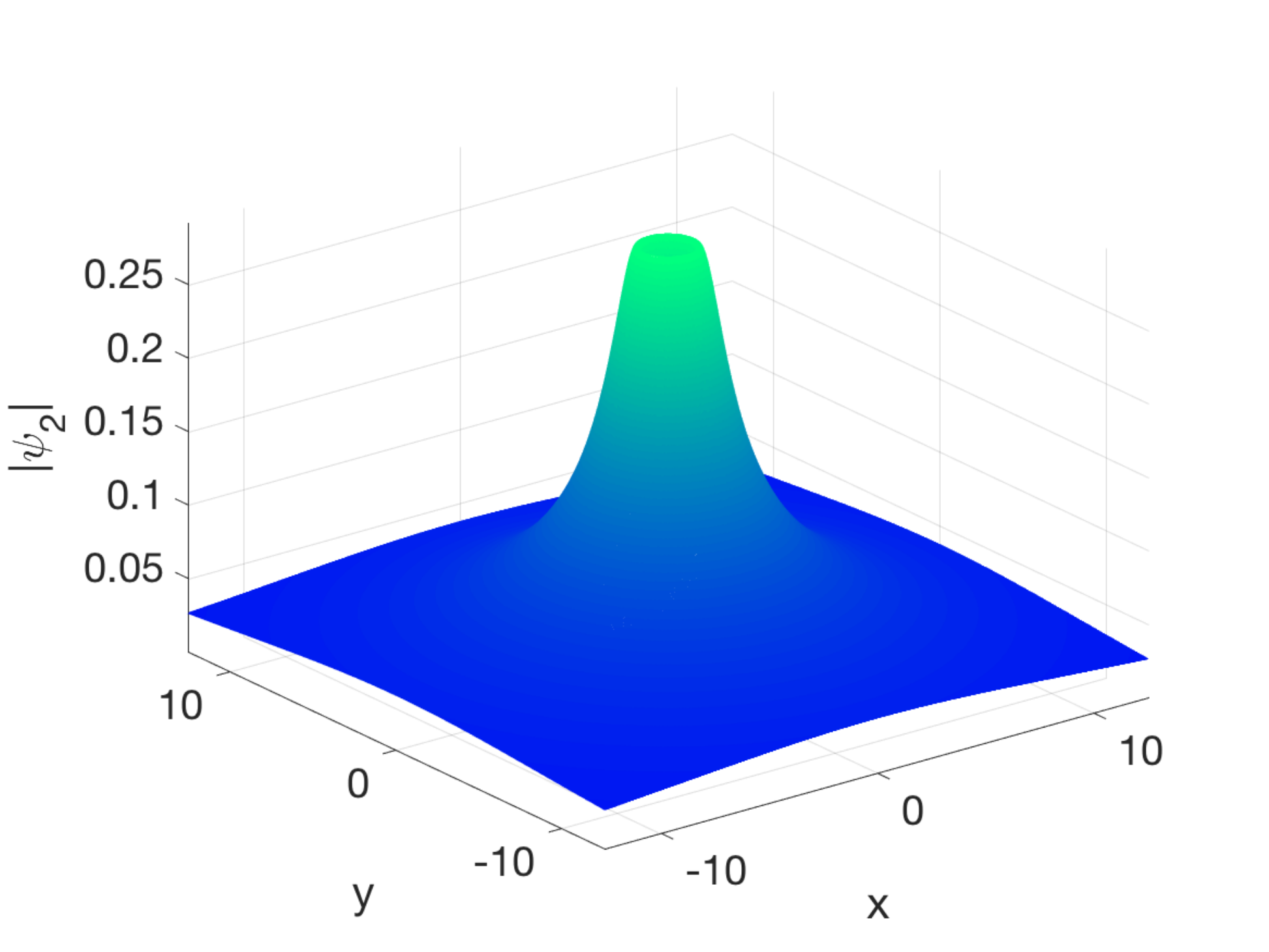}
 \caption{Numerical solution to the Dirac system \eqref{eq:1.10}--\eqref{eq:psi-norm} with a Gaussian 
 potential for $k=0$ and $\epsilon=1$. Left:  modulus of $\psi_{1}$.  
 Right: modulus of $\psi_{2}$.  }
 \label{gausspsi12k0e1}
\end{figure}

\section{Numerical examples}
\label{sec:numerical-examples}
In this section we test the conjectures formulated in the previous 
sections for several examples with and without radial symmetry. We 
first address the case of a Gaussian potential for various values of 
$\epsilon$ and compare the solution to the  Dirac system  
\eqref{eq:1.10} for 
sufficiently large $k$ to the leading order semiclassical solution 
built from the solution to  the eikonal problem. A similar study  is 
presented for a non-radially symmetric potential in the Schwartz 
class. 

\subsection{Gaussian potential}
\label{sec:numerical-Gaussian}
As a first computational example outside the realm of potentials $A(x,y)$ for which the eikonal problem \eqref{eq:eikonal}--\eqref{eq:f-norm} has a known solution, we consider here
the Gaussian
\begin{equation}
    A(x,y)=A(r) = \ee^{-r^{2}}
    \label{gauss}
\end{equation}
as a canonical example of a smooth, rapidly decaying, and radially-symmetric potential.
Since the reflection coefficient is a function of $|k|$ only for radial potentials such as \eqref{gauss},
we will here restrict attention to real positive $k$.

Firstly, we numerically solve the eikonal problem for this potential
using the series approach of 
Section~\ref{sec:numerical-Fourier-series} with discretization parameter $N_{r}=40$ and 
200 terms in the series (\ref{gseries}).
%
The coefficients $\{c_{n}(r)\}$ as computed via \eqref{gseries} have $L^{\infty}(\mathbb{R}_+)$ norms 
exhibiting algebraic decay as $n\to\infty$ as suggested by 
Figure~\ref{gausskcoeff}, where a log-log plot of $\|c_{n}(\cdot)\|_{\infty}$ 
is shown on the left. The essentially linear behavior of the plot for  
larger values of $n$ indicates algebraic (predominantly power-law) decay as $n\to\infty$.
We can fit the norms $\{\|c_{n}(\cdot)\|_{\infty}\}$ with a least squares 
method to $\ln \|c_{n}(\cdot)\|_{\infty}\sim -\alpha n-\beta \ln n - \gamma$ 
and find 
    $\alpha=10^{-4}$, $\beta=    1.0951$ and 
    $\gamma=1.1122$ for values of $n>20$ (the results do not change 
    much if the fitting is done for $n>50$). The results of the 
    fitting can be seen in the right-hand panel of 
    Figure~\ref{gausskcoeff} in the form of the quantity $\Delta:=\ln 
    \|c_{n}(\cdot)\|_{\infty}-(-\alpha n-\beta \ln n - \gamma)$. 
The fact that $\alpha$ is essentially zero while $\beta$ is finite is strong numerical evidence that the series \eqref{gseries} converges for $|k|>\tfrac{1}{2}$ and diverges for $|k|<\tfrac{1}{2}$.  Of course this threshold value of $|k|=\tfrac{1}{2}$ is the known exact value for the Lorentzian potential, but for the Gaussian $A(x,y)=\ee^{-(x^2+y^2)}$ the best analytical estimate we have is, as explained in Section~\ref{sec:Gaussian}, that the $\delbar$ derivative of \eqref{gseries} converges in $\mathrm{W}(\mathbb{R}^2)$ if $|k|\ge 1$.  In general, there is obviously nothing special about the value $|k|=\tfrac{1}{2}$; indeed if $\{c_n(r)\}_{n=0}^\infty$ are the coefficients for the potential $A(r)$, then  from \eqref{c0}--\eqref{cn} we see that $\{M^{2n+2}c_n(r)\}_{n=0}^\infty$ are the coefficients for the rescaled potential $MA(r)$ for any $M>0$, and it follows that if $|k|=\tfrac{1}{2}$ is the threshold value for $A(r)$, then $|k|=\tfrac{1}{2}M$ is the threshold value for $MA(r)$.  The coincidence of threshold values for the Gaussian and Lorentzian potentials is perhaps related to the fact that for both potentials $\|A^2\|=1$ in $\mathrm{W}(\mathbb{R}^2)$ as well as in $L^\infty(\mathbb{R}^2)$, as explained in Sections~\ref{sec:Gaussian} and \ref{sec:Lorentzian} respectively.
In any case, since it is known from the explictly-solvable Lorentzian case that upon decreasing $|k|$ through the convergence threshold singularities appear in the solution $f(x,y;k)$ of the eikonal problem \eqref{eq:eikonal}--\eqref{eq:f-norm} at certain points in the $(x,y)$-plane, we may reasonably conjecture that some kind of singularity formation for a critical value of $|k|$ is a generic feature at least for radial potentials.
\begin{figure}[htb!]
  \includegraphics[width=0.49\textwidth]{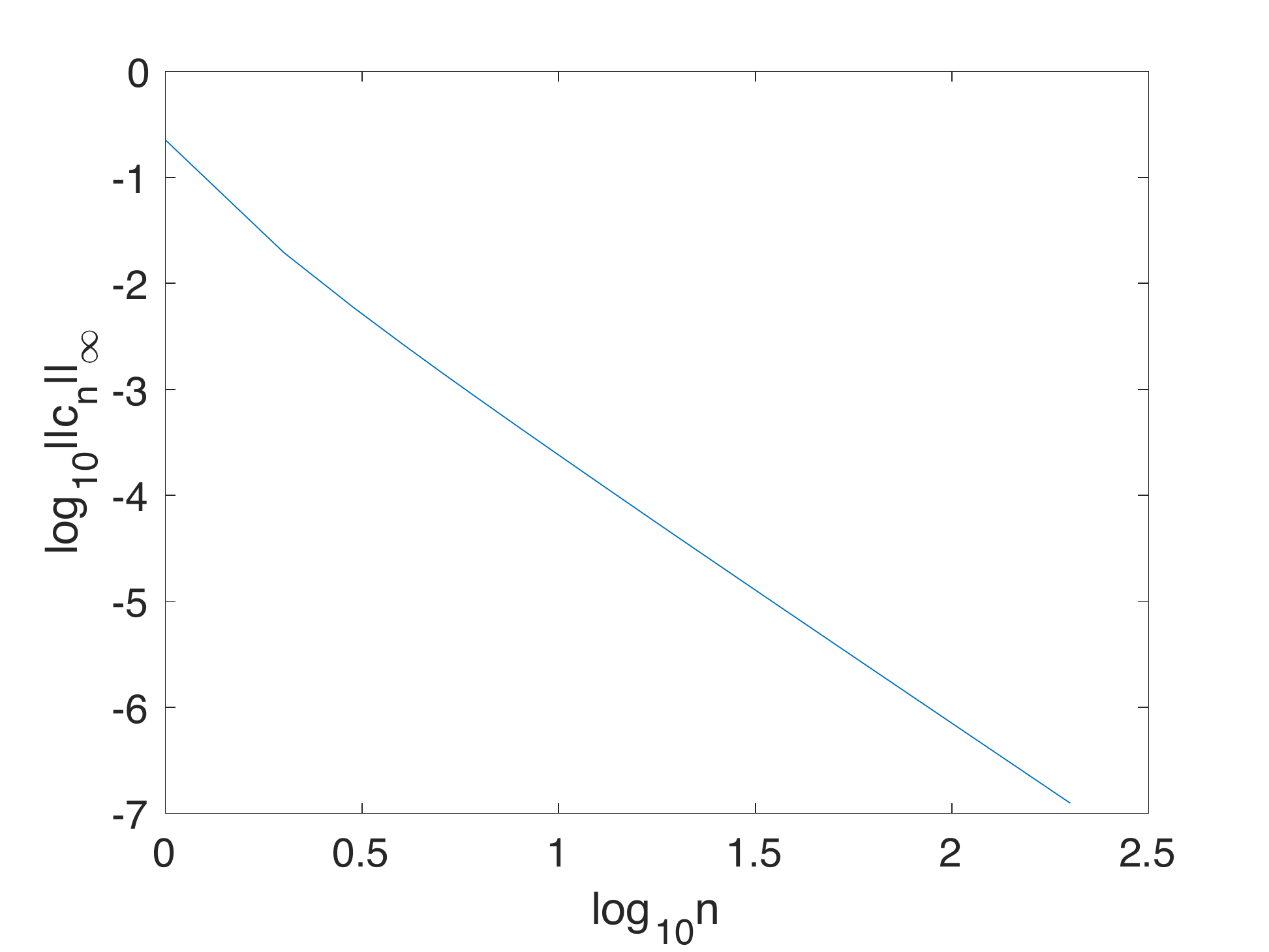}
  \includegraphics[width=0.49\textwidth]{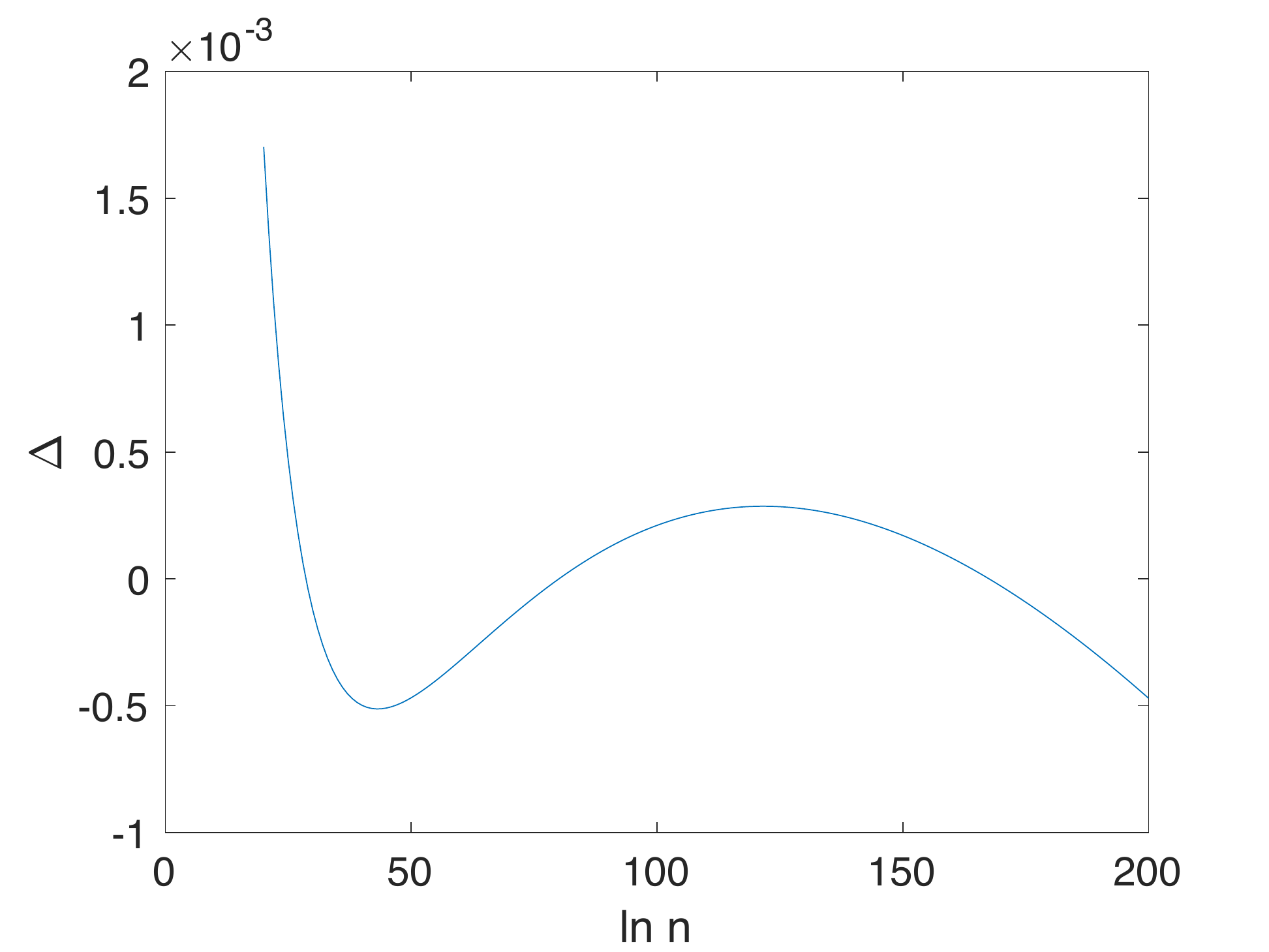}
 \caption{Left:  a log-log plot of the $L^\infty(\mathbb{R}_+)$ norms of the coefficients $\{c_{n}(r)\}$
 appearing in the series solution \eqref{gseries} of the 
 eikonal problem \eqref{eq:eikonal}--\eqref{eq:f-norm} for a Gaussian 
 potential on the left.  Right: the quantity $\Delta:=\ln 
    \|c_{n}(\cdot)\|_{\infty}-(-\alpha n-\beta \ln n - \gamma)$ obtained 
    after a linear least-squares regression.  }
 \label{gausskcoeff}
\end{figure}

To solve the Dirac system \eqref{eq:1.10} with normalization \eqref{eq:psi-norm} for the Gaussian potential 
for various values of $\epsilon$, we use the approach of 
Section~\ref{sec:numerical-dbar} with $N_{x}=N_{y}=2^{12}$ Fourier 
modes for $(x,y)\in 4[-\pi,\pi]\times 4[-\pi,\pi]$. The first
row of Figure~\ref{fig:conjecture} shows plots of the
modulus (scaled by $\ee^{-kz/\epsilon}$)  of 
the components of the solution obtained for $k=1$ and $\epsilon=1/16$. 

In order to compare solutions to the eikonal problem  
\eqref{eq:eikonal}--\eqref{eq:f-norm} for a given potential as well 
as the corresponding leading-order normalization function 
$\alpha_{0}$ to a solution to the $\epsilon$-dependent direct scattering problem  
\eqref{eq:1.10}--\eqref{eq:psi-norm}, we have to interpolate from the mixed 
Chebychev-Fourier (polar coordinate) grid used for $g$ and $\alpha_{0}$ to the two-dimensional 
Fourier (Cartesian) grid used for the computation of $\psi_{1}$ and $\psi_{2}$. 
There are efficient ways to do this. For simplicity we use here 
simply the definition of the spectral approximations of $g$ and 
$\alpha_{0}$. A function $f$ is approximated in each of the radial 
domains under consideration as 
\begin{equation}
    f(r,\phi)\approx 
    \sum_{n=0}^{N_{r}}\sum_{m=-N/2+1}^{N/2}c_{nm}T_{n}(l)e^{\mathrm{i}m\phi}.
    \label{cnm}
\end{equation}
Thus for given spectral coefficients $c_{nm}$, the corresponding 
function can be computed for arbitrary values of $r$ and $\phi$. For 
the Gaussian potential, the solution to the eikonal problem 
\eqref{eq:eikonal}--\eqref{eq:f-norm} can be seen after  
interpolation to a Cartesian grid in Figure~\ref{Fig:ggauss}.
\begin{figure}[htb!]
  \includegraphics[width=0.49\textwidth]{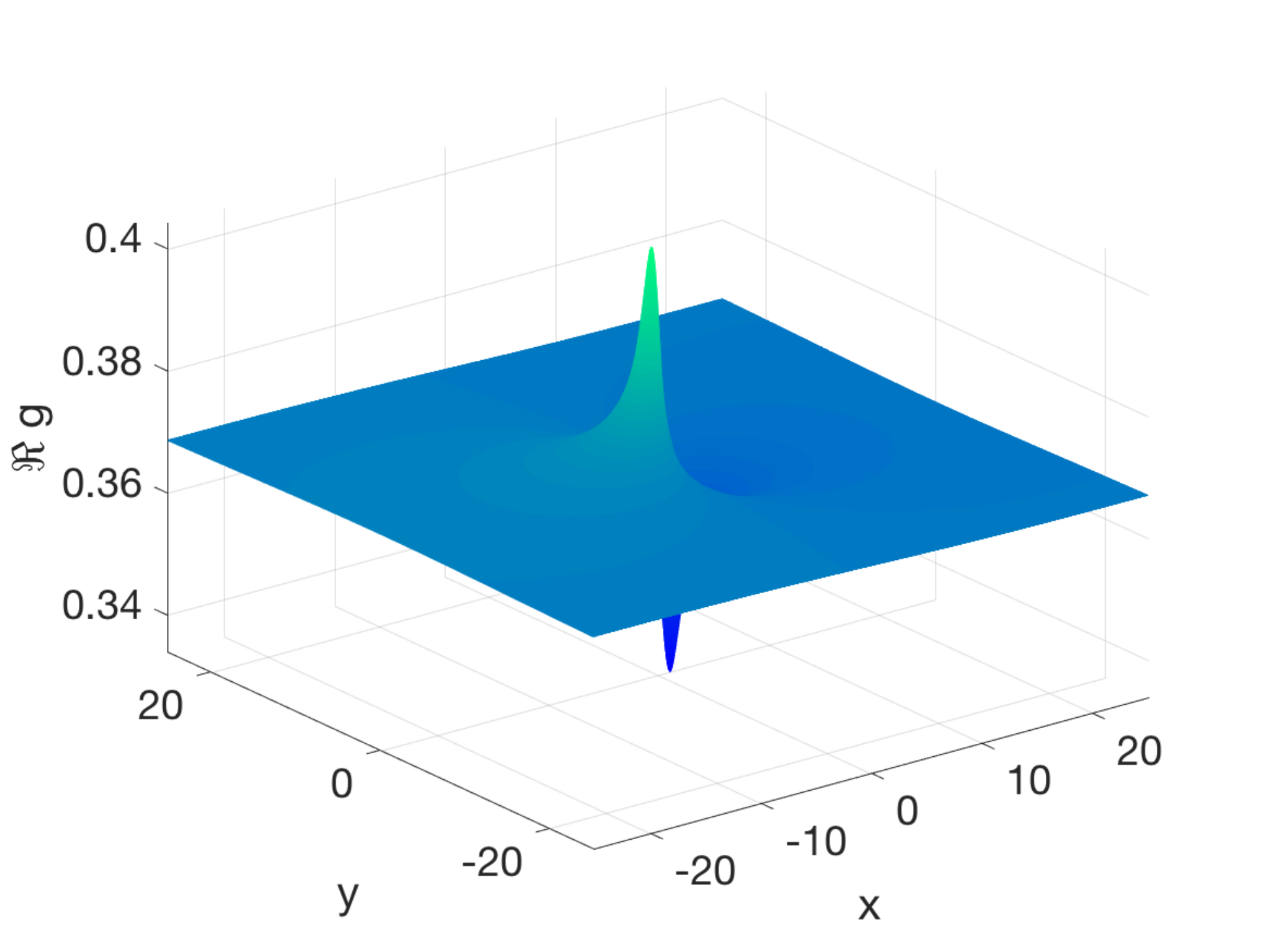}
  \includegraphics[width=0.49\textwidth]{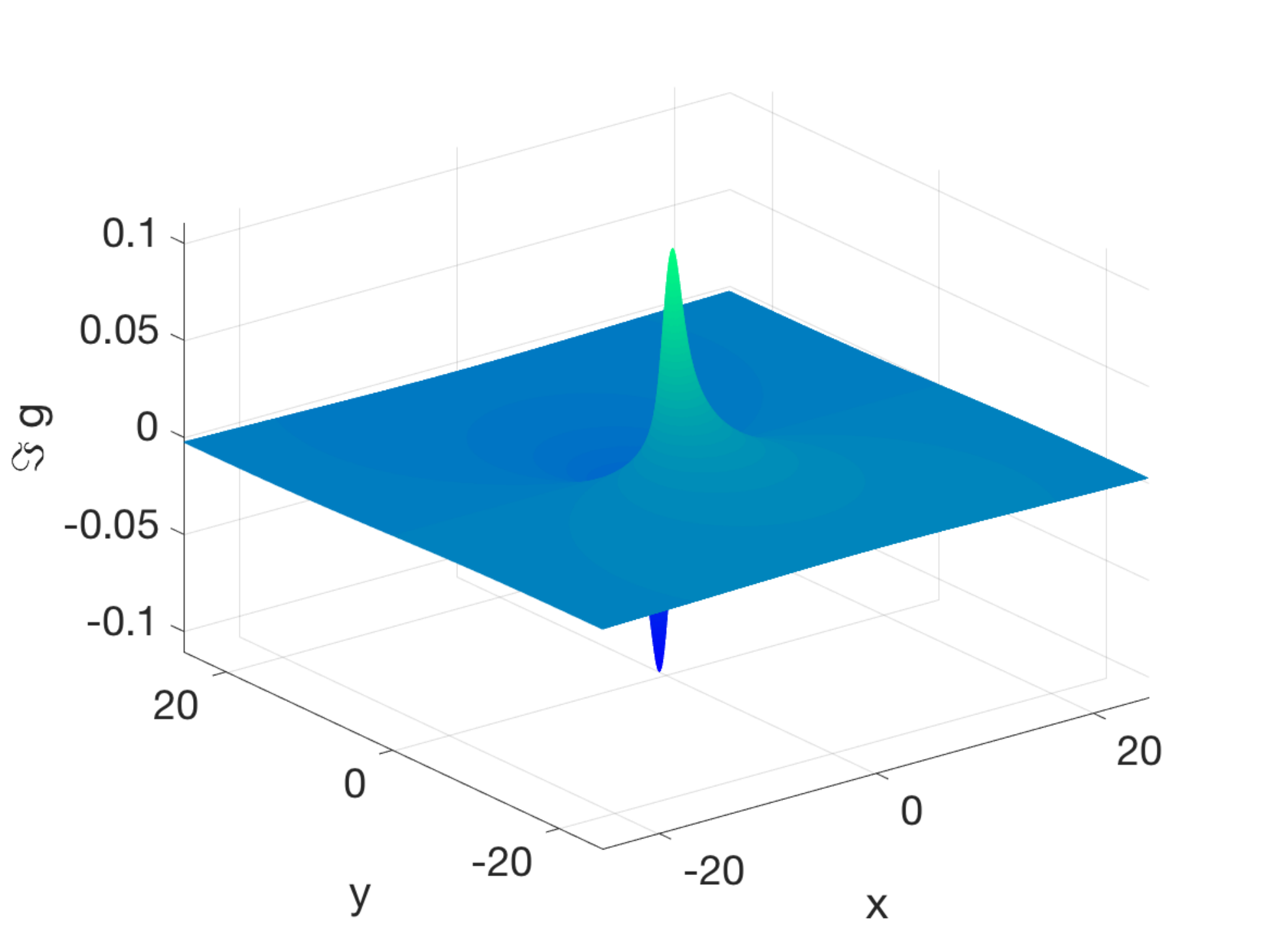}
 \caption{Solution to the 
 eikonal problem \eqref{eq:eikonal}--\eqref{eq:f-norm} for the Gaussian 
 potential with $k=1$.  Left: real part, right: imaginary part. }
 \label{Fig:ggauss}
\end{figure}
The corresponding Cartesian interpolation of the leading order normalization function $\alpha_{0}$ 
can be seen in Figure~\ref{Fig:alphagauss}.
\begin{figure}[htb!]
  \includegraphics[width=0.49\textwidth]{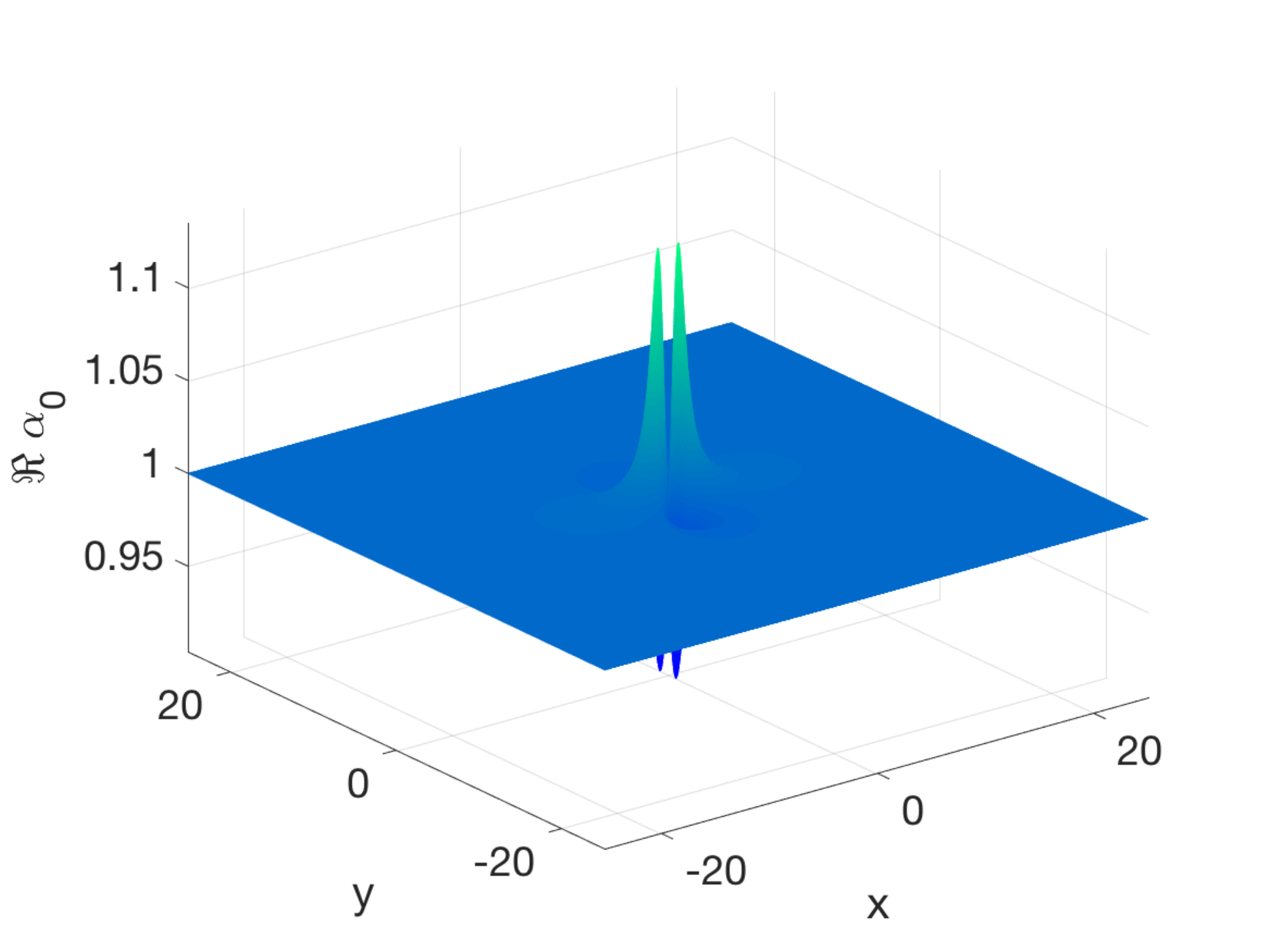}
  \includegraphics[width=0.49\textwidth]{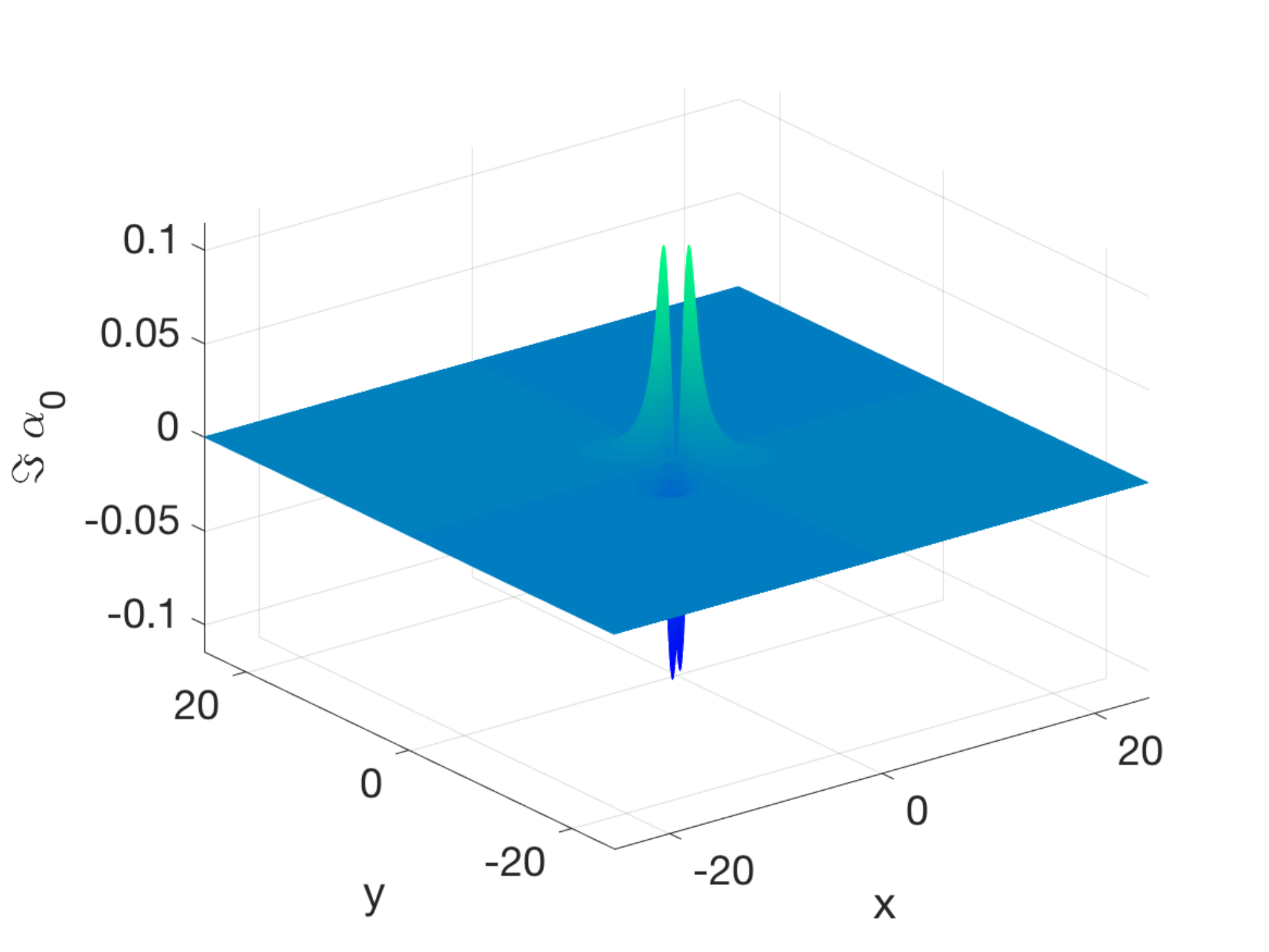}
 \caption{Leading order normalization function $\alpha_{0}$ for the Gaussian 
 potential with $k=1$.  Left: real part, right: imaginary part. }
 \label{Fig:alphagauss}
\end{figure}

With the numerical computations of $f$ and $\alpha_0$ complete, we may construct the leading term of the formal WKB approximation described in Section~\ref{sec:WKB} for the solution of the direct scattering problem \eqref{eq:1.10}--\eqref{eq:psi-norm}.  We now are in a position to compare this approximation with numerically-computed solutions to the direct scattering problem 
obtained as described in Section~\ref{sec:numerical-dbar}.  To quantify the comparison, we use \eqref{eq:conjecture-equation} for $S(x,y)\equiv 0$ to define 
the 
quantities
\begin{equation}
    \Delta_{1}:=\left|\psi_{1}\ee^{-f/\epsilon}-\frac{\alpha_0}{k}\del f\right|
    \label{D1}
\end{equation}
and 
\begin{equation}
    \Delta_{2}:=\left|\psi_{2}\ee^{-f/\epsilon}-\frac{\alpha_0 
    A}{2k}\right|.
    \label{D2}
\end{equation}
Conjecture~\ref{conjecture:WKB} asserts that both of these quantities should be proportional to $\epsilon$ as $\epsilon\downarrow 0$. For the Gaussian potential at $k=1$ we plot 
$\Delta_{1}$ and $\Delta_2$ for four values of $\epsilon$ in Figures~\ref{Delta1} and \ref{Delta2} respectively.
\begin{figure}[htb!]
  \includegraphics[width=\textwidth]{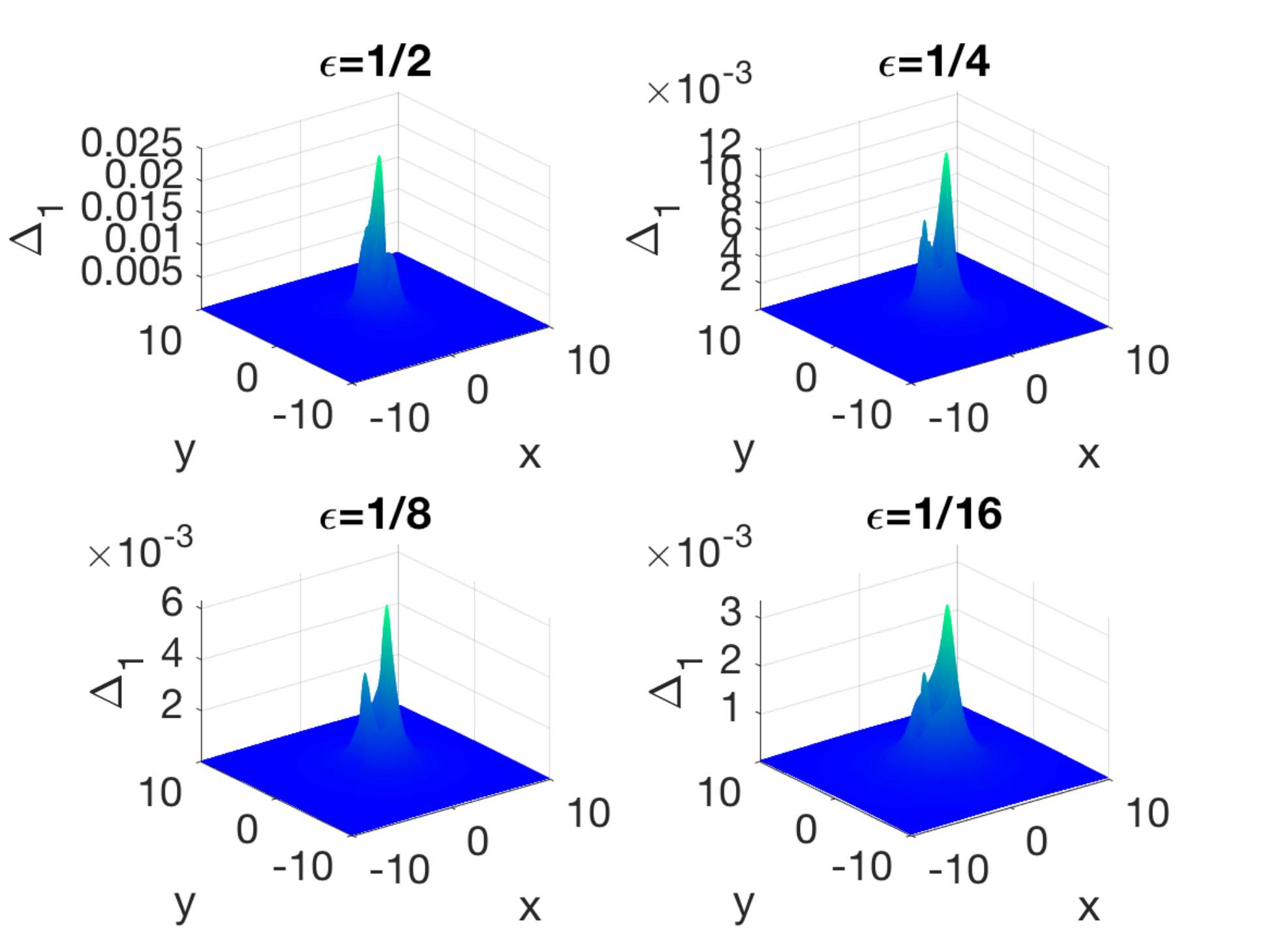}
 \caption{The quantity $\Delta_{1}$ of \eqref{D1} for the Gaussian 
 potential with $k=1$ for $\epsilon=1/2,1/4,1/8,1/16$.}
 \label{Delta1}
\end{figure}
\begin{figure}[htb!]
  \includegraphics[width=\textwidth]{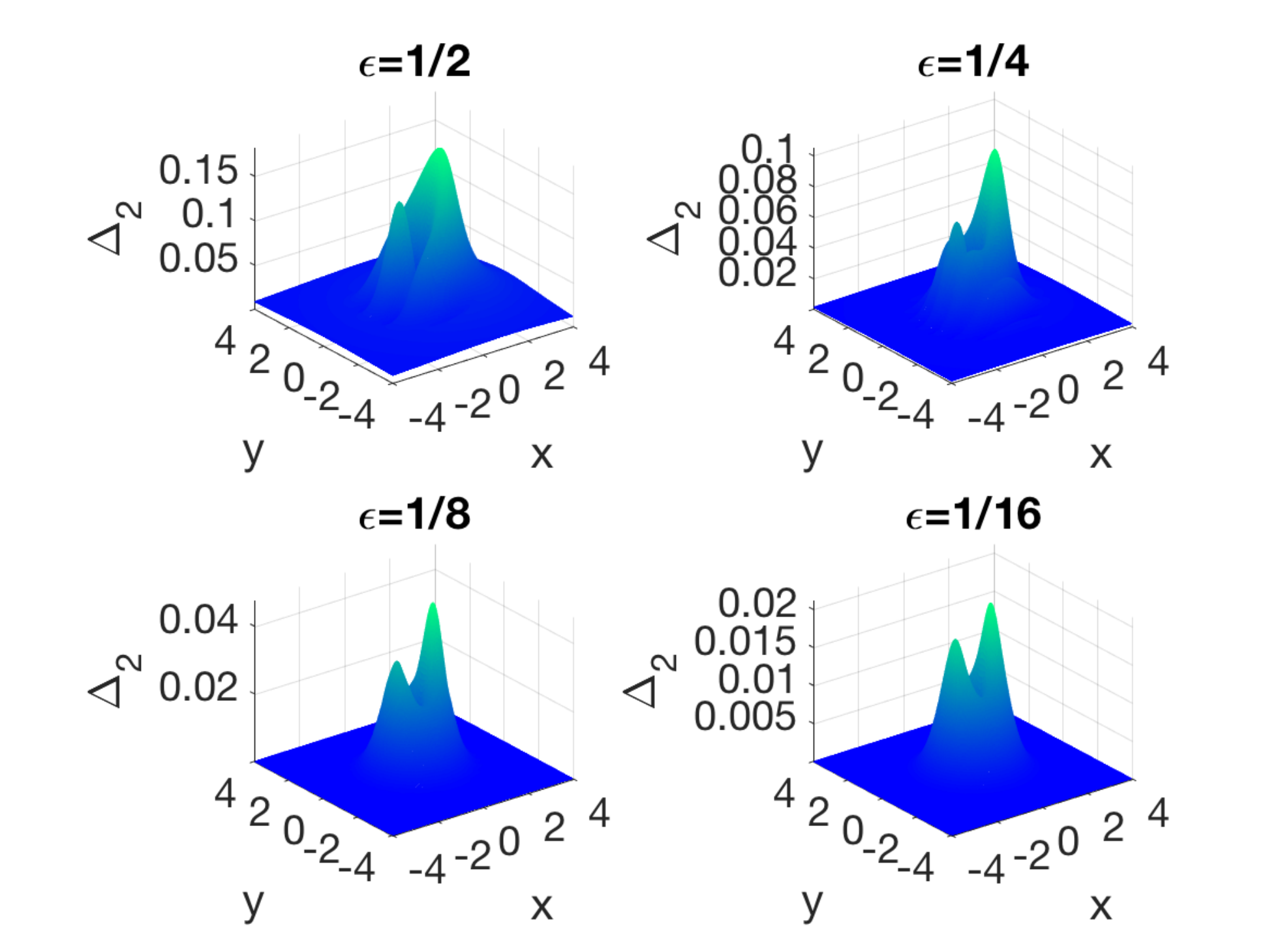}
 \caption{The quantity $\Delta_{2}$ of \eqref{D2} for the Gaussian 
 potential with $k=1$ for $\epsilon=1/2,1/4,1/8,1/16$. }
 \label{Delta2}
\end{figure}

The expected scaling in 
$\epsilon$ can indeed be seen from these plots, but it is even more obvious from
the results of a linear regression to determine the best fit to the logarithms of the  
$L^{\infty}$ norms of $\Delta_{1}$ and $\Delta_2$ as functions of $\ln(\epsilon)$ as is shown 
in Figure~\ref{Deltareg}. The data for the regression is calculated  
for the values $\epsilon=2^{0},2^{-1},\ldots,2^{-5}$, although we should keep in mind that for the larger values of $\epsilon$, accuracy of the WKB approximation might not be expected.
On the serial computers we used for our numerical simulations, we cannot go much lower 
than $\epsilon=0.04$ for lack of resolution.  The precise results of the linear regression are as follows.  In the left panel of 
Figure~\ref{Deltareg}, it can be seen that 
$\log_{10}\|\Delta_{1}\|_{\infty}\sim \alpha \log_{10}\epsilon+\beta$ 
with $\alpha=0.99$ and $\beta=-1.24$. In the same way we get for 
$\Delta_{2}$ the values $\alpha=0.99$ and $\beta=-0.46$ as can be 
seen in the right panel of Figure~\ref{Deltareg}. \emph{Thus in both cases 
the expected linear dependence in $\epsilon$ predicted by Conjecture~\ref{conjecture:WKB} is numerically confirmed.}
\begin{figure}[htb!]
  \includegraphics[width=0.49\textwidth]{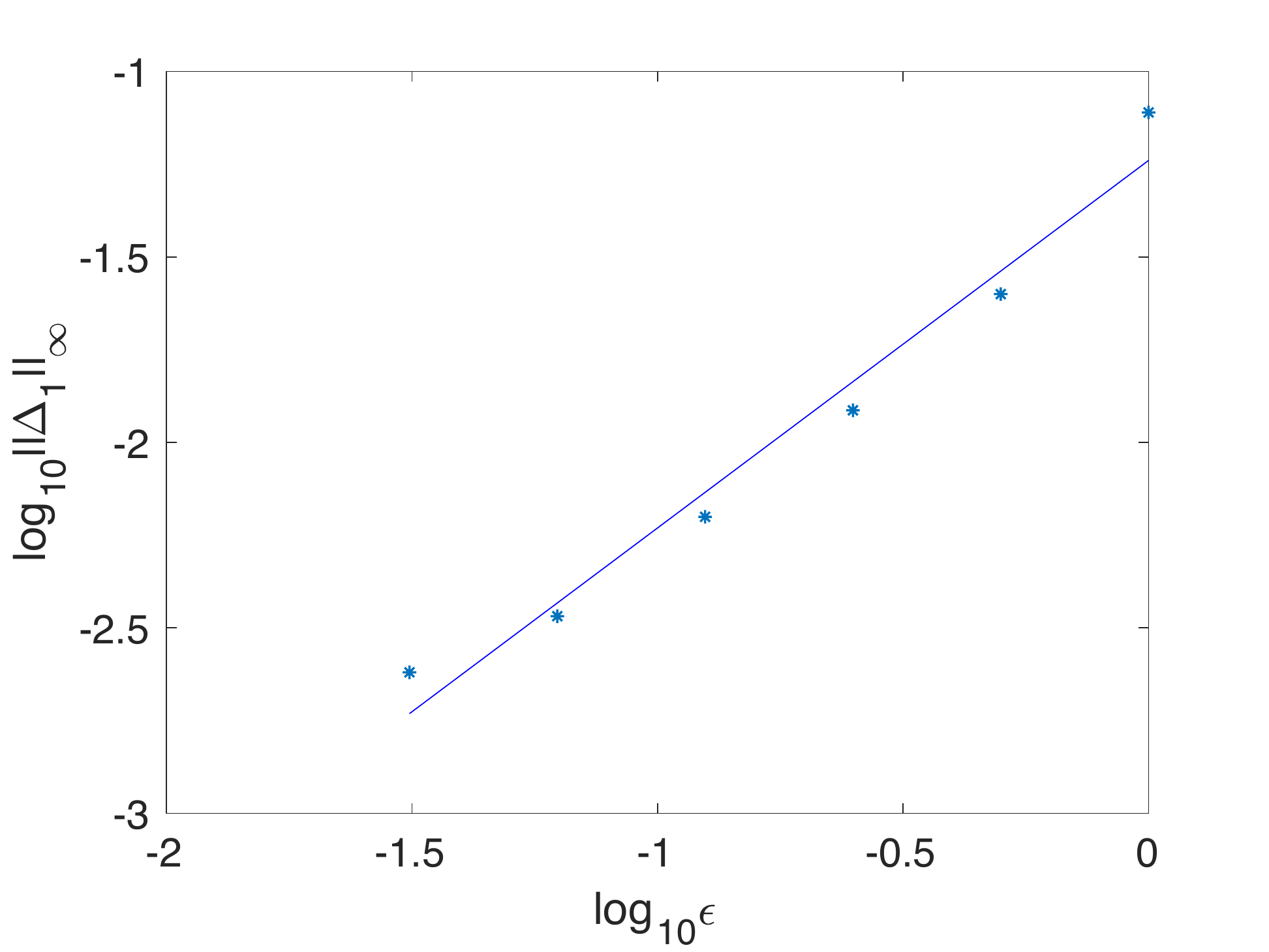}
  \includegraphics[width=0.49\textwidth]{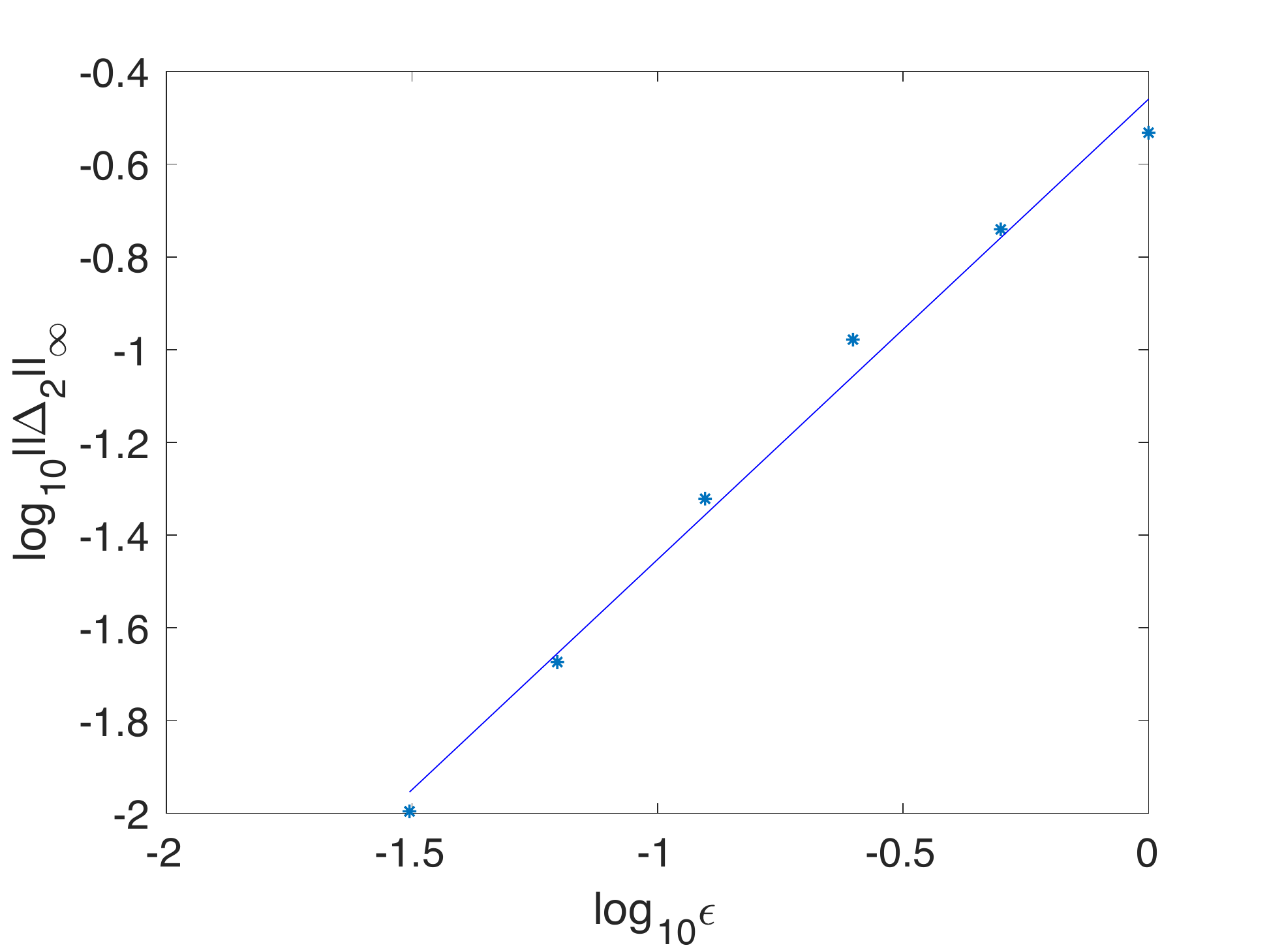}
 \caption{Dependence of $\|\Delta_{1}\|_{\infty}$ (left panel) and $\|\Delta_{2}\|_{\infty}$ (right panel) on $\epsilon$ for $k=1$, together with the result of linear least-squares regression for the logarithms. }
 \label{Deltareg}
\end{figure}

To show that the good agreement between numerics and conjecture is 
not due to a special choice of the spectral parameter $k$, we make similar
plots as shown in Figure~\ref{Deltareg} for two more values of $k$.
The upper and lower rows of Figure~\ref{Deltareg2k} correspond to $k=0.75$ and $k=1.25$ respectively.
(Note that the solution of the 
eikonal problem is expected to become singular for sufficiently small $k$.)
Even though for $\epsilon$ as large as $\epsilon=1$, $\Delta_1$ might not be expected to be small, still the regression line taking the corresponding data into account has the slope $0.97$. For $\Delta_{2}$ the slope of the 
regression line is $0.99$. For $k=1.25$ we find that the slope of the 
line for $\Delta_{1}$ is $0.88$, and for $\Delta_{2}$ it is $1.03$. 
Thus in all cases the results are compatible with the expected $O(\epsilon)$
scaling.  The slopes (exponents) obtained from regression would be expected to be even closer to $1$ if numerical simulations for smaller values of $\epsilon$ were performed; however such experiments are out of reach for the serial computer we used for our simulations. 
\begin{figure}[htb!]
  \includegraphics[width=0.49\textwidth]{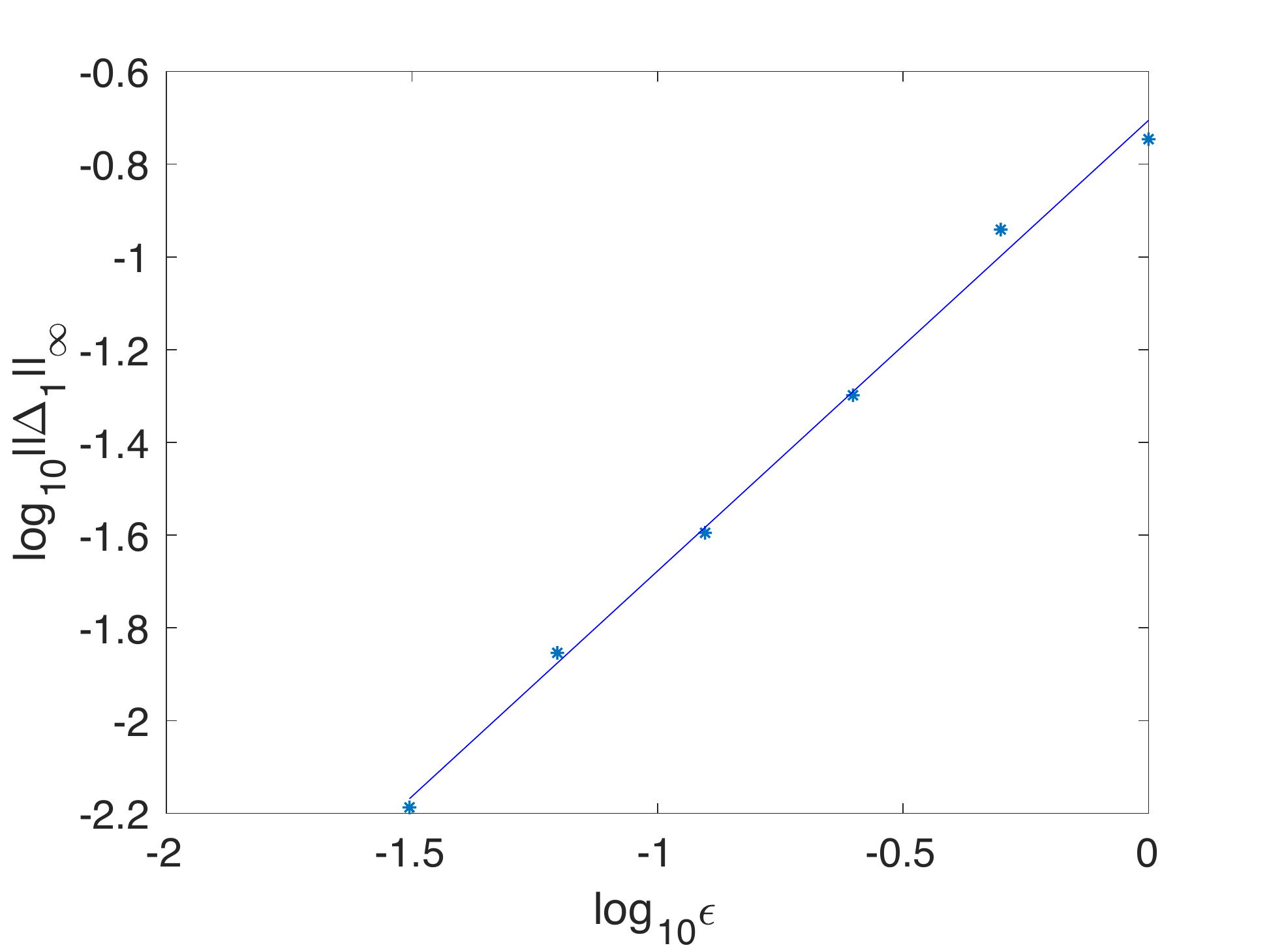}
  \includegraphics[width=0.49\textwidth]{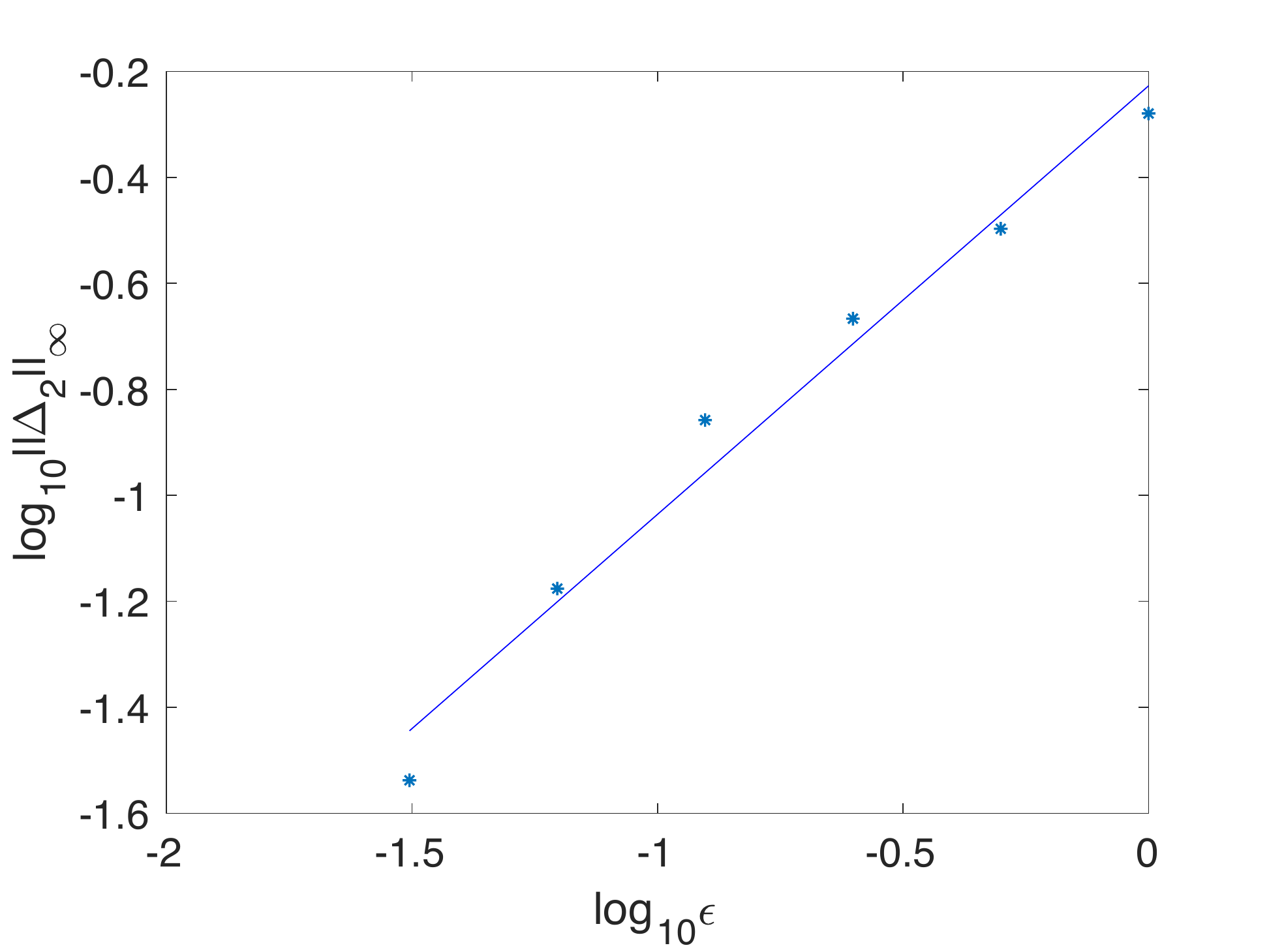}\\
    \includegraphics[width=0.49\textwidth]{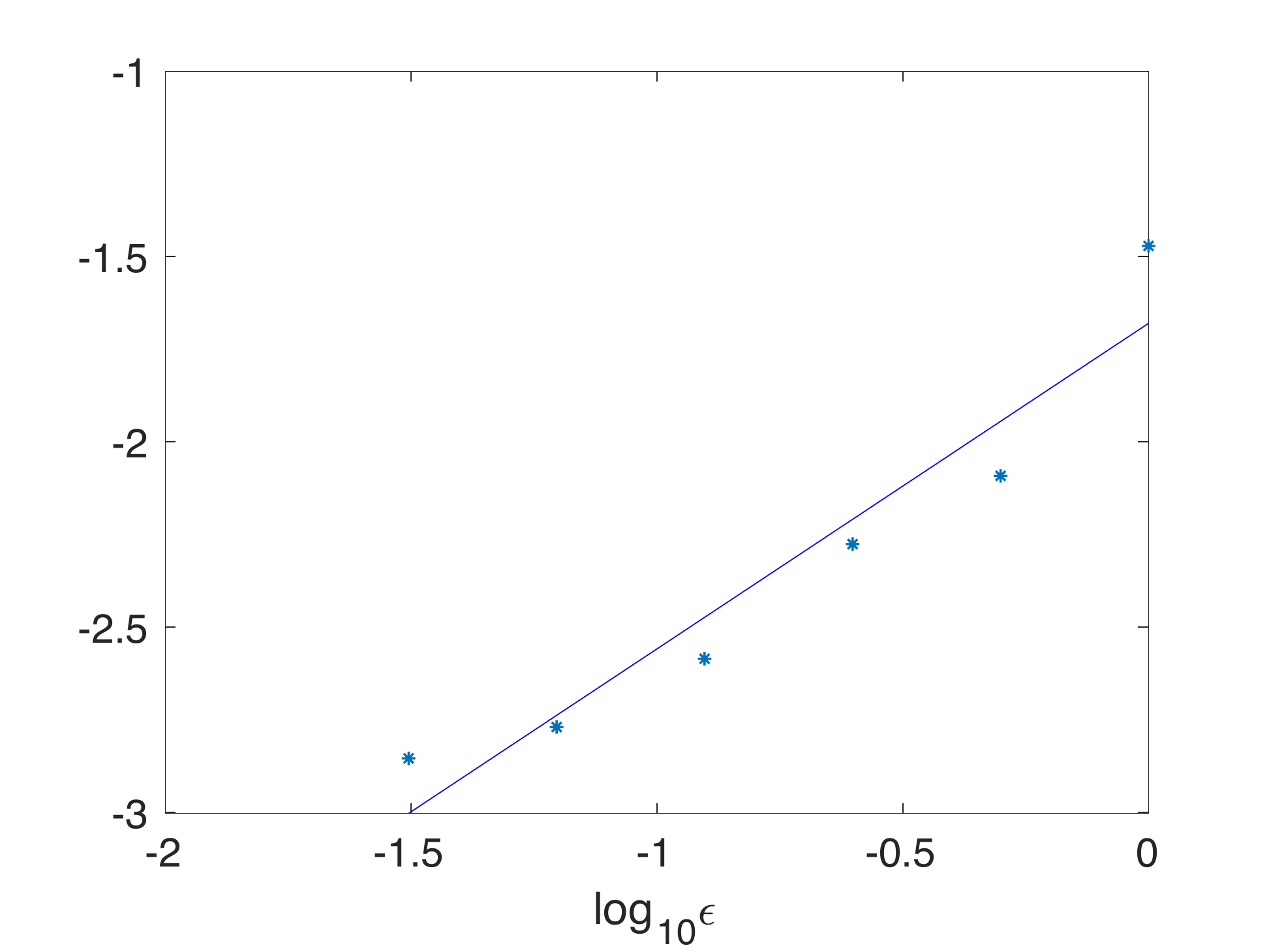}
  \includegraphics[width=0.49\textwidth]{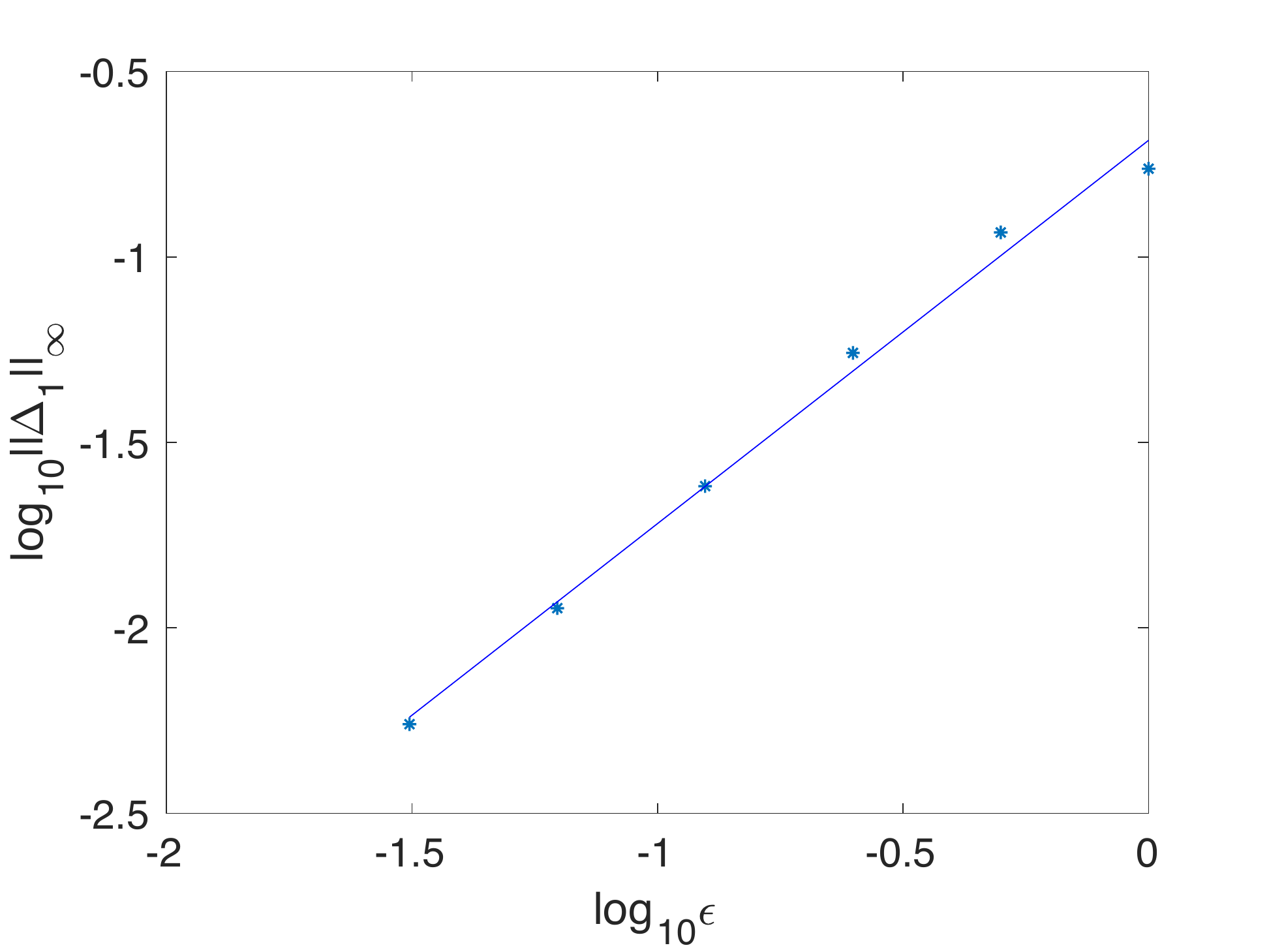}
 \caption{Dependence of $\|\Delta_{1}\|_{\infty}$ (left column) and $\|\Delta_2\|_\infty$ (right column) on $\epsilon$ together with the results of  
 linear least-squares regression for the logarithms; upper row for $k=0.75$, lower 
 row for $k=1.25$. }
 \label{Deltareg2k}
\end{figure}

\begin{myremark} A comparison between the WKB approximation and the numerical solution of the direct scattering problem can be made only if the eikonal problem has a global solution, hence allowing the construction of the WKB approximation globally in the $(x,y)$-plane.  According to Theorem~\ref{theorem:eikonal-solve}, this is guaranteed for $|k|$ sufficiently large.  The lower bound on $|k|$ sufficient to guarantee a global solution is given in \eqref{eq:mod-k-lower-bound}.  In Section~\ref{sec:k-large} it is shown that for phase-free potentials ($S(x,y)\equiv 0$) the lower bound \eqref{eq:mod-k-lower-bound} can be optimized by choice of the constant $B$ to $|k|>\sqrt{\|A^2\|_\mathrm{W}}$, and in Section~\ref{sec:Gaussian} the lower bound is calculated for the Gaussian potential to be $|k|>1$.  
%
However, even the optimized lower bound is only a sufficient condition for the global solvability of the eikonal problem \eqref{eq:eikonal}--\eqref{eq:f-norm}.  Since the hypotheses of Conjecture~\ref{conjecture:WKB} only refer to the existence of a global solution of \eqref{eq:eikonal}--\eqref{eq:f-norm}, we chose in our study to deal with values of $k$ for which the eikonal problem can be solved numerically (which as pointed out above appears to be possible for $|k|$ larger than $\tfrac{1}{2}$), even if those values lie on or within the optimal radius $|k|=1$ for Theorem~\ref{theorem:eikonal-solve} to make a theoretical prediction about the eikonal problem.
\end{myremark}

\subsection{Potential without radial symmetry}
Next we consider the numerical solution of the eikonal problem \eqref{eq:eikonal}--\eqref{eq:f-norm} and coincident construction of the leading-order WKB approximation together with the numerical solution of the $\epsilon$-dependent direct scattering problem \eqref{eq:1.10}--\eqref{eq:psi-norm} in the case of a phase-free potential ($S\equiv 0$) and an amplitude $A(x,y)$ in the class 
of rapidly decaying smooth functions, but now \emph{without radial 
symmetry even asymptotically for large $|z|$}. Concretely, we consider the potential
\begin{equation}
    A(x,y)=\ee^{-(x^{2}+5y^{2}+3xy)}.
    \label{nonsym}
\end{equation}

To solve the Dirac system \eqref{eq:1.10}--\eqref{eq:psi-norm} for the  potential 
\eqref{nonsym}
for various values of $\epsilon$, we once more use the approach of 
Section~\ref{sec:numerical-dbar} with $N_{x}=N_{y}=2^{12}$ Fourier 
modes for $(x,y)\in 4[-\pi,\pi]\times 4[-\pi,\pi]$. The modulus of 
the solutions obtained for $k=1$ and $\epsilon=1/16$ can be seen 
in Figure~\ref{psi12nr}. 
\begin{figure}[htb!]
  \includegraphics[width=0.49\textwidth]{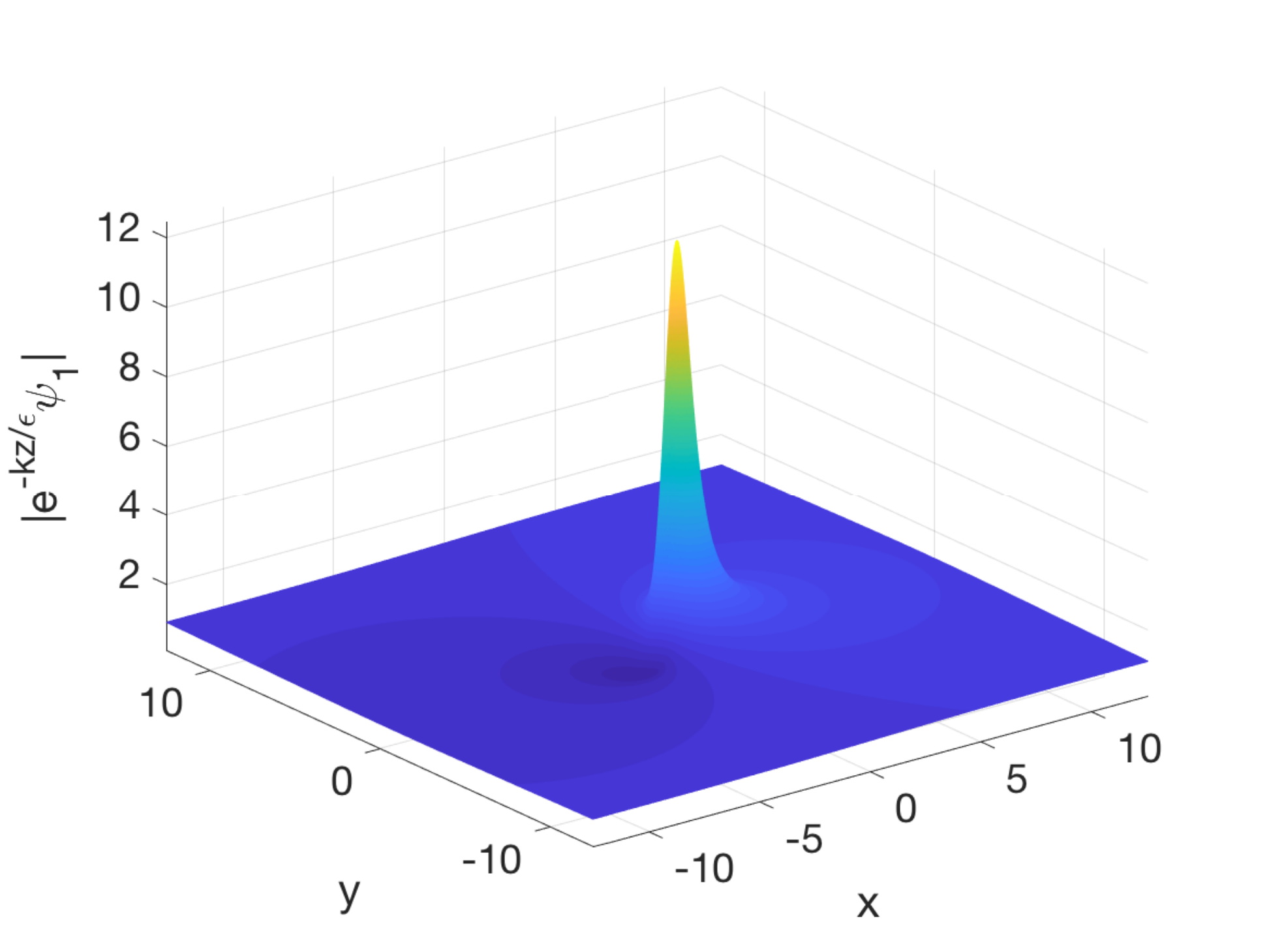}
  \includegraphics[width=0.49\textwidth]{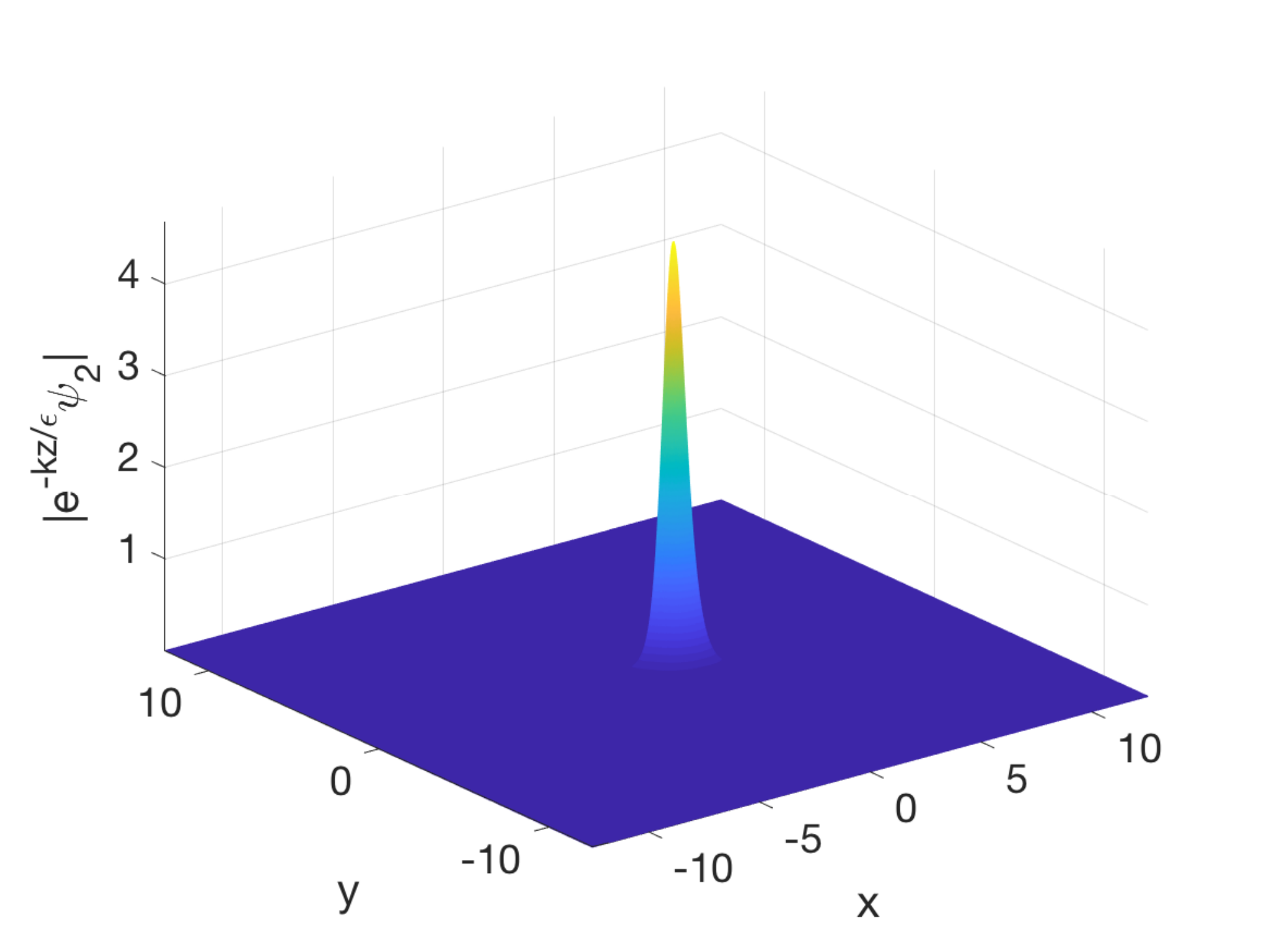}
 \caption{
    Numerical solutions to the Dirac system \eqref{eq:1.10}--\eqref{eq:psi-norm} with  
 potential \eqref{nonsym} for $k=1$ and $\epsilon=1/16$. Left:  the modulus 
 of $\ee^{-kz/\epsilon}\psi_{1}$.  Right: the modulus 
 of $\ee^{-kz/\epsilon}\psi_{2}$.}
 \label{psi12nr}
\end{figure}

Since the potential \eqref{nonsym} is not radially symmetric, the numerical series-based approach described in Section~\ref{sec:numerical-Fourier-series} does not apply, so we must use instead an iterative approach to the eikonal problem as described in Section~\ref{sec:eikonal-iterative}, and it turns out that we will also
need higher resolution in $\phi$ than for radially-symmetric potentials to effectively solve for $g=f-kz$.  We use $N_{r}=64$ Chebychev polynomials and $N_{\phi}=128$ Fourier modes for the 
case $k=1$. The real and imaginary part of the function 
$g(x,y;k)=f(x,y;k)-kz$ are plotted  in the left and right panels of Figure~\ref{gnr} respectively.
\begin{figure}[htb!]
  \includegraphics[width=0.49\textwidth]{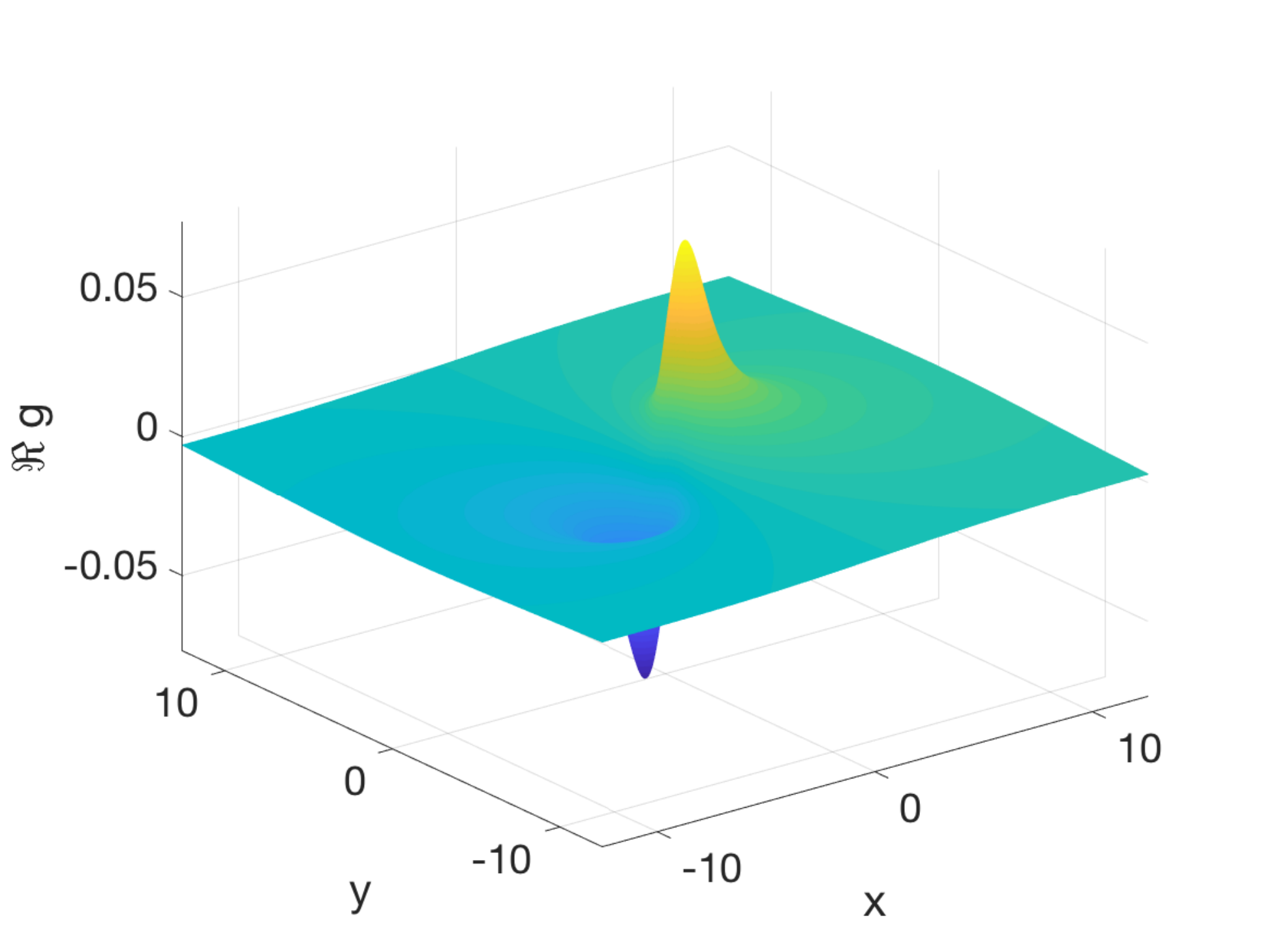}
  \includegraphics[width=0.49\textwidth]{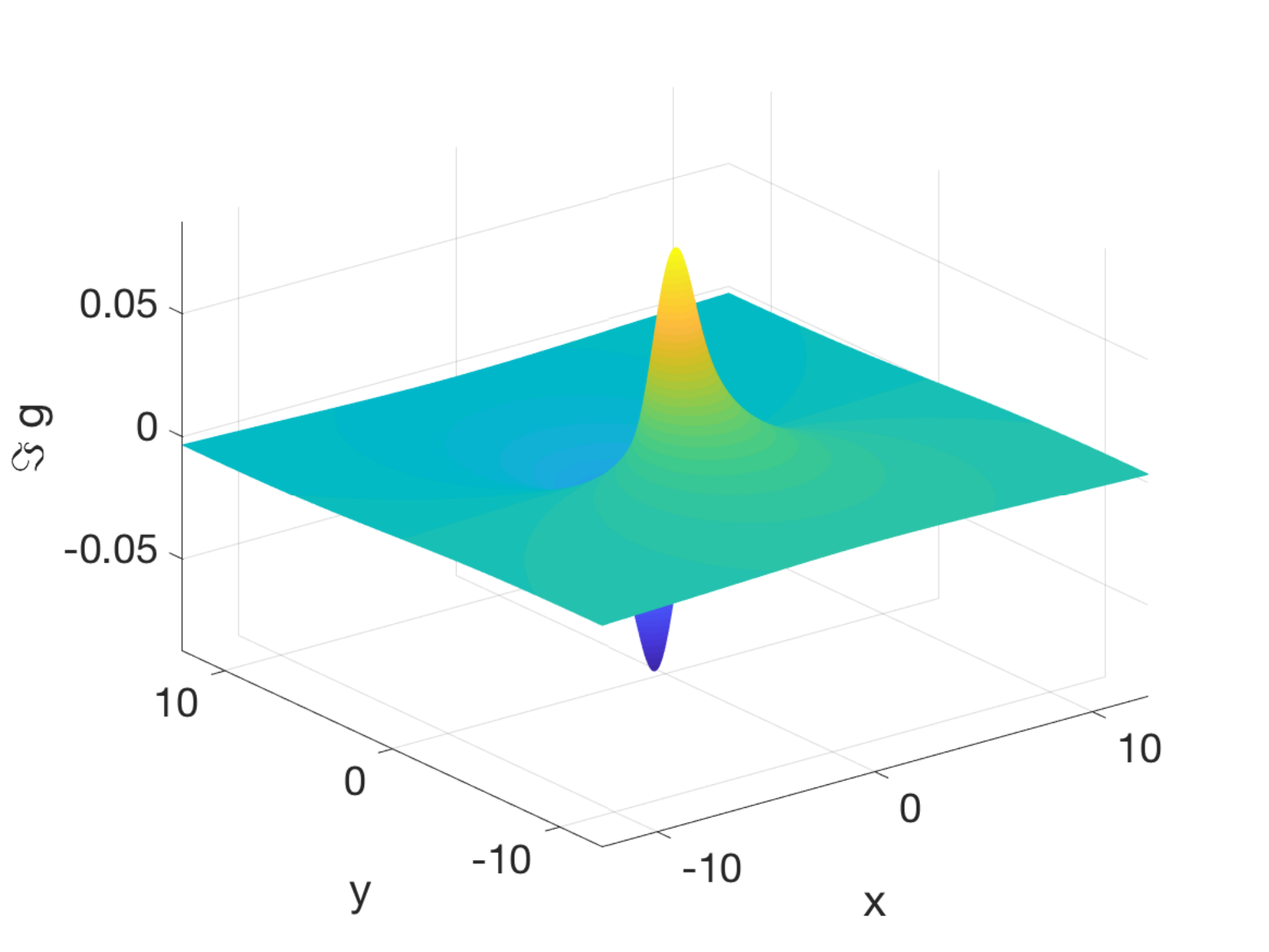}
 \caption{Numerical solution $g=f-kz$ to the eikonal problem \eqref{eq:eikonal}--\eqref{eq:f-norm} with $k=1$ for the 
 potential \eqref{nonsym} with $S\equiv 0$.  Left:  $r\leq1$. Right:
 $r\geq1$.}
 \label{gnr}
\end{figure}
The corresponding spectral coefficients $r\leq1$ and $r\geq1$ are shown in Figure~\ref{gaussasymk1coeff}, indicating that the 
solution is well resolved. 
\begin{figure}[htb!]
  \includegraphics[width=0.49\textwidth]{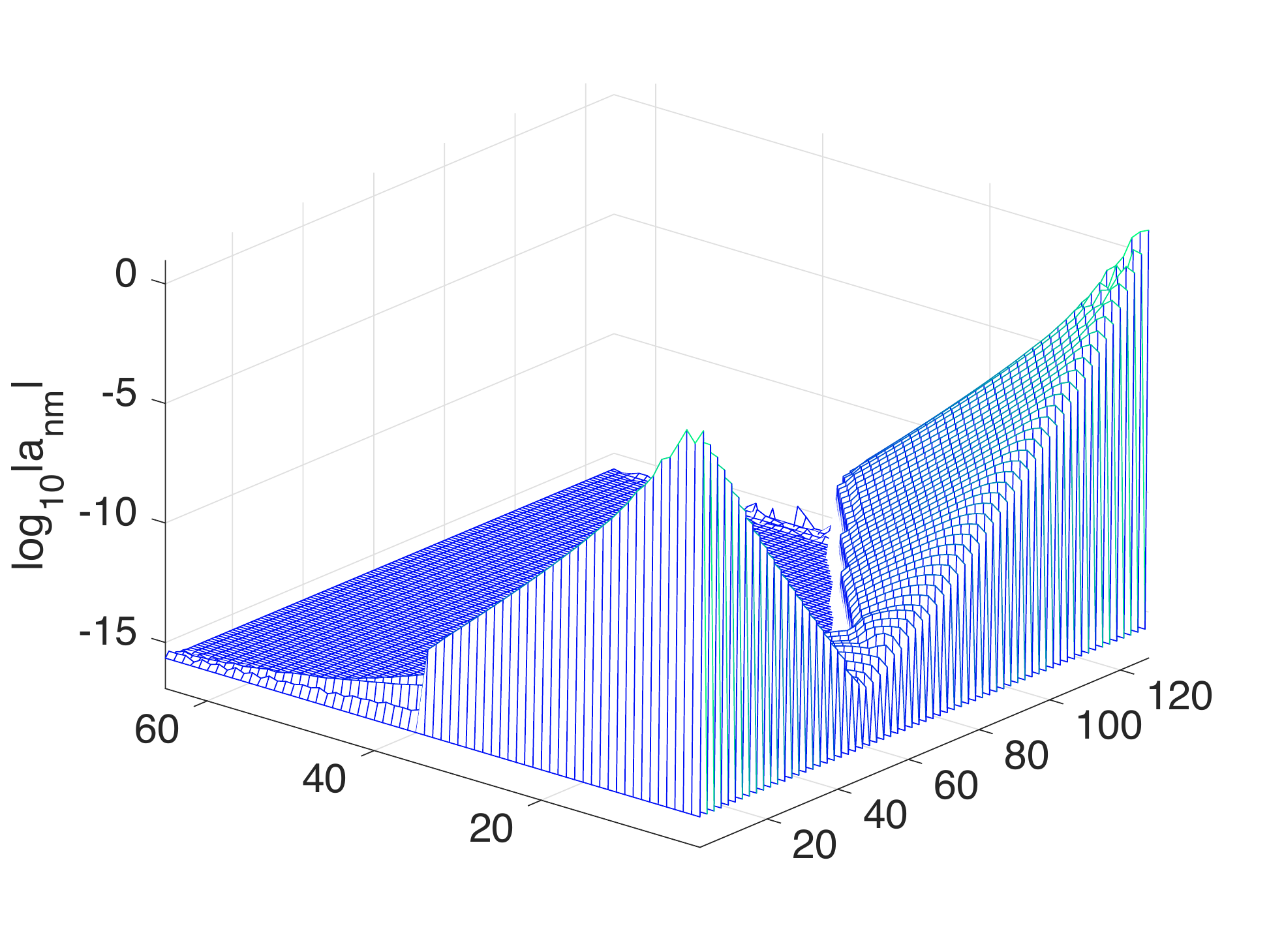}
  \includegraphics[width=0.49\textwidth]{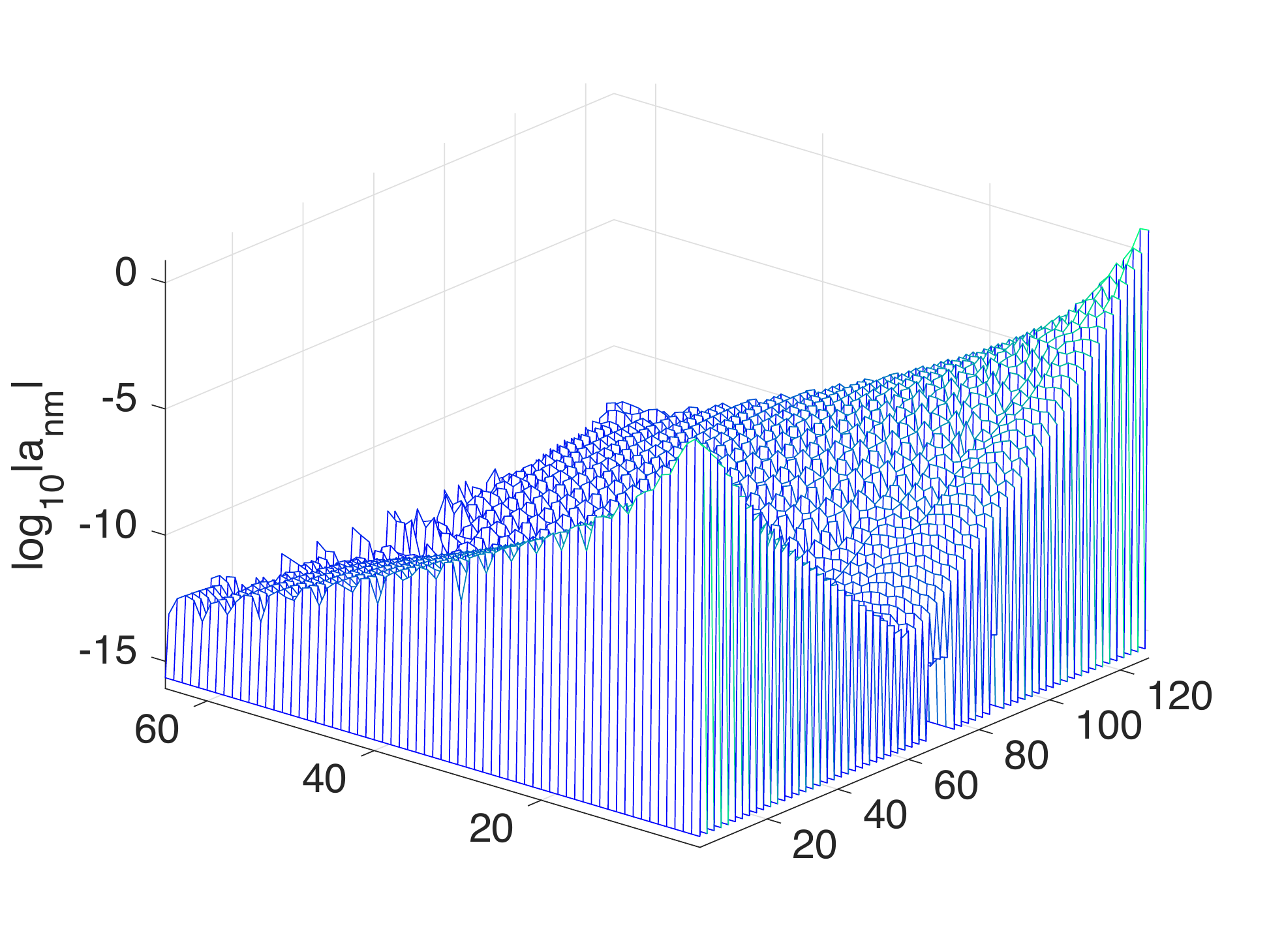}
 \caption{Chebychev and Fourier spectral coefficients plotted over the $(n,m)$-grid for the 
 solution shown in Figure~\ref{gnr}. Left:  the coefficients for $r\leq1$. Right: 
 the coefficients for $r\geq1$.}
 \label{gaussasymk1coeff}
\end{figure}

Next, we numerically solve for the leading-order normalization function $\alpha_0(x,y;k)$ also for $k=1$ as described in Section~\ref{sec:numerical-alpha}.  The real and imaginary parts of the the numerically-computed $\alpha_{0}(x,y;1)$ 
can be seen in Figure~\ref{alphanr}.
\begin{figure}[htb!]
  \includegraphics[width=0.49\textwidth]{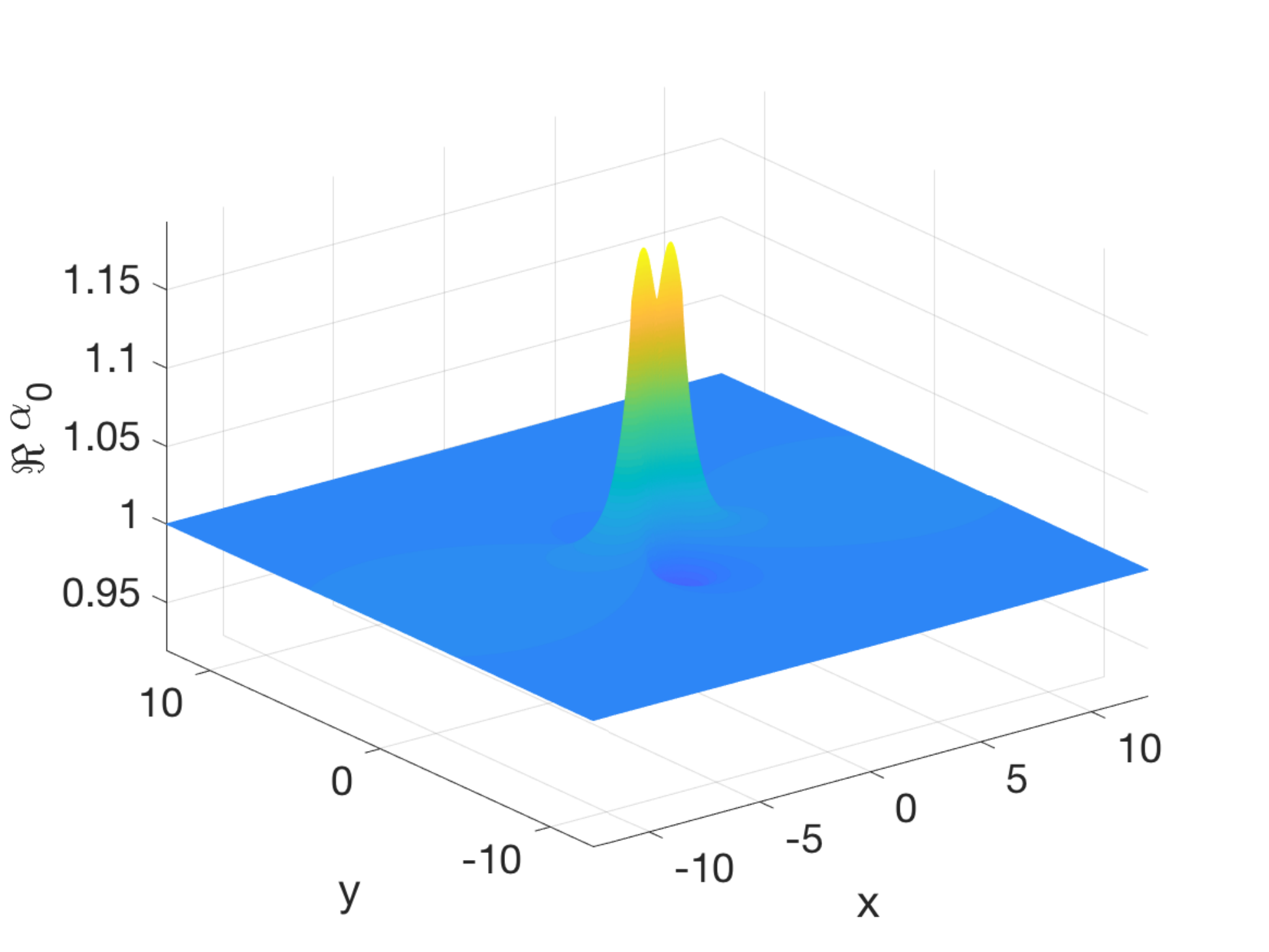}
  \includegraphics[width=0.49\textwidth]{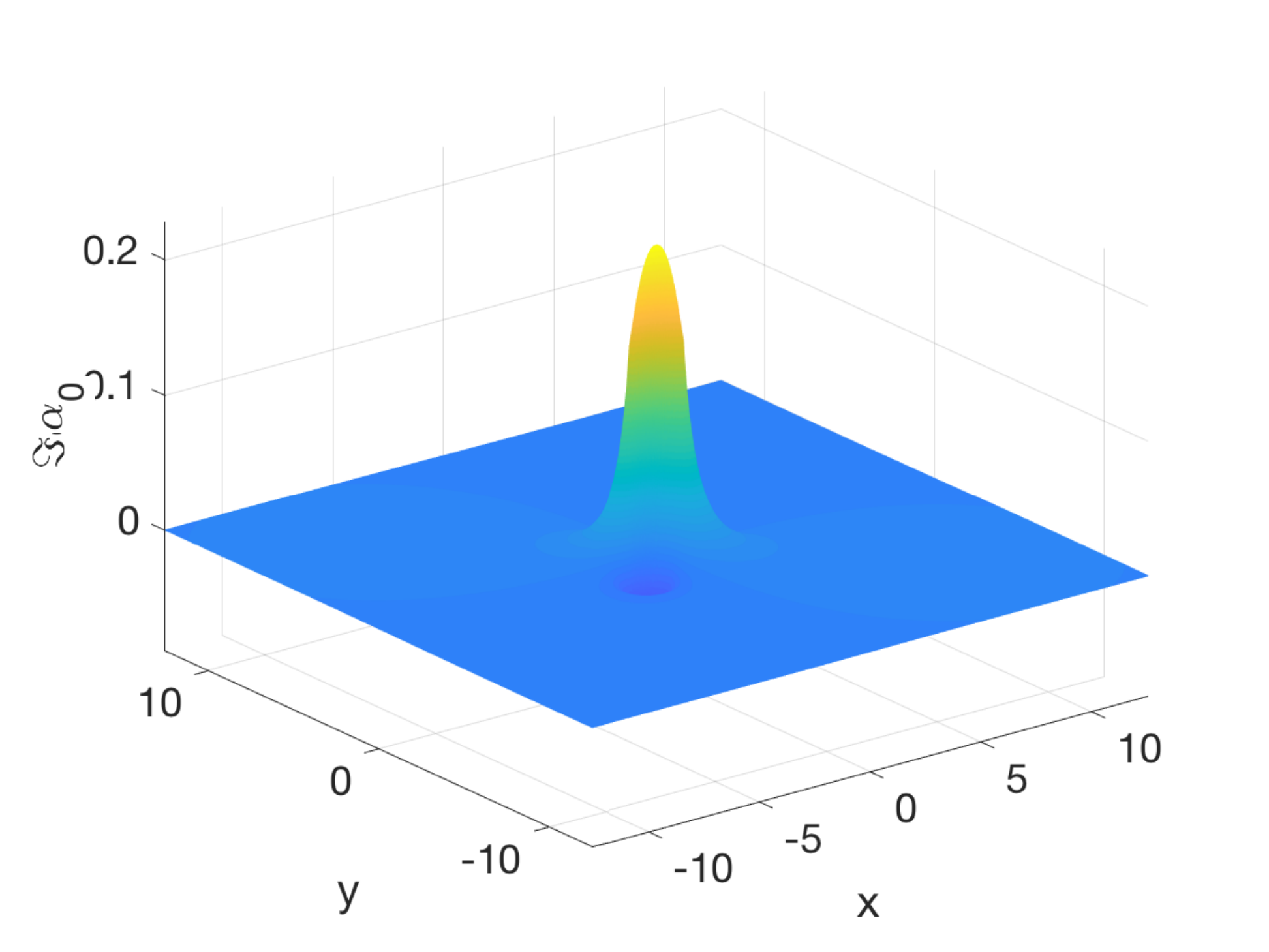}
 \caption{The numerically-computed leading-order normalization function $\alpha_{0}$ for the potential \eqref{nonsym} without radial symmetry at $k=1$. 
Left: real part, right: imaginary part. }
 \label{alphanr}
\end{figure}

Given $f$ and $\alpha_0$, we may again compare the numerical solution to the 
Dirac system \eqref{eq:1.10}--\eqref{eq:psi-norm} to the leading term of the formal WKB 
approximation described in Section~\ref{sec:WKB}.   For the 
potential \eqref{nonsym} at $k=1$ we plot 
$\Delta_{1}$ and $\Delta_2$ defined by \eqref{D1}--\eqref{D2} for four values of $\epsilon$ in Figures~\ref{Delta1nr} and \ref{Delta1nr} respectively.
\begin{figure}[htb!]
  \includegraphics[width=\textwidth]{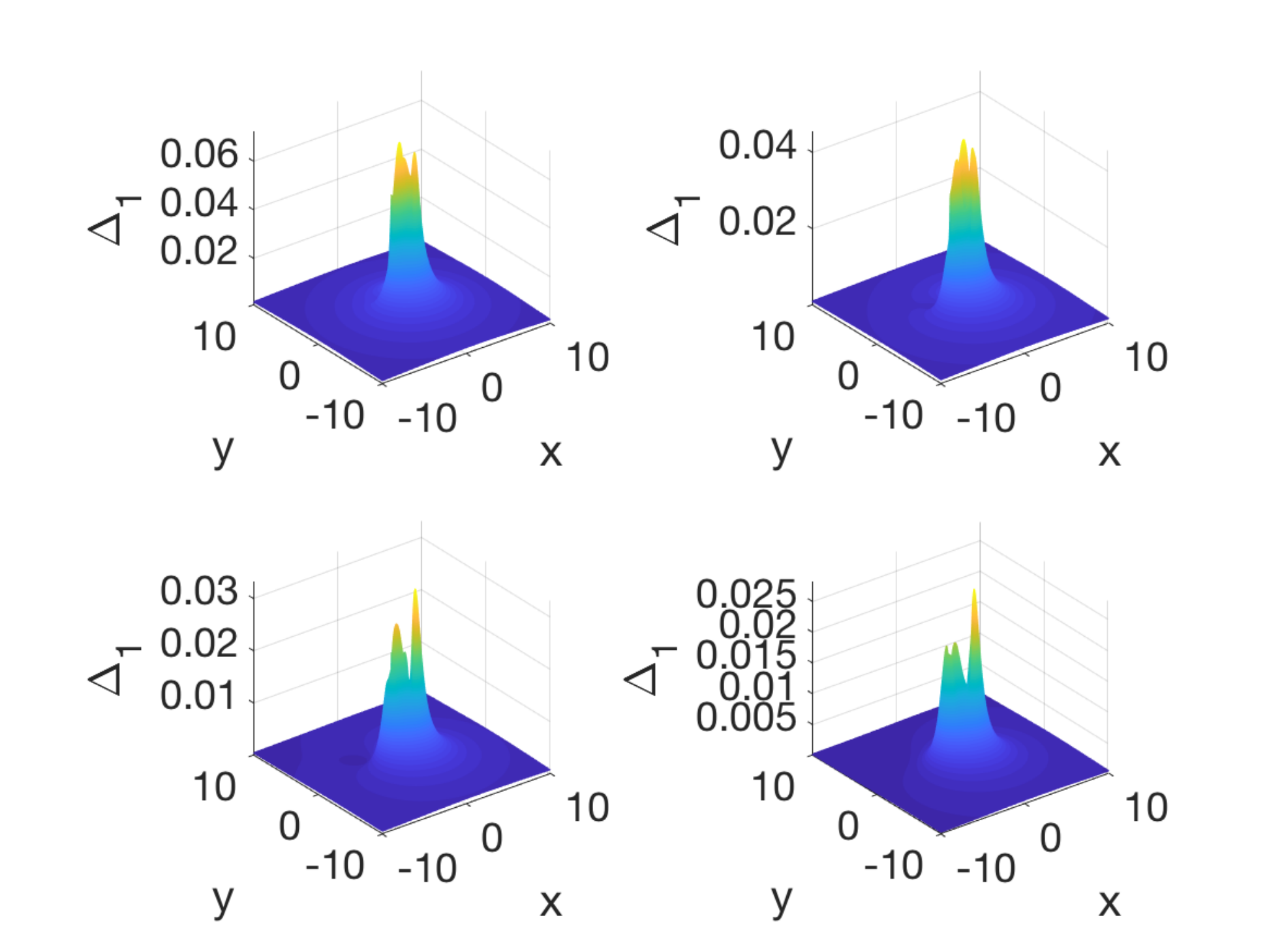}
 \caption{The quantity $\Delta_{1}$ of \eqref{D1} for the 
 potential \eqref{nonsym} with $k=1$ for $\epsilon=1/2,1/4,1/8,1/16$.}
 \label{Delta1nr}
\end{figure}
\begin{figure}[htb!]
  \includegraphics[width=\textwidth]{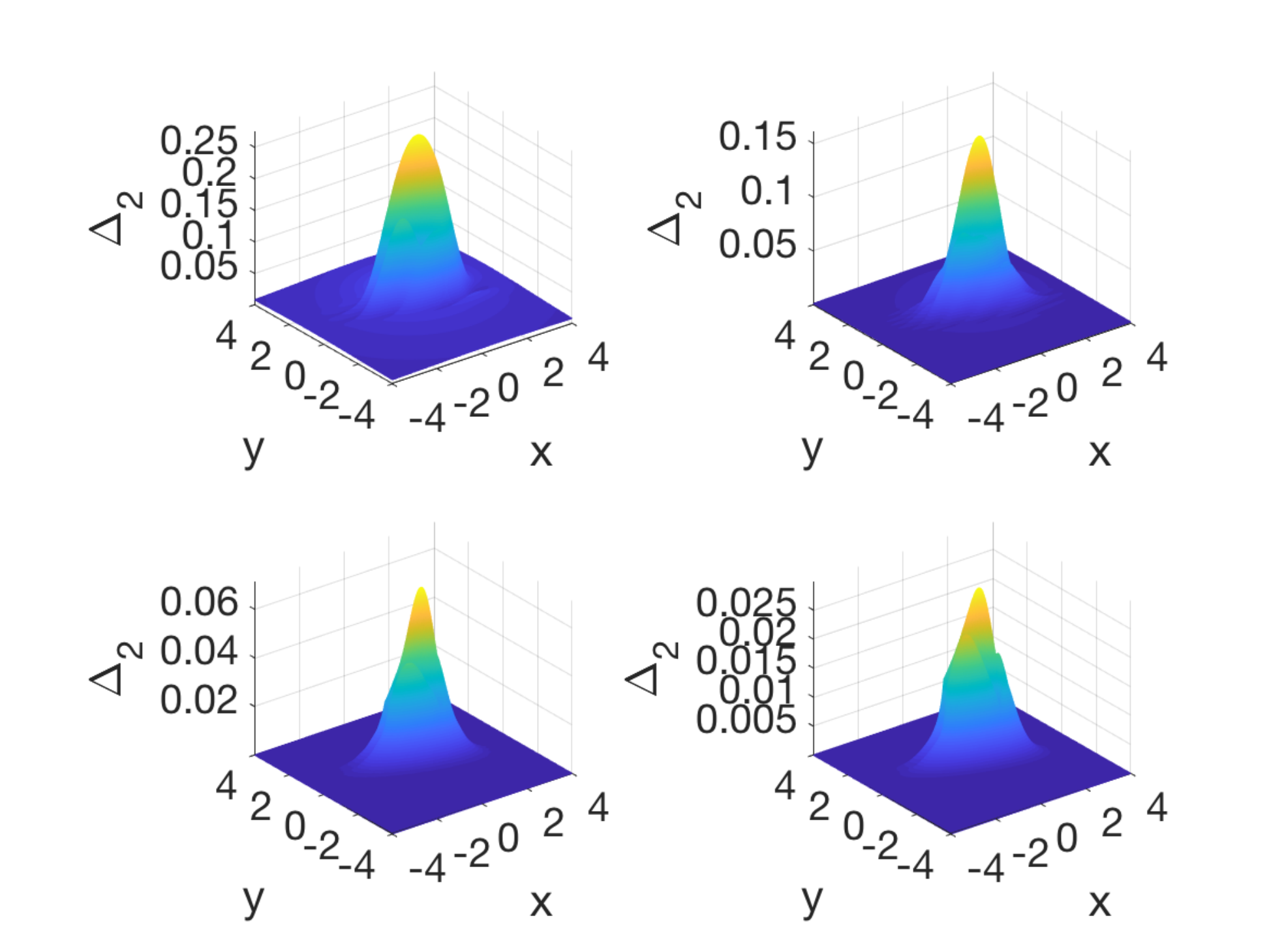}
 \caption{The quantity $\Delta_{2}$ of \eqref{D2} for the
 potential \eqref{nonsym} with $k=1$ for $\epsilon=1/2,1/4,1/8,1/16$. }
 \label{Delta2nr}
\end{figure}

It is clear that the numerical treatment of potentials lacking radial symmetry is 
considerably more challenging than for radially symmetric potentials such as the Gaussian considered in Section~\ref{sec:numerical-Gaussian}. Thus the 
numerical errors for small values of $\epsilon$ are larger, and we would 
need access to parallel computers in order to get the same accuracy 
as in the Gaussian case for a given small $\epsilon$. Nonetheless we computed the quantities 
$\Delta_{1}$ and $\Delta_{2}$ of \eqref{D1} and \eqref{D2} 
respectively for the same values of $\epsilon$ as in the Gaussian case. In 
Figure~\ref{Delta12nr} we plot the $L^\infty(\mathbb{R}^2)$-norms of these quantities for the potential 
\eqref{nonsym} for various values of $\epsilon$ and compare the data in a log-log plot with lines of slope $1$, which would correspond to the $O(\epsilon)$ relative error predicted by Conjecture~\ref{conjecture:WKB}. Obviously the somewhat surprising good agreement for 
values of $\epsilon\approx 1$ observed in the Gaussian case is not present here, 
and for small values of $\epsilon$ the above mentioned resolution 
problems in the solution of the Dirac system \eqref{eq:1.10}--\eqref{eq:psi-norm} are 
visible. Nonetheless compatibility with the conjectured scaling 
proportional to $\epsilon$ can be recognized. 
\begin{figure}[htb!]
  \includegraphics[width=0.49\textwidth]{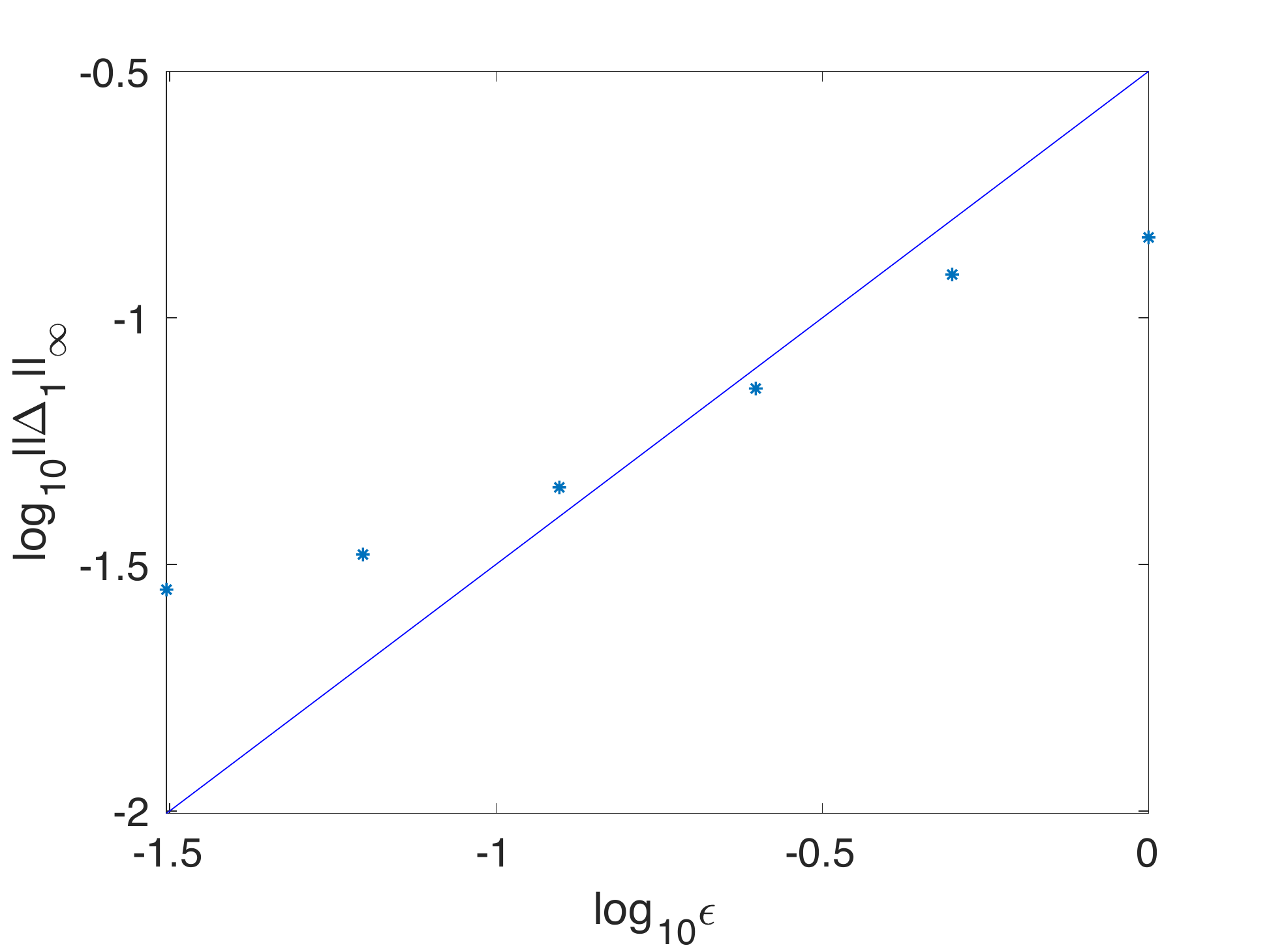}
  \includegraphics[width=0.49\textwidth]{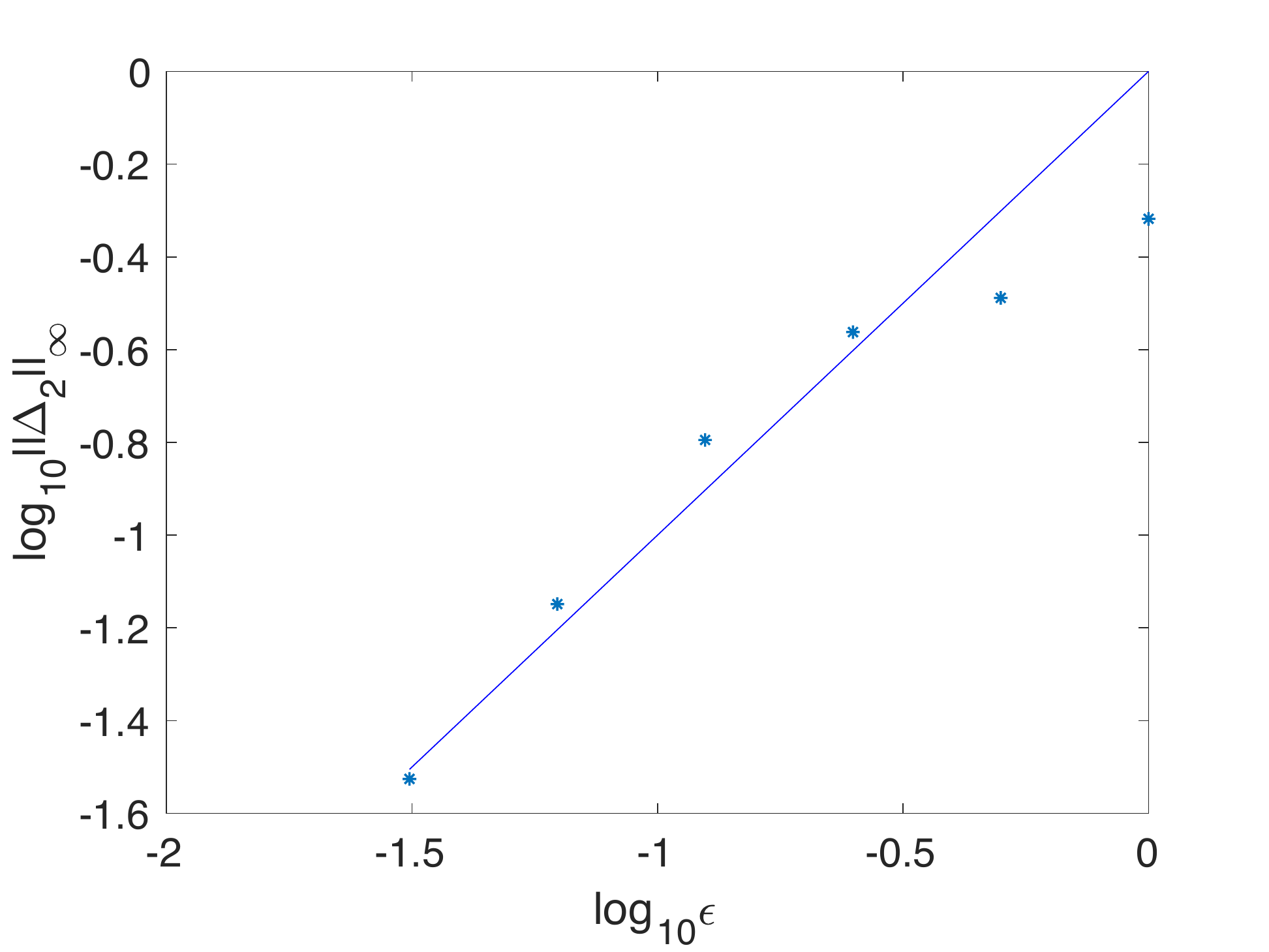}
 \caption{Dependence of $\|\Delta_{1}\|_{\infty}$ (left panel) and 
 $\|\Delta_{2}\|_{\infty}$ (right panel) on $\epsilon$, together 
 with lines with slope $1$. }
 \label{Delta12nr}
\end{figure}
\emph{Thus, our numerical computations also confirm Conjecture~\ref{conjecture:WKB} for certain potentials outside the class of radially-symmetric functions.}  We leave the numerical study of potentials $A(x,y)\ee^{\ii S(x,y)/\epsilon}$ for which $S(x,y)\not\equiv 0$ and the investigation of Conjecture~\ref{conjecture:WKB} in such cases for the future.


\begin{thebibliography}{99}

\bibitem{AblowitzC91} M. J. Ablowitz and P. A. Clarkson, \textit{Solitons, Nonlinear Evolution Equations and Inverse Scattering}, London Mathematical Society Lecture Note Series, vol.\@ 149, Cambridge University Press, Cambridge, UK, 1991.

\bibitem{AblowitzF82} M. J. Ablowitz and A. S. Fokas, ``The inverse scattering problem for multidimensional $2+1$ problems,'' in \textit{Nonlinear Phenomena (Oaxtepec, 1982)}, pp.\@ 137--183, Lecture Notes in Physics \textbf{189}, 1983.

\bibitem{AblowitzF83} M. J. Ablowitz and A. S. Fokas, ``Method of solution for a class of multidimensional nonlinear evolution equations,'' \textit{Phys.\@ Rev.\@ Lett.\@} \textbf{51}, 7--10, 1983.

\bibitem{AblowitzF84} M. J. Ablowitz and A. S. Fokas, ``On the inverse scattering transform of multidimensional nonlinear equations related to first-order systems in the plane,'' \textit{J. Math.\@ Phys.\@} \textbf{25}, 2494--2505, 1984.

\bibitem{AstalaIM09} K.~Astala, T.~Iwaniec, and G.~Martin, \textit{Elliptic Partial Differential Equations and Quasiconformal Mappings in the Plane}, Princeton Mathematical Series, vol.\@ 48, Princeton University Press, Princeton, NJ, 2009.

\bibitem{BealsC85} R. Beals and R. R. Coifman, ``Multidimensional inverse scattering and nonlinear partial differential equations,'' in \textit{Pseudodifferential Operators and Applications (Notre Dame, Ind., 1984)}, pp.\@ 45--70, \textit{Proc.\@ Sympos.\@ Pure Math.\@} \textbf{43}, Amer.\@ Math.\@ Soc., Providence, RI, 1985.

\bibitem{BealsC89} R. Beals and R. R. Coifman, ``Linear spectral problems, nonlinear equations and the $\delbar$ method,'' \textit{Inverse Problems} \textbf{5}, 87--130, 1989.

     \bibitem{birem} M.~Birem and C.~Klein, ``Multidomain spectral method 
     for Schr\"odinger equations,''  \textit{Adv.\@ 
     Comp.\@ Math.\@} \textbf{42}, 395--423, 2016.  (DOI: 10.1007/s10444-015-9429-9) 

    \bibitem{DeiftVZ97}
P. A. Deift, S. Venakides, and X. Zhou, ``New results in small dispersion KdV by an extension of the steepest descent method for Riemann-Hilbert problems,'' \textit{Internat. Math. Res. Notices}, no. 6, 286--299, 1997.
\bibitem{dlmf} NIST Digital Library of Mathematical Functions.  \texttt{http://dlmf.nist.gov/}, Release 1.0.6 of 2013-05-06.  Online companion to \cite{OLBC10}.
\bibitem{DubrovinN89} B. A. Dubrovin and S. P. Novikov, ``Hydrodynamics of weakly deformed soliton lattices.  Differential geometry and Hamiltonian theory,'' \textit{Russ.\@ Math.\@ Surv.\@} \textbf{44}, 35--124, 1989.
\bibitem{DymMcK72}
H. Dym and H. P. McKean, \textit{Fourier Series and Integrals}, Academic Press, New York, 1972.
\bibitem{FerapontovK04}
E. Ferapontov and K. K. Khusnutdinova, ``On the integrability of $(2+1)$-dimensional quasilinear systems,'' \textit{Comm.\@ Math.\@ Phys.\@} \textbf{248}, 187--206, 2004.
\bibitem{Fokas83}
A. S. Fokas, ``Inverse scattering of first-order systems in the plane related to nonlinear multidimensional equations,'' \textit{Phys.\@ Rev.\@ Lett.\@} \textbf{51}, 3--6, 1983.
\bibitem{joerg}J.~Frauendiener, ``Calculating initial data for the conformal Einstein equations
by pseudo-spectral methods,'' \textit{J. Comp.\@ Appl.\@ 
Math.\@} \textbf{109}, 475--491, 1999.
\bibitem{Grenier98}
E. Grenier, ``Semiclassical limit of the nonlinear Schr\"odinger equation in small time,'' \textit{Proc.\@ Amer.\@ Math.\@ Soc.\@} \textbf{126}, 523--530, 1998.
\bibitem{JinLM99} Shan Jin, C. D. Levermore, and D. W. McLaughlin, ``The semiclassical limit of the defocusing NLS hierarchy,'' \textit{Comm.\@ Pure Appl.\@ Math.\@} \textbf{52}, 613--654, 1999.
\bibitem{KleinMcL}C.~Klein and K.~McLaughlin, ``Spectral approach to D-bar 
     problems,''  \textit{Comm.\@ Pure Appl.\@ Math.\@} \textbf{70}, 1052--1083, 2017.  (DOI: 10.1002/cpa.21684)
\bibitem{KleinR14}
C. Klein and K. Roidot, ``Numerical study of the semiclassical limit of the Davey-Stewartson II equations,'' \textit{Nonlinearity} \textbf{27}, 2177--2214, 2014.
\bibitem{Konopelchenko07}
B. G. Konopelchenko, ``Quasiclassical generalized Weierstra\ss\ representation and dispersionless DS equation,'' \textit{J. Phys.\@ A:  Math.\@ Theor.\@} \textbf{40}, F995--F1004, 2007.
\bibitem{tau} C. Lanczos, ``Trigonometric interpolation of empirical 
and analytic functions,'' \textit{J. Math.\@ and Phys.\@} \textbf{17}, 123--199, 1938.


\bibitem{LiebL01} E. H. Lieb and M. Loss, \textit{Analysis}, second edition, Graduate Studies in Mathematics, Volume 14, American Math. Society, Providence, RI, 2001.   
\bibitem{Madelung26} E. Madelung, ``Quantum theory in hydrodynamic form,'' \textit{Zeitschr.\@ Phys.\@} \textbf{40}, 322--326, 1926.
\bibitem{ManakovS06}
S. V. Manakov and P. M. Santini, ``Inverse scattering problem for vector fields and the Cauchy problem for the heavenly equation,'' 
\textit{Phys.\@ Lett.\@ A} \textbf{359}, 613--619, 2006. 
\bibitem{Miller06} P. D. Miller, \textit{Applied Asymptotic Analysis}, Graduate Studies in Mathematics vol.\@ 75, Amer.\@ Math.\@ Soc.\@, Providence, RI, 2006.
\bibitem{Miller16} P. D. Miller, ``On the generation of dispersive shock waves,'' \textit{Physica D} \textbf{333}, 66--83, 2016.
\bibitem{MillerQ15} P. D. Miller and Z.-Y. Qin, ``Initial-boundary value problems for the defocusing nonlinear Schr\"odinger equation in the semiclassical limit,'' \textit{Stud.\@ Appl.\@ Math.\@} \textbf{134}, 276--362, 2015.
\bibitem{NachmanRT17} A. I. Nachman, I. Regev, and D. I. Tataru, ``A nonlinear Plancherel theorem with applications to global well-posedness for the defocusing Davey-Stewartson equation and to the inverse boundary value problem of Calderon,'' \texttt{arXiv:1708.04759}, 2017.
\bibitem{OLBC10} F. W. J. Olver, D. W. Lozier, R. F. Boisvert, and C. W. Clark, editors.  \textit{NIST Handbook of Mathematical Functions}.  Cambridge University Press, New York, NY, 2010.  Print companion to \cite{dlmf}.
\bibitem{Perry12} P. Perry, ``Global well-posedness and long-time asymptotics for the defocussing Davey-Stewartson II equation in $H^{1,1}(\mathbb{C})$.  With an appendix by Michael Christ,'' \textit{J. Spectr. Theory} \textbf{6}, 429--481, 2016. 
\bibitem{GMRES}Y. Saad and M. Schultz, ``GMRES: A generalized minimal residual algorithm for solving nonsymmetric
linear systems,'' \textit{SIAM J. Sci.\@ Comput.\@} \textbf{7}, 856--869, 1986.
\bibitem{SatsumaY74} J. Satsuma and N. Yajima, ``Initial value problem of one-dimensional self-modulation of nonlinear waves in dispersive media,'' \textit{Supp.\@ Prog.\@ Theo.\@ Phys.\@} \textbf{55}, 284--306, 1974.
\bibitem{Sung94} L.-Y.\@ Sung, ``An inverse scattering transform for the Davey-Stewartson II equations,'' \textit{J. Math.\@ Anal.\@ Appl.\@} \textbf{183}, 121--154 (part I), 289--325 (part II), 477--494 (part III), 1994.
\bibitem{TovbisV00} A. Tovbis and S. Venakides, ``The eigenvalue problem for the focusing nonlinear Schr\"odinger equation:  new solvable cases,'' \textit{Physica D} \textbf{146}, 150--164, 2000.
\bibitem{trefethen} L. N. Trefethen, \textit{Spectral Methods in
    Matlab}, SIAM, Philadelphia, PA, 2000.
\bibitem{Whitham65} G. B. Whitham, ``Non-linear dispersive waves,'' \textit{Proc.\@ Roy.\@ Soc.\@ Lond.\@ Ser.\@ A} \textbf{283}, 238--261, 1965.
\bibitem{Mathematica}
Wolfram Research, Inc., \textit{Mathematica}, Version 10.1, Champaign, IL, 2015.
\bibitem{Yi14} G. Yi, \textit{Investigation of Integrable Dynamical Systems and Evolution Partial Differential Equations}, PhD Dissertation, University of Roma ``La Sapienza,'' 2014.
\bibitem{ZakharovS73} V. E. Zakharov and A. B. Shabat, ``Interaction between solitons in a stable medium,'' \textit{Sov.\@ Phys.\@ JETP} \textbf{37}, 823--828, 1973.
\end{thebibliography}
\end{document}